%% file: thesis.tex
\documentclass[a4paper,11pt]{book}
\usepackage{amsmath,amsfonts,amssymb,amsthm}
\usepackage{latexsym}
\usepackage{color}
\usepackage{graphicx}
\usepackage{tikz}
\usepackage{hyperref}
\hypersetup{
    colorlinks=false,
    pdfborder={0 0 0},
}

\usepackage{enumerate}

\usepackage[a4paper,vmargin=2.5cm,inner=2.5cm,outer=3cm]{geometry}
\usepackage{fancyhdr}

\input{macros}

\numberwithin{equation}{chapter}

\title{Characterization Theorem for Local Operators in Factorizing Scattering Models}
\author{Daniela Cadamuro}

\begin{document}

\thispagestyle{empty}

\begin{center}
\vspace*{4cm}

{\Huge\bfseries A~Characterization~Theorem for~Local~Operators
in~Factorizing~Scattering~Models

}

\vspace{2cm}

{\Large
Dissertation \\
zur Erlangung des mathematisch-naturwissenschaftlichen Doktorgrades\\
``Doctor rerum naturalium''\\
der Georg-August-Universit\"at G\"ottingen

}

\vspace{2cm}
{\Large
vorgelegt von}

\vspace{1cm}

{\Large
Daniela Cadamuro\\
aus Torino}

\vspace{1cm}
{\Large
G\"ottingen, 2012}
\end{center}

\clearpage

\thispagestyle{empty}

\vspace*{0.9\textheight}

{\large
\noindent
Referent: Prof.~Karl-Henning Rehren

\noindent
Koreferentin: Prof.~Laura Covi

\noindent
Tag der m\"undlichen Pr\"ufung: 26.~Oktober 2012

\cleardoublepage



\tableofcontents

\listoffigures


\input{introduction}

\input{generaldefinitions}

\input{arakiexpansion}

\input{operators}

\input{theoremstatement}

\input{from_a_to_fp}
\input{from_fp_to_f}

\input{from_f_to_a}

\input{localexamples}

\input{conclusions}

\appendix

\input{warped}

\input{boundaryvalues}

\input{graphs}

\input{acknowledgements}


\addcontentsline{toc}{chapter}{\bibname}

\bibliographystyle{thesis}
\bibliography{../integrable}

\end{document}

%% file: macros.tex
\newtheorem{definition}{Definition}[chapter]
\newtheorem{lemma}[definition]{Lemma}
\newtheorem{proposition}[definition]{Proposition}
\newtheorem{theorem}[definition]{Theorem}
\newtheorem{corollary}[definition]{Corollary}
\newtheorem{conjecture}[definition]{Conjecture}

\definecolor{detailsgray}{gray}{0.3}

\def \cbb{\mathbb{C}}

\def \nbb{\mathbb{N}}
\def \rbb{\mathbb{R}}

\def \zbb{\mathbb{Z}}

\def \acal {\mathcal{A}}

\def \ccal {\mathcal{C}}
\def \dcal {\mathcal{D}}

\def \fcal {\mathcal{F}}
\def \gcal {\mathcal{G}}
\def \hcal {\mathcal{H}}
\def \ical {\mathcal{I}}

\def \kcal {\mathcal{K}}

\def \mcal {\mathcal{M}}
\def \ncal {\mathcal{N}}
\def \ocal {\mathcal{O}}
\def \pcal {\mathcal{P}}
\def \qcal {\mathcal{Q}}
\def \rcal {\mathcal{R}}
\def \scal {\mathcal{S}}

\def \ucal {\mathcal{U}}

\def \wcal {\mathcal{W}}

\def \bfk  {\mathfrak{B}}

\newcommand{\betav}{\boldsymbol{\beta}}
\newcommand{\thetav}{\boldsymbol{\theta}}
\newcommand{\etav}{\boldsymbol{\eta}}
\newcommand{\lambdav}{\boldsymbol{\lambda}}
\newcommand{\nuv}{\boldsymbol{\nu}}
\newcommand{\zetav}{\boldsymbol{\zeta}}
\newcommand{\piv}{\boldsymbol{\pi}}
\newcommand{\xiv}{\boldsymbol{\xi}}

\newcommand{\av}{\boldsymbol{a}}
\newcommand{\bv}{\boldsymbol{b}}
\newcommand{\cv}{\boldsymbol{c}}
\newcommand{\ev}{\boldsymbol{e}}
\newcommand{\nv}{\boldsymbol{n}}
\newcommand{\pv}{\boldsymbol{p}}
\newcommand{\qv}{\boldsymbol{q}}
\newcommand{\xv}{\boldsymbol{x}}

\newcommand{\zv}{\boldsymbol{z}}
\newcommand{\zerov}{\boldsymbol{0}}

\def \. { \,\! }

\def\clap#1{\hbox to 0pt{\hss#1\hss}}

\def\mathclap{\mathpalette\mathclapinternal}

\def\mathclapinternal#1#2{\clap{$\mathsurround=0pt#1{#2}$}}

\def \cdotarg { \, \cdot \, }

\DeclareMathOperator{\im}{Im}
\DeclareMathOperator{\re}{Re}

\newcommand{\id}{\operatorname{id}}

\def \st {^\ast}

\DeclareMathOperator{\supp}{supp}

\def \expltext#1 {\\ \text{\footnotesize{ (#1) }}\\}
\def \intercomm#1 {\\ \text{\footnotesize{ (#1) }}\\}
\def \undercomm#1 {\underset{\text{\scriptsize{ (#1) }}}}
\def \overcomm#1 {\overset{\text{\scriptsize{ (#1) }}}}

\newcommand{\lhs}[1]{\mathrm{l.h.s.}(\ref{#1})}
\newcommand{\rhs}[1]{\mathrm{r.h.s.}\eqref{#1}}

\newcommand{\Hil}{\mathcal{H}}

\newcommand{\fpn}{\Hil^{\mathrm{f}}}
\newcommand{\fpno}{\Hil^{\omega,\mathrm{f}}}
\newcommand{\fpnp}{Q}

\newcommand{\boundedops}{\bfk(\Hil)}

\newcommand{\qf}{\mathcal{Q}}

\newcommand{\A}{\mathcal{A}}
\newcommand{\M}{\mathcal{M}}

\newcommand{\hscalar}[2]{\langle #1 , #2 \rangle }
\newcommand{\bighscalar}[2]{\big\langle  #1 , #2 \big\rangle }

\newcommand{\gnorm}[2]{\lVert #1 \rVert_{#2}}

\newcommand{\Biggnorm}[2]{\Big\lVert #1 \Big\rVert_{#2}}
\newcommand{\onorm}[2]{\lVert #1 \rVert_{#2}^\omega}
\newcommand{\bigonorm}[2]{\big\lVert #1 \big\rVert_{#2}^\omega}

\newcommand{\leftwedge}{\mathcal{W}'}
\newcommand{\rightwedge}{\mathcal{W}}

\newcommand{\ad}{a^\dagger}
\newcommand{\zd}{z^\dagger}

\newcommand{\lvector}[2]{\boldsymbol{\ell}_{#1}(#2)}
\newcommand{\rvector}[2]{\boldsymbol{r}_{#1}(#2)}

\newcommand{\cmeA}[2]{f_{#1}^{\lbrack #2 \rbrack}}

\newcommand{\cme}[2]{\mathchoice{\cmeA{#1}{#2}}{\cmeA{#1}{#2}}{\cmeA{#1}{#2}}{\cmeA{#1}{#2}}}
\newcommand{\cmelong}[2]{f_{#1} \left\lbrack #2 \right\rbrack}

\newcommand{\oa}{\varpi}

\DeclareMathOperator{\strip}{S}

\newcommand{\affiliated}{\mathrel{\eta}}

\newcommand{\tube}{\mathcal{T}}

\newcommand{\perms}[1]{\mathfrak{S}_{#1}}
\DeclareMathOperator{\sign}{sign}

\DeclareMathOperator*{\res}{res}

\DeclareMathOperator*{\domain}{dom}

\DeclareMathOperator{\cch}{cch}
\DeclareMathOperator{\ich}{ich}
\DeclareMathOperator{\ach}{ach}

\DeclareMathOperator{\nodes}{nodes}

%% file: introduction.tex
\chapter{Introduction}

Relativistic quantum field theories are described by the set of local observables, which are linear bounded or unbounded operators associated with regions of Minkowski space. These observables have the physical meaning of ``measurements'' which take place in a finite space and in a finite period of time. In the case of bounded operators, the set of observables forms von Neumann algebras associated with spacetime regions, which in order to gain any physical interpretation, need to fulfil some properties. We list here these properties by paying attention especially to the physical motivation behind them.

The first property states that an algebra $\acal(\ocal_1)$ associated with the region $\ocal_1$ includes all the operators of another algebra $\acal(\ocal_2)$ if $\ocal_2 \subset \ocal_1$. This reflects the fact that measurements performed in a certain region of the spacetime include also all the measurements performed on a smaller region, which is included in the previous one. The second property is called Einstein's causality, which says that no signals can travel faster than the velocity of light. This means that measurements performed in space-like separated regions cannot interfere with each other, and therefore, by Heisenberg's uncertainty relations, the corresponding operators must commute.
The third property concerns with the principle of covariance of the theory. This implies that the algebra of observables must transform covariantly under spacetime symmetry transformation of the region. Mathematically, it means that there must exist a strongly continuous representation of the spacetime symmetry group acting on the algebra.
The stability of the matter requires a positive energy spectrum in all Lorentz frames, and therefore that the joint spectrum of the generators of the spacetime translations is contained in the forward light cone.
Finally, we require the existence of a unique vector in the Hilbert space of the theory which has energy and momentum zero, and represents the vacuum state.

The problem is to construct models of quantum field theories in this setting, by exhibiting algebras of local observables fulfilling all these properties. With the exception of the free field theory, this is in general a difficult task due to the complicate structure that local observables have in the presence of non trivial interaction. There are the results of Glimm and Jaffe \cite{GliJaf:quantum_physics} on the construction of simpler and lower dimensional models with interaction. But in the case of $3+1$ spacetime dimensions this is still nowadays an open problem.

In particular we focus in models in $1+1$ spacetime dimensions with factorizing scattering matrices, and we are interested to study the content of local observables in these theories. Note that for models with one particle species and without inner degrees of freedom, a factorizing scattering matrix is in fact just given by a function in one variable (the \emph{rapidity} $\theta$).
We would look for the existence of these local observables in any mathematical framework: As algebras of bounded operators \cite{Haa:LQP}, or as Wightman fields \cite{StrWig:PCT}, or as closed operators affiliated with the algebras of bounded operators.

One approach to this problem is the so called \emph{form factor programme} \cite{Smirnov:1992,BabujianFoersterKarowski:2006}. Here, one starts from the scattering matrix $S$ as an input, and construct the Wightman $n$-point functions of the theory with the $S$-matrix that we started with. For this, one expands expectation values of local observables in a series of form factors. Here as local observables, we intend pointlike localized quantum fields, and a form factor is the expectation value of this field operator between asymptotic scattering states. However, as expectation values of local observables, the form factors must fulfil a number of properties, and by solving these conditions, they can be computed explicitly. There are explicit examples of form factors in various models, such as the Sinh-Gordon \cite{FringMussardoSimonetti:1993}, the Sine-Gordon models \cite{BabujianFringKarowskiZapletal:1993}, the Ising model \cite{BergKarowskiWeisz:1979}, and many more.

Then, one computes the Wightman $n$-point functions from the form factors by introducing in the vacuum expectation values of the local fields a complete basis in terms of asymptotic scattering states. As a result, the Wightman $n$-point functions are expressed by an infinite series expansion in terms of form factors. We write down here the example of the two-point function:
\begin{equation}\label{FFtwopointfunct}
\langle \Omega, A(x)A(0) \Omega \rangle = \sum_{n=0}^{\infty}\frac{1}{n!}\int d\theta_1 \ldots \int d\theta_n\, e^{-ix\cdot \sum_{k=1}^{n}p(\theta_k)}|\langle \Omega|A(0)|\theta_1,\ldots,\theta_n\rangle_{\text{in}}|^{2},
\end{equation}
where $|\theta_1,\ldots,\theta_n\rangle_{\text{in}}$ are the incoming particle states depending on the rapidities $\theta_j$. By computing all the $n$-point functions using this method, one would be able to construct the local observable as operator-valued distribution (Wightman reconstruction theorem, \cite{StrWig:PCT}).

However, this approach hides a subtle difficulty, that is controlling the convergence of infinite series expansion of the type \eqref{FFtwopointfunct}: In despite of some progress in \cite{BabujianKarowski:2004}, this problem remains still open.

A different approach was due to Schroer \cite{Schroer:1999} who proposed to construct algebras of local observables indirectly in terms of algebras of observables with a weaker notion of localization. He started with a Hilbert space representation of the Zamolodchikov-Faddeev algebra in terms of creation and annihilation operators $z,\zd$ which satisfy a deformed version of the canonical commutation relations, which already involves
the scattering function. Then he constructed field operators, similarly to the free field theory, by taking linear combination of $z,\zd$. These operators can be consistently be interpreted as being localized in unbounded regions, called \emph{wedges}. In particular, we have that fields localized in space-like separated wedges commute. Then one can pass to algebras of bounded operators associated with wedges by taking certain bounded functions of the fields and considering the von Neumann algebra generated by them.

By viewing bounded regions in spacetime, for example double cone regions, as the intersection of left and right wedges, one can correspondingly  obtain the set of local observables associated with the double cone as the intersection of the respective sets of observables associated with the right and left wedges. One can see this on the level of von Neumann algebras, but we will consider it on a more general level, see Sec.~\ref{sec:weaklocality}.

The remaining problem in this approach is to show that this intersection is non-trivial, namely that it does not contain only multiples of the identity operator. Lechner proved this in his Ph.D. thesis for a large class of two dimensional models with factorizing scattering matrices \cite{Lechner:2006,Lechner:2008} using a very abstract argument from the Tomita-Takesaki modular theory, the so called \emph{modular nuclearity condition}. In this way, instead of directly constructing the local operators, one can guarantee the non-triviality of the double cone algebras by giving an abstract condition on the underlying wedge algebras. From a technical level, Lechner proved this condition by analysing the analyticity and boundedness properties of the matrix elements of the wedge local operators.

While Lechner proved that the double cone algebras are non-trivial, we do not know much about the explicit form of these local observables.  This because, while we know explicitly the generators of the wedge algebras, the passage to the von Neumann algebras adds many observables as weak-limit points, about which much less is known. It is these limit points which are contained in the intersection.

Our task is to give more information on the structure of these local observables. For this, we expand the local observables in a series expansion and we analyse the analyticity and boundedness properties of the single terms in the expansion, corresponding to the localization of the observable in a bounded region of spacetime.

To clarify the idea at the basis of our approach, we first consider the situation of the free field theory. Araki proved \cite{Ara:lattice} in the theory of a free scalar real massive field, in a slightly different notation, that for every bounded operator on Fock space there exists a unique expansion in terms of a string of normal-ordered creation and annihilation operators $a,\ad$ of the free field theory, depending on the rapidities $\theta_j,\eta_j$:
\begin{equation}\label{eq:expansionfree}
  A = \sum_{m,n=0}^\infty \int \frac{d\thetav \, d\etav}{m!n!} f_{m,n}(\thetav,\etav) a^{\dagger}(\theta_1)\cdots a^\dagger(\theta_m) a(\eta_1)\cdots a(\eta_n),
\end{equation}
where the coefficients $f_{m,n}$ (generalized functions) are given as vacuum expectation value of a string of nested commutators:
\begin{equation}\label{eq:coefffree}
 f_{m,n}(\thetav,\etav) = \bighscalar{\Omega}{ [a(\theta_m), \ldots [a(\theta_1),[\ldots[A,\ad(\eta_n)]\ldots ,\ad(\eta_1)] \ldots ] \;\Omega  }.
\end{equation}
Note that this expansion holds for any $A$, independently from its localization properties: Whether $A$ is localized in a space-time point, or in a bounded region, or in an unbounded region such as in a wedge, or completely delocalized.

As next step, one looks for analyticity and boundedness properties of the coefficients $f_{m,n}$ corresponding to the localization of $A$ in a bounded region of spacetime. To obtain this, we can express $a,\ad$ in \eqref{eq:coefffree} in terms of the Fourier transforms of time-zero fields and use the fact that the localization of the field in a bounded region of spacetime represents a certain support restriction in position space which corresponds, by Fourier transformation, to certain analyticity and boundedness properties of the coefficients $f_{m,n}$ in momentum space; this ideas is at the basis of the well-known Paley-Wiener theorem \cite[Thm.IX.16]{ReedSimon:1975-2}. So, one finds that if $A$ is localized in a bounded region, then the expansion coefficients are entire analytic and fulfil Paley-Wiener type of bounds.
For more technical details of this proof, this can be seen as a special case of the construction we will work out in Chapter~\ref{sec:localitythm}, in particular see Theorem~\ref{thm:localequiv}.

Schroer and Wiesbrock \cite{SchroerWiesbrock:2000-1} proposed to generalize the expansion \eqref{eq:expansionfree} to $1+1$ dimensional theories of one type of scalar massive particle with factorizing scattering matrices, by replacing $a,\ad$ with the annihilation and creation operators $z,\zd$ satisfying the algebraic relations of the
Zamolodchikov-Faddeev algebra depending on a given scattering function.
\begin{equation}\label{eq:expansionz}
  A = \sum_{m,n=0}^\infty \int \frac{d\thetav \, d\etav}{m!n!} f_{m,n}(\thetav,\etav) z^{\dagger}(\theta_1)\cdots z^\dagger(\theta_m) z(\eta_1)\cdots z(\eta_n).
\end{equation}
Now, one would look again for analyticity properties of the expansion coefficients $f_{m,n}$ and bounds of its analytic continuation, corresponding to the localization of $A$ in bounded region of spacetime.

In the case of observables localized in bounded regions of spacetime, Schroer and Wiesbrock \cite{SchroerWiesbrock:2000-1} expected the following scenario: The coefficients $f_{m,n}$ of an observable $A$ localized in a double cone are boundary values of meromorphic functions on the entire rapidity multi-variables complex plane with specific growth behaviour in certain real direction in the complex plane and following a certain pole structure with residue given by an infinite system of recursion relation for the expansion coefficients.

Our programme aims to make these expectations more precise in the class of $1+1$ dimensional models with factorizing scattering matrices studied by Lechner \cite{Lechner:2008}, where however our class of scattering function do not need to fulfil certain regularity conditions imposed by Lechner and our observables are not restricted only to the class of bounded operators. This programme is developed in several steps, that we are going to explain in the following.

First, we will prove that for every quadratic form (and therefore for bounded and unbounded operators, as well) $A$ there exists a unique expansion \eqref{eq:expansionz}. We will provide an explicit expression (see Eq.~\eqref{eq:fmndef}) for the expansion coefficients $f_{m,n}$  in terms of matrix elements of $A$, involving the scattering function $S$. It is not obvious how to relate this expression to a formula similar to \eqref{eq:coefffree}. For this purpose, we will introduce the notion of \emph{warped convolution} used in deformation methods for the construction of quantum field theories by several authors \cite{GrosseLechner:2007,Lechner:2011},\cite{BuchholzSummers:2008,BuchholzSummersLechner:2011}.
Here, Buchholz, Summers and Lechner made use of the warped convolution integral to deform wedge-local observables of any theory in order to construct interacting models in arbitrary spacetime dimensions; in $1+1$ dimensions, this yields models with a factorizing scattering matrix. We will use this notion to define a ``deformed commutator'' that depends on the scattering function and fulfils a certain ``deformed'' version of the standard properties of a commutator. Then, by replacing in \eqref{eq:coefffree} $a,\ad$ with the ``deformed'' annihilators and creators $z,\zd$ and the commutator with the deformed commutator, one obtains a generalization of \eqref{eq:coefffree} to the class of factorizing scattering models described by \cite{GrosseLechner:2007}.

Note that the expansion \eqref{eq:expansionz} is similar to the form factor expansion, but it is not identical to it. In particular, the basis of our expansion is in the operators $z,\zd$, rather than in the asymptotic free creators and annihilators $a_{\text{in}}, \ad_{\text{in}}$. Heuristically, in the basis of $z,\zd$ one may expect that it is easier to control the convergence of the infinite series in \eqref{eq:expansionz} for local operators, since these $z,\zd$ are related to the notion of wedge locality. In fact, we will discuss this convergence in an example in Chapter~\ref{sec:localexamples}. This would not be possible in the basis in terms of $a_{\text{in}}, \ad_{\text{in}}$, since these operators are completely unrelated with local objects. Another advantage of our construction is that it applies to the model of Lechner, which is fully constructed, while, as far we know, there are no completely constructed models in the form factor programme.

We will discuss the properties of the expansion coefficients that are independent of the localization of $A$. In particular, we will study how the coefficients $f_{m,n}$ behaves under spacetime symmetry transformations of $A$ (see Chapter~\ref{sec:araki}), such as the spacetime translations. Of particular interest to us is the behaviour of $f_{m,n}$ under spacetime reflections, since it encodes the interaction of the model and it will play an important role in the analysis of observables localized in bounded regions, as we will see in Chapter~\ref{sec:ftoa}.

As next step, we will deal with the problem of convergence of the infinite series expansion in \eqref{eq:expansionz}. Note that since this expansion is expressed in terms of the unbounded operators $z,\zd$, it is more natural that it describes unbounded objects, rather than bounded operators. As a consequence, we have established this expansion on the level of quadratic forms. As a quadratic form, $A$ can be unbounded at high energies and high particle numbers, however, for our characterization of the local observables, we are considering quadratic forms of a specific ``regularity'' class, where this singular behaviour is in a certain way ``controlled''. We are thinking here to some kind of generalized $H$ bounds of the type introduced by Jaffe \cite{Jaffe:1967}, see Sec.~\ref{sec:qforms}. Extra conditions on the summability of certain $\omega$-norms of $f_{m,n}$ (see Sec.~\ref{sec:zgen} for definitions) would imply an extension of the quadratic form to a closed, possibly unbounded, operator (see Sec.~\ref{sec:Amnbounds}, in particular Prop.~\ref{proposition:summable1}).

For our characterization of the local observables, we will need a notion of locality that is therefore adapted to the level of quadratic forms, called \emph{$\omega$-locality}, see Sec.~\ref{sec:weaklocality}. This kind of locality is ``weaker'' than the usual notion of locality, however, we will show that a quadratic form that is $\omega$-local and moreover can be extended to a closed operator, is \emph{affiliated} with the local algebras of bounded operators, see Prop.~\ref{proposition:locality}.

In the third step of our programme, we want to identify the necessary and sufficient conditions on the expansion coefficients $f_{m,n}$ in \eqref{eq:expansionz} that make $A$ local in a bounded region.

In the case of operators localized in wedges, we can refer partially to the results of Lechner \cite{Lechner:2008}. However, we recall that our context is less restrictive than Lechner's setting, since we do not assume that our observable $A$ is necessarily a bounded operator and we do not need certain regularity conditions on the scattering function, used by Lechner. On the level of quadratic forms, we find that due to the localization of $A$ in a wedge, the coefficients $\cme{m,n}{A}$ are boundary values of a common analytic function, i.e. $\cme{m,n}{A}(\thetav,\etav)=F_{m+n}(\thetav,\etav+i\piv)$, where $F_k$ are analytic in the area $0 < \im\zeta_1 < \ldots < \im \zeta_k < \pi$.

In the case of quadratic forms localized in double cones, we will find that the localization of $A$ in the shifted right and left wedges (which identify the double cone) implies, via a rather geometrical construction, involving \emph{graphs} and tube domains on the rapidity multi-variables complex plane, the meromorphic continuation of the functions $F_k$ to the entire rapidity multi-variables complex plane. We will show that these functions $F_k$ fulfil an infinite system of recursion relations, and have a rich pole structure due to these recursion relations and the poles of the $S$-matrix. We compute explicitly the expression of the residua at the poles, given by the recursion relations; we note that these residua vanish in the case $S=1$, corresponding to the free field theory, and the functions $F_k$ become entire analytic. In the case $S=-1$, the same situation holds for $k$ even, namely when the operator creates even number of particles from the vacuum. Further, we will find that these functions $F_k$ fulfil certain properties of symmetry and periodicity, which depend on the scattering function $S$. We will compute certain pointwise bounds of these functions along specific lines on certain graphs in the rapidity multi-variables complex plane, and also specific $L^2$-like bounds on certain nodes of these graphs.

For clarity, we will consider conditions on three levels: We will establish conditions on the quadratic forms $A$; conditions on analytic functions $f_{m,n}$ when the imaginary part of the argument is restricted to a certain graph $\gcal$ (formal definition will be given in Sec.~\ref{sec:fkppositive}); conditions on meromorphic/analytic functions $F_k$. Then we will show that these conditions are equivalent, yielding a \emph{theorem of characterization for $\omega$-local quadratic forms} in bounded spacetime regions (see Chapter~\ref{sec:localitythm}).

As already mentioned, we can show that a set of functions $F_k$ which fulfil the conditions Def.~\ref{def:conditionF} for the characterization of a $\omega$-local quadratic form, and the condition in Prop.~\ref{proposition:summable1} for the extension of the quadratic form to a closed operator, defines, using the expansion \eqref{eq:expansionz}, a closed, possibly unbounded, operator affiliated with the local algebras of bounded operators (see Prop.~\ref{proposition:locality}).

As next step in our programme, we will use the sufficient conditions established before to construct explicit examples of local operators. We will present two examples in the case $S=-1$. In one example we admit only a finite number of coefficient functions $F_k$ for even $k$; the other example contains an infinite family of coefficient functions for odd $k$. In particular, we will show\footnote{up to the rigorous verification of a certain numerical estimate, see Conjecture~\ref{conj:Tmbounds}, which is however very plausible} in the second example, where the infinite sum can possibly diverge, that our condition for the extension of a $\omega$-local quadratic form to a closed operator affiliated with the local algebras is fulfilled.

Finally, we will propose in Chapter~\ref{sec:localexamples} an approach for finding examples of local operators in the case of a general scattering function, without having verified all the conditions discussed in Chapters~\ref{sec:localitythm} and \ref{sec:operators}.
However, our approach is a natural generalization of the construction of examples for $S=-1$ studied in Chapter~\ref{sec:localexamples}. Completing the general construction would be an important achievement of our programme, since the explicit form of local observables in the presence of highly non-trivial interaction has been for long time an open problem.

This thesis is organized as follows: We introduce in Chapter~\ref{sec:preliminaries} the general mathematical framework, partially similar to \cite{Lechner:2008}. In Chapter~\ref{sec:araki} we will prove existence and uniqueness of the expansion \eqref{eq:expansionz} for any quadratic form $A$; moreover, we analyse the properties (independent of locality) of this expansion, and in particular its behavior under spacetime symmetries. In Chapter~\ref{sec:operators} we identify the conditions on $f_{m,n}$, so that $A$ is a closable operator affiliated with the local algebra. In Chapters~\ref{sec:localitythm}, \ref{sec:atofp}, \ref{sec:fptof} and \ref{sec:ftoa}, we formulate and prove a \emph{theorem of characterization for $\omega$-local quadratic forms}, which gives the necessary and sufficient conditions on the coefficients $f_{m,n}$ that make $A$ $\omega$-local in a bounded region. Using the conditions of Chapter~\ref{sec:operators} and Chapter~\ref{sec:localitythm}, we will construct explicit examples of local observables in the case $S=-1$ in Chapter~\ref{sec:localexamples}. In the final appendix, we will discuss in particular a generalization of the formula of the string of nested commutators \eqref{eq:coefffree} to a certain class of factorizing scattering models described by \cite{GrosseLechner:2007}, and its relation with the notion of warped convolution integral introduced by \cite{BuchholzSummers:2008}, see Appendix~\ref{sec:warped}. Finally, we will discuss conclusions and outlook in Chapter~\ref{sec:conclusion}. Chapters~\ref{sec:preliminaries}, \ref{sec:araki} and Appendix~\ref{sec:warped} are material of one of the joint papers with H. Bostelmann \cite{BostelmannCadamuro:expansion-wip}. We will deal with the characterization of locality in another paper \cite{BostelmannCadamuro:characterization-wip} and with the concrete examples and Chapter~\ref{sec:operators} in \cite{BostelmannCadamuro:examples-wip}.

%% file: generaldefinitions.tex
\chapter{General definitions}\label{sec:preliminaries}

\section{Minkowski space}

In the present thesis, the spacetime is given by the 1+1 dimensional Minkowski space $\rbb^2$ with vectors $x=(x_0,x_1)$ and scalar product $x \cdot y = x_0 y_0-x_1 y_1$. The symmetry group of Minkowski space is the Poincar\'e group $\mathcal{P}$ which includes the translations in time and space $x\mapsto x+c,\; c\in \rbb^2$, the space reflection $x\mapsto (x_0,-x_1)$, the time reflection $x\mapsto (-x_0,x_1)$ and the Lorentz boosts:
\begin{equation}
x\mapsto \left( \begin{array}{ccc} \cosh\lambda & \sinh\lambda \\ \sinh\lambda & \cosh \lambda  \end{array} \right)x, \quad \lambda \in \rbb.
\end{equation}
We denote with $\mathcal{P}_{+}$ the subgroup of the Poincar\'e group consisting of the translations, boosts and the total space-time reflection $x \rightarrow -x$.

We are in particular interested in wedge-shaped regions of spacetime called \emph{wedges}. We have the standard \emph{right wedge} $\rightwedge$ with edge at the origin, which is the set
\begin{equation}
\rightwedge:=\{ x\in \mathbb{R}^{2}: x_{1}>|x_{0}| \};
\end{equation}
cf.~Fig.~\ref{fig:rightwedge}, and the standard \emph{left wedge} $\leftwedge$, which is defined as the causal complement of $\rightwedge$. We also consider the translates of the standard right and left wedges, $\wcal_x := \rightwedge + x$  and $\wcal_y' = \leftwedge + y = (\wcal_y)'$ with $x,y \in \rbb^2$.

\begin{figure}
\begin{center}
\input{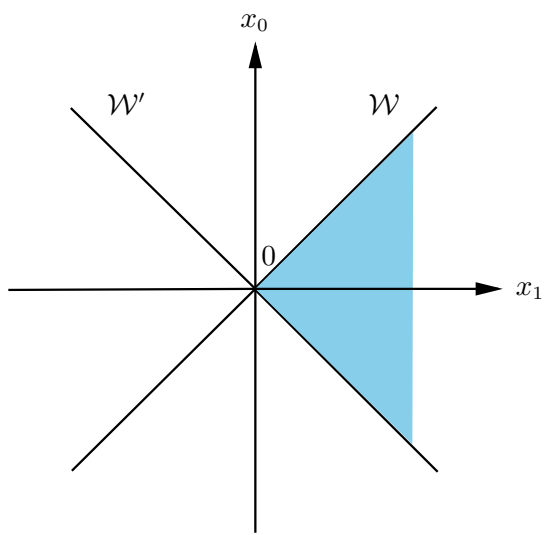}
\caption{The standard right and left wedges} \label{fig:rightwedge}
\end{center}
\end{figure}

We will consider the intersection of the translated right and left wedges $\mathcal{O}_{x,y}=\mathcal{W}_{x}\cap \mathcal{W}_{y}'$, with $x,y \in \mathbb{R}^{2}$, $y-x\in \mathcal{W}$, which is called the \emph{double cone}.

Of particular interest to us is the double cone of radius $r$ and centre the origin cf.~Fig.~\ref{fig:doublecone}, which is defined as:
$\mathcal{O}_{r}=\mathcal{W}_{-r}\cap \mathcal{W}_{r}'$, where $\mathcal{W}_{r}:=\mathcal{W}_{r\pmb{e}^{(1)}}=\mathcal{W} + r \pmb{e}^{(1)}$ and $\mathcal{W}_{r}'=(\mathcal{W}_{r})'$.

\begin{figure}
\begin{center}
\input{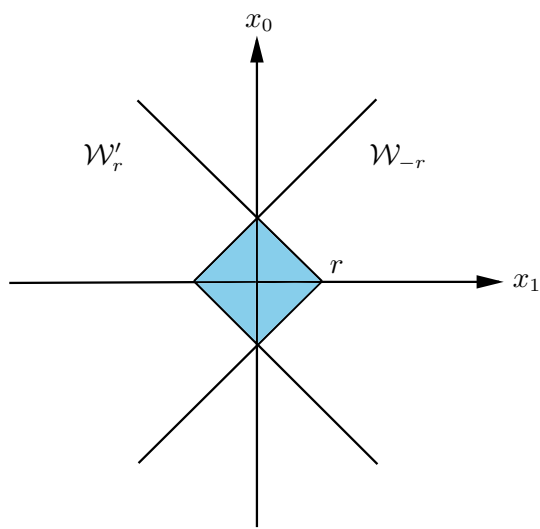}
\caption{The double cone} \label{fig:doublecone}
\end{center}
\end{figure}

\section{Scattering function and its properties}\label{sec:spropr}

We are focusing on theories with factorizing scattering matrices, namely theories where the scattering amplitudes between the outgoing particle and the incoming particle factorize in the product of the S-matrix of the free theory $S^{\text{free}}$ and a scattering function $S$:
\begin{equation}
S_{n,n}(\thetav; \thetav')=S^{\text{free}}(\thetav;\thetav')\prod_{1\leq \ell \leq n}S(|\theta_{k}-\theta_{\ell}|).
\end{equation}
where we set $\thetav:=(\theta_1,\ldots,\theta_n)$, and where the variables $\thetav, \thetav'$ are related to the momenta of the incoming and outgoing particles.

We can find examples of such theories within the completely integrable models, see for example \cite{ZamolodchikovZamolodchikov:1979}. The scattering function $S$ is a function defined by the following properties.

\begin{definition}\label{def:Smatrix}
Let $\strip(0,\pi)$ denote the strip $\rbb+i(0,\pi)$ in the complex plane. We denote with $\mathcal{S}$ the class of scattering functions $S$ satisfying the following properties:
\begin{enumerate}
\item Analytic on $\mathrm{S}(0,\pi)$ and smooth on the boundary,
\item Symmetry relation $S(\theta+i\pi)=S(\theta)^{-1}=S(-\theta)=\overline{S(\theta)}$, $(\theta\in \rbb)$,
\item Bounded on $\mathrm{S}(0,\pi)$,

\end{enumerate}
\end{definition}

Remark: in the present work we do not need the regularity condition used in \cite[Def.~3.3]{Lechner:2008}.

\section{S-symmetry} \label{sec:ssymm}

Following \cite[p.~53]{Lechner:2006}, we introduce an action $D_{n}$ of the permutation group  $\perms{n}$ on $L^{2}(\mathbb{R}^{n})$, acting as
\begin{equation}
(D_{n}(\sigma)f)(\thetav)= S^{\sigma}(\thetav)f(\thetav^{\sigma}),\quad \sigma \in \perms{n}.
\end{equation}
where $\thetav^{\sigma}=(\theta_{\sigma(1)},\ldots,\theta_{\sigma(n)})$ and the factors $S^\sigma$ ($\sigma \in \mathfrak{S}_n$) are given by:
\begin{equation}\label{eq:Sperm}
S^\sigma (\thetav) := \prod_{\substack{i<j \\ \sigma(i)>\sigma(j)}} S(\theta_{\sigma(i)}-\theta_{\sigma(j)}).
\end{equation}
\begin{lemma}
The factors $S^\sigma$ fulfil a composition law, that can be found in \cite[p.~54]{Lechner:2006}:
\begin{equation}\label{eq:Scompose}
 S^{\sigma\circ\rho} (\thetav) =  S^\sigma(\thetav)S^{\rho}(\thetav^\sigma).
\end{equation}
\end{lemma}
\begin{proof}
First we consider the case where $\rho$ is the transposition which exchanges the indices $k$ and $k+1$. Following \cite[Formula (4.1.16)]{Lechner:2006} and using the definition \eqref{eq:Sperm}, we have:
\begin{multline}
S^{\sigma \circ \rho}(\thetav)=\prod_{\substack{i<j \\ \sigma \circ \rho(i)>\sigma \circ \rho(j)}} S(\theta_{\sigma \circ \rho(i)}-\theta_{\sigma \circ \rho(j)})\\
=\prod_{\substack{i<j \\ \sigma \circ \rho(i)>\sigma \circ \rho(j) \\ (i,j)\neq (k,k+1)}} S(\theta_{\sigma(i)}-\theta_{\sigma(j)})\prod_{\substack{i=k, j=k+1 \\ \sigma \circ \rho(k)>\sigma \circ \rho(k+1)}} S(\theta_{\sigma (k+1)}-\theta_{\sigma(k)})\\
=\prod_{\substack{i<j \\ \sigma(i)>\sigma(j) \\ (i,j)\neq (k,k+1)}} S(\theta_{\sigma(i)}-\theta_{\sigma(j)})\prod_{\substack{i=k, j=k+1 \\ \sigma(k+1)>\sigma(k)}} S(\theta_{\sigma (k+1)}-\theta_{\sigma(k)}).\label{complawproof}
\end{multline}
If $\sigma(k+1)>\sigma(k)$, then
\begin{eqnarray}
\text{r.h.s.}\eqref{complawproof} &=&\prod_{\substack{i<j \\ \sigma(i)>\sigma(j) }} S(\theta_{\sigma(i)}-\theta_{\sigma(j)})S(\theta_{\sigma(k+1)}-\theta_{\sigma(k)})\nonumber\\
&=&\prod_{\substack{i<j \\ \sigma(i)>\sigma(j) }} S(\theta_{\sigma(i)}-\theta_{\sigma(j)})S(\theta_{k+1}^{\sigma}-\theta_{k}^{\sigma})\nonumber\\
&=&\prod_{\substack{i<j \\ \sigma(i)>\sigma(j) }} S(\theta_{\sigma(i)}-\theta_{\sigma(j)})S(\theta^{\sigma}_{\rho(k)}-\theta^{\sigma}_{\rho(k+1)})\nonumber\\
&=&S^{\sigma}(\thetav)S^{\rho}(\thetav^{\sigma}).
\end{eqnarray}
Notice that in the equation above the product $\prod_{i<j,\; \sigma(i)>\sigma(j) } S(\theta_{\sigma(i)}-\theta_{\sigma(j)})$ includes in principle the case $(i,j)=(k,k+1)$, but this case does not contribute with an $S$-factor because of the condition $\sigma(k+1)>\sigma(k)$; notice also that in the last equality we made use of the following: $S^{\rho}(\thetav)=\prod_{i<j,\; \rho(i)>\rho(j) }S(\theta_{\rho(i)}-\theta_{\rho(j)})=S(\theta_{\rho(k)}-\theta_{\rho(k+1)})=S(\theta_{k+1}-\theta_{k})$.

If $\sigma(k+1)<\sigma(k)$:
\begin{eqnarray}
\text{r.h.s.}\eqref{complawproof} &=&\Big( \prod_{\substack{i<j \\ \sigma(i)>\sigma(j) }}S(\theta_{\sigma(i)}-\theta_{\sigma(j)})\Big)S(\theta_{\sigma(k+1)}-\theta_{\sigma(k)})\nonumber\\
&=& S^{\sigma}(\thetav)S(\theta^{\sigma}_{k+1}-\theta_{k}^{\sigma})\nonumber\\
&=& S^{\sigma}(\thetav)S^{\rho}(\thetav^{\sigma}).
\end{eqnarray}
Notice that in the equation above the product $\prod_{i<j,\;  \sigma(i)>\sigma(j) }S(\theta_{\sigma(i)}-\theta_{\sigma(j)})$ includes the factor $S(\theta_{\sigma(k)}-\theta_{\sigma(k+1)})$; this means that we had to multiply this product with the inverse $S(\theta_{\sigma(k+1)}-\theta_{\sigma(k)})$.

Now, we apply induction hypothesis three times as follows. Let $\tau$ be a transposition, $\sigma, \rho$ be general permutations, we have:
\begin{multline}
S^{\sigma \circ (\rho \circ \tau)}(\thetav)=S^{(\sigma \circ \rho)\circ \tau}(\thetav)=S^{\sigma \circ \rho}(\thetav)S^{\tau}(\thetav^{\sigma \circ \rho})\\
=S^{\sigma}(\thetav)S^{\rho}(\thetav^{\sigma})S^{\tau}((\thetav^{\sigma})^{\rho})=S^{\sigma}(\thetav)S^{\rho\circ \tau}(\thetav^{\sigma}).
\end{multline}
\end{proof}
Using this composition law, it follows \cite[page 830]{Lechner:2008} that $D_{n}$ defines a unitary representation of $\mathfrak{S}_n$ on $L^{2}(\mathbb{R}^{n})$ and that $P^{S}_{n}:=\frac{1}{n!}\sum_{\sigma \in \mathfrak{S}_n}D_{n}(\sigma)$ is the orthogonal projection onto the space of \emph{$S$-symmetric functions} in $L^{2}(\mathbb{R}^{n})$, namely functions such that:
\begin{equation}
f(\thetav)=S^{\sigma}(\thetav)f(\thetav^{\sigma}).
\end{equation}
We denote the $S$-symmetrization of a function with $\operatorname{Sym}_S f$ and it is given by:
\begin{equation}\label{Symf}
\operatorname{Sym}_S f(\thetav)=\frac{1}{n!}\sum_{\sigma \in \perms{n}}S^{\sigma}(\thetav)f(\thetav^{\sigma}).
\end{equation}
We will use $\operatorname{Sym}_S $ etc. also for more general functions and for distributions. If the function depends on several variables and we want to symmetrize only with respect to some of them, we will write $\operatorname{Sym}_{S,\thetav}$. The choice of variables for the symmetrization can be of importance, as the formula
\begin{equation}\label{symstheta}
\operatorname{Sym}_{S,\thetav}\delta^{(n)}(\thetav-\thetav')=
\operatorname{Sym}_{S^{-1},\thetav'}\delta^{(n)}(\thetav-\thetav').
\end{equation}
shows.

\section[Single-particle space, \ldots]{Single-particle space, $S$-symmetric Fock space, space-time symmetries}\label{sec:hilbertspace}

We will focus our attention on models with only one sort of scalar particle with mass $\mu >0$. As in the free scalar field, our single particle space is then $\mathcal{H}_{1}=L^{2}(\mathbb{R},d\theta)$, where $\theta$ (``rapidity'') is related to the particle momentum by
\begin{equation}
p(\theta):=\mu \begin{pmatrix} \cosh\theta \\ \sinh\theta \end{pmatrix}, \quad \theta\in \mathbb{R}.
\end{equation}
Using the subspace of ``$S$-symmetrized'' wave functions introduced in Sec.~\ref{sec:ssymm}, we define our Hilbert space $\Hil$ of the theory as the $S$-symmetrized Fock space over $\mathcal{H}_{1}$:
\begin{equation}
\mathcal{H}:=\bigoplus_{n=0}^{\infty}\mathcal{H}_{n},
\end{equation}
where $\mathcal{H}_{n}$ is the $n$-particle space: $\mathcal{H}_{n}:= P^{S}_{n}\mathcal{H}_{1}^{\otimes n}$, with $\Hil_0 = \cbb\Omega$. We denote the projection onto $\Hil_k$ with $P_k$, we define $Q_k:=\sum_{j=0}^{k}P_k$, and we denote the space of finite particle number states with $\fpn = \bigcup_{k}Q_{k}\mathcal{H}$, $\fpn \subset \Hil$ dense.

We denote  with $U(x,\lambda)$ the unitary, strongly continuous representation of the boosts, $U(0,\lambda)$, and of the translations, $U(x,0)$, on $\mathcal{H}$; we have $U(x,\lambda)=U(x,0)U(0,\lambda)$. This representation acts as, $\Psi \in \mathcal{H}$,
\begin{equation}
(U(x,\lambda)\Psi)_{n}(\thetav):=\exp\Big( i\sum_{k=1}^{n}p(\theta_{k})\cdot x\Big) \Psi_{n}(\thetav-\lambdav), \quad \lambdav=(\lambda,\ldots,\lambda).
\end{equation}
and we denote with $U(j)=J$ the anti-unitary representation of the reflection $j(x):=-x$ on $\mathcal{H}$ , which acts as, $\Psi \in \mathcal{H}$,
\begin{equation}
(U(j)\Psi)_{n}(\thetav):= \overline{\Psi_{n}(\theta_{n},\ldots,\theta_{1})}.
\end{equation}
It is important for later to fix the conventions for the Fourier transform: Let $g\in \mathcal{S}(\mathbb{R}^{n})$, we set
\begin{align}
\tilde g(p) &:=\frac{1}{2\pi}\int d^{2}x g(x)e^{ip\cdot x}
= \frac{1}{2\pi}\int d^2 x g(x)e^{i p_0 x_0}e^{-i p_1 x_1},
\\
 g^{\pm}(\theta) &:=\frac{1}{2\pi}\int d^{2}x g( x)e^{\pm ip(\theta)\cdot x} = \tilde g (\pm p(\theta)) .\label{fouriertransformg}
\end{align}

\section{Jaffe class functions} \label{sec:jaffee}

We know that in Wightman quantum field theory, quantum fields (and associated objects) localized at a point in space-time must be unbounded operators.
Their singular behaviour can be explained by thinking of the uncertainty relation in quantum mechanics. Indeed, measurements which take place in a finite region of space-time need that a big quantity of energy and momentum is transferred. Therefore, expectation values of quantum fields between states with good behaviour at high energies should be non-singular. For this reason, one usually considers operators which fulfil polynomial bounds at high energy, namely Wightman fields $\phi(x)$ such that $(1+H)^{-\ell} \phi(x) (1+H)^{-\ell}$ is bounded for some $\ell>0$ \cite{FreHer:pointlike_fields}.
We can absorb the above condition on the high-energy behaviour of quantum fields into the choice of the class of test functions space with which one smears these quantum fields. Usually we take this class to be the Schwartz space, just recall the known book \cite{StrWig:PCT}. But actually according to Jaffe \cite{Jaffe:1967} this choice is too restrictive and we can extend the class of smearing functions to a more general family. This includes  energy bounds which instead of being only of polynomial type growth in energy, can be ``almost exponential'' growth like $\exp \omega(E)$, where the function $\omega$ can almost grow linearly in $E$. To read more about this see also \cite{ConstantinescuThaleimer:1974}.
The generalized class of distributions associated to this more general space of test functions was studied by \cite{Bjoerck:1965}, but according to the paper Beurling already presented a certain generalized distribution theory before (see the related citations in the paper of Bjoerck).

In this thesis, we are going to adopt Jaffe's point of view with some little variations, since with the aim of constructing examples of local operators, we would like to consider a more general class of operators as possible.
In the following we list the properties that we require the function $\omega$ (the so called \emph{indicatrix}) to fulfil.

\begin{definition}\label{def:indicatrix}
An \emph{indicatrix} is a smooth function $\omega:[0,\infty) \to [0,\infty)$ with the following properties.
\begin{enumerate}
\renewcommand{\theenumi}{($\omega$\arabic{enumi})}
\renewcommand{\labelenumi}{\theenumi}
\item \label{it:omegamonoton}
$\omega$ is monotonously increasing;
\item \label{it:sublinear}
$\omega(p+q)\leq \omega(p)+\omega(q)$ for all $p,q \geq 0$ \emph{(sublinearity)};
\item \label{it:omegagrowth}
$\displaystyle{\int_{0}^{\infty}\frac{\omega(p)}{1+p^{2}}\;dp < \infty}$.
\end{enumerate}
We call $\omega$ an \emph{analytic indicatrix} if, in addition, there exists a function $\oa$ on the upper half plane $\rbb + i [0,\infty)$, analytic in the interior and continuous at the boundary, such that
\begin{enumerate}
\setcounter{enumi}{3}
\renewcommand{\theenumi}{($\omega$\arabic{enumi})}
\renewcommand{\labelenumi}{\theenumi}

\item \label{it:omegaeven}
$\re\oa(p)=\re\oa(-p)$ for all $p \geq 0$;

\item \label{it:omegaestimate}
There exist $a_\omega,b_\omega>0$ such that $\omega(|z|) \leq \re\oa(z) \leq a_\omega \omega(|z|)+b_\omega$ for all $z \in \rbb+i[0,\infty)$.

\end{enumerate}
\end{definition}

We have chosen these properties as general as possible so that one can find a large range of examples. One  example, which in terms of $\exp \omega(E)$ reads as the usual polynomial growth energy behaviour, is the following for some $\beta>0$:
\begin{equation}\label{eq:omegalog}
\omega(p) = \frac{\beta}{2}\log(1+p), \quad
\oa(z)=\beta[\operatorname{Log}(i+z)+1].
\end{equation}

\begin{lemma}
The example \eqref{eq:omegalog} matches the definition \ref{def:indicatrix}.
\end{lemma}
\begin{proof}
In this example $\omega$  is clearly a continuous function $[0,\infty)\rightarrow [0,\infty)$. It fulfils the \emph{subadditivity} property due to \cite[Proposition 1.3.6]{Bjoerck:1965}. Moreover, there holds:
$\int_{0}^{\infty}\frac{\log(1+p)}{1+p^{2}}\;dp < \infty$.

The function $\oa$ is analytic on $\mathbb{R}+i[0,\infty)$.

We have $\operatorname{Re}\oa(p+iq)=\frac{\beta}{2}[\log|p+i(q+1)|^{2}+2]\geq \frac{\beta}{2}[\log(p^{2}+q^{2}+1)+2]\geq \frac{\beta}{2}\log(|p+iq|+1)$ since  $\frac{\log(x^{2}+1)+2}{\log(x+1)}\geq 1$ for $x\geq 0$. This proves the property \ref{it:omegaestimate}(part 1).

The property \ref{it:omegaestimate}(part 2) follows from a short computation: $\operatorname{Re}\oa(p+iq)=\beta \log |p+i(q+1)|+\beta\leq \omega(|p+iq|)+\beta$ since $|p+i(q+1)|\leq |p+iq|+1$.

The property \ref{it:omegaeven} is also fulfilled because
 $\operatorname{Re}\oa(p)=\beta +\beta \log|i+p|=\beta+\beta\log\sqrt{1+p^{2}}$.
\end{proof}

A second class of examples with stronger growth in $p$ is, with $0 < \alpha < 1$,
\begin{equation}\label{eq:omegapower}
\omega(p)= p^{\alpha}\cos\Big( \frac{\alpha \pi}{2}\Big),\quad
\oa(z)= i^{-\alpha}(z+i)^{\alpha}.
\end{equation}

\begin{lemma}
The example \eqref{eq:omegapower} matches the definition \ref{def:indicatrix}.
\end{lemma}
\begin{proof}
In this example $\omega$ is again a continuous function $[0,\infty)\rightarrow [0,\infty)$; it is increasing and concave, since $\alpha < 1$. Subadditivity then follows by \cite[Proposition 1.2.1]{Bjoerck:1965}. Moreover, there holds for $0<\alpha<1$:
$\int_{0}^{\infty}\frac{p^{\alpha}}{1+p^{2}}\;dp < \infty$.

The function $\oa$ is analytic on $\mathbb{R}+i[0,\infty)$.  To prove the property \ref{it:omegaestimate} we compute:
\begin{multline}\label{reombar}
\operatorname{Re}\oa(p+iq)=\operatorname{Re}i^{-\alpha}(p+(q+1)i)^{\alpha}
=\operatorname{Re}i^{-\alpha}\exp(\alpha\operatorname{Log}(p+(q+1)i))\\
=\operatorname{Re}\exp[i\frac{\pi}{2}(-\alpha)+\alpha \log |p+(q+1)i|+i\alpha\operatorname{arg}(p+(q+1)i)]\\
=|p+(q+1)i|^{\alpha}\cos((-\alpha)\frac{\pi}{2}+\alpha\operatorname{arctan}\frac{q+1}{p})\\
\geq |p+(q+1)i|^{\alpha}\cos(\frac{\alpha\pi}{2}).
\end{multline}
where in the last inequality we made use of the fact that $\alpha\arctan\frac{q+1}{p} \in [\alpha \epsilon,\alpha(\pi-\epsilon)]$ and therefore $(-\alpha)\frac{\pi}{2}+\alpha\operatorname{arctan}\frac{q+1}{p}\in [-\alpha\frac{\pi}{2}+\epsilon,\alpha\frac{\pi}{2}-\epsilon]$.

This proves the property \ref{it:omegaestimate}(part 1).

From \eqref{reombar} we have also
\begin{equation}
\begin{aligned}
\operatorname{Re}\oa(p+iq)&\leq&|p+(q+1)i|^{\alpha}\\
&=&[p^{2}+(q+1)^{2}]^{\alpha/2}\\
&\leq& (|p+iq|+1)^{\alpha}\\
&\leq& c |p+iq|^{\alpha}+d\cdot 1^{\alpha}\\
&=&c' \omega(|p+iq|)+d.
\end{aligned}
\end{equation}
The fourth inequality follows from the following fact. The function $f(a,b)=(|a|^{\alpha}+|b|^{\alpha})/(|a|+|b|)^{\alpha}$ has the property to be homogeneous of order $0$, that is: $f(\lambda a, \lambda b)=f(a,b)$ for all $\lambda>0$; hence, for arbitrary $(a,b) \neq (0,0)$ we can rescale the argument of the function $f$, $(a,b)=\lambda (c,d)$ with $\lambda>0$, such that $(c,d)\in S^{1}\subset \rbb^{2}$. Then, we notice that $f$ is clearly continuous, positive and non-zero on the unit circle $S^{1}$. Hence, $f$ is bounded there and we can find two positive real constants $m,M >0$ such that $m \leq f(a,b)\leq M$ for all $(a,b) \in S^{1}$. This implies $m(|a|+|b|)^{\alpha}\leq |a|^{\alpha} + |b|^{\alpha} \leq M(|a|+|b|)^{\alpha}$ (*), where we can choose $m=M^{-1}$. Notice that the point $(a,b)=(0,0)$ fulfils the inequality (*) trivially.

This proves the property \ref{it:omegaestimate}(part 2). To prove the property \ref{it:omegaeven} we compute:
\begin{equation}
\operatorname{Re}\oa(p)=\operatorname{Re}(e^{i(-\alpha)\frac{\pi}{2}}e^{\alpha\operatorname{Log}(p+i)})
=\operatorname{Re}\exp(\alpha\log\sqrt{1+p^{2}}+i\alpha\operatorname{arctan}\frac{1}{p}-i\alpha\frac{\pi}{2}).
\end{equation}
Using the relation $\operatorname{arctan}\frac{1}{p}=-\operatorname{arctan}p +\frac{\pi}{2}$, we find
\begin{multline}
\operatorname{Re}\oa(p)=\exp\alpha\log\sqrt{1+p^{2}}\operatorname{Re}\exp i(\alpha\operatorname{arctan}(-p))\\
=|i+p|^{\alpha}\cos(\alpha \operatorname{arctan}(-p))=|i+p|^{\alpha}\cos(\alpha \operatorname{arctan}(p)).
\end{multline}
\end{proof}
Let $\omega$ be an indicatrix and let $\ocal$ be an open set in Minkowski space. We consider the following space of functions with compact support in $\ocal$,
\begin{equation}
 \dcal^\omega(\ocal) := \{ f \in \dcal(\ocal) : \theta \mapsto e^{\omega(\cosh \theta)} f^\pm(\theta) \text{ is bounded and square integrable} \}.\label{domega}
\end{equation}
We are not interested in equipping $\dcal^\omega(\ocal)$ with a topology, even if one can find in \cite{Bjoerck:1965,ConstantinescuThaleimer:1974} methods on how to topologize these kind of spaces. We are rather interested to know ``how many'' elements $f$ the space $\dcal^\omega(\ocal)$ contains. If for example $\omega$ is of the form \eqref{eq:omegalog}, or bounded by this, then $e^{\omega(p)}$ is clearly bounded by a power of $p$; due to Paley-Wiener theorem the product $e^{\omega(p)} \tilde{f}(p)$ is bounded for any $f\in \dcal(\ocal) := C^{\infty}_{0}(\ocal)$: this because the Fourier transform of a function $f\in C^{\infty}_{0}(\ocal)$ is entire analytic and bounded by a polynomial in $p$ at infinity. Hence, in this case $\dcal^\omega(\ocal) =  \dcal(\ocal) := C^{\infty}_{0}(\ocal)$. See also \cite[Proposition~1.3.6]{Bjoerck:1965}.
If instead we consider a faster growing $\omega$, it is not obvious a priori that $\dcal^\omega(\ocal)$ contains any non-zero element. It is condition \ref{it:omegagrowth} on how fast $\omega$ needs to grow to be decisive for nontriviality. Indeed, it was shown in \cite[Theorem~1.3.7]{Bjoerck:1965} that condition \ref{it:omegagrowth} is equivalent to the fact that one can find functions $f$ (``local units'') in $\dcal^\omega(\ocal)$ with $0 \leq f \leq 1$, with $f=1$ on any given compact set $\kcal\subset \ocal$, and $f=0$ outside any given neighbourhood of $\kcal$, such that a certain norm of $f$ (see \cite[Definition~1.3.1]{Bjoerck:1965}) is finite. However, this bound is related to our bound in \eqref{domega} as a consequence of \cite[Definition~1.3.25]{Bjoerck:1965} and \cite[Corollary~1.4.3]{Bjoerck:1965}. The square integrability in \eqref{domega} is a consequence of \cite[Definition~1.3.25]{Bjoerck:1965} for $\lambda =2$; indeed if $|e^{2\omega(\cosh\theta)}f^{\pm}(\theta)|$ is bounded by a constant, then $e^{\omega(\cosh\theta)}f^{\pm}(\theta)$ is bounded by $e^{-\omega(\cosh\theta)}$, which can be integrated due to \ref{it:omegagrowth}.
Notice that it suffices to show this for $\omega(p)>p^{\alpha}$, $0<\alpha<1$, or for $\omega(p)> \log p$. Indeed, in the case where $\omega$ is not greater than $p^{\alpha}$, we can define $\omega'(p):= \omega(p) +p^{\alpha}\geq p^{\alpha}$; if we can find local units in $\dcal^{\omega'}(\ocal)$, then we can also find such local units in $\dcal^{\omega}(\ocal)$ as well, since $\dcal^{\omega'}(\ocal)\subset \dcal^{\omega}(\ocal)$ (as the condition in $\dcal^{\omega'}(\ocal)$ is stricter).
All this is equivalent to say that the space $\dcal^\omega(\ocal)$ is non-trivial.

We can approximate any functions in $\dcal(\ocal)$ with functions in $\dcal^\omega(\ocal)$ by considering the convolutions of the smooth functions with compact support in $\ocal$ with these local units. Since the convolution in Fourier space is just a multiplication, their product still decays rapidly in momentum space and is again in $\dcal^\omega(\ocal)$. By performing the limit of the convoluted function in the $\dcal(\ocal)$ topology as the local units in $\dcal^\omega(\ocal)$ approaches the delta distribution, we obtain that the convoluted function converges to the function in $\dcal(\ocal)$ (see \cite[Theorem 1.3.16]{Bjoerck:1965} for more details on this argument). Hence, one finds that $\dcal^\omega(\ocal)$ is actually dense in $\dcal(\ocal)$, in the $\dcal(\ocal)$ topology.

The importance of the condition \ref{it:omegagrowth} becomes evident also if one considers the example of a function $f$ whose Fourier transform fulfils the bound: $|\tilde{f}(p)|\leq e^{-|p|}$, namely $\omega(p)=|p|$. In this case, we have that the integral $\int \tilde{f}(p)e^{ip(x+iy)} dp= f(x+iy)$ converges if $|y|<1$ since the integrand is bounded by $|\tilde{f}(p)e^{ip(x+iy)}|\leq e^{-|p|}e^{-yp}$ (and therefore by $e^{-p(1+y)}$ for $p>0$ and by $e^{+p(1-y)}$ for $p<0$). Hence the function $f(x+iy)$ is defined on the strip $|y|<1$ and consequently cannot have compact support. Thus,
$\dcal^\omega(\ocal)$ is trivial (see also \cite{Bjoerck:1965} before Theorem~1.3.7).

For functions in $\dcal^\omega(\ocal)$, one can derive Paley-Wiener type estimates on their Fourier transform \cite[Sec~1.4]{Bjoerck:1965}. We use the following variant in our context.

Most of the material in the rest of this section is due to H. Bostelmann.

\begin{proposition}\label{prop:omegapw}
Let $\omega$ be an analytic indicatrix, $r\in \rbb$, and $f \in \dcal^\omega(\wcal_r)$. Then $ f^-$ extends to an analytic function on the strip $\strip(0,\pi)$, continuous on its closure, and one has $f^-(\theta+i\pi)=f^+(\theta)$. For fixed $\ell\in\nbb_0$, there exists $c>0$ such that
\begin{equation}\label{eq:fmstrip}
  \Big\lvert \frac{d^\ell f^-}{d\zeta^\ell} (\theta+ i \lambda) \Big\rvert \leq c (\cosh \theta)^\ell e^{-\mu r \cosh \theta \sin \lambda} e^{-\omega(\cosh \theta)/a_\omega}
\quad \text{for all }\theta \in \rbb, \; \lambda \in [0,\pi].
\end{equation}
\end{proposition}

\begin{proof}
 Since $f$ has compact support, its Fourier transform $\tilde f$ and $f^\pm$ are actually entire, and the relation $f^-(\zeta \pm i\pi)=f^+(\zeta)$ follows by direct computation from definition \eqref{fouriertransformg} and the fact that $p(\theta+i\pi)=-p(\theta)$. We first prove the bound \eqref{eq:fmstrip} in the case $\ell=0$, $r=0$. We consider the function $g$ on $\strip(0,\pi)$ (note that $\sinh$ maps the strip into the upper half plane and $\oa$ is defined there), defined by
\begin{equation}
   g(\zeta) := f^-(\zeta) e^{\oa(\sinh \zeta)/a_\omega}.\label{gdeffminus}
\end{equation}
For $\zeta=\theta+i\lambda$ in the closed strip, one has
\begin{equation}\label{eq:omegacd}
 \re \oa(\sinh \zeta)/a_\omega \leq \omega( \lvert\sinh \zeta\rvert)+b_\omega/a_\omega \leq \omega(\cosh \theta)+b_\omega/a_\omega,
\end{equation}
where in the first inequality we used \ref{it:omegaestimate} (right inequality) and in the second inequality we used \ref{it:omegamonoton} and the fact that $|\sinh (\theta +i\lambda)|\leq \cosh\theta$.
Since $f \in \dcal^\omega(\rightwedge)$, we have by definition $\sup_{\theta\in\rbb} |\exp(\omega(\cosh\theta))f^{-}(\theta)|  < \infty$, hence it follows that
\begin{equation}
   \sup_{\theta\in\rbb} |g(\theta)| \leq e^{b_\omega/a_\omega} \sup_{\theta\in\rbb} |e^{\omega(\cosh\theta)}f^{-}(\theta)|  < \infty.
\end{equation}
This means that $g$ is bounded on $\rbb$, and by a similar computation involving $f^+$, we have that it is bounded also on the line $\rbb+i\pi$ (since $f^-(\theta +i\pi)=f^+(\theta)$).

In the interior of the strip, we know that $f^-(\zeta)$ is bounded since $\supp f \subset \rightwedge$: see \cite[Proposition 4.2.6]{Lechner:2006}; therefore,
\begin{equation}
  \lvert g(\theta+i\lambda) \rvert  \leq e^{\omega(\cosh \theta)} \sup_{\zeta'\in\strip(0,\pi)} |f^- (\zeta') |,
\end{equation}
where we have used \eqref{eq:omegacd}.

Hence, we have shown that $g$ is uniformly bounded in $\lambda$ in the interior of the strip and grows for large $\theta$ like $e^{\omega(\cosh \theta)}$. However, $g$ is bounded in  real direction at the boundary of the strip. By application of the maximum modulus principle we would like to show that $g$ is actually bounded on the entire strip by the maximum which is attained at the boundary. For this, $g$ must grow not
too ``fast'' for $\theta \rightarrow \infty$ in the interior of the strip. According to \cite[Theorem~3]{HardyRogosinski:1946} it suffices if $g$ behaves like $e^{\omega(p)}$ with $(\omega(p)/ p) \rightarrow 0$ for $p\rightarrow \infty$. Since $\omega(p)=o(p)$ due to \ref{it:omegagrowth} and the function $\log |g|$ is subharmonic and bounded by $\log|g|\leq \omega(\cosh\theta)\approx e^{\theta}$, we can apply a Phragm\'en-Lindel\"of argument \cite[Theorem~3]{HardyRogosinski:1946} to $\log |g|$ and show that this function is actually bounded on the strip, and takes its maximum at the boundary. Therefore, the function $g$ is bounded on the strip for $\theta\rightarrow \infty$ uniformly in $\lambda$.

In other words, from \eqref{gdeffminus} we have,
\begin{equation}\label{eq:fmbound}
  \lvert f^-(\zeta) \rvert \leq c \, \lvert e^{-\oa(\sinh \zeta)/a_\omega} \rvert\quad \text{for all } \zeta \in \strip(0,\pi)
\end{equation}
with some $c>0$. We estimate
\begin{equation}
  \re \oa(\sinh \zeta) \geq \omega(|\sinh \zeta|) \geq \omega (\cosh \theta - 1) \geq \omega(\cosh\theta)-\omega(1),
\end{equation}
where in the first inequality we used \ref{it:omegaestimate}; in the second inequality we used \ref{it:omegamonoton} together with the relations $|\sinh\zeta|\geq |\sinh\theta|$ and $\sinh\theta = \cosh\theta -e^{-\theta}$ with $e^{-\theta}\leq 1$ for $\theta >0$; in the third inequality we made use of \ref{it:sublinear}: $\omega(\cosh\theta)=\omega(\cosh\theta -1 +1) \leq \omega(\cosh\theta -1) + \omega(1)$.

Inserted into \eqref{eq:fmbound}, this gives \eqref{eq:fmstrip} for $r=0$, $\ell=0$.

For the case $r\neq 0$, $\ell=0$, we note that $f^-(\zeta)=\exp(-i \mu r \sinh \zeta) h^-(\zeta)$ with $h \in \dcal^\omega(\rightwedge)$ and by applying the result before to $h^-(\zeta)$, we find \eqref{eq:fmstrip} for $r\neq 0$, $\ell=0$. By analogous arguments, the same estimate \eqref{eq:fmstrip} holds for $f^+(\zeta)$, $\zeta \in \strip(0,\pi)$, if $f \in \dcal^\omega(\wcal_{-r}')$ (see \cite[Proposition 4.2.6]{Lechner:2006}).

For $r=0$, $\ell > 0$, we proceed as follows. Since $f$ has compact support and $\rightwedge$ is open, we can choose $s>0$ such that $f \in \dcal^\omega( \rightwedge )\cap \dcal^\omega( \wcal_s')$.

Because of the relation $f^-(\zeta \pm i\pi)=f^+(\zeta)$, we have that $f^-$ in the strip $\strip(-\pi,0)$ or $\strip(\pi,2\pi)$ corresponds to $f^+$ in the strip $\strip(0,\pi)$ (which is bounded if $f$ is localized in $\wcal_s'$ due to \cite[Proposition 4.2.6]{Lechner:2006}). Hence, using the above result for $f^-$ and $f^+$, we have the estimate
\begin{equation}\label{eq:fmestoutside}
  \lvert f^- (\theta+ i \lambda) \rvert \leq c  e^{-\omega(\cosh \theta)/a_\omega} \cdot
\begin{cases}
  1 & \text{if } \lambda \in (0,\pi),  \\
 e^{\mu s \cosh \theta |\sin \lambda|} & \text{if } \lambda \in (-\pi,0) \cup (\pi,2\pi).
\end{cases}
\end{equation}
We use Cauchy's formula to estimate the derivatives of $f^-$: For any $t>0$, we have
\begin{equation}
   \Big\lvert \frac{d^\ell f^-}{d\zeta^\ell} (\zeta) \Big\rvert = \frac{\ell!}{2\pi} \int_{|\zeta-\zeta'|=t} \frac{|f^-(\zeta')|}{|\zeta-\zeta'|^{\ell +1}}d\zeta'
   \leq \ell! \,t^{-\ell} \sup_{|\zeta-\zeta'|=t} |f^-(\zeta')|.
\end{equation}
where we took into account that the length of the integral path is $2\pi t$.

Here $\zeta\in\strip(0,\pi)$, but parts of the circle $|\zeta-\zeta'|=t$ might be outside this strip. With $t < \pi/2$, this circle is within the strips $\strip(-\pi,0)\cup \strip(0,\pi)$ or $\strip(0,\pi)\cup \strip(\pi,2\pi)$, so we can use the estimates \eqref{eq:fmestoutside} and we obtain for large $\theta$, taking into account that $|\sin\lambda|\leq |\lambda|\leq t$,
\begin{equation}
   \Big\lvert \frac{d^\ell f^-}{d\zeta^\ell} (\theta+i \lambda) \Big\rvert
  \leq \ell! c t^{-\ell} e^{\mu s t \cosh (\theta+t) } e^{-\omega(\cosh(\theta-t))/a_\omega}.
\end{equation}
(Notice that assuming ``large $\theta$'' ensures for example that $\cosh(\theta-t)<\cosh\theta$ and $\cosh(\theta+t)>\cosh\theta$ and therefore that the two estimates on the exponentials above hold.)

We choose $t = 1/\cosh\theta$. Using $\cosh(\theta-t) \geq \cosh \theta - c'$, $\cosh (\theta+t) \leq \cosh \theta + c'$ with some $c'>0$ (notice that the first inequality can be proved by showing that the function $y(x)=\cosh\Big( x- \frac{1}{\cosh x}\Big)-\cosh x$ is bounded below by some negative constant; with analogous argument we prove also the second inequality), and using \ref{it:sublinear}, we obtain a constant $c''>0$ such that
\begin{equation}
   \Big\lvert \frac{d^\ell f^-}{d\zeta^\ell} (\theta+i \lambda) \Big\rvert
  \leq c'' (\cosh\theta)^{\ell} e^{-\omega(\cosh\theta)/a_\omega}.
\end{equation}
For large $-\theta$, the computation is analogous. This gives \eqref{eq:fmstrip}.

Finally, for the case $\ell>0$, $r \neq 0$ we compute the derivatives of\\ $f^-(\zeta)=\exp(-i \mu r \sinh \zeta) h^-(\zeta)$.  The result before applies to the factor $\frac{d^k h^-}{d\zeta^k} (\zeta)$; noting that
$d^k \exp(i \mu r \sinh \zeta) / d\zeta^k$, $0 \leq k \leq \ell$, is bounded by $c_k (\cosh\theta)^k \exp(\mu r \cosh \theta \sin \lambda)$ with constants $c_k>0$, we obtain \eqref{eq:fmstrip}.
\end{proof}

\section{Quadratic forms} \label{sec:qforms}

We introduce a dense subspaces of our Hilbert space $\Hil$, which is related to a fixed indicatrix $\omega$. We call it $\Hil^\omega$ and it is defined as $\Hil^\omega := \{\psi \in \Hil : \|e^{\omega(H/\mu)} \psi\| < \infty\}$. We denote, for fixed $k$, $\Hil^\omega_k = \Hil^\omega \cap \Hil_k$, and $\fpno = \Hil^{\omega} \cap \fpn$.
We consider for test functions $g \in \dcal(\rbb^m)$ the following norm,
\begin{equation}
  \onorm{g}{2} := \gnorm{\thetav \mapsto e^{\omega(E(\thetav))}g(\thetav)}{2},
\end{equation}
where $E$ is the dimensionless energy function,
\begin{equation}
   E(\thetav) = \sum_{j=1}^{m}p_0(\theta_{j})/\mu = \sum_{j=1}^m \cosh \theta_j .
\end{equation}
Now we denote by $\qf^\omega$ the space of quadratic forms (or more precisely, sesquilinear forms) $A$ on $\fpno\times\fpno$, namely,
\begin{equation}
   A: \fpno \times \fpno \to \cbb, \quad  (\psi,\chi) \mapsto \hscalar{\psi}{ A \chi},
\end{equation}
such that the following norms are finite for any $k \in \nbb_0$:
\begin{equation}\label{eq:aomeganorm}
 \gnorm{A}{k}^{\omega} := \frac{1}{2} \gnorm{Q_k A e^{-\omega(H/\mu)}Q_k}{} + \frac{1}{2} \gnorm{Q_k e^{-\omega(H/\mu)} A Q_k}{}.
\end{equation}
As we can see from \eqref{eq:aomeganorm}, quadratic forms $A \in \qf^\omega$ can be unbounded because of their behaviour at high energies (notice the energy damping factor $\exp(-\omega(H/\mu))$) and at high particle numbers (notice the projector on $\Hil_k$, $Q_{k}$).

We note that space-time translations and reflections act on $\qf^\omega$ by adjoint action of $U(\cdotarg)$, and leave this space invariant since they commute with $H$ and $Q_{k}$:
\begin{eqnarray}
\gnorm{ U(x)AU(x)^{*}}{k}^{\omega} &=& \frac{1}{2} \gnorm{Q_k U(x)A U(x)^{*} e^{-\omega(H/\mu)}Q_k}{} + \frac{1}{2} \gnorm{Q_k e^{-\omega(H/\mu)} U(x) A U(x)^{*} Q_k}{}\nonumber\\
&=& \frac{1}{2} \gnorm{U(x) Q_k A  e^{-\omega(H/\mu)} Q_k U(x)^{*}}{} + \frac{1}{2} \gnorm{U(x) Q_k e^{-\omega(H/\mu)}  A  Q_k U(x)^{*}}{}\nonumber\\
 &=& \gnorm{A}{k}^{\omega}.
\end{eqnarray}
and similarly for $U(j)$.

The adjoint action of Lorentz boosts $U(0,\lambda)$ maps $\qf^\omega$ into $\qf^{\omega'}$:
\begin{eqnarray}
\gnorm{ U(\lambda)AU(\lambda)^{*}}{k}^{\omega'} &=&  \frac{1}{2} \gnorm{Q_k U(j)A U(j)^{*} e^{-\omega'(H/\mu)}Q_k}{} + \frac{1}{2} \gnorm{Q_k e^{-\omega'(H/\mu)} U(j) A U(j)^{*} Q_k}{}\nonumber\\
&=& \frac{1}{2} \gnorm{U(j) Q_k A  e^{-\omega'(H'/\mu)} Q_k U(j)^{*}}{} + \frac{1}{2} \gnorm{U(j) Q_k e^{-\omega'(H'/\mu)}  A  Q_k U(j)^{*}}{}\nonumber\\
&=& \frac{1}{2} \gnorm{ Q_k A  e^{-\omega'(H'/\mu)} Q_k }{} + \frac{1}{2} \gnorm{ Q_k e^{-\omega'(H'/\mu)}  A  Q_k }{}.\label{UlambdaA}
\end{eqnarray}
where $H':= U(0,\lambda) H U(0,\lambda)^\ast$. By recalling that the boost and the Hamiltonian act on functions $f \in L^{2}(\rbb^{2})$ as: $(U(\lambda)f)(\theta)=f(\theta +\lambda)$ and $(H f)(\theta)=\mu\cosh\theta f(\theta) = p_0(\theta) f(\theta)$, we have:
\begin{eqnarray}
(U(\lambda) (H f))(\theta)=(H  f)(\theta +\lambda) &=& \mu \cosh(\theta +\lambda) f(\theta + \lambda)\nonumber\\
&=& (\mu \cosh\theta \cosh \lambda +\mu \sinh\theta \sinh \lambda)f(\theta + \lambda) \nonumber\\
& =& (E(\theta)\cosh \lambda + p_1 \sinh \lambda)f(\theta + \lambda).
\end{eqnarray}
Now, we want to compare $\exp(-\omega(H/\mu))$ with $\exp(-\omega'(H'/\mu))$. Let $c>0$ be such that $H' \geq c H$. Since $[H',cH]=0$, we can proceed from $H' \geq c H$ to $\exp(-\omega'(H'/\mu))\leq \exp(-\omega'(c H))$. Defining $\omega'(p)=\omega(p/c)$, we have that $\exp(-\omega'(c H))=\exp(-\omega(H/\mu))$. Hence, we find that $\exp(-\omega'(H'/\mu))\leq \exp(-\omega(H/\mu))$. By \eqref{UlambdaA},\\ we have $\gnorm{U(\lambda)AU(\lambda)^{*}}{k}^{\omega'} \leq \gnorm{A}{k}^{\omega}$ with $\omega'(p) = \omega(p/c)$. This implies that the adjoint action of the Lorentz boosts maps $\qf^\omega$ into $\qf^{\omega'}$ with $\omega'(p) = \omega(cp)$. So, we could in principle modify the definition of $\qf^\omega$ by requiring that the norm $\gnorm{A}{k}^{\omega(\beta \cdotarg)} < \infty$ for \emph{some} $\beta$ (depending on $A$), so that $\qf^\omega$ is fully Poincar\'e invariant; but we remain here with the definition \eqref{eq:aomeganorm}, which is simpler.

\section{Generalized annihilation and creation operators} \label{sec:zgen}

Similar to the Fock representation of the CCR algebra, Lechner \cite{Lechner:2006} introduced a representation of the Zamolodchikov algebra using modified creation and annihilation operators $z,\zd$ on $\mathcal{H}$.
These operators are defined on $\fpn$ by
\begin{eqnarray}
(\zd(f)\Phi)_{n} &:=& \sqrt{n}P^{S}_{n}(f\otimes \Phi_{n-1}),\\
z(f)&:=& \zd(\overline{f})^{*}.
\end{eqnarray}
where $\Phi \in \fpn$, $f \in \mathcal{H}_{1}$.
It was shown in \cite{Lechner:2006} that these satisfy the relations of the Zamolodchikov algebra:
\begin{eqnarray}
\zd(\theta)\zd(\eta) &=& S(\theta-\eta)\zd(\eta)\zd(\theta),\nonumber\\
z(\theta)z(\eta) &=& S(\theta-\eta)z(\eta)z(\theta),\nonumber\\
z(\theta)\zd(\eta) &=& S(\eta-\theta)\zd(\eta)z(\theta) + \delta(\theta-\eta)\cdot 1_{\mathcal{H}}.\label{zamoloalgebra}
\end{eqnarray}
The $\zd(\theta),z(\eta)$ are distributions, or can also be seen as quadratic forms on the domain $(\fpn \cap \dcal(\rbb^{k})) \times (\fpn \cap \dcal(\rbb^{k}))$; when smeared with test functions $f\in \scal(\rbb)$, $\zd(f)$, $z(f)$ are unbounded operators on $\fpn$, but their unboundedness is  related to the particle number, as we can see from the following computation of their norms (setting $\omega=0$, the following Lemma holds still true):
\begin{lemma}
In generalization of \cite[Eq.~(3.14)]{Lechner:2008}, we have for $\ell \in \nbb_0$ and $f \in \scal(\rbb)$,
\begin{equation}\label{omegaz}
 \gnorm{ e^{\omega(H/\mu)} \zd(f) e^{-\omega(H/\mu)} Q_\ell  }{}
  \leq \sqrt{\ell+1} \onorm{f}{2},
\quad
 \gnorm{ e^{\omega(H/\mu)} z(f) e^{-\omega(H/\mu)} Q_\ell  }{}
  \leq \sqrt{\ell} \gnorm{f}{2},
\end{equation}
if the right-hand side is finite.
\end{lemma}
\begin{proof}
Formula \eqref{omegaz} (left equation) is equivalent to say that for every $\psi \in Q_{\ell}\mathcal{H}^{\omega}$, we have
\begin{equation}
|| e^{\omega(H/\mu)}z^{\dagger}(f)\psi ||\leq ||f||_{2}^{\omega}\sqrt{\ell +1} ||e^{\omega(H/\mu)}\psi||.
\end{equation}
Due to Pythagoras it suffices to prove this for $\psi \in \mathcal{H}^{\omega}_{\ell}$.

We have
\begin{equation}
\Big( e^{\omega(H/\mu)}z^{\dagger}(f)\psi\Big)(\thetav)= \sqrt{\ell +1}\operatorname{Sym}_{S,\thetav}\Big(e^{\omega(E(\thetav))}f(\theta_{1})\psi(\theta_{2},\ldots,\theta_{\ell +1})  \Big).
\end{equation}

By application of Cauchy-Schwarz, we have
\begin{multline}
||e^{\omega(H/\mu)}z^{\dagger}(f)\psi||^{2}\leq (\ell +1) \int e^{2\omega(E(\thetav))} |f(\theta_{1})|^{2} |\psi(\hat{\thetav})|^{2}\; d\theta_{1} d\hat{\thetav}\\
\leq (\ell +1) (||f||_{2}^{\omega})^{2} ||e^{\omega(E(\hat{\thetav}))}\psi(\hat{\thetav})||^{2}_{2}\\
 = (\ell +1) (||f||^{\omega}_{2})^{2} ||e^{\omega(H/\mu)}\psi||^{2}.
\end{multline}
where in the second inequality we made use of the sublinearity of $\omega$  \ref{it:sublinear}: $e^{2\omega(E(\thetav))}=e^{2\omega(E(\theta_{1}) +E(\hat{\thetav}))}\leq e^{2\omega(E(\theta_{1}))}e^{2\omega(E(\hat{\thetav}))}$.

Analogously, for $z$ we have
\begin{equation}
\Big( e^{\omega(H/\mu)}z(f)\psi\Big)(\thetav)= \sqrt{\ell }\int d\theta_{1}\;\operatorname{Sym}_{S,\hat{\thetav}}\Big(e^{\omega(E(\hat\thetav))}f(\theta_{1})\psi(\theta_{1},\ldots,\theta_{\ell })  \Big).
\end{equation}
Using Cauchy-Schwarz and the monotonicity of $\omega$ \ref{it:omegamonoton}, we have
\begin{equation}
||e^{\omega(H/\mu)}z(f)\psi||^{2}\leq \ell  \int e^{2\omega(E(\thetav))} |f(\theta_{1})|^{2} |\psi(\thetav)|^{2}\; d\thetav\\
\leq \ell  \int  e^{2\omega(E(\thetav))}|f(\theta_{1})|^{2} |\psi(\thetav)|^{2}\; d\thetav.
\end{equation}
Applying again Cauchy-Schwarz in the variable $\theta_{1}$, we find
\begin{eqnarray}
||e^{\omega(H/\mu)}z(f)\psi||^{2}&\leq& \ell (\gnorm{f}{2})^2 ||e^{\omega(E(\thetav))}\psi(\thetav)||^{2}_{2}\nonumber\\
&=&\ell  (\gnorm{f}{2})^2 ||e^{\omega(H/\mu)}\psi||^{2}.
\end{eqnarray}
\end{proof}

We want to define an extension of $z,\zd$ to normal-ordered products of these annihilators and creators, which are multilinear operators in a suitable class of ``smearing functions''. Formally this is given by
\begin{equation}
  z^{\dagger m} z^n(f) = \int d^m \theta\, d^n \eta f(\thetav,\etav) \underbrace{\zd(\theta_1)\ldots \zd(\theta_m)z(\eta_1)\ldots z(\eta_n)}_{ =: z^{\dagger m} (\thetav) z^n (\etav)}.
\end{equation}
This is given by the definitions of $z,\zd$ above if $f$ is ``factorizable'', namely if it is of the form $f(\thetav,\etav)=f_1(\theta_1)\ldots f_{m+n}(\eta_n)$, or is a linear combination of such functions. Lechner in \cite[Lemma 4.1.2]{Lechner:2006} extended the definition to arbitrary $f \in L^2(\rbb^{m+n})$.
Actually, the class of ``smearing functions'' that we will need is even more general than this (see Prop.~\ref{pro:zzdcrossnorm}). To define such class we first introduce for a distribution $f\in \dcal(\rbb^{m+n})'$, the (possibly infinite) norms
\begin{align}\label{eq:crossnorm}
\gnorm{f}{m \times n} &:= \sup \Big\{ \big\lvert \! \int f(\thetav,\etav) g(\thetav) h(\etav)  d^m\theta d^n\eta \,\big\rvert : \\
\notag
&\qquad\qquad g \in \dcal(\rbb^m), \, h \in \dcal(\rbb^n), \,  \gnorm{g}{2} \leq 1, \, \gnorm{h}{2} \leq 1\Big\},
\\
\onorm{f}{m \times n} &:= \frac{1}{2} \gnorm{e^{-\omega(E(\thetav))}f(\thetav,\etav)}{m \times n}
+ \frac{1}{2} \gnorm{f(\thetav,\etav)e^{-\omega(E(\etav))}}{m \times n}.\label{eq:crossnorm2}
\end{align}
We also consider
\begin{equation}\label{eq:fullcrossnorm}
  \gnorm{f}{\times} := \sup\big\{ \Big\lvert  \int d^k\thetav f(\thetav) g_1(\theta_1)\cdots g_k(\theta_k) \Big\rvert : g_1,\ldots,g_k \in \dcal(\rbb), \, \gnorm{g_j}{2} \leq 1 \big\}.
\end{equation}
We note that these norms fulfil some properties. First, we have the following Lemma.
\begin{lemma}
If $f_L\in \ccal^\infty(\rbb^m)$, $f_R\in \ccal^\infty(\rbb^n)$, $f_1,\ldots,f_k\in\ccal^\infty(\rbb)$ are bounded, then
\begin{align} \label{eq:mnnormfactor}
  \onorm{f_L(\thetav)f(\thetav,\etav)f_R(\etav)}{m \times n} &\leq \gnorm{f_L}{\infty}  \onorm{f}{m \times n} \gnorm{f_R}{\infty},
\\ \label{eq:crossnormfactor}
  \gnorm{f(\thetav) \prod_j f_j(\theta_j)}{\times} &\leq \gnorm{f}{\times} \prod_j \gnorm{f_j}{\infty}.
\end{align}
\end{lemma}
\begin{proof}
This can be proved by absorbing $f_L, f_R, f_j$ into the test functions $g,h,g_j$, respectively. Let us prove \eqref{eq:mnnormfactor} first.
Applying definition \eqref{eq:crossnorm2}, we find
\begin{equation}
\onorm{f_{L}f f_{R}}{m \times n} = \frac{1}{2} \gnorm{e^{-\omega(E(\thetav))}f_{L}(\thetav)f(\thetav,\etav)f_{R}(\etav)}{m \times n}
+ \frac{1}{2} \gnorm{f_{L}(\thetav)f(\thetav,\etav)f_{R}(\etav)e^{-\omega(E(\etav))}}{m \times n}.\label{proofeq:mnnormfactor}
\end{equation}
We consider the first norm on the right hand side of the above equation. By \eqref{eq:crossnorm}, we have:
\begin{multline}
\gnorm{e^{-\omega(E(\thetav))}f_{L}(\thetav)f(\thetav,\etav)f_{R}(\etav)}{m \times n}\\
 =\sup_{\substack{ \gnorm{g}{2} \leq 1\\ \gnorm{h}{2} \leq 1}} \frac{\big\lvert \! \int e^{-\omega(E(\thetav))}f_{L}(\thetav)f(\thetav,\etav)f_{R}(\etav) g(\thetav) h(\etav)  d^m\theta d^n\eta \,\big\rvert}{\gnorm{g}{2}\gnorm{h}{2}}.
\end{multline}
We call $g'(\thetav):=f_L(\thetav) g(\thetav)$ and $h'(\etav)= f_R(\etav) h(\etav)$.  By multiplying the above equation with $\gnorm{g'}{2}/\gnorm{g'}{2}$ and taking $\gnorm{g'}{2}\leq \gnorm{f_L}{\infty}\gnorm{g}{2}$ into account, we find:
\begin{multline}
\gnorm{e^{-\omega(E(\thetav))}f_{L}(\thetav)f(\thetav,\etav)f_{R}(\etav)}{m \times n} \\
\leq \gnorm{f_L}{\infty}\gnorm{f_R}{\infty}\cdot \sup_{\substack{ \gnorm{g}{2} \leq 1\\ \gnorm{h}{2} \leq 1}} \frac{\big\lvert \! \int f(\thetav,\etav) g'(\thetav) h'(\etav)  d^m\theta d^n\eta \,\big\rvert}{\gnorm{g'}{2}\gnorm{h'}{2}}\\
= \gnorm{f_L}{\infty}\gnorm{f_R}{\infty}\cdot \gnorm{e^{-\omega(E(\thetav))}f(\thetav,\etav)}{m \times n}.
\end{multline}
We can apply the same argument to the second norm on the right hand side of \eqref{proofeq:mnnormfactor}; hence, we find \eqref{eq:mnnormfactor}.

By a similar method one can prove \eqref{eq:crossnormfactor}.
\end{proof}
Another property is that if $g \in \dcal(\rbb^m)$, $g' \in \dcal(\rbb^{m'})$, and if $g \cdot g' \in \dcal(\rbb^{m+m'})$ is the product of $g,g'$ in \emph{independent} variables, then $\gnorm{g\cdot g'}{2} = \gnorm{g}{2} \gnorm{g'}{2}$, and also $\onorm{g\cdot g'}{2} \geq \onorm{g}{2} \gnorm{g'}{2}$ due to monotonicity of $\omega$:
\begin{eqnarray}
\onorm{g\cdot g'}{2} &=& \gnorm{(\thetav ,\thetav')\mapsto e^{\omega(E(\thetav,\thetav'))}g(\thetav)g'(\thetav')}{2}\nonumber\\
&\geq & \gnorm{(\thetav ,\thetav')\mapsto e^{\omega(E(\thetav))}g(\thetav)g'(\thetav')}{2}\nonumber\\
&=& \onorm{g}{2} \gnorm{g'}{2}.\label{proofggponorm}
\end{eqnarray}
This gives the following Lemma:
\begin{lemma}
\begin{equation} \label{eq:mnnormproduct}
   \onorm{f \cdot f'}{(m+m')\times (n+n')} \leq \onorm{f}{m \times n} \gnorm{f'}{m' \times n'}.
\end{equation}
\end{lemma}
\begin{proof}
Applying definition \eqref{eq:crossnorm2}, we have
\begin{multline}
\onorm{f \cdot f'}{(m+m')\times (n+n')}
=\frac{1}{2} \gnorm{e^{-\omega(E(\thetav,\thetav'))}(f\cdot f')(\thetav,\thetav',\etav,\etav')}{(m+m')\times (n+n')} \\
+ \frac{1}{2} \gnorm{(f\cdot f')(\thetav,\thetav',\etav,\etav')e^{-\omega(E(\etav,\etav'))}}{(m+m')\times (n+n')}.\label{proofonormffp}
\end{multline}
We consider the first norm on the right hand side of the above equation. By \eqref{eq:crossnorm}, we have:
\begin{multline}
\gnorm{e^{-\omega(E(\thetav,\thetav'))}(f\cdot f')(\thetav,\thetav',\etav,\etav')}{(m+m')\times (n+n')}\\
 =\sup_{\substack{ \gnorm{g}{2} \leq 1\\ \gnorm{h}{2} \leq 1}} \frac{\big\lvert \! \int e^{-\omega(E(\thetav,\thetav'))}f(\thetav,\etav)f'(\thetav',\etav')g(\thetav,\thetav') h(\etav,\etav')  d^{m}\thetav d^{m'}\thetav' d^{n}\etav d^{n'}\etav' \,\big\rvert}{\gnorm{g}{2}\gnorm{h}{2}}.\label{proofgnormffp}
\end{multline}
Referring to \cite[Prop.~2.6.12]{KadRin:algebras1}, we can consider the special case $g(\thetav,\thetav')=g'(\thetav)\cdot g''(\thetav')$ (the same for $h$); Indeed, \cite[Prop.~2.6.12]{KadRin:algebras1} tells us that the supremum over the special functions which are factorizable equals the supremum over more general functions in $L^2$. Hence, using \eqref{proofggponorm}, we find
\begin{equation}
\begin{aligned}
\text{ r.h.s.}&\eqref{proofgnormffp} \\
& = \sup_{\substack{ \gnorm{g'\cdot g''}{2} \leq 1\\ \gnorm{h'\cdot h''}{2} \leq 1}} \frac{\big\lvert \! \int e^{-\omega(E(\thetav))}f(\thetav,\etav)f'(\thetav',\etav')g'(\thetav) g''(\thetav') h'(\etav) h''(\etav')  d^{m}\thetav d^{m'}\thetav' d^{n}\etav d^{n'}\etav' \,\big\rvert}{\gnorm{g'}{2}\gnorm{g''}{2}\gnorm{h'}{2}\gnorm{h''}{2}}\\
&=\sup_{\substack{ \gnorm{g'}{2} \neq 0\\ \gnorm{h'}{2} \neq 0}} \frac{\big\lvert \! \int e^{-\omega(E(\thetav))}f(\thetav,\etav)g'(\thetav) h'(\etav) d^{m}\thetav d^{n}\etav \,\big\rvert}{\gnorm{g'}{2}\gnorm{h'}{2}} \times \\
 &\quad\times \sup_{\substack{ \gnorm{g''}{2} \neq 0\\ \gnorm{h''}{2} \neq 0}} \frac{\big\lvert \! \int f'(\thetav',\etav')g''(\thetav') h''(\etav')  d^{m'}\thetav' d^{n'}\etav' \,\big\rvert}{\gnorm{g''}{2}\gnorm{h''}{2}}\\
&= \onorm{f}{m \times n} \gnorm{f'}{m' \times n'}.
\end{aligned}
\end{equation}
We can apply the same argument to the second norm on the right hand side of \eqref{proofonormffp}; hence, we find \eqref{eq:mnnormproduct}.
\end{proof}
Finally, we have the following Lemma:
\begin{lemma}
\begin{equation}\label{eq:crossnormcomparison}
   \gnorm{f(\thetav) \prod_{j=1}^{m+n}e^{-\omega(\cosh \theta_{j})} }{\times} \leq \gnorm{f}{m \times n}^{\omega} \leq \frac{1}{2}\Big(\gnorm{e^{-\omega(E(\thetav))}f(\thetav,\etav)}{2}+\gnorm{e^{-\omega(E(\etav))}f(\thetav,\etav)}{2}\Big),
\end{equation}
\end{lemma}
\begin{proof}
As for the left inequality in \eqref{eq:crossnormcomparison}, using the monotonicity and the sublinearity of $\omega$:
\begin{equation}
\prod_{j=1}^{m+n}e^{-\omega(\cosh\theta_{j})}\leq e^{-\omega(E(\thetav))}\leq e^{-\omega(\sum_{j=1}^{m}\cosh\theta_{j})}, \quad \prod_{j=1}^{m+n}e^{-\omega(\cosh\theta_{j})}\leq e^{-\omega(\sum_{j=m+1}^{m+n}\cosh\theta_{j})},
\end{equation}
we find:
\begin{equation}
 \gnorm{f(\thetav) \prod_{j=1}^{m+n}e^{-\omega(\cosh \theta_{j})} }{\times} \leq \frac{1}{2}\gnorm{f(\thetav)e^{-\omega(\sum_{j=1}^{m}\cosh \theta_{j})} }{\times}+\frac{1}{2}\gnorm{f(\thetav) e^{-\omega(\sum_{j=m+1}^{m+n}\cosh \theta_{j})} }{\times}.\label{proofleftineq}
\end{equation}
Since the factorizable functions are a special case of the larger set of $L^2$ functions, and the supremum over a larger set of functions is larger than the supremum over a smaller set of functions, we have from definition \eqref{eq:crossnorm}:
\begin{equation}
\text{ r.h.s. }\eqref{proofleftineq}\leq \frac{1}{2}\gnorm{f(\thetav)e^{-\omega(\sum_{j=1}^{m}\cosh \theta_{j})} }{m \times n}+ \frac{1}{2}\gnorm{f(\thetav)e^{-\omega(\sum_{j=m+1}^{m+n}\cosh \theta_{j})} }{m \times n}.
\end{equation}
The right hand side of the equation above is \eqref{eq:crossnorm2}. This implies the left inequality in \eqref{eq:crossnormcomparison}.

The right inequality in \eqref{eq:crossnormcomparison} is a consequence of the application of the Cauchy-Schwarz inequality to definition \eqref{eq:crossnorm} in the case $\omega=0$, and to definition \eqref{eq:crossnorm2} in the case $\omega \neq 0$:
\begin{multline}
\gnorm{f}{m \times n}^{\omega}
\leq \frac{1}{2}\sup_{\substack{ \gnorm{g}{2} \leq 1\\ \gnorm{h}{2} \leq 1}} \frac{\Big(\int d^m\theta d^n\eta\;| e^{-\omega(E(\thetav))}f(\thetav,\etav)|^{2}\Big)^{1/2} \gnorm{g}{2} \gnorm{h}{2}}{\gnorm{g}{2}\gnorm{h}{2}}\\
+\frac{1}{2}\sup_{\substack{ \gnorm{g}{2} \leq 1\\ \gnorm{h}{2} \leq 1}} \frac{\Big(\int d^m\theta d^n\eta\;| e^{-\omega(E(\etav))}f(\thetav,\etav)|^{2}\Big)^{1/2} \gnorm{g}{2} \gnorm{h}{2}}{\gnorm{g}{2}\gnorm{h}{2}}\\
=\frac{1}{2}\gnorm{e^{-\omega(E(\thetav))}f(\thetav,\etav)}{2} + \frac{1}{2}\gnorm{e^{-\omega(E(\etav))}f(\thetav,\etav)}{2}.
\end{multline}
\end{proof}
However, equality in \eqref{eq:crossnormcomparison} does in general not hold: As a counterexample in the case $\omega=0$, consider $f(\theta_1,\theta_2,\theta_3) = \delta(\theta_1-\theta_2)/(1+\theta_3^2) + \delta(\theta_1-\theta_3)/(1+\theta_2^2)$;
then $\gnorm{f}{\times}$ and $\onorm{f}{1 \times 2}$ are finite but $\onorm{f}{2 \times 1}$ and $\gnorm{f}{2}$ are infinite, as we can see from the following direct computation using the definitions \eqref{eq:fullcrossnorm}, \eqref{eq:crossnorm2}:

By \eqref{eq:fullcrossnorm}, we compute
\begin{equation}
\gnorm{f}{\times} = \sup\big\{ \Big\lvert  \int d^3\thetav f(\thetav) g_1(\theta_1)g_2(\theta_2)g_3(\theta_3) \Big\rvert : g_j \in \dcal(\rbb), \, \gnorm{g_j}{2} \leq 1 \big\}.
\end{equation}
The absolute value of $f$ integrated with the functions $g_j$ can be estimated using the Cauchy-Schwarz inequality:
\begin{multline}
\Big\lvert  \int d^3\thetav f(\thetav) g_1(\theta_1)g_2(\theta_2)g_3(\theta_3) \Big\rvert\\
\leq \prod_{j=1}^{3}\gnorm{g_j}{2}\Big(\gnorm{\theta_3 \rightarrow 1/(1+\theta_3^2)}{2} + \gnorm{\theta_2 \rightarrow 1/(1+\theta_2^2)}{2} \Big)<\infty.
\end{multline}
By definition \eqref{eq:crossnorm2} with $\omega=0$, we have:
\begin{equation}
\onorm{f}{1 \times 2} = \sup \Big\{ \big\lvert \! \int f(\thetav) g(\theta_1) h(\theta_2,\theta_3)  d^3\thetav \,\big\rvert : g \in \dcal(\rbb^1), \, h \in \dcal(\rbb^2), \,  \gnorm{g}{2} \leq 1, \, \gnorm{h}{2} \leq 1\Big\}.
\end{equation}
The absolute value of $f$ integrated with the functions $g,h$ can be estimated using the Cauchy-Schwarz inequality:
\begin{equation}
\big\lvert \! \int f(\thetav) g(\theta_1) h(\theta_2,\theta_3)  d^3\thetav \,\big\rvert
\leq \gnorm{g}{2}\gnorm{h}{2}\Big(\gnorm{\theta_3 \rightarrow 1/(1+\theta_3^2)}{2} + \gnorm{\theta_2 \rightarrow 1/(1+\theta_2^2)}{2} \Big)<\infty.
\end{equation}
Clearly $\gnorm{f}{2}$ is not finite due to the delta distributions. As for $\onorm{f}{2 \times 1}$ with $\omega=0$, we compute:
\begin{equation}
\onorm{f}{2 \times 1} = \sup \Big\{ \big\lvert \! \int f(\thetav) g(\theta_1,\theta_2) h(\theta_3)  d^3\thetav \,\big\rvert : g \in \dcal(\rbb^1), \, h \in \dcal(\rbb^2), \,  \gnorm{g}{2} \leq 1, \, \gnorm{h}{2} \leq 1\Big\}.
\end{equation}
The absolute value of $f$ integrated with the functions $g,h$ can be estimated using the Cauchy-Schwarz inequality:
\begin{multline}
\big\lvert \! \int f(\thetav) g(\theta_1,\theta_2) h(\theta_3)  d^3\thetav \,\big\rvert\\
\leq  \big\lvert \! \int d\theta_1 g(\theta_1,\theta_1) \,\big\rvert \cdot \gnorm{h}{2} \gnorm{\theta_3 \rightarrow 1/(1+\theta_3^2)}{2} + \gnorm{h}{2} \gnorm{g}{2} \gnorm{\theta_2 \rightarrow 1/(1+\theta_2^2)}{2}.
\end{multline}
which is infinite because $\sup \Big\{ \big\lvert \! \int g(\theta_1,\theta_1) d\theta_1 \,\big\rvert : g \in \dcal(\rbb^1), \,  \gnorm{g}{2} \leq 1\Big\}$ is not finite.

We now define the multilinear annihilation and creation operators $z^{\dagger m} z^n(f)$ as follows.
For an arbitrary distribution $f \in \dcal(\rbb^{m+n})'$ and with vectors $\psi \in \Hil_k \cap \dcal(\rbb^k)$, $\chi \in \Hil_l \cap \dcal(\rbb^l)$, we set:
\begin{equation}\label{eq:zzqf}
   \hscalar{\chi}{ z^{\dagger m} z^n(f) \psi} =
    \frac{\sqrt{k!(k-n+m)!}}{(k-n)!}
   \int d^{k-n}\lambda\, d^m \theta\, d^n \eta \, \bar \chi(\thetav,\lambdav) f(\thetav,\etav) \psi(\eta_n\ldots\eta_1,\lambdav)
\end{equation}
if $\ell=k-n+m$, and $=0$ otherwise. Because of the relation $(z(\eta)\Psi_{k})(\thetav)=\sqrt{k}\cdot \Psi_{k}(\eta,\thetav)$, this extends the previous definition of the annihilators and creators. Now the question is whether the quadratic form \eqref{eq:zzqf} can be extended to $\fpno \times \fpno$, or even to an (unbounded) operator on $\fpno$. A sufficient condition for that is $\gnorm{f}{m \times n}^\omega < \infty$, as the following proposition shows.

\begin{proposition} \label{pro:zzdcrossnorm}

If $f \in \dcal(\rbb^{m+n})'$ with $\|f\|^{\omega}_{m \times n} < \infty$, then $z^{\dagger m} z^n(f)$ extends to an operator on $\fpno$, and
\begin{equation}
\big\|  z^{\dagger m} z^n(f)e^{- \omega(H/\mu)} Q_k \big\| \leq 2 \frac{\sqrt{k!(k-n+m)!}}{(k-n)!}\|f\|^{\omega}_{m\times n}.\label{eq:cross}
\end{equation}
Moreover,
\begin{equation}
\gnorm{ z^{\dagger m} z^n(f) }{k}^\omega \leq 2 \frac{k!}{(k-\max(m,n))!}\|f\|^{\omega}_{m\times n}.\label{eq:zzdanorm}
\end{equation}
\end{proposition}

\begin{proof}

For $\psi \in\mathcal{H}_{k} \cap \dcal(\rbb^k)$ and $\chi\in \mathcal{H}_{l}\cap \dcal(\rbb^l)$, with $l=k-m+n$, one has by the definition in \eqref{eq:zzqf},
\begin{equation}\label{eq:zzmatrixelm}
\begin{aligned}
\left| \langle  \chi| z^{\dagger m} z^n(f)|\psi \rangle\right|
& =\frac{\sqrt{k!(k-n+m)!}}{(k-n)!}\left| \int d\thetav d\etav d\lambdav\; \overline{\chi(\thetav,\lambdav)}\psi(\eta_n\ldots\eta_1,\lambdav)f(\thetav,\etav)\right|\\
& \leq \frac{\sqrt{k!(k-n+m)!}}{(k-n)!}\int d\lambdav \left| \int d\thetav d\etav\; \overline{\chi(\thetav,\lambdav)}\psi(\eta_n\ldots\eta_1,\lambdav)f(\thetav,\etav)\right|\\
& \leq  2 \frac{\sqrt{k!(k-n+m)!}}{(k-n)!} \gnorm{f}{m \times n}^{\omega} \gnorm{\chi}{2}\Big(\int d\lambdav\int d\etav\; |\psi(\etav,\lambdav)|^{2}e^{2\omega(E(\etav))}\Big)^{1/2}\\
& \leq 2 \frac{\sqrt{k!(k-n+m)!}}{(k-n)!}\onorm{f}{m \times n}  \gnorm{\chi}{2} \gnorm{e^{\omega(H/\mu)}\psi}{2}
\end{aligned}
\end{equation}
where in the third inequality we made use of the estimate
\begin{equation}
\left| \int d\thetav d\etav\; \overline{\chi(\thetav,\lambdav)}\psi(\eta_n\ldots\eta_1,\lambdav)f(\thetav,\etav)\right|\leq 2 \gnorm{f}{m \times n}^{\omega} \gnorm{\chi(\cdot,\lambdav)}{2} \gnorm{e^{\omega(E(\cdotarg))}\psi(\cdot,\lambdav)}{2}.\label{proofmulticreann}
\end{equation}
and of the Cauchy-Schwarz inequality, and in the fourth inequality we made use of the relation $e^{2\omega(E(\etav))}\leq e^{2\omega(E(\etav,\lambdav))}$.

Remark: The estimate \eqref{proofmulticreann} follows from \eqref{eq:crossnorm2}:
\begin{multline}
2\onorm{f}{m\times n}\geq \gnorm{e^{-\omega(E(\etav))}f(\thetav,\etav)}{m\times n}\\
=\sup_{\substack{ \gnorm{\chi(\cdot,\lambdav)}{2}\leq 1 \\ \gnorm{\psi(\cdot,\lambdav)}{2}\leq 1}} \frac{\big\lvert \! \int f(\thetav,\etav)e^{-\omega(E(\etav))}\overline{\chi(\thetav,\lambdav)}\psi(\eta_n,\ldots,\eta_1,\lambdav)   d^m\theta d^n\eta \,\big\rvert }{\gnorm{\chi(\cdot,\lambdav)}{2}\gnorm{\psi(\cdot,\lambdav)}{2}}\\
=\sup_{\substack{ \gnorm{\chi(\cdot,\lambdav)}{2}\leq 1 \\ \gnorm{e^{\omega(E(\cdot))}\psi'(\cdot,\lambdav)}{2} \leq 1}} \frac{\big\lvert \! \int f(\thetav,\etav)\overline{\chi(\thetav,\lambdav)}\psi'(\eta_n,\ldots,\eta_1,\lambdav)  d^m\theta d^n\eta \,\big\rvert }{\gnorm{\chi(\cdot,\lambdav)}{2}\gnorm{e^{\omega(E(\cdot))}\psi'(\cdot,\lambdav)}{2}}.\label{proofmulticreann2}
\end{multline}
by shifting the denominator in the last line of \eqref{proofmulticreann2} to the left hand side of the equation (here we called $\psi'(\eta_n,\ldots,\eta_1,\lambdav):=e^{-\omega(E(\etav))}\psi(\eta_n,\ldots,\eta_1,\lambdav)$).

Since $\psi$ and $\chi$ were chosen from dense sets in the corresponding  spaces, and since the matrix elements \eqref{eq:zzmatrixelm} vanish if $\ell \neq k-n+m$, we can extend $z^{\dagger m} z^n(f)$ to a bounded operator on $\hcal_k^\omega$ with norm
\begin{equation}\label{eq:pkbound}
\left\|z^{\dagger m} z^n(f)P_{k}e^{-\omega(H/\mu)}  \right\|\leq 2\frac{\sqrt{k!(k-n+m)!}}{(k-n)!}\gnorm{f}{m\times n}^{\omega}.
\end{equation}
This works for any $k$. For $k\neq k'$, the images of $z^{\dagger m} z^n(f)P_{k}$ and $z^{\dagger m} z^n(f)P_{k'}$ are orthogonal; hence \eqref{eq:cross} follows from \eqref{eq:pkbound} using Pythagoras' theorem. Explicitly: For $\Psi:=\sum_{j=1}^{k}\Psi_{j}$,
\begin{multline}
\left\|   \int d^{m}\theta d^{n}\eta f(\theta,\eta)z^{\dagger m}(\theta)z^{n}(\eta)e^{-\omega(H/\mu)}\Psi\right\|^{2}\\
=\sum_{j=1}^{k}\left\| \int d\theta d\eta f(\theta,\eta)z^{\dagger m}(\theta)z^{n}(\eta)e^{-\omega(H/\mu)}\Psi_{j}\right\|^{2}.\label{eq:pythagorasQP}
\end{multline}
Using \eqref{eq:pkbound}, we find:
\begin{equation}
\text{l.h.s.}\eqref{eq:pythagorasQP}\leq \sum_{j=1}^{k}\frac{j!(j-n+m)!}{((j-n)!)^{2}}(2\gnorm{f}{m \times n}^{\omega})^{2}||\Psi_{j}||_{2}^{2}.
\end{equation}
Hence, we have:
\begin{eqnarray}
\text{l.h.s.} \eqref{eq:pythagorasQP} &\leq& \frac{k!(k-n+m)!}{((k-n)!)^{2}}(2\gnorm{f}{m \times n}^{\omega})^{2}\sum_{j=1}^{k}||\Psi_{j}||_{2}^{2}\\
&=&\frac{k!(k-n+m)!}{((k-n)!)^{2}}(2\gnorm{f}{m \times n}^{\omega})^{2} ||\Psi||_{2}^{2}.\nonumber
\end{eqnarray}
by application of Pythagoras.
This implies:
\begin{equation}
\left\|  \int d^{m}\theta d^{n}\eta f(\theta,\eta)z^{\dagger m}(\theta)z^{n}(\eta)e^{-\omega(H/\mu)}Q_{k}\right\|^{2}\leq \frac{k!(k-n+m)!}{((k-n)!)^{2}}(2||f||_{m\times n}^{\omega})^{2}.\label{crossq}
\end{equation}
This concludes the proof of Eq.\eqref{eq:cross}.

To prove \eqref{eq:zzdanorm}, we compute for $n\geq m$, using \eqref{eq:cross},
\begin{eqnarray}\label{eq:zzdenergynorm}
 \gnorm{Q_k z^{\dagger m} z^n(f) e^{-\omega(H/\mu)} Q_k}{} &\leq&  2 \frac{\sqrt{k!}\sqrt{k!}}{(k-n)!}  \gnorm{f}{m\times n}^{\omega}\nonumber\\
 &=& 2\frac{k!}{(k-\max(m,n))!}\gnorm{f}{m\times n}^{\omega}.
\end{eqnarray}
Similarly, using \eqref{eq:cross}, for $m>n$ we have
\begin{eqnarray}
\text{l.h.s.\eqref{eq:zzdenergynorm}} &\leq&  \gnorm{Q_k z^{\dagger m} z^n(f) e^{-\omega(H/\mu)} Q_{k-m+n}}{}\nonumber\\
&\leq& 2\frac{\sqrt{(k-m+n)!}\sqrt{(k-m+n-n+m)!}}{(k-m+n-n)!}\gnorm{f}{m\times n}^{\omega}\nonumber\\
&\leq& 2\frac{k!}{(k-m)!}\gnorm{f}{m\times n}^{\omega}\nonumber\\
&=& 2\frac{k!}{(k-\max(m,n))!}\gnorm{f}{m\times n}^{\omega}.
\end{eqnarray}
Moreover, we note that, in the sense of quadratic forms, $(z^{\dagger m} z^n(f))^\ast = z^{\dagger n} z^m(f^\ast)$, where $f^\ast(\thetav,\etav) = \overline{f(\eta_n,\ldots,\eta_1,\theta_m,\ldots,\theta_1)}$, and where one finds $\gnorm{f^\ast}{n \times m} = \gnorm{f}{m \times n}$.
Another application of \eqref{eq:zzdenergynorm} then gives
\begin{equation}
 \gnorm{Q_k e^{-\omega(H/\mu)} z^{\dagger m} z^n(f)  Q_k}{}
 = \gnorm{Q_k z^{\dagger n} z^m(f^\ast) e^{-\omega(H/\mu)} Q_k}{}
\leq 2 \frac{k!}{(k-m)!}  \gnorm{f}{m\times n}^{\omega},
\end{equation}
and thus \eqref{eq:zzdanorm} is proven.
\end{proof}
Remark: The equality  $\gnorm{f^\ast}{n \times m} = \gnorm{f}{m \times n}$ follows from a short computation:
\begin{multline}
\gnorm{f^{*}}{n\times m}=\sup_{\substack{ \gnorm{g}{2}\leq 1 \\ \gnorm{h}{2}\leq 1}} \frac{\big\lvert \! \int f^{*}(\thetav,\etav) g(\etav) h(\thetav)  d^m\eta d^n\theta \,\big\rvert }{\gnorm{g}{2}\gnorm{h}{2}}\\
=\sup_{\substack{ \gnorm{g}{2}\leq 1 \\ \gnorm{h}{2}\leq 1}}  \frac{\big\lvert \! \int \overline{f(\eta_m,\ldots,\eta_1,\theta_n, \ldots ,\theta_1)} g(\etav) h(\thetav)  d^m\eta d^n\theta \,\big\rvert }{\gnorm{g}{2}\gnorm{h}{2}}\\
=\sup_{\substack{ \gnorm{g}{2}\leq 1 \\ \gnorm{h}{2}\leq 1}}  \frac{\big\lvert \! \overline{\int f(\thetav,\etav) \overline{g(\theta_m, \ldots ,\theta_1)} \overline{ h(\eta_n, \ldots ,\eta_1)}  d^m\theta d^n\eta }\,\big\rvert }{\gnorm{g}{2}\gnorm{h}{2}}\\
=\sup_{\substack{ \gnorm{g'}{2}\leq 1 \\ \gnorm{h'}{2}\leq 1}}  \frac{\big\lvert \! \int f(\thetav,\etav) g'(\thetav)  h'(\etav)  d^n\eta d^m\theta \,\big\rvert }{\gnorm{g'}{2}\gnorm{h'}{2}}\\
=\gnorm{f}{m \times n}.
\end{multline}
where in the fourth equality we called $g'(\thetav):=\overline{g(\theta_m, \ldots ,\theta_1)}$ (similar definition for $h'$) and we used that $\gnorm{\overline{g(\theta_m, \ldots ,\theta_1)}}{2}=\gnorm{g}{2}$ (same for $h$).

\section{Fields and local operators} \label{sec:fieldloc}

Following \cite{Lechner:2006}, and analogous to free field theory, we can define a quantum field $\phi$ as, $f\in \mathcal{S}(\mathbb{R}^{2})$,
\begin{equation}
\phi(f):= \zd(f^{+}) + z(f^{-}).
\end{equation}
As shown in \cite[Proposition 4.2.2]{Lechner:2006}, this field has similar mathematical properties to the free scalar field: It is defined on $\fpn$, and essentially selfadjoint for real-valued $f$. Moreover, $\phi$ has the Reeh-Schlieder property, transforms covariantly under the representation $U(x,\lambda)$ of $\mathcal{P}_{+}$, and it solves the Klein-Gordon equation.

However, $\phi$ is strictly local only if $S=1$. For generic $S$, the field is only localized in an infinitely extended wedge -- rather than at a space-time point -- in the following sense. Let us introduce the ``reflected'' Zamolodchikov operators, $\psi \in \mathcal{H}_{1}$,
\begin{equation}
z(\psi)':=Jz(\psi)J, \quad \zd(\psi)':= J \zd(\psi)J,
\end{equation}
and define another field $\phi'$ as, $f\in \mathcal{S}(\mathbb{R}^{2})$,
\begin{equation}
\phi'(f):=J\phi(f^{j})J,\quad f^{j}(x):=\overline{f(-x)}.
\end{equation}
It has been shown in \cite[Proposition 4.2.6]{Lechner:2006} that the two fields $\phi, \phi'$ are relatively \emph{wedge-local:} For real-valued test functions $f,g$ with $\supp f \subset \mathcal{W}'$ and $\supp g \subset \mathcal{W}$, one finds that $[e^{i\phi(f)^-},e^{i\phi'(g)^-}]=0$. Hence, we can understand $\phi(x)$ and $\phi'(y)$ as being localized in the shifted left wedge $\mathcal{W}'_{x}$ and in the shifted right wedge $\mathcal{W}_{y}$, respectively.

This result is obtained by computing the commutation relations of $z,z^\dagger$ with $z', z^{\dagger \prime}$ \cite[Lemma 4.2.5]{Lechner:2006}: Let $g\in \mathcal{H}_{1}$. The following holds in the sense of operator-valued distributions on $\fpn$:
\begin{align}
 [z(\overline{g})',z^{\dagger}(\theta)]&=B^{g,\theta},
 &
[z^{\dagger}(\overline{g})',z(\theta)]&=-(B^{\bar g,\theta})^{*} \label{comzpz}
\\
[z(g)',z(\theta)]&=0, & [z^{\dagger}(g)',z^{\dagger}(\theta)]&=0,\label{zzp}
\end{align}
where $B^{g,\theta}=\oplus_{n=0}^{\infty}B_{n}^{g,\theta}$ and $B_{n}^{g,\theta}$ acts on the $n$-particle Hilbert space as a multiplication operator:
\begin{equation}\label{multop}
B_{n}^{g,\theta}(\theta_{1},\ldots,\theta_{n})=g(\theta)\prod_{j=1}^{n}S(\theta-\theta_{j}).
\end{equation}
Instead of working with unbounded (closed) operators, we can also work with associated von Neumann algebras: We define the ``\emph{wedge algebra}'' as:
\begin{equation}
\M = \{e^{i\phi(f)^-} \,|\, f \in \scal_\rbb(\rbb^2), \, \supp f \subset \wcal' \}'.
\end{equation}
Remark: we can restrict this definition to smaller sets $f\in \mathcal{D}_{\mathbb{R}}(\mathcal{W}')$, or even to $f\in \mathcal{D}_{\mathbb{R}}^{\omega}(\mathcal{W}')$. This does not change $\M$ since $\mathcal{D}_{\mathbb{R}}^{\omega}(\mathcal{W}')$ is dense in $\mathcal{D}_{\mathbb{R}}(\mathcal{W}')$ in the $\mathcal{D}$-topology (see \cite{Bjoerck:1965}) and $\mathcal{D}_{\mathbb{R}}(\mathcal{W}')$ is dense in $ \scal_\rbb(\rbb^2)$ with respect to test functions with support in $\wcal'$. Moreover, the set of operators $e^{i\phi(f)}$ with $f$ in these restricted domains is dense in $\M$ because the map $f\mapsto e^{i\phi(f)}$ is continuous in the strong operator topology (see for example \cite{ReedSimon:1975-2}).

We can extend this definition to define algebras associated with any wedge in $\rbb^2$: As shown in \cite[Proposition 4.4.1]{Lechner:2006}, the triple $(\M,U(x),\mathcal{H})$ satisfies the defining properties of a standard right wedge algebra in the sense of \cite[Definition 2.1.1]{Lechner:2006} and the associated map $\mathcal{W} \mapsto \mathcal{A}(\mathcal{W})$ (where here $\mathcal{W}$ is a generic wedge) is a local net of von Neumann algebras with the properties in \cite[Proposition 4.4.1]{Lechner:2006}.

We can extend this definition to bounded regions by taking intersections of wedge algebras. Namely, the local algebra of a \emph{double cone} $\mathcal{O}_{x,y}=\mathcal{W}_{x}\cap \mathcal{W}_{y}'$, $x,y \in \mathbb{R}^{2}$, $y-x\in \mathcal{W}$, is defined as
\begin{equation}
\mathcal{A}(\mathcal{O}_{x,y}):=\mathcal{A}(\mathcal{W}_{x})\cap \mathcal{A}(\mathcal{W}_{y})'.
\end{equation}
It has been shown in \cite{Lechner:2006} that $\mathcal{O}\mapsto \mathcal{A}(\mathcal{O})$, where
\begin{equation}
\mathcal{A}(\mathcal{O}):= \Big( \bigcup_{\mathcal{O}_{x,y}\subset \mathcal{O}}\mathcal{A}(\mathcal{O}_{x,y}) \Big)'',
\end{equation}
is a covariant, local net of von Neumann algebras fulfilling the standard axioms of a local quantum field theory. Here it is not \emph{a priori} clear that the algebras $\A(\ocal)$ are nontrivial, i.e., that they contain any operator except for multiples of the identity. However, Lechner proved \cite{Lechner:2008} that at least for regions $\ocal$ of a certain minimum size,  the vacuum vector $\Omega$ is indeed cyclic for $\A(\ocal)$, of which it follows that the algebras are type $\mathrm{III}_1$ factors \cite{BuchholzLechner:2004}.

%% file: arakiexpansion.tex
\chapter{The Araki expansion} \label{sec:araki}

We consider quantum field theory on 1+1 dimensional Minkowski space. We know that in the case of a real scalar free field, any operator $A$ on Fock space can be decomposed as
\begin{equation}\label{eq:expansionfree2}
  A = \sum_{m,n=0}^\infty \int \frac{d^m \theta \, d^n \eta}{m!n!} f_{m,n}(\thetav,\etav) a^{\dagger}(\theta_1)\cdots a^\dagger(\theta_m) a(\eta_1)\cdots a(\eta_n),
\end{equation}
where $\theta_j$, $\eta_j$ are rapidities and where the (generalized) functions $f_{m,n}$ can be written down explicitly in terms of a string of nested commutators:
\begin{equation}
 f_{m,n}(\thetav,\etav) = \bighscalar{\Omega}{ [a(\theta_1),[ \ldots a(\theta_m), [\ldots [A, \, a^\dagger(\eta_n)] \ldots a^\dagger(\eta_1)] \ldots ] \Omega}.
\end{equation}
Araki has shown in \cite{Ara:lattice} (in a different notation) that every bounded operators $A$ has such decomposition.

In the following section we aim to establish an analogue of the series expansion \eqref{eq:expansionfree2} in terms of the deformed creators and annihilators $z,\zd$ in our models with factorizing scattering matrix. Moreover, we aim to establish this expansion for arbitrary bounded operators, and more generally for unbounded operators and quadratic forms. This is an important ingredient for a characterization theorem for local operators which we will formulate in Sec.~\ref{sec:localitythm}.

\section{Contractions}

In this section we will introduce some of our notation, similar to \cite{Lechner:2008} but with conventions slightly more convenient for our purposes.

We consider for $A \in \qf^\omega$ the matrix element
\begin{equation}
 \hscalar{\zd(\theta_1) \cdots \zd(\theta_m) \Omega}{A \zd(\eta_n) \cdots \zd(\eta_1) \Omega}
  =: \hscalar{\lvector{}{\thetav} }{ A \rvector{}{\etav}},
\end{equation}
A \emph{contraction} $C$ is a triple $C=(m,n,\{(l_1,r_1),\ldots,(l_{k},r_k)\})$, where $m,n\in\nbb_0$, $1 \leq l_j \leq m$ and $m+1 \leq r_j \leq m+n$, and both the $l_j$ and the $r_j$ are pairwise different among each other. We denote $\ccal_{m,n}$ the set of all contractions for fixed $m$ and $n$, and write $|C|:=k$ for the length of a contraction (in other words, the number of elements of the set in the third entry in the definition of $C$).

Using this notion of a contraction, we can consider (``contracted'') matrix elements $\hscalar{\lvector{C}{\thetav}}{ A \rvector{C}{\etav}}$,
where
\begin{align}
   \lvector{C}{\thetav} &:= \zd(\theta_1) \cdots \widehat{\zd(\theta_{l_1})} \cdots \widehat{\zd(\theta_{l_{|C|}})} \cdots \zd(\theta_m) \Omega,
\\
   \rvector{C}{\etav} &:= \zd(\eta_n) \cdots \widehat{\zd(\eta_{r_1-m})} \cdots \widehat{\zd(\eta_{r_{|C|}-m})} \cdots \zd(\eta_1) \Omega,
\end{align}
and where the hats indicate that the marked elements have been omitted in the sequence.

We note that $\lvector{C}{\cdotarg}$ is an $\Hil$-valued distribution on $\dcal(\rbb^{m-|C|})$, namely when smearing each $\zd$ with test functions in $\dcal(\rbb)$, $\lvector{C}{f}$ is a vector in $\Hil$. Actually, its values are in $\fpno$; indeed, given any function $f$ smooth and of compact support, $\lvector{C}{f}:=\int d^m\theta \, f(\thetav) \lvector{C}{\thetav}$ is a vector of finite particle number, and has the norm
\begin{equation}
   \gnorm{e^{\omega(H/\mu)} \lvector{C}{f}}{} \leq  \sqrt{(m-|C|)!}\onorm{f}{2}.
\end{equation}
This inequality is a generalization of \eqref{omegaz} first part, in the case $\ell =0$. Namely, we can follow the same computation as in the proof of \eqref{omegaz}, but setting $\ell =0$ and considering ${\zd}^{m}(f)$ instead of $\zd(f)$. By explicit computation:
\begin{multline}
 \gnorm{e^{\omega(H/\mu)} \lvector{C}{f}}{}^{2}
 =\int d^{m-|C|} \hat{\theta} d^{m'-|C|} \hat{\theta'} \overline{f(\hat{\thetav}')} f(\hat{\thetav})
 \langle e^{\omega(H/\mu)}\zd(\theta'_{1}) \cdots \\
    \cdots \widehat{\zd(\theta'_{l_1})} \cdots \widehat{\zd(\theta'_{l_{|C|}})} \cdots \zd(\theta'_{m}) \Omega, e^{\omega(H/\mu)}\zd(\theta_1) \cdots \widehat{\zd(\theta_{l_1})} \cdots \widehat{\zd(\theta_{l_{|C|}})} \cdots \zd(\theta_m) \Omega  \rangle\\
 \leq (m-|C|)!\int d^{m-|C|}\hat{\theta} |f(\hat{\thetav})|^2 e^{2\omega(\sum_{j=1}^{m}\cosh\widehat{\theta_j})}\\
=(m-|C|)!(\onorm{f}{2})^2.\label{prooflC}
\end{multline}
where in the second inequality we used the $S$-symmetry of $f$.

This holds similarly for $\rvector{C}{\cdotarg}$. Therefore, for fixed $A \in \qf^\omega$, the matrix element $\hscalar{\lvector{C}{\thetav}}{ A \rvector{C}{\etav}}$ is a well-defined distribution on $\dcal(\rbb^{m+n-2|C|})'$.

We associate with a contraction $C\in \ccal_{m,n}$ the following quantities:
\begin{align}
\delta_{C}(\thetav,\etav) &:= \prod_{j=1}^{|C|}\delta(\theta_{l_j}-\eta_{r_j-m}), \label{eq:deltac}\\
 S_C(\thetav,\etav) &:= \prod_{j=1}^{|C|}\prod_{m_{j}=l_{j}+1}^{r_{j}-1}S^{(m)}_{m_{j},l_{j}}\cdot \prod_{\substack{r_{i}<r_{j} \\ l_{i}<l_{j}}} S^{(m)}_{l_{j},r_{i}},\label{eq:sc}
\end{align}
where we used the notation
\begin{equation}
S_{a,b}(\thetav):=S(\theta_{a}-\theta_{b}),\quad S_{a,b}^{(m)} := \left\{\begin{array}{l} S_{b,a} \quad a \leq m <b, b\leq m <a \\S_{a,b} \quad \text{otherwise}\end{array} \right.
\end{equation}
We will often not write down explicitly the arguments $\thetav,\etav$ where they are clear from the context. We will see the use of the above expressions later in the present thesis.

It will become also very useful the fact that we can express the factors $S_C$ in terms of the expressions $S^\sigma$ associated with permutations $\sigma$, as the following Lemma shows.

\begin{lemma}\label{lemma:scpermute}
 There holds
\begin{equation}\label{Scontrperm}
 \delta_{C} S_C(\thetav,\etav) = \delta_{C}S^{\sigma}(\thetav)S^{\rho}(\etav),
\end{equation}
where
\begin{equation}\label{eq:scpermutations}
 \begin{aligned}
\sigma &=\begin{pmatrix} 1 & & \ldots & &  & m \\ 1 & \ldots \;\hat{l}\; \ldots & m &  l_{1} & \ldots & l_{|C|} \end{pmatrix}, \\
\rho &=\begin{pmatrix} m+1 & & & \ldots & & m+n \\ r_{|C|} & \ldots & r_{1} & m+1 & \ldots \; \hat{r}\; \ldots & m+n \end{pmatrix}.
 \end{aligned}
\end{equation}
\end{lemma}

\emph{Remarks}: $\hat l$ indicates that we leave out the $l_j$ from the sequence; $\hat r$ analogously. The permutations $\sigma,\rho$ are not unique since one can permute the pairs of the contraction. However,
the right hand side of \eqref{Scontrperm} is independent of this choice since the extra $S$-factors associated with different permutations $\sigma,\rho$ of the same pairs would cancel each other due to the delta distributions.

\begin{proof}

Considering the above remark, we can assume that $r_{1} < \ldots < r_{|C|} $. From the definition of $S^\sigma$, $S^\rho$ in Eq.~\eqref{eq:Sperm} with $\sigma, \rho$ given by \eqref{eq:scpermutations}, we can read off that
\begin{align}\label{proofCsigma}
 S^{\sigma} &= \prod_{j=1}^{|C|}\prod_{p_{j}=l_{j}+1}^{m}S_{p_{j},l_{j}}\cdot \prod_{\substack{i<j \\ l_{i}<l_{j}}} S_{l_{i},l_{j}}, \\
 S^{\rho}   &= \prod_{j=1}^{|C|}\prod_{q_{j}=m+1}^{r_{j}-1}S_{r_{j},q_{j}}.\label{proofCrho}
\end{align}
Computing the product $S^{\sigma} S^{\rho}$ from \eqref{proofCsigma} and \eqref{proofCrho}, and taking the factor $\delta_C$ into account, we find that
$\delta_{C} S^{\sigma} S^{\rho} = \delta_{C} S_C$ with $S_C$ defined as in \eqref{eq:sc}.
\end{proof}

We will also need to consider \emph{compositions} of contractions. Given the contractions $C \in \ccal_{m,n}$ and $C' \in \ccal_{m-|C|,n-|C|}$, the \emph{composed contraction}, where the indices are contracted first with $C$, then with $C'$, is defined as $C \dot \cup C' \in \ccal_{m,n}$, $C \dot \cup C'=(m,n,\{(l_1,r_1),\ldots,(l_k,r_k),(l'_1,r'_1),\ldots,(l'_{k'},r'_{k'})\})$. This definition should be intuitively clear; on the other hand, note that it involves a \emph{renumbering} of the indices in $C'$ before taking the set union $C \dot \cup C'$; we will often avoid to indicate  this renumbering explicitly.
With respect to this composition of contractions, also the factors $\delta_C$ and $S_C$ compose in a certain way, as the following lemma shows. Here $\hat \thetav \in \rbb^{m-|C|}$ indicates that $\thetav$ has the components $\theta_{l_1},\ldots,\theta_{l_{|C|}}$ left out; analogously for $\hat \etav$.

\begin{lemma}\label{lemma:contractcompose}

Let $C \in \ccal_{m,n}$ and $C'\in \ccal_{m-|C|,n-|C|}$. There holds
\begin{equation}
\delta_{C}(\thetav,\etav) \delta_{C'} (\hat\thetav,\hat\etav) S_C(\thetav,\etav) S_{C'}(\hat\thetav,\hat\etav) = \delta_{C\dot{\cup}C'}(\thetav,\etav) S_{C\dot{\cup}C'}(\thetav,\etav) \label{scc}.
\end{equation}
\end{lemma}

\begin{proof}
From the definition \eqref{eq:deltac} it is clear that $\delta_{C}\delta_{C'}=\delta_{C \dot\cup C'}$.
Using this and Lemma~\ref{lemma:scpermute}, it remains to show that
\begin{equation}\label{Scomposperm}
  S^{\sigma}(\thetav)S^{\sigma'}(\hat \thetav ) = S^{\sigma''}(\thetav),
\quad
  S^{\rho}(\etav) S^{\rho'}(\hat \etav ) = S^{\rho''}(\etav),
\end{equation}
where $\sigma,\rho$, $\sigma',\rho'$, $\sigma'',\rho''$ are the permutations associated with $C$, $C'$, and $C\dot\cup C'$, respectively, by Eq.~\eqref{Scontrperm}.
We note that $\sigma'$ is given explicitly by
\begin{equation}
\sigma'=\begin{pmatrix} 1 & & \ldots \hat{l} \ldots &  & & m \\ 1 & \ldots \hat{l}\; \hat{l'} \ldots & m &  l'_{1} & \ldots & l'_{|C'|} \end{pmatrix} \in \perms{m-|C|}.
\end{equation}
We can consider $\sigma'$ as an element of $\perms{m}$ by extending the permutation matrix in the following way:
\begin{equation}
\sigma'=\begin{pmatrix} 1 & & \ldots \hat{l} \ldots & &  & m  & l_{1} & \ldots & l_{|C|} \\ 1 & \ldots \hat{l'}\; \hat{l} \ldots & m  &  l'_{1} & \ldots & l'_{|C'|} & l_{1} & \ldots & l_{|C|}  \end{pmatrix}.
\end{equation}
With this, $S^{\sigma'}(\hat\thetav)=S^{\sigma'}(\thetav^{\sigma})$. Using the composition law in Eq.~\eqref{eq:Scompose}, one has
\begin{equation}
 S^{\sigma''}(\thetav) = S^{\sigma'}(\thetav^\sigma) S^{\sigma}(\thetav) = S^{\sigma'}(\hat\thetav)S^{\sigma}(\thetav).
\end{equation}
where $\sigma''$ is given by:
\begin{equation}
\sigma'' := \sigma \circ \sigma'=\begin{pmatrix} 1 & & \ldots & &  & m \\ 1 & \ldots \;\hat{l'}\;\hat{l}\; \ldots & m &  l'_{1} & \ldots & l'_{|C'|}\;l_{1} & \ldots & l_{|C|} \end{pmatrix}
\end{equation}
One notices that this permutation is indeed associated with $C \dot\cup C'$ by Eq.~\eqref{Scontrperm}.

We obtain in a similar way the second part of Eq.~\eqref{Scomposperm}, and hence we find the result of this lemma.
\end{proof}

\section{Contracted matrix elements} \label{sec:fmndef}

Given any quadratic form $A \in \mathcal{Q}^{\omega}$, we define its \emph{fully contracted matrix elements} $\cme{m,n}{A}$ by
\begin{equation}\label{eq:fmndef}
\cme{m,n}{A}(  \thetav,  \etav) :=
\sum_{C\in \ccal_{m,n}} (-1)^{|C|} \delta_C \, S_C (\thetav,\etav) \,
\hscalar{ \lvector{C}{\thetav} }{ A \,\rvector{C}{\etav} }.
\end{equation}
These are very similar to Lechner's contracted matrix elements $\langle \cdotarg \rangle_{m+n,m}^\text{con}$ introduced in \cite{Lechner:2008}; the relation between our $\cme{m,n}{A}$ and Lechner's contracted matrix elements, in notation used there, is $\cme{m,n}{A} = \langle J A^\ast J \rangle_{m+n,m}^\text{con}$.

We can show that $\cme{m,n}{A}$ are well-defined distributions in $\dcal(\rbb^{m+n})'$. Indeed, due to our remark after \eqref{prooflC}, the contracted matrix elements $\hscalar{ \lvector{C}{\thetav} }{ A \,\rvector{C}{\etav} }$ are well-defined distributions in $\dcal(\rbb^{m+n})'$. The product of the delta distributions $\delta_C$ with $\hscalar{ \lvector{C}{\thetav} }{ A \,\rvector{C}{\etav} }$ is well-defined because they depend on mutually different variables.
Since $S$ is smooth (see Definition \ref{def:Smatrix}), the product of $S_C$ with $\delta_C \cdot \hscalar{ \lvector{C}{\thetav} }{ A \,\rvector{C}{\etav} }$ is also well-defined. Therefore, the quantity $\delta_C \cdot S_C \cdot \hscalar{ \lvector{C}{\thetav} }{ A \,\rvector{C}{\etav} }$ is a well-defined distribution in $\dcal(\rbb^{m+n})'$.

We can even show more, namely that the norms $\gnorm{\cme{m,n}{A}}{m\times n}^{\omega}$ are finite; we prove this in the following Proposition.
\begin{proposition}\label{proposition:fmnbound}
 For $m,n \in \nbb_0$, there is a constant $c_{mn}$ such that for all $A \in \qf^{\omega}$,
\begin{equation}
\onorm{\cme{m,n}{A} }{m\times n} \leq c_{mn} \onorm{A}{m+n}.\label{boundcme}
\end{equation}
\end{proposition}

\begin{proof}
Applying the triangle inequality we find
\begin{equation}
\onorm{\cme{m,n}{A} }{m\times n} \leq \sum_{C\in \ccal_{m,n}} \onorm{\delta_C \, S_C \,
\hscalar{ \lvector{C}{\thetav} }{ A \,\rvector{C}{\etav} } }{m \times n} .
\end{equation}
By Lemma~\ref{lemma:scpermute} the factor $S_C (\thetav,\etav)$ factorizes to $S^\sigma(\thetav)S^\rho(\etav)$; applying Eq.~\eqref{eq:mnnormfactor}, we have
\begin{eqnarray}
\onorm{\delta_C \, S_C \,
\hscalar{ \lvector{C}{\thetav} }{ A \,\rvector{C}{\etav} } }{m \times n} &\leq& \gnorm{S^{\sigma}}{\infty}\onorm{\delta_C \,\hscalar{ \lvector{C}{\thetav} }{ A \,\rvector{C}{\etav} } }{m\times n}\gnorm{S^{\rho}}{\infty} \nonumber\\
&=& \onorm{\delta_C \,\hscalar{ \lvector{C}{\thetav} }{ A \,\rvector{C}{\etav} } }{m\times n}.
\end{eqnarray}
where in the third equality we estimated the uniform norms of $S^\sigma, S^\rho$ by 1 (since on the real axis, $|S(\theta)|=1$ for all $\theta$). The individual factors of $\delta_C$ and the matrix element $\hscalar{ \lvector{C}{\thetav} }{ A \,\rvector{C}{\etav} } $ depend on mutually different variables; so, we can apply repeatedly \eqref{eq:mnnormproduct} to $\onorm{\delta_C \, \hscalar{ \lvector{C}{\thetav} }{ A \,\rvector{C}{\etav} } }{m \times n}$, we find
\begin{equation} \label{eq:mntriangle}
\onorm{\cme{m,n}{A} }{m\times n} \leq \sum_{C\in \ccal_{m,n}} \Big(\prod_{j=1}^{|C|} \gnorm{\delta(\theta_{l_j}-\eta_{r_j})}{1 \times 1} \Big) \onorm{
\hscalar{ \lvector{C}{\thetav} }{ A \,\rvector{C}{\etav} } }{(m-|C|) \times (n-|C|)} .
\end{equation}
By application of Cauchy-Schwarz one easily sees that $\gnorm{\delta(\theta-\eta)}{1 \times 1}=1$; we can show that
\begin{equation}
 \onorm{\hscalar{ \lvector{C}{\thetav} }{ A \,\rvector{C}{\etav} } }{(m-|C|) \times (n-|C|)}
 \leq \sqrt{(m-|C|)!}\,\sqrt{(n-|C|)!} \;\onorm{A}{m+n}.\label{onormmatrix}
\end{equation}
The inequality \eqref{onormmatrix} can be proved as follows. From \eqref{eq:crossnorm2}, we have that
\begin{multline}
\onorm{\hscalar{ \lvector{C}{\thetav} }{ A \,\rvector{C}{\etav} } }{(m-|C|) \times (n-|C|)}
=\frac{1}{2}\gnorm{e^{-\omega(E(\hat{\thetav}))}\hscalar{ \lvector{C}{\thetav} }{ A \,\rvector{C}{\etav} }}{(m-|C|) \times (n-|C|)}\\
+\frac{1}{2}\gnorm{\hscalar{ \lvector{C}{\thetav} }{ A \,\rvector{C}{\etav} }e^{-\omega(E(\hat{\etav}))}}{(m-|C|) \times (n-|C|)}.\label{proofonormmatrix}
\end{multline}
We consider the first norm on the right hand side of the above equation. By \eqref{eq:crossnorm}, we have that
\begin{multline}
\gnorm{e^{-\omega(E(\hat{\thetav}))}\hscalar{ \lvector{C}{\thetav} }{ A \,\rvector{C}{\etav} }}{(m-|C|) \times (n-|C|)}\\
=\sup_{\substack{ \gnorm{g}{2} \leq 1\\ \gnorm{h}{2} \leq 1}} \frac{\big\lvert \! \int e^{-\omega(E(\hat{\thetav}))}\hscalar{ \lvector{C}{\thetav} }{ A \,\rvector{C}{\etav} }g(\hat{\thetav}) h(\hat{\etav})  d^{m-|C|}\hat{\theta} d^{n-|C|}\hat{\eta}  \,\big\rvert}{\gnorm{g}{2}\gnorm{h}{2}}.
\end{multline}
We can estimate the absolute value of $\hscalar{ \lvector{C}{\thetav} }{ A \,\rvector{C}{\etav} }$ integrated with the functions $g,h$ using the Cauchy-Schwarz inequality:
\begin{multline}
|\hscalar{ e^{-\omega(H/\mu)}\lvector{C}{g} }{ A \,\rvector{C}{h} }|\\
=|\hscalar{ \lvector{C}{g} }{ Q_{m-|C|}e^{-\omega(H/\mu)} A Q_{n-|C|} \,\rvector{C}{h} }|\\
\leq \sqrt{(m-|C|)!} \sqrt{(n-|C|)!} \gnorm{g}{2} \gnorm{h}{2} \gnorm{Q_{m-|C|}e^{-\omega(H/\mu)} A Q_{n-|C|}}{}\\
\leq \sqrt{(m-|C|)!} \sqrt{(n-|C|)!} \gnorm{g}{2} \gnorm{h}{2} \gnorm{Q_{m+n}e^{-\omega(H/\mu)} A Q_{m+n}}{}.
\end{multline}
Hence, we have:
\begin{multline}
\gnorm{e^{-\omega(E(\hat{\thetav}))}\hscalar{ \lvector{C}{\thetav} }{ A \,\rvector{C}{\etav} }}{(m-|C|) \times (n-|C|)}\\
\leq \sqrt{(m-|C|)!} \sqrt{(n-|C|)!} \gnorm{Q_{m+n}e^{-\omega(H/\mu)} A Q_{m+n}}{}.
\end{multline}
By a similar method, we find for the second norm on the right hand side of \eqref{proofonormmatrix}:
\begin{multline}
\gnorm{\hscalar{ \lvector{C}{\thetav} }{ A \,e^{-\omega(E(\hat{\etav}))} \rvector{C}{\etav} }}{(m-|C|) \times (n-|C|)}\\
\leq \sqrt{(m-|C|)!} \sqrt{(n-|C|)!} \gnorm{Q_{m+n} A e^{-\omega(H/\mu)} Q_{m+n}}{}.
\end{multline}
By definition \eqref{eq:aomeganorm} and \eqref{proofonormmatrix} the two estimates above imply \eqref{onormmatrix}. Inserting \eqref{onormmatrix} into \eqref{eq:mntriangle}, we find the result \eqref{boundcme}.
\end{proof}

We present here an alternative proof of Proposition \ref{proposition:fmnbound}, close to Lechner \cite[Lemma~4.3]{Lechner:2008}.

\begin{proof}
We write $\boldsymbol{\theta_{l}}:=(\theta_{l_{1}},\ldots,\theta_{l_{|C|}})$ and $\delta_{C}S_C=\delta_{C}S_{C}^{L}S_{C}^{R}$, where $S_{C}^{L}$ and $S_{C}^{R}$ are functions of the variables $\{ \theta_{1},\ldots,\theta_{m}\}$ and $(\{ \eta_{1},\ldots,\eta_{n}\}\backslash \{ \eta_{r_{1}-m},\ldots,\eta_{r_{|C|}-m}\})\cup \{\boldsymbol{\theta_{l}}\}$, respectively.
For $g\in \mathcal{D}(\mathbb{R}^{m-|C|}),\, h\in \mathcal{D}(\mathbb{R}^{n-|C|})$, we consider
\begin{eqnarray}
G_{\boldsymbol{\theta_{l}}}^{L} &:=& S_{C}^{L}\cdot g_{\boldsymbol{\theta_{l}}},\\
H_{\boldsymbol{\theta_{l}}}^{R} &:=& S_{C}^{R}\cdot h_{\boldsymbol{\theta_{l}}}.
\end{eqnarray}
where $g_{\boldsymbol{\theta_{l}}}:=g(\theta_{1},\ldots,\boldsymbol{\theta_{l}},\ldots,\theta_{m})$, $h_{\boldsymbol{\theta_{l}}}:=h(\eta_{n},\ldots,\boldsymbol{\theta_{l}},\ldots,\eta_{1})$ and where the functions $G^{L}_{\boldsymbol{\theta_{l}}}\in \mathcal{D}(\mathbb{R}^{m-|C|})$ and $H^{R}_{\boldsymbol{\theta_{l}}}\in \mathcal{D}(\mathbb{R}^{n-|C|})$ depend parametrically on $\boldsymbol{\theta_{l}}\in \mathbb{R}^{|C|}$.

We consider:
\begin{multline}
\left|\int d^{m}\thetav d^{n}\etav\; \cme{m,n}{A}(\thetav,\etav)g(\thetav)h(\etav)\right|=
\Big| \int d^{m}\thetav d^{n}\etav \sum_{C\in \ccal_{m,n}}(-1)^{|C|}\delta_{C}S_C
 \times \\
 \times \langle \ell_{C}(\thetav),e^{-\omega(H/\mu)}Q_{m-|C|}A Q_{n-|C|}r_{C}(\etav)\rangle
   e^{\omega(E(\widehat{\thetav}))}g(\thetav)h(\etav)\Big|.\label{crossbounodeproof}
\end{multline}
After integrating over the delta distributions, we find:
\begin{equation}
\text{ r.h.s. }\eqref{crossbounodeproof} = \left|\sum_{C\in \ccal_{m,n}}(-1)^{|C|}\int d^{|C|}\pmb{\theta_{l}}\;\langle e^{\omega(E(\cdot))}G_{\pmb{\theta_{l}}}^{R},Q_{m-|C|}e^{- \omega(H/\mu)}A Q_{n-|C|}H_{\pmb{\theta_{l}}}^{L}\rangle\right|.
\end{equation}
Using the Cauchy-Schwarz inequality, we find:
\begin{equation}
\begin{aligned}
&\text{r.h.s. }\eqref{crossbounodeproof}\\
&\leq \sum_{C\in \ccal_{m,n}}\sqrt{(n-|C|)!}\sqrt{(m-|C|)!}\int d^{|C|}\pmb{\theta_{l}}\;||e^{\omega(E(\cdot))}G_{\pmb{\theta_{l}}}^{R}||_{2}\cdot ||Q_{m}e^{-\omega(H/\mu)}AQ_{n}||\cdot ||H_{\pmb{\theta_{l}}}^{L}||_{2}\\
&=\sum_{C\in \ccal_{m,n}}\sqrt{(n-|C|)!}\sqrt{(m-|C|)!}\int d^{|C|}\pmb{\theta_{l}}\; ||e^{\omega(E(\cdot))}g_{\pmb{\theta_{l}}}||_{2}\cdot ||Q_{m}e^{-\omega(H/\mu)}AQ_{n}||\cdot ||h_{\pmb{\theta_{l}}}||_{2}\\
&\leq \Big(\sum_{C\in \ccal_{m,n}}\sqrt{(n-|C|)!}\sqrt{(m-|C|)!}\Big) ||e^{\omega(E(\cdot))}g||_{2}\cdot ||Q_{m}e^{-\omega(H/\mu)}AQ_{n}||\cdot ||h||_{2}\\
&=\Big(\sum_{C\in \ccal_{m,n}}\sqrt{(n-|C|)!}\sqrt{(m-|C|)!}\Big) ||g||^{\omega}_{2}\cdot ||Q_{m}e^{-\omega(H/\mu)}AQ_{n}||\cdot ||h||_{2}.
\end{aligned}
\end{equation}
where in the third inequality we made use of the monotonicity of $\omega$: $e^{\omega(E(\widehat{\thetav}))}\leq e^{\omega(E(\widehat{\thetav})+E(\pmb{\theta_{l}}))}$, and of the Cauchy-Schwarz inequality.

Using the inequality $a!b!\leq (a+b)!$, we find
\begin{equation}
\text{ r.h.s. }\eqref{crossbounodeproof}\leq \Big( \sum_{C\in \ccal_{m,n}}1\Big)\sqrt{(m+n)!}||g||^{\omega}_{2}\cdot ||Q_{m}e^{-\omega(H/\mu)}AQ_{n}||\cdot ||h||_{2}.
\end{equation}
The number of contractions in $\ccal_{m,n}$ is given by
\begin{equation}
|\ccal_{m,n}|:=\sum_{k}\left|\left\{C\in \ccal_{m,n}\; |\; |C|=k\right\}\right|=\sum_{k}\left( \begin{array}{c}
m  \\
k \end{array} \right) \left( \begin{array}{c}
n  \\
k \end{array} \right)k!
\end{equation}
We denote $c'_{mn}:=|\ccal_{m,n}|\sqrt{(m+n)!}$, we have
\begin{equation}
\left|\int d^{m}\thetav d^{n}\etav\; \cme{m,n}{A}(\thetav,\etav)g(\thetav)h(\etav)\right|\leq c'_{mn}||g||^{\omega}_{2}\cdot ||Q_{m}e^{-\omega(H/\mu)}AQ_{n}||\cdot ||h||_{2}.
\end{equation}
Calling $g'(\thetav):=e^{\omega(E(\thetav))}g(\thetav)$, the equation above can be rewritten as follows:
\begin{equation}
\left|\int d^{m}\thetav d^{n}\etav\; e^{-\omega(E(\thetav))}\cme{m,n}{A}(\thetav,\etav)g'(\thetav)h(\etav)\right|\leq c'_{mn}||g'||_{2}\cdot ||Q_{m}e^{-\omega(H/\mu)}AQ_{n}||\cdot ||h||_{2}.\label{gpstima}
\end{equation}
Analogously, we compute the same norm as before but with $e^{\omega(E(\widehat{\thetav}))}$ replaced by $e^{\omega(E(\widehat{\etav}))}$:
\begin{multline}
\left|\int d^{m}\thetav d^{n}\etav\; \cme{m,n}{A}(\thetav,\etav)g(\thetav)h(\etav)\right|=\Big| \int d^{m}\thetav d^{n}\etav \sum_{C\in \ccal_{m,n}}(-1)^{|C|}\delta_{C}S_C\times\\
 \times \langle \ell_{C}(\thetav),Q_{m-|C|}A Q_{n-|C|}e^{-\omega(H/\mu)}r_{C}(\etav)\rangle
  g(\thetav)e^{\omega(E(\widehat{\etav}))} h(\etav)\Big|.
\end{multline}
We find in this case:
\begin{equation}
\left|\int d^{m}\thetav d^{n}\etav\; \cme{m,n}{A}(\thetav,\etav)g(\thetav)h(\etav)\right|\leq c''_{mn}||g||_{2}\cdot ||Q_{m}Ae^{-\omega(H/\mu)}Q_{n}||\cdot ||h||_{2}^{\omega}.
\end{equation}
Calling $h'(\etav):=e^{\omega(E(\etav))}h(\etav)$, the equation above can be rewritten as follows:
\begin{equation}
\left|\int d^{m}\thetav d^{n}\etav\; \cme{m,n}{A}(\thetav,\etav)e^{-\omega(E(\etav))}g(\thetav)h'(\etav)\right|\leq c''_{mn}||g||_{2}\cdot ||Q_{m}Ae^{-\omega(H/\mu)}Q_{n}||\cdot ||h'||_{2}.\label{hpstima}
\end{equation}
By summing the left and right hand sides of \eqref{gpstima} and \eqref{hpstima}, we find:
\begin{multline}
\frac{1}{2}\left|\int d^{m}\thetav d^{n}\etav\; e^{-\omega(E(\thetav))}\cme{m,n}{A}(\thetav,\etav)g(\thetav)h(\etav)\right|\\
+\frac{1}{2}\left|\int d^{m}\thetav d^{n}\etav\; \cme{m,n}{A}(\thetav,\etav)e^{-\omega(E(\etav))}g(\thetav)h(\etav)\right|\\
\leq \frac{1}{2}c'_{mn}||g||_{2}\cdot ||Q_{m}e^{-\omega(H/\mu)}AQ_{n}||\cdot ||h||_{2} + \frac{1}{2}c''_{mn}||g||_{2}\cdot ||Q_{m}Ae^{-\omega(H/\mu)}Q_{n}||\cdot ||h||_{2}\\
\leq \frac{1}{2}c_{mn}||g||_{2}\cdot ||Q_{m+n}e^{-\omega(H/\mu)}AQ_{m+n}||\cdot ||h||_{2} + \frac{1}{2}c_{mn}||g||_{2}\cdot ||Q_{m+n}Ae^{-\omega(H/\mu)}Q_{m+n}||\cdot ||h||_{2}.
\end{multline}
where we denoted $c_{mn}:=c'_{mn}+c''_{mn}$.

Taking the supremum of the left and right hand sides of the above equation over $\gnorm{g}{2}\leq 1$ and $\gnorm{h}{2}\leq 1$, we find \eqref{boundcme}.
\end{proof}

\section{S-symmetry of the coefficients} \label{sec:fmnsymm}

\begin{proposition}\label{proposition:fmnsymm}
 The distributions $\cme{m,n}{A}$ are $S$-symmetric in the first $m$ and last $n$ variables separately; that is, for permutations
  $\pi \in \perms{m}$ and $\tau \in \perms{n}$,
\begin{equation}
   \cme{m,n}{A}(\thetav,\etav) = S^\pi(\thetav) S^\tau(\etav)\cme{m,n}{A}(\thetav^\pi,\etav^\tau).
\end{equation}
\end{proposition}

\begin{proof}
We consider only the case $\tau = \operatorname{id}$; the arguments that we make for $\pi$ then apply to $\tau$ analogously. It also suffices to consider the case where $\pi$ is a transposition due to the representation property of $S^\pi$, see Sec.~\ref{sec:ssymm}.

In \eqref{eq:fmndef} we denote:
\begin{equation}\label{eq:cmesum}
\cme{m,n}{A}(\thetav,\etav )=\sum_{C\in \ccal_{m,n}} T_{C}(\thetav ,\etav ), \quad \text{where}\;
 T_{C}(\thetav ,\etav ):= (-1)^{|C|} \delta_C \, S_C \hscalar{ \lvector{C}{\thetav} }{ A \,\rvector{C}{\etav} }.
\end{equation}
We want to compute
\begin{equation}\label{Tpi}
 T_{C}(\thetav^{\pi},\etav)= (-1)^{|C|} \delta_C(\thetav^{\pi},\etav) \, S^{\sigma}(\thetav^{\pi}) \,
 S^{\rho}(\etav) \, \hscalar{ \lvector{C}{\thetav^\pi} }{ A \,\rvector{C}{\etav} },
\end{equation}
where $\thetav ^{\pi}=(\theta_{1},\ldots,\theta_{k+1},\theta_{k},\ldots,\theta_{m})$, and where we made use of Eq.~\eqref{Scontrperm}, namely $\sigma,\rho$ are the permutations corresponding to $C$ by Lemma~\ref{lemma:scpermute}.
We distinguish the four cases where each of the indices $k$ and $k+1$ can be either contracted or non-contracted in $C$.

\begin{enumerate}[(a)]

 \item Consider the case where both $k$ and $k+1$ are not contracted. Since $\delta_{C}$ depends only on the contracted variables, we have  $\delta_{C}(\thetav^{\pi},\etav)=\delta_{C}(\thetav,\etav)$. Moreover, $\hscalar{ \lvector{C}{\thetav^\pi} }{ A \,\rvector{C}{\etav} } = S(\theta_{k}-\theta_{k+1}) \hscalar{ \lvector{C}{\thetav} }{ A \,\rvector{C}{\etav} }$ due to the exchange relations of the Zamolodchikov algebra.
 As for the factor $S^\sigma(\thetav^\pi)$:
 Since $\pi$ is a transposition, then for all $1\leq j\leq |C|$ we have that either $k,k+1> l_{j}$ or $k,k+1< l_{j}$. In both cases it follows from \eqref{eq:scpermutations} that $S^{\sigma}(\thetav^{\pi})=S^{\sigma}(\thetav)$.
 In total, we obtain
\begin{equation}\label{eq:tcpi-a}
S^\pi(\thetav) T_{C}(\thetav^{\pi},\etav)= T_{C}(\thetav,\etav).
\end{equation}

\item Consider the case where both $k$ and $k+1$ are contracted. Given $C=(m,n,$\\$\{(l_{1},r_{1}),\ldots,(k,r),(k+1,r'),\ldots,(l_{|C|},r_{|C|})\})$, let $C':=(m,n,\{(l_{1},r_{1}),\ldots,(k,r'),(k+1,r),\ldots,(l_{|C|},r_{|C|})\})$. Then we have that $\delta_C(\thetav^\pi,\etav) = \delta_{C'}(\thetav,\etav)$ and also, $\lvector{C}{\thetav^\pi}=\lvector{C}{\thetav}=\lvector{C'}{\thetav}$, since both $\delta_C$ and $\lvector{C}{\cdotarg}$ do not depend on the contracted variables. Regarding the S-factors, we write using Eq.~\eqref{eq:Scompose},
    \begin{equation} \label{eq:spifactor}
    S^{\sigma}(\thetav^{\pi}) = S^{\pi \sigma}(\thetav)(S^{\pi}(\thetav))^{-1}
    \end{equation}
The permutation $ \pi \circ \sigma$ is given by
    \begin{equation}
    \pi \circ \sigma=\begin{pmatrix}1 && \ldots && m \\
    1 && \ldots \; k+1 \; k\;  \ldots && m\\
    1 & \ldots \hat r \ldots & m & r_1 &\ldots \; k+1 \; k\; \ldots & r_k\end{pmatrix}
    \end{equation}

    One finds that $\pi \circ \sigma$ corresponds to $C'$ above in the sense of Eq.~\eqref{Scontrperm}, with the same $\rho$ for both $C$ and $C'$. Combining all this into \eqref{Tpi}, we obtain
    \begin{equation}\label{eq:tcpi-b}
    S^\pi(\thetav) T_{C}(\thetav^{\pi},\etav)=T_{C'}(\thetav,\etav).
    \end{equation}
    Note that the contraction $C'$ is again of type (b).

    \item Consider the case where $k$ is contracted, but $k+1$ is not contracted.
    Given \\ $C=( m,n,\{(l_{1},r_{1}) \ldots (k,r) \ldots (l_{|C|},r_{|C|})\})$,
    let \\ $C':=(m,n,\{(l_{1},r_{1})\ldots (k+1,r)\ldots (l_{|C|},r_{|C|}) \})$. Then, we have that $\delta_C(\thetav^\pi,\etav) = \delta_{C'}(\thetav,\etav)$      and $\lvector{C}{\thetav^\pi}=\lvector{C'}{\thetav}$. Moreover, using Eq.~\eqref{eq:Scompose}, we write $S^{\sigma}(\thetav^{\pi}) = S^{\pi   \sigma}(\thetav) S^\pi(\thetav)^{-1}  $ where the permutation $\pi\circ \sigma$, which is given by
    \begin{equation}
    \pi\circ \sigma= \begin{pmatrix}1 && \ldots && m \\
    1 && \ldots\; k+1\; k\;  \ldots && m \\
    1 & \ldots \hat r\; k\; \widehat{k+1}\; \ldots & m & r_1 &\ldots \;k+1\; \ldots & r_k\end{pmatrix},
    \end{equation}
corresponds to $C'$ in the sense of Eq.~\eqref{Scontrperm}, with the same $\rho$ for both $C$ and $C'$. Combining all in \eqref{Tpi}, we arrive at
    \begin{equation}\label{eq:tcpi-c}
    S^\pi(\thetav) T_{C}(\thetav^{\pi},\etav)= T_{C'}(\thetav,\etav).
    \end{equation}

  \item  $k+1$ is contracted, but $k$ is not contracted. This case is analogous to (c); indeed, the contraction $C'$ in (c) is exactly of type (d).

\end{enumerate}

Summing over all contractions $C$ in \eqref{eq:cmesum}, we obtain from \eqref{eq:tcpi-a}, \eqref{eq:tcpi-b}, \eqref{eq:tcpi-c} that $S^\pi(\thetav)\cme{m,n}{A}(\thetav^\pi,\etav ) =  \cme{m,n}{A}(\thetav,\etav )$ as claimed. The details of this summation argument are as follows.

Using the above results in (a), (b), (c), (d), we compute $\sum_{C\in\ccal_{m,n}}T_{C}(\thetav ^{\pi},\etav )$ and we find:
\begin{equation}
\begin{aligned}
\cme{m,n}{A}(\thetav^\pi ,\etav ) &=\sum_{C\in\ccal_{m,n}}T_{C}(\thetav ^{\pi},\etav )\\
&=\sum_{C\in \ccal^{(a)}_{m,n}}T_{C}(\thetav ^{\pi},\etav )\;+\sum_{C\in \ccal^{(b)}_{m,n}}T_{C}(\thetav ^{\pi},\etav )\;+\sum_{C\in {\ccal}^{(c)}_{m,n}}T_{C}(\thetav ^{\pi},\etav )\;+\sum_{C\in \ccal^{(d)}_{m,n}}T_{C}(\thetav ^{\pi},\etav )\\
&=(S^{\pi})^{-1}\sum_{C\in \ccal^{(a)}_{m,n}}T_{C}(\thetav ,\etav )\;+\;(S^{\pi})^{-1}\sum_{C\in \ccal^{(b)}_{m,n}}T_{C'}(\thetav ,\etav )\\
&\quad +\;(S^{\pi})^{-1}\sum_{C\in {\ccal}^{(c)}_{m,n}}T_{C'}(\thetav ,\etav )\;+\; (S^{\pi})^{-1}\sum_{C\in {\ccal}^{(d)}_{m,n}}T_{C'}(\thetav,\etav)\\
&=(S^{\pi})^{-1}\Big( \sum_{C\in\ccal^{(a)}_{m,n}}T_{C}(\thetav,\etav)\;+\;\sum_{C'\in{\ccal}^{(b)}_{m,n}}T_{C'}(\thetav,\etav)\\
&\quad +\;\sum_{C'\in\ccal^{(d)}_{m,n}}T_{C'}(\thetav,\etav)+\sum_{C'\in{\ccal}^{(c)}_{m,n}}T_{C'}(\thetav,\etav)\Big)\\
&=(S^{\pi})^{-1}\Big( \sum_{C\in\ccal^{(a)}_{m,n}}T_{C}(\thetav,\etav)\;+\;\sum_{C\in{\ccal}^{(b)}_{m,n}}T_{C}(\thetav,\etav)\\
&\quad +\;\sum_{C\in\ccal^{(d)}_{m,n}}T_{C}(\thetav,\etav)+\sum_{C\in{\ccal}^{(c)}_{m,n}}T_{C}(\thetav,\etav)\Big)\\
&=(S^{\pi})^{-1}\sum_{C\in\ccal_{m,n}}T_{C}(\thetav ,\etav ).
\end{aligned}
\end{equation}
where in the sum over $C\in {\ccal}^{(d)}_{m,n}$ in the third equality we have that $C'$ is defined correspondingly as in (c).
\end{proof}

\section{Inversion formula for matrix elements} \label{sec:fmninv}

In Eq.~\eqref{eq:fmndef}, we defined  $\cme{m,n}{A}$ as a certain sum over matrix elements of $A$. We can now invert this formula in the sense given by the following Proposition.
\begin{proposition}\label{proposition:fmninversion}
For any $A \in \qf^{\omega}$,
\begin{equation}\label{eq:fmninverse}
\hscalar{ \lvector{}{\thetav} }{ A\, \rvector{}{\etav} }
= \sum_{C\in \ccal_{m,n}} \delta_C \,S_C  \, f_{m-|C|,n-|C|}^{[A]}( \hat \thetav, \hat \etav).
\end{equation}
\end{proposition}
(Here $\hat\thetav,\hat\etav$ denotes the variables which are obtained from $\thetav,\etav$ by leaving out the components which are contracted in $C$.)
\begin{proof}
Inserting \eqref{eq:fmndef} into the right-hand side of \eqref{eq:fmninverse}, we need to show:
\begin{equation}\label{insert}
\hscalar{ \lvector{}{\thetav} }{A  \rvector{}{\etav} }
 =\sum_{C\in\ccal_{m,n}}\delta_{C}S_C(\thetav,\etav )
 \;\;\sum_{\mathclap{C'\in\ccal_{m-|C|,n-|C|}}}\;\; (-1)^{|C'|}\delta_{C'}S_{C'}(\hat{\thetav },\hat\etav)
 \hscalar{ \lvector{C \dot\cup C'}{\thetav} }{A \rvector{C \dot\cup C'}{\etav} }.
\end{equation}
Using Lemma~\ref{lemma:contractcompose}, we find
\begin{equation}
\rhs{insert} = \;\; \sum_{\mathclap{\substack{C\in\ccal_{m,n} \\ C'\in\ccal_{m-|C|,n-|C|}}}} \;\;
(-1)^{|C'|}\delta_{C\dot{\cup}C'}S_{C\dot{\cup}C'}(\thetav,\etav )\; \hscalar{ \lvector{C \dot\cup C'}{\thetav} }{A \rvector{C \dot\cup C'}{\etav} }.
\end{equation}
We denote $D:=C\dot{\cup}C'$; then, we can reorganize the sum over $C$ and $C'$ in the following way:
\begin{equation}\label{eq:dsum}
\rhs{insert}=\sum_{D\in\ccal_{m,n}}\Big(\sum_{D=C\dot{\cup}C'}(-1)^{|C'|}\Big)\delta_{D}S_{D}
\hscalar{ \lvector{D}{\thetav} }{A \rvector{D}{\etav} }.
\end{equation}
We compute the inner sum in the above formula (at fixed $D$) using the \emph{binomial formula}:
\begin{equation}
\sum_{D=C\dot{\cup}C'}(-1)^{|C'|}=\sum_{j=0}^{|D|}(-1)^{j}\binom{|D|}{j}=\begin{cases} 0 \quad&\mathrm{if}\quad |D|\geq 1,\\ 1 &\mathrm{if}\quad |D|=0,\end{cases}
\end{equation}
Hence, we find that the right hand side of \eqref{eq:dsum} gives $\hscalar{ \lvector{}{\thetav} }{ A\, \rvector{}{\etav} }$ as claimed.
\end{proof}

\section{Basis property} \label{sec:fmnbasis}

A next step which is useful in order to establish a series expansion for any quadratic form $A \in \qf^{\omega}$ in terms of the $\cme{m,n}{A}$ is to prove that the $z^{\dagger m}(\thetav)z^n(\etav)$ form a ``dual basis'' for the contracted matrix elements $\cme{m,n}{A}$.

\begin{proposition}\label{proposition:fmnbasis}
In the sense of distributions, there holds:
\begin{equation}\label{basisprop}
\cmelong{m,n}{z^{\dagger m'}(\thetav')z^{n'}(\etav')} (\thetav,\etav)
= m!n! \delta_{m,m'}\delta_{n,n'}
\operatorname{Sym}_{S,\thetav} \delta^m(\thetav-\thetav')
\operatorname{Sym}_{S,\etav}\delta^n(\etav-\etav').
\end{equation}

\end{proposition}

\begin{proof}
We consider $A:=z^{\dagger m'}(\thetav ')z^{n'}(\etav ')$. If $m-m'\neq n-n'$, then all matrix elements of $A$ in \eqref{eq:fmndef} vanish, hence $\cme{m,n}{A} =0$ and the claim follows. Therefore, we consider in the following $k:=m-m'=n-n'$.

If $k<0$, then all the matrix elements in \eqref{eq:fmndef} vanish again, and we have $\cme{m,n}{A} =0$; hence, the claim follows.

If $k=0$, then we can directly compute that
\begin{multline}\label{eq:k0case}
\cme{m,n}{A}(\thetav,\etav)
=\hscalar{ z^{\dagger m}(\thetav )\Omega }{ z^{\dagger m}(\thetav') z^{n}(\etav') J z^{\dagger n}(\etav) \Omega}
\\
=\hscalar{ z^{\dagger m}(\thetav )\Omega }{ z^{\dagger m}(\thetav ')\Omega} \hscalar{ \Omega}{z^{n}(\etav ')J z^{\dagger n}(\etav )\Omega}
=m!n!\operatorname{Sym}_{S,\thetav} \delta^m(\thetav-\thetav')\operatorname{Sym}_{S,\etav}\delta^ n(\etav-\etav'),
\end{multline}
So, again the claim follows.

If $k>0$, we prove that $\cme{m,n}{A} =0$ using induction on $k$. For given $k$, we can assume that this statement is true for some $k'$ in place of $k$, if $0<k'<k$. Now, in Prop.~\ref{proposition:fmninversion} we have $0<k_{C}<k$, and $m_{C}-m'=n_{C}-n'=: k_{C}$, $m_{C}:=m-|C|$, $n_{C}:=n-|C|$. Thus we can apply the induction hypothesis: $\cme{m_{C},n_{C}}{A} =0$. For the case $k_{C}=0$, \eqref{eq:k0case} applies; for the case $k_{C}<0$, the above argument for $k<0$ applies.

Therefore, from Prop.~\ref{proposition:fmninversion}, we have
\begin{multline}\label{eq:mtxk}
\hscalar{ \lvector{}{\thetav} }{ A \rvector{}{\etav}}
=\sum_{C\in \ccal_{m,n}}\delta_{C}S_Cf_{m-|C|,n-|C|}^{[A]}(\hat{\thetav },\hat{\etav })\\
=\cme{m,n}{A}(\thetav ,\etav )
  +\sum_{\substack{C\in \ccal_{m,n}  \\ |C|=k}}\delta_{C}S_C
   m'!n'!\operatorname{Sym}_{S,\hat{\thetav}} \delta^{m'}(\hat{\thetav }-\thetav ')\operatorname{Sym}_{S,\hat{\etav}} \delta^{ n'}(\hat{\etav }-\etav ').
\end{multline}
Using also \eqref{symstheta}, it therefore suffices to show that
\begin{equation}\label{matrixd}
\hscalar{ \lvector{}{\thetav} }{ A \rvector{}{\etav}}
=
m'!n'! \operatorname{Sym}_{S^{-1},\thetav'} \operatorname{Sym}_{S^{-1},\etav'}
 \sum_{\substack{C\in \ccal_{m,n}  \\ |C|=k}}\delta_{C}S_C \delta^{m'}(\hat{\thetav }-\thetav ') \delta^{n'}(\hat{\etav }-\etav ').
\end{equation}
Since both $\hscalar{\lvector{}{\thetav}}{A\rvector{}{\etav}}$ and $\cme{m,n}{A}$ are $S$-symmetric in the variables $\thetav,\etav$, we know from \eqref{eq:mtxk} that the right hand side of \eqref{matrixd} must be $S$-symmetric too; we can therefore take the $S$-symmetric part of each term in the sum.
\begin{multline}
\text{r.h.s.}\eqref{matrixd}\\
=
\operatorname{Sym}_{S,\thetav} \operatorname{Sym}_{S,\etav}\Big(m'!n'! \operatorname{Sym}_{S^{-1},\thetav'} \operatorname{Sym}_{S^{-1},\etav'}
 \sum_{\substack{C\in \ccal_{m,n}  \\ |C|=k}}\delta_{C}S_C \delta^{m'}(\hat{\thetav }-\thetav ') \delta^{n'}(\hat{\etav }-\etav ')\Big).
\end{multline}
We rewrite $S_C$ as $S^\sigma S^\rho$, where $\sigma$ and $\rho$ are the permutations corresponding to $C$ by Eq.~\eqref{Scontrperm}. For the moment let us consider a single term in the above sum, that implies:
\begin{multline}
 \operatorname{Sym}_{S,\thetav} \operatorname{Sym}_{S,\etav}\Big(\delta_{C}S_{C}m'!n'!\operatorname{Sym}_{S^{-1},\thetav'} \delta^{m'}(\hat{\thetav }-\thetav ')\operatorname{Sym}_{S^{-1},\etav'}\delta^{ n'}(\hat{\etav }-\etav ')\Big)\\
=\operatorname{Sym}_{S,\thetav} \operatorname{Sym}_{S,\etav}\Big(\delta_{C}S^{\sigma}(\thetav)S^{\rho}(\etav)m'!n'!\operatorname{Sym}_{S^{-1},\thetav'} \delta^{m'}(\hat{\thetav }-\thetav ')\operatorname{Sym}_{S^{-1},\etav'}\delta^{ n'}(\hat{\etav }-\etav ')\Big).\label{singtermsum}
\end{multline}
Let us consider the contraction  $C_{0}:=\{ (m,1),(m-1,2),\ldots (m-k,k+1) \}$ corresponding to $S_{C_{0}}=1$, and therefore (in the sense of formula \eqref{Scontrperm}) to $S^{\sigma_{C_{0}}}=S^{\rho_{C_{0}}}=1$.
We write a generic contraction $C$ of length $|C|=k$ as a permutation acting on the contraction $C_{0}$.
In particular, $\delta_{C}$ in the formula above rewrites in terms of $C_{0}$ as: $\delta_{C}(\thetav,\etav)=\prod_{j=1}^{|C|}\delta(\theta_{l_{j}}-\eta_{r_{j}-m})=\prod_{\substack{ i>m-|C| \\ j<|C| }}\delta(\tilde{\theta}_{i}-\tilde{\eta}_{j})\big\vert_{\substack{  \tilde{\theta}=\theta^{\sigma}\\\tilde{\eta}=\eta^{\rho} }} = \delta_{C_{0}}(\thetav^{\sigma},\etav^{\rho})$. Hence, each piece in \eqref{matrixd} becomes
\begin{multline}
\text{l.h.s.}\eqref{singtermsum}=m'!n'!\operatorname{Sym}_{S,\thetav} \operatorname{Sym}_{S,\etav}\\
\Big(\delta_{C_{0}}(\thetav^{\sigma},\etav^{\rho})S^{\sigma}(\thetav) S^{\rho}(\etav)\operatorname{Sym}_{S^{-1},\thetav'} \delta^{m'}(\widehat{\thetav^{\sigma} }^{C_{0}}-\thetav ')\operatorname{Sym}_{S^{-1},\etav'}\delta^{ n'}(\widehat{\etav^{\rho} }^{C_{0}}-\etav ')\Big).
\end{multline}
where $\widehat{\thetav}^{C_{0}}$, $\widehat{\etav}^{C_{0}}$ indicate that the variables are contracted in $C_{0}$.

Then, using the formula $\operatorname{Sym}_{S,\thetav}(S^{\sigma}(\thetav)g(\thetav^{\sigma}))=\operatorname{Sym}_{S,\thetav} (g)$ to simplify the expression, we find
\begin{multline}
\text{l.h.s.}\eqref{singtermsum}=
m'!n'!\operatorname{Sym}_{S,\thetav} \operatorname{Sym}_{S,\etav} \\
\Big(\prod_{j=1}^{|C_{0}|}\delta_{m-j+1,m+j}(\thetav,\etav) \operatorname{Sym}_{S^{-1},\thetav'}  \prod_{j=1}^{m-|C_{0}|}\delta(\theta_{j}-\theta'_{j}) \operatorname{Sym}_{S^{-1},\etav'}  \prod_{j=|C_{0}|+1}^{n}\delta(\eta_{j}-\eta_{j-|C_{0}|}') \Big).
\end{multline}
In view of \eqref{symstheta} the symmetrization in $\thetav',\etav'$ can be dropped in favour of the symmetrization in $\thetav,\etav$, we find
\begin{multline}
\text{l.h.s.}\eqref{singtermsum}\\
=m'!n'!\operatorname{Sym}_{S,\thetav}\operatorname{Sym}_{S,\etav} \Big(\prod_{j=1}^{|C_{0}|}\delta_{m-j+1,m+j}(\thetav,\etav) \prod_{j=1}^{m-|C_{0}|}\delta(\theta_{j}-\theta'_{j})   \prod_{j=|C_{0}|+1}^{n}\delta(\eta_{j}-\eta_{j-|C_{0}|}') \Big).
\end{multline}
Hence, we have
\begin{multline}
\text{r.h.s.}\eqref{matrixd}
=\sum_{\substack{C\in \ccal_{m,n}  \\ |C|=k}}m'!n'!\operatorname{Sym}_{S,\thetav}\operatorname{Sym}_{S,\etav} \\ \Big(\prod_{j=1}^{|C_{0}|}\delta_{m-j+1,m+j}(\thetav,\etav) \prod_{j=1}^{m-|C_{0}|}\delta(\theta_{j}-\theta'_{j})   \prod_{j=|C_{0}|+1}^{n}\delta(\eta_{j}-\eta_{j-|C_{0}|}') \Big).
\end{multline}
We find that the terms in the sum do not actually depend on $C$; the sum, which contains $\binom{n}{k}\binom{m}{k}k!$ terms, can then be computed:
\begin{equation}
m'!n'!\sum_{|C|=k}1=m'!n'!\left(\begin{array}{l} m \\k\end{array}\right)\left(\begin{array}{l} n \\k\end{array}\right)k!=\frac{(m'+k)!(n'+k)!}{m'!k!n'!k!}k!m'!n'!=\frac{m!n!}{k!}.
\end{equation}
Hence, we have
\begin{multline}
\text{r.h.s.}\eqref{matrixd}\\
=\frac{m!n!}{k!}\operatorname{Sym}_{S,\thetav}\operatorname{Sym}_{S,\etav} \Big(\prod_{j=1}^{|C_{0}|}\delta(\theta_{m-j+1}-\eta_{j}) \prod_{j=1}^{m-|C_{0}|}\delta(\theta_{j}-\theta'_{j})   \prod_{j=|C_{0}|+1}^{n}\delta(\eta_{j}-\eta_{j-|C_{0}|}') \Big).
\end{multline}
However, the left hand side of \eqref{matrixd} gives:
\begin{multline}
\hscalar{ \lvector{}{\thetav} }{ A \rvector{}{\etav}}
=\operatorname{Sym}_{S,\thetav}\operatorname{Sym}_{S,\etav}\\
 \Big(\prod_{j=1}^{m'}\delta(\theta'_{j},\theta_{j}) \prod_{i=1}^{n'}\delta(\eta'_{i}-\eta_{i+k})\prod_{l=1}^{m-m'}\delta(\theta_{l+m'}-\eta_{k-l+1}) \Big)\frac{n!}{(n-n')!}\frac{m!}{(m-m')!}k!.
\end{multline}
which is obtained by first computing the matrix element with unsymmetrized creators and annihilators, and using $(z(\eta)h)_{n}(\boldsymbol{\xi})=\sqrt{n-1}h_{n}(\eta,\hat{\boldsymbol{\xi}})$.

This shows \eqref{matrixd} and therefore concludes the proof.
\end{proof}

\section{Uniqueness of Araki expansion} \label{sec:expansionunique}

\begin{proposition}\label{proposition:expansionunique}
For any $m,n\in\nbb_0$, let $g_{mn}\in\dcal(\rbb^{m+n})'$ with $\gnorm{g_{mn}}{m \times n}^{\omega} < \infty$. Then,
\begin{equation}\label{arakig}
A := \sum_{m,n=0}^\infty \int \frac{d^m \theta  \,d^n \eta}{m!n!} \, g_{mn}(\thetav,\etav) \,z^{\dagger m}(\thetav)z^n(\etav)
\end{equation}
defines an element of $\mathcal{Q}^{\omega}$, and $\cme{m,n}{A}(\thetav,\etav) = \operatorname{Sym}_{S,\thetav} \operatorname{Sym}_{S,\etav} g_{mn}(\thetav,\etav)$.
\end{proposition}

\begin{proof}
Since we want to show that $A$ is a quadratic form defined between finite particle number vectors, it is enough to consider \eqref{arakig} evaluated between vectors of finite particle number; hence, the sum on the right hand side of \eqref{arakig} is finite. By Prop.~\ref{pro:zzdcrossnorm}, we know that since $\onorm{g_{mn}}{m \times n}< \infty$, every summand is a well-defined quadratic form in $\mathcal{Q}^{\omega}$; and therefore the sum is in $\mathcal{Q}^{\omega}$ as well; the ``integral is finite'' due to the bound $\onorm{g_{mn}}{m \times n}< \infty$ (see this by computing the scalar product of $A$ between finite particle number vectors explicitly). It remains to show that $\cme{m,n}{A} = \operatorname{Sym}_S g_{mn}$:
\begin{multline}
\cme{m,n}{A}(\thetav ,\etav )=\sum_{m',n'\geq 0}\int \frac{d^{m'}\theta'd^{n'}\eta'}{m'!n'!}g_{m'n'}(\thetav ',\etav ')
\cmelong{m,n}{z^{\dagger m'}(\thetav ')z^{n'}(\etav ')}(\thetav,\etav)\\
=\sum_{m',n'\geq 0}\int \frac{d^{m'}\theta'd^{n'}\eta'}{m'!n'!}g_{m'n'}(\thetav ',\etav ')m!n! \delta_{m,m'}\delta_{n,n'}
\operatorname{Sym}_{S,\thetav} \delta^m(\thetav-\thetav')
\operatorname{Sym}_{S,\etav} \delta^ n(\etav-\etav')\\
=\operatorname{Sym}_{S,\thetav} \operatorname{Sym}_{S,\etav} g_{mn}(\thetav ,\etav ),
\end{multline}
where in the second equality we made use of Prop.~\ref{proposition:fmnbasis}.
\end{proof}

\section{Existence of the Araki expansion} \label{sec:expansionexist}

We can now show that any $A \in \qf^{\omega}$ can be expanded into a series with coefficients the $\cme{m,n}{A}$.
\begin{theorem}\label{theorem:arakiexpansion}
If $A \in \mathcal{Q}^{\omega}$, then in the sense of quadratic forms,
\begin{equation}\label{eq:arakiexp}
A = \sum_{m,n=0}^\infty \int \frac{d^m \theta \, d^n \eta}{m!n!} \,\cme{m,n}{A} (\thetav,\etav)
z^{\dagger m}(\thetav)z^ n(\etav).
\end{equation}
\end{theorem}

\begin{proof}
According to Prop.~\ref{proposition:fmnbound}, since $A \in \mathcal{Q}^{\omega}$, we have that $\gnorm{\cme{m,n}{A}}{m\times n}^{\omega}<\infty$, thus by Prop.~\ref{proposition:expansionunique} the right-hand side of \eqref{eq:arakiexp} exists in $\qf^{\omega}$. To establish equality in \eqref{eq:arakiexp}, we need to show that both sides agree in all matrix elements. In view of Prop.~\ref{proposition:fmninversion}, it suffices to show that they agree in all $\cme{m,n}{\cdotarg}$; that is, we need to prove
\begin{equation}
\cme{m,n}{A}(\thetav ,\etav )=\sum_{m',n'\geq 0}\int \frac{d^{m'}\theta' d^{n'}\eta'}{m'!n'!}\cme{m',n'}{A}(\thetav ',\etav ')f_{mn}[z^{\dagger m'}z^{n'}](\thetav ,\etav ).
\end{equation}
But this is the case by Prop.~\ref{proposition:expansionunique}, with $g_{mn}=\cme{m,n}{A}$:
\begin{multline}
\cme{m,n}{A}(\thetav ,\etav )=\sum_{m',n'\geq 0}\int \frac{d^{m'}\theta' d^{n'}\eta'}{m'!n'!}\cme{m',n'}{A}(\thetav ',\etav ')f_{mn}[z^{\dagger m'}z^{n'}](\thetav ,\etav )\\
=\sum_{m',n'\geq 0}\int \frac{d^{m'}\theta' d^{n'}\eta'}{m'!n'!}\cme{m',n'}{A}(\thetav ',\etav ')m!n! \delta_{m,m'}\delta_{n,n'}
\operatorname{Sym}_S \delta^m(\thetav-\thetav')
\operatorname{Sym}_S\delta^ n(\etav-\etav')\\
=\cme{m,n}{A}(\thetav ,\etav ),
\end{multline}
where we have used that the $\cme{m,n}{A}$ are $S$-symmetric by Prop.~\ref{proposition:fmnsymm}.
\end{proof}

\section{Behavior of coefficients under translations and boosts}\label{sec:ftransboost}

The \emph{Araki coefficients} $\cme{m,n}{A}$ behave under Poincar\'e transformations as follows.
\begin{proposition}\label{proposition:fmnpoincare}
For any $A \in \mathcal{Q}^{\omega}$, $x \in \rbb^2$, and $\lambda \in \rbb$,
\begin{equation}\label{eq:fmnpoincare}
\cmelong{m,n}{U(x,\lambda)AU(x,\lambda)^\ast}(\thetav,\etav) = \exp\Big(
i\sum_{j=1}^{m}p(\theta_{j})\cdot x-i\sum_{i=1}^{n}p(\eta_{i})\cdot x
\Big) \cme{m,n}{A}(\thetav-\lambdav,\etav-\lambdav),
\end{equation}
where $\thetav-\lambdav = (\theta_1-\lambda, \ldots, \theta_m-\lambda)$, similarly for $\etav$.
\end{proposition}

\begin{proof}
From the definition \eqref{eq:fmndef}, we have
\begin{equation}
\cmelong{m,n}{U(x,\lambda)AU(x,\lambda)^{*}}(\thetav ,\etav )=\sum_{C\in \ccal_{m,n}} (-1)^{|C|} \delta_C \, S_C \,
\hscalar{ \lvector{C}{\thetav}}{ U(x,\lambda)A U(x,\lambda)^{*} \rvector{C}{\etav} }.
\end{equation}
Since $U(x,\lambda)^\ast\rvector{C}{\etav} = \exp(-i \sum_{k \not\in \{r_j\}} p(\eta_{k-m})\cdot x) \,\rvector{C}{\etav-\lambdav}$, similarly for $\lvector{C}{\thetav}$, we find
\begin{multline}
\cmelong{m,n}{U(x,\lambda)AU(x,\lambda)^{*}}(\thetav ,\etav )= \exp(i \sum_{t \not\in \{l_j\}} p(\theta_{t})\cdot x-i \sum_{k \not\in \{r_j\}} p(\eta_{k-m})\cdot x)\times\\
\times \sum_{C\in \ccal_{m,n}} (-1)^{|C|} \delta_C \, S_C \,
\hscalar{ \lvector{C}{\thetav - \lambdav}}{ A  \rvector{C}{\etav-\lambdav} }.
\end{multline}
Since the factors $\delta_C$, $S_C$ depend only on differences of rapidities, they are invariant under $U(0,\lambda)$. Hence, the result follows.
\end{proof}

\section{Behavior of coefficients under reflections} \label{sec:freflect}
The behavior of the coefficients under space-time reflections (which are represented by antiunitaries $J$) is a bit more involved than the one under translations and boosts.
To study how the coefficients behaves under space-time reflections, we introduce for any contraction $C=(m,n,\{(\ell_j,r_j)\})$, the ``reflected`` contraction $C^J$ given by
$C^J=(n,m,\{(r_j-m,\ell_j+n)\})$. This contraction is the one that is obtained from $C$ by exchanging $\ell$ with $r$, and $m$ with $n$.
We also introduce a factor associated with $C$ that will become relevant in later computations:
\begin{equation}
 R_C := \prod_{j=1}^{|C|} \Big(1-\prod_{p_j=1}^{m+n} S^{(m)}_{l_j,p_j} \Big).\label{rc}
\end{equation}
This factor is related in a certain sense to the interaction of the model; indeed we note that in the free case $S=1$, one has $R_C = \delta_{|C|,0}$, and in the case $S=-1$ one has $R_C=0$ when $m+n$ is even.
\begin{lemma}
$R_C$ has the property
\begin{equation}\label{eq:RSJ}
\delta_{C}(\etav,\thetav)S_{C}(\etav,\thetav)R_{C}(\etav,\thetav)
=(-1)^{|C|}\delta_{C^{J}}(\thetav,\etav)S_{C^{J}}(\thetav,\etav)R_{C^{J}}(\thetav,\etav),
\end{equation}
\end{lemma}
\begin{proof}
Using Lemma~\ref{lemma:scpermute}, we rewrite \eqref{eq:RSJ} equivalently as
\begin{equation}\label{eq:RSJperm}
\delta_{C}(\etav,\thetav) S^\sigma(\etav) S^\rho(\thetav) R_{C}(\etav,\thetav)
=(-1)^{|C|}\delta_{C^{J}}(\thetav,\etav)S^{\sigma'}(\thetav) S^{\rho'}(\etav)R_{C^{J}}(\thetav,\etav),
\end{equation}
where $\sigma,\rho$ and $\sigma',\rho'$ correspond to the contractions $C$ and $C^J$, respectively, in the sense of Eq.~\eqref{Scontrperm}. We note that $\sigma$ is literally the same as in \eqref{eq:scpermutations}, instead $\rho'$ works out to be
\begin{equation}
\rho'=\begin{pmatrix} 1 & & \ldots & &  & m
\\   l_{|C|} & \ldots & l_{1} & 1 & \ldots \;\hat{l}\; \ldots & m & \end{pmatrix} .\quad
\end{equation}
By introducing the permutation
\begin{equation}
\pi=\begin{pmatrix}1 && \ldots \hat{l} \ldots && m & l_1 & \ldots & l_{|C|}
\\ l_{|C|} & \ldots & l_1 & 1 && \ldots \hat{l} \ldots && m \end{pmatrix}
\end{equation}
we notice that $\rho' = \sigma\circ 	\pi$, and therefore we have that $S^{\rho'}(\etav) = S^{\pi}(\etav^\sigma) S^\sigma(\etav)$ by Eq.~\eqref{eq:Scompose}. Correspondingly, one finds that $S^{\sigma'}(\thetav) = S^{\tau}(\thetav^\rho) S^\rho(\thetav)$ where the permutation $\tau$ is defined analogously to $\pi$. Explicitly, one has
\begin{align}
S^{\pi}(\etav^{\sigma}) &= \prod_{j=1}^{|C|}\prod_{p_{j}=1}^{m}S_{l_{j},p_{j}},\\
S^{\tau}(\thetav^{\rho}) &= \prod_{j=1}^{|C|}\prod_{m_{j}=m+1}^{m+n}S_{m_{j},r_{j}}.
\end{align}
and therefore
\begin{equation}
S^{\pi}(\etav^{\sigma}) S^{\tau}(\thetav^{\rho}) = \prod_{j=1}^{|C|}\prod_{p_{j}=1}^{m+n}S^{(m)}_{l_{j},p_{j}}(\etav,\thetav).
\end{equation}
We can see that
\begin{equation}
  R_{C^J}(\thetav,\etav) S^{\pi}(\etav^{\sigma}) S^{\tau}(\thetav^{\rho}) = (-1)^{|C|} R_{C}(\etav,\theta).
\end{equation}
Indeed, using \eqref{rc}, the formula above writes explicitly as:
\begin{equation}
\prod_{j=1}^{|C^{J}|}\Big( 1- \prod_{p_{j}=1}^{m+n}S_{r_{j}-m,p_{j}}^{(n)}(\thetav,\etav) \Big) \cdot \Big( \prod_{j=1}^{|C|}\prod_{p_{j}=1}^{m+n}S^{(m)}_{l_{j},p_{j}}(\etav,\thetav) \Big) = (-1)^{|C|} \prod_{j=1}^{|C|}\Big( 1- \prod_{p_{j}=1}^{m+n}S_{l_{j},p_{j}}^{(m)}(\etav,\thetav)\Big).
\end{equation}
So, to get the equality in the equation above, in particular we have to show that:
\begin{equation}
\prod_{p=1}^{m+n}S_{l_j,p}^{(m)}(\etav,\thetav) \cdot \prod_{p=1}^{m+n}S_{r_j-m,p}^{(n)}(\thetav,\etav)=1.
\end{equation}
To see this, consider the first of the two factors. On the support of $\delta_C$, we can replace $l_j$ with $r_j$, but noting that $l_j\leq m$ and $r_j >m$, so that the $S^{(m)}$ factors are turned into their inverse:
\begin{equation}
\prod_{p=1}^{m+n}S^{(m)}_{l_j,p}(\etav,\thetav)=\prod_{p=1}^{m+n}S_{r_j,p}^{(m)}(\etav,\thetav)^{-1}.
\end{equation}
Now we renumber the variables: The $r_j$-th component of $(\etav,\thetav)$ is $\thetav_{r_j -m}$, or alternatively speaking, the $(r_j - m)$-th component of $(\thetav,\etav)$. We get from there,
\begin{equation}
S_{r_j,p}^{(m)}(\etav,\thetav)=S_{r_j -m,p'}^{(n)}(\thetav,\etav),
\end{equation}
where $p'=p+n$ if $p\leq m$, and $p'=p-m$ if $p>m$. (Note that if $p$ is ``left'' then $p'$ is ``right'' and vice versa.)

Taking the product over all $p$, inserting into the above, and renaming the product index, we have
\begin{equation}
\prod_{p=1}^{m+n}S_{l_j,p}^{(m)}(\etav,\thetav)=\prod_{p=1}^{m+n}S^{(n)}_{r_j -m,p}(\thetav,\etav)^{-1}.
\end{equation}
and that is what we claimed.
\end{proof}

We will see now in the following Proposition how the factor $R_C$ enters the formula which describes the behaviour of the coefficients $\cme{m,n}{A}$ under the action of space-time reflections.
\begin{proposition}\label{proposition:fmnreflected}
For any $A \in \qf^{\omega}$,
\begin{equation} \label{eq:fmnreflected}
\cme{m,n}{J A^\ast J}(\thetav,\etav) = \sum_{C \in \ccal_{m,n}}
(-1)^{|C|} \delta_C S_C
R_C(\thetav,\etav)
\cme{n-|C|,m-|C|}{A}(\hat \etav ,\hat\thetav).
\end{equation}
\end{proposition}

\begin{proof}
First we notice that given any contraction $C\in\ccal_{m,n}$, we have that $J \lvector{C}{\cdotarg} = \rvector{C^J}{\cdotarg}$ and $J \rvector{C}{\cdotarg} = \lvector{C^J}{\cdotarg}$. Replacing $A$ with $JA^{*}J$ in the definition of $\cme{m,n}{A}$, given by Eq.~\eqref{eq:fmndef}, we obtain:
\begin{equation}
\cme{m,n}{JA^{*}J}(\thetav,\etav)=
\sum_{C\in\ccal_{m,n}}(-1)^{|C|}\delta_{C}S_{C}(\thetav,\etav)
\hscalar{ \lvector{C^J}{\etav} }{ A \, \rvector{C^J}{\thetav}}.
\end{equation}
Using Prop.~\ref{proposition:fmninversion} in the formula above, we find
\begin{equation}\label{maineq}
\cme{m,n}{JA^{*}J}(\thetav,\etav)=
\;\;\sum_{\mathclap{\substack{C\in\ccal_{m,n} \\ C'\in\ccal_{n-|C|,m-|C|} }}}\;\;
(-1)^{|C|}\delta_{C} S_{C}(\thetav,\etav ) \delta_{C'} S_{C'}(\hat{\etav},\hat{\thetav}) \cme{n-|C|-|C'|,m-|C|-|C'|}{A}(\Hat{\Hat{\etav}},\Hat{\Hat{\thetav}}),
\end{equation}
where $\hat{\thetav},\hat\etav$ indicates that the variables contracted in $C$ are left out, instead $\Hat{\Hat{\thetav}},\Hat{\Hat\etav}$ indicates that the variables which are contracted $C \dot\cup C^{\prime J}$ are left out.
Now we apply Lemma~\ref{lemma:contractcompose} (with $C^{\prime J}$ in place of $C'$), and setting $D:=C \dot\cup C^{\prime J}$, we reorganize the sum as follows:
\begin{equation}\label{eq:dsumJ}
\cme{m,n}{JA^{*}J}(\thetav,\etav)=\sum_{D\in\ccal_{m,n}}(-1)^{|D|} \delta_{D} S_{D}(\thetav,\etav)
\underbrace{\Big(\sum_{D = C \dot\cup C^{\prime J}} (-1)^{|C'|} \frac{S_{C^{\prime}}(\hat\etav,\hat\thetav)}{S_{C^{\prime J}}(\hat\thetav,\hat\etav)} \Big)}_{(\ast)}
\cme{n-|D|,m-|D|}{A}(\hat{\hat{\etav}},\hat{\hat{\thetav}}).
\end{equation}
It remains to compute the sum $(\ast)$.
Using Eq.~\eqref{eq:RSJ}, and taking the product with $\delta_D$ into account, we have
\begin{equation}\label{eq:dsumrewritten}
 (\ast) = \sum_{D = C \dot\cup C^{\prime J}} \frac{R_{C^{\prime J}}(\hat\thetav,\hat\etav)}{R_{C^{\prime}}(\hat\etav,\hat\thetav)}
  = \sum_{D = C \dot\cup C^{\prime J}}  \prod_{j \in \{r_i'-n\}} \frac{1-a_j}{1-a_j^{-1}}=\sum_{C^{\prime J}\subset D}\prod_{j\in \{ r_i'-n\}}\frac{1-a_j}{1-a_j^{-1}} \cdot \prod_{ k \in \{  \ell_j\}} 1,
\end{equation}
where
\begin{equation}
a_j:=\prod_{p=1}^{m+n} S_{j,p}^{(m )}(\thetav,\etav), \quad j\in \{ r_i'-n\}.
\end{equation}
(here we used the fact that
\begin{equation}
\delta_{C} \prod_{p=1}^{m+n} S_{j,p}^{(m)}(\thetav,\etav)= \delta_{C} \prod_{p=1}^{m+n} S_{j,p}^{(m-|C|)}(\hat\thetav,\hat\etav) \cdot \prod_{q=1}^{|C|}S_{r_q,j}S_{j,l_q},
\end{equation}
where the last two $S$-factors cancel due to the delta distribution.)

and where
\begin{equation}
a_j^{-1}=\prod_{p=1}^{m+n}S_{l'_j,p}^{(n)}(\etav,\thetav).\label{ajinv}
\end{equation}
To see \eqref{ajinv}, consider the right hand side of the equation above. On the support of $\delta_D$, we can replace $l'_j$ with $r'_j$, but noting that $l'_j\leq n$ and $r'_j >n$, so that the $S^{(n)}$ factors are turned into their inverse:
\begin{equation}
\prod_{p=1}^{m+n}S^{(n)}_{l'_j,p}(\etav,\thetav)=\prod_{p=1}^{m+n}S_{r'_j,p}^{(n)}(\etav,\thetav)^{-1}.
\end{equation}
Now we renumber the variables: The $r'_j$-th component of $(\etav,\thetav)$ is $\thetav_{r'_j -n}$, or alternatively speaking, the $(r'_j - n)$-th component of $(\thetav,\etav)$. We get from there,
\begin{equation}
S_{r'_j,p}^{(n)}(\etav,\thetav)=S_{r'_j -n,p'}^{(m)}(\thetav,\etav),
\end{equation}
where $p'=p+m$ if $p\leq n$, and $p'=p-n$ if $p>n$. (Note that if $p$ is ``left'' then $p'$ is ``right'' and vice versa.)

Taking the product over all $p$, inserting into the above, and renaming the product index, we have
\begin{equation}
\prod_{p=1}^{m+n}S_{l'_j,p}^{(n)}(\etav,\thetav)=\prod_{p=1}^{m+n}S^{(m)}_{r'_j -n,p}(\thetav,\etav)^{-1}.
\end{equation}
and that is what we claimed in \eqref{ajinv}.

Using the following \emph{distributional law}:
\begin{equation}
\prod_{a\in A}(\alpha_{a}+\beta_{a})=\sum_{B\subset A}\prod_{b\in B}\alpha_{b}\prod_{d\in B^{c}}\beta_{d},
\end{equation}
we find
\begin{equation}
 (\ast) =  \prod_{j \in \{\ell_j\} \cup \{r_j'-n\}} \Big(1 + \frac{1-a_j}{1-a_j^{-1}}\Big)
=  \prod_{j \in \{\ell_j\} \cup \{r_j'-n\}} (1-a_j)
 = R_D(\thetav,\etav).
\end{equation}
Inserting this result into \eqref{eq:dsumJ} concludes the proof.
\end{proof}

%% file: operators.tex
\chapter{Operators and quadratic forms}\label{sec:operators}

Most of the material in the following chapter is due to H. Bostelmann.

\section{Locality of quadratic forms} \label{sec:weaklocality}

In the previous chapter we discussed the existence and uniqueness of the Araki decomposition for any quadratic form $A \in \qf^\omega$. So in order to discuss the locality properties of $A$ in terms of the Araki decomposition, we need a notion of locality that is adapted to quadratic forms $A \in \qf^\omega$. This locality is studied in terms of commutators with the wedge-local field $\phi$, but we have first to clarify in which sense these commutators are well defined.

If $f \in \dcal^\omega(\rbb^2)$, then $\zd(f^+)$ maps $\fpno$ into $\fpno$, so $A\zd(f^+)$ is well-defined as a quadratic form; indeed, we have the following Lemma:
\begin{lemma}
For $f \in \dcal^\omega(\rbb^2)$ one has $\onorm{A\zd(f^+)}{k} \leq \sqrt{k+1}\onorm{f^+}{2}\onorm{A}{k+1}$.
\end{lemma}
\begin{proof}
From \eqref{eq:aomeganorm} we have:
\begin{equation}
\onorm{A\zd(f^+)}{k}=\frac{1}{2}||e^{-\omega(H/\mu)}Q_{k}A\zd(f^+)Q_{k}||+\frac{1}{2}||Q_{k}A\zd(f^+)Q_{k}e^{-\omega(H/\mu)}||\label{onormAzd}
\end{equation}
We estimate the first norm on the right hand side of \eqref{onormAzd} as follows:
\begin{eqnarray}
||e^{-\omega(H/\mu)}Q_{k}A\zd(f^+)Q_{k}|| &=&  ||e^{-\omega(H/\mu)}Q_{k}AQ_{k+1}\zd(f^+)Q_{k}||\nonumber\\
&\leq& ||e^{-\omega(H/\mu)}Q_{k}AQ_{k+1}||\cdot ||\zd(f^+)Q_{k}||\nonumber\\
&\leq& ||e^{-\omega(H/\mu)}Q_{k+1}AQ_{k+1}||\cdot ||f^+||_{2}\sqrt{k+1}\nonumber\\
&\leq& ||e^{-\omega(H/\mu)}Q_{k+1}AQ_{k+1}|| \sqrt{k+1} \onorm{f^+}{2}.
\end{eqnarray}
And similarly, for the second norm we find:
\begin{eqnarray}
||Q_{k}A\zd(f^+)Q_{k}e^{-\omega(H/\mu)}|| &=& ||Q_{k}AQ_{k+1}e^{-\omega(H/\mu)}e^{\omega(H/\mu)}\zd(f^+)Q_{k}e^{-\omega(H/\mu)}||\nonumber\\
&\leq& ||Q_{k+1}AQ_{k+1}e^{-\omega(H/\mu)}||\cdot ||e^{\omega(H/\mu)}\zd(f^+)e^{-\omega(H/\mu)}Q_{k}||\nonumber\\
&\leq& ||Q_{k+1}AQ_{k+1}e^{-\omega(H/\mu)}||\sqrt{k+1}\onorm{f^+}{2}.
\end{eqnarray}
Putting together these two estimates, we find from \eqref{onormAzd} the claimed result.
\end{proof}
In analogous way, we have that also the product $\zd(f^+)A$, and the products of $A$ with $z(f^-)$, $\phi(f)$, $\phi'(f)$ from the left or the right are well-defined in $\qf^\omega$; this implies that we can define the commutator $[A,\phi(f)]:=A\phi(f)-\phi(f)A \in \qf^\omega$. Given this, we can define our notion of locality as follows.

\begin{definition}\label{definition:omegalocal}
  Let $A \in \qf^\omega$. We say that $A$ is \emph{$\omega$-local} in $\wcal_x$ (the right wedge with edge at $x$) iff
\begin{equation}\label{eq:acommute}
     [A,\phi(f)] = 0
\quad
\text{for all }
    f \in \dcal^\omega(\wcal_x'), \text{ as a relation in $\qf^\omega$}.
\end{equation}
$A$ is called $\omega$-local in $\wcal_x'$ iff $J A^\ast J$ is $\omega$-local in $\wcal_x$. $A$ is called $\omega$-local in the double cone $\ocal_{x,y} = \wcal_x \cap \wcal'_y$ iff it is $\omega$-local in both $\wcal_x$ and $\wcal'_y$.
\end{definition}

 In the following lemma we characterize better this notion of locality. Such a Lemma is formulated only for the standard right wedge $\wcal$, but actually it holds for other regions in analogous way.

\begin{lemma}\label{lemma:localitychar}
 Let $\omega$ be an indicatrix, and $A \in \qf^\omega$. The following conditions are equivalent:
\begin{enumerate}
\renewcommand{\theenumi}{(\roman{enumi})}
\renewcommand{\labelenumi}{\theenumi}

 \item \label{it:charlocal}
    $A$ is $\omega$-local in $\wcal$.
 \item \label{it:charcommutator}
    $[A,\phi(f)]= 0$ for all $f \in \dcal^\omega(\wcal')$.
 \item \label{it:charomegavar}
    For every $\psi,\chi\in\fpno$, there exists an indicatrix $\omega'$ such that
   $\hscalar{\psi}{[A,\phi(f)]\chi} = 0$ for all $f \in \dcal^{\omega'}(\wcal')$.
 \item \label{it:chartempered}
    For every $\psi,\chi\in\fpno$, it holds that
    $\hscalar{\phi(x)\psi}{A \chi} = \hscalar{\psi}{A \phi(x) \chi}$ for $x \in\wcal'$,
    in the sense of tempered distributions.
\end{enumerate}
\end{lemma}

\begin{proof}
First we note that \ref{it:chartempered} is well defined: Indeed, since $\gnorm{A}{k}^\omega < \infty$, the matrix element $\hscalar{\psi}{A \chi}$ is well defined (by continuous extension) if $\psi,\chi \in \fpn$ and at least \emph{one} of $\psi,\chi$ is in $\Hil^{\omega}$.

Since $\phi(f) \fpn \subset \fpn$, and since the map $\scal(\rbb^2) \to \Hil$, $f \mapsto \phi(f)\chi$ is continuous with respect to the Schwarz and the Hilbert space norms in the corresponding spaces, the map $f \mapsto\hscalar{\psi}{A \phi(f) \chi}$ is a tempered distribution; the same holds for $\hscalar{\phi(\bar f)\psi}{A \chi}$ analogously.---Now the equivalence \ref{it:charlocal}$\Leftrightarrow$\ref{it:charcommutator} is true due to the definition; \ref{it:charcommutator}$\Rightarrow$\ref{it:charomegavar} is trivial since \ref{it:charomegavar} is a special case of \ref{it:charcommutator} where we choose $\omega'=\omega$; \ref{it:charomegavar}$\Rightarrow$\ref{it:chartempered} follows because $\dcal^{\omega'}(\wcal')$ is dense in $\dcal(\wcal')$ (see \cite{Bjoerck:1965}); and \ref{it:chartempered}$\Rightarrow$\ref{it:charcommutator} holds because $\dcal^\omega(\wcal')\subset\scal(\rbb^2)$ since $\dcal(\wcal')\subset\scal(\rbb^2)$ (and $\dcal^{\omega'}(\wcal')$ is dense in $\dcal(\wcal')$).
\end{proof}

The notion of $\omega$-locality is clearly weaker than the usual notion of locality. If $A$ is just a quadratic form we cannot write down well-defined commutators of $A$ for example with the unitary operators $\exp i\phi(f)^-$, or with a general operator $B \in \A(\wcal_x')$, and the notion of $\omega$-locality does not give us any information regarding such commutators. However, $\omega$-locality is not so weak as it might appear at first glance: In the following section we will try to clarify the relation between $\omega$-locality and the usual notion of locality.

\section{Relations to usual notions of locality}  \label{sec:localityequiv}

Now we would like to pass from $\omega$-local quadratic forms to \emph{operators} which are $\omega$-local and also to investigate the relation of $\omega$-locality to usual notions of locality. We are aiming to characterize closed operators which are affiliated with the local algebras. We found that this class of operators are still manageable to characterize within the framework of the Araki expansion. On the other hand, smeared pointlike fields are typically closable where they exist, see \cite{FreHer:pointlike_fields,Wollenberg:1985}, and we will also find some closable operators in our models explicitly, see Sec.~\ref{sec:localexamples}.

\begin{proposition}\label{proposition:locality}
Let $\rcal$ be one of the regions $\wcal_x$, $\wcal_y'$, $\ocal_{x,y}$ for some $x,y\in\rbb^2$.
\begin{enumerate}
 \renewcommand{\theenumi}{(\roman{enumi})}
 \renewcommand{\labelenumi}{\theenumi}

\item \label{it:abounded}
  Let $A$ be a bounded operator; then $A$ is $\omega$-local in $\rcal$ if and only if $A \in \A(\rcal)$.

\item \label{it:aclosed}
  Let $A$ be a closed operator
  with core $\fpno$, and $\fpno \subset \domain{A^\ast}$.
  Suppose that
\begin{equation}\label{eq:weyldomaincond}
\forall g \in \dcal_\rbb^\omega(\rbb^2): \quad  \exp (i \phi(g)^-) \fpno \subset \domain{A}.
\end{equation}
  Then $A$ is $\omega$-local in $\rcal$ if and only if it is affiliated with $\A(\rcal)$.

\item \label{it:aising}
 In the case $S=-1$, statement \ref{it:aclosed} is true even without the condition \eqref{eq:weyldomaincond}.

\end{enumerate}
\end{proposition}

\emph{Remark on the conditions}:
 We could make the condition \eqref{eq:weyldomaincond} in the hypothesis \ref{it:aclosed} of Prop.~\ref{proposition:locality} a bit weaker by requiring that $\exp (i t \phi(g)^-) \fpno \subset \domain{A}$ for small $|t|$, and we could also restrict to $\supp g \subset \rcal'$. With this, one can see that the proof of this Proposition (see below) proceeds analogously.

\begin{proof}
We will prove this proposition only in the case where $\rcal$ is the standard right wedge $\rcal=\wcal$ and its associated algebra   $\A(\rcal)=\A(\wcal)=\M$. The other cases  $\rcal=\wcal_x$ or $\rcal=\wcal_y'$ can be obtained from the case before by applying Poincar\'e transformations (and using the property of covariance of the associated algebra); analogously the case $\rcal=\ocal_{x,y}$ can be obtained from the case $\rcal=\wcal$ by applying Poincar\'e transformations to the right and left wedges of the intersection.

We notice that \ref{it:abounded} is a special case of \ref{it:aclosed} since a bounded operator is in particular a closed operator with domain the entire Hilbert space $\Hil$. Also, obviously a bounded operator is affiliated to the von Neumann algebra $\A(\rcal)$ if and only if it is an element of this von Neumann algebra. We will now prove \ref{it:aclosed}; for this
we consider an operator $A$ which is closed and $\omega$-local in $\wcal$. We need to show that $A$ commutes with the unitaries $\exp(i\phi(g)^-)$ where $g \in \dcal^\omega_\rbb(\wcal ')$, in a way that is compatible with the domain of $A$, in the sense that each unitary $\exp(i\phi(g)^-)$ in $\A(\rcal')$ should carry the domain of $A$, $\domain{A}$, onto itself and satisfy the commutation relation there (``affiliation'' of $A$ with the von Neumann algebra $\A(\rcal)$). So, let $g \in \dcal^\omega_\rbb(\wcal ')$, we consider the ``cut-off'' series expansion:
\begin{equation}
 B_n:=\sum_{k=0}^n \frac{i^k}{k!}(\phi(g)^-)^k,
\end{equation}
with $n \in \nbb_0$. This is an operator defined at least on $\fpno$, since $\phi$ can be applied finitely many times to finite particle vectors; also its adjoint is defined there, since $B_n^{*}$ is the same as $B_n$ with $i$ replaced by $-i$, considering that $\phi$ is essentially self-adjoint on the space of vectors of finite particle number.

Since $A$ is $\omega$-local in $\wcal$ and since $B_n$ has a series expansion in terms of $\phi(g)$, with $g \in \dcal^\omega_\rbb(\wcal ')$, we have, as a consequence of Lemma~\ref{lemma:localitychar}\ref{it:chartempered}, that $B_n$ commutes with $A$:
\begin{equation}\label{eq:phincommute}
     \bighscalar{B_n^\ast \psi}{ A \chi } = \bighscalar{ A^{\ast} \psi}{ B_n \chi }
\quad
\text{for all }
     \quad \psi,\chi \in \fpno,
\end{equation}
where we used that $\psi \in \domain{A\st}$, since in the above equation we took $A$ to the left side of the scalar product, applied to $\psi$.
Since $\psi$ and $\chi$ are analytic vectors for $\phi(g)$, we have for $n \to \infty$ that the exponential series for $B_n$ converges in strong operator topology, namely that $B_n\chi \to B\chi$ and $B_n\st\psi \to B\st\psi$, where $B:=\exp i \phi(g)^-$.

Equation~\eqref{eq:phincommute} implies in the limit $n \to \infty$:
\begin{equation}\label{eq:gfcommute}
     \bighscalar{B\st \psi}{ A \chi } = \bighscalar{ A^{\ast} \psi}{ B \chi }
\quad
\text{for all }
     \quad \psi,\chi \in \fpno.
\end{equation}
By hypothesis \eqref{eq:weyldomaincond}, $B \chi \in \domain{A}$; this implies that we can take $A\st$ in \eqref{eq:gfcommute} to the right side of the scalar product, that is $\hscalar{A\st \psi}{B \chi} = \hscalar{\psi}{A B \chi}$. Since $B$ is bounded, we can do the same with $B\st$ in \eqref{eq:gfcommute}, we have $\bighscalar{B\st \psi}{ A \chi }=\bighscalar{ \psi}{ BA \chi }$; since $\psi$ can be chosen from a dense set in $\Hil$, we conclude that
\begin{equation}\label{eq:bafpncommute}
    B A \chi = A B \chi
\quad
\text{for all }
     \chi \in \fpno.
\end{equation}
We would like now to generalize \eqref{eq:bafpncommute} to general vectors $\chi \in \domain{A}$ and to more general operators $B$ in the commutant of $\M$.

First, we consider a general vector $\chi \in \domain{A}$. Since $\fpno$ is a core for $A$, we can find a sequence of vectors $(\chi_j)$ in $\fpno$ such that $\chi_j \to \chi$,  and also $A \chi_j \to A \chi$ in Hilbert space norm since $A$ is a closed operator. Since $A \chi_j \to A \chi$ we compute from Eq.~\eqref{eq:bafpncommute}:
\begin{equation}\label{eq:balimit}
  A B \chi_j = B A \chi_j \to B A \chi.
\end{equation}
Since $B$ is bounded we have that $ B \chi_j \to B \chi$; moreover since $A$ is closed and $B\chi_j \in \domain{A}$ (by \eqref{eq:weyldomaincond}), we have that $B \chi \in \domain{A}$. Hence, from \eqref{eq:balimit} we have that
\begin{equation}\label{eq:badcommute}
     AB\chi = BA\chi
\quad 
\text{for all }
     \chi \in \domain{A},\;
      B = \exp(i \phi(g)^-),\; g \in \dcal^\omega_\rbb(\wcal').
\end{equation}
By doing a similar computation as in \eqref{eq:balimit}, we can see that the same result then holds if $B$ is a finite product of Weyl operators $\exp(i \phi(g)^-)$, or a linear combination of product of Weyl operators, or the strong operator limit of linear combinations of products of Weyl operators. Thus, by the double commutant theorem, \eqref{eq:badcommute} holds for all $B \in \{\exp(i \phi(g)^-) \,\big|\, g \in \dcal^\omega_\rbb(\wcal')  \}''=\M'$. This because by the double commutant theorem the closure of $\{\exp(i \phi(g)^-) \,\big|\, g \in \dcal^\omega_\rbb(\wcal')  \}$ in the strong operator topology is equal to the bicommutant of $\{\exp(i \phi(g)^-) \,\big|\, g \in \dcal^\omega_\rbb(\wcal')  \}$, which is in turn equal to $\M'$. Due to \eqref{eq:badcommute} and the fact that $B\in \M'$, this means that $A$ is affiliated with $\M$ (we denote this by $A \affiliated \M$), as claimed.

For the converse, we consider $A \affiliated \M$ and $g \in \dcal_\rbb(\wcal')$. We need to show that $A$ is $\omega$-local in $\wcal$. For any $t \in \rbb$, we have that $\exp i t \phi(g)^- \in \M'$, since we can write $\exp i t\phi(g)^-=\exp i \phi(tg)^-$ and we know that $\M'$ is generated by $\exp i \phi(g)^-$ with $g \in \dcal_\rbb(\wcal')$. The fact that $A$ is affiliated with $\M$ implies that $A$ commutes with $\exp i t \phi(g)^-$, $g \in \dcal_\rbb(\wcal')$, in the following sense:
\begin{equation}\label{eq:wlconverse}
   \forall \psi,\chi \in \fpno\; \forall t \in \rbb: \quad
    \hscalar{e^{-i t \phi(g)^-} \psi}{A \chi}
=   \hscalar{A^\ast \psi}{e^{i t \phi(g)^-} \chi}.
\end{equation}
To pass from the Weyl operators back to the fields $\phi$ (which is needed for $\omega$-locality), we notice that since $\psi,\chi$ are analytic vectors for $\phi(g)$, we can think of these $\exp(i t \phi(g)^-)$ as series expansion in terms of $\phi$; in particular, we have that these functions are analytic in $t$ and therefore both sides of \eqref{eq:wlconverse} are real analytic in $t$. We can then compute the derivatives at $t=0$ of both sides of \eqref{eq:wlconverse} and equal the corresponding derivatives; hence we find:
\begin{equation}
\forall \psi,\chi \in \fpno\; \forall t \in \rbb: \quad
    \hscalar{\phi(g) \psi}{A \chi}
=   \hscalar{A^\ast \psi}{\phi(g) \chi}.
\end{equation}
This implies that $A$ is $\omega$-local in $\wcal$ by Lemma~\ref{lemma:localitychar}\ref{it:chartempered}. This completes the proof of \ref{it:aclosed}.

To prove \ref{it:aising}, note that in the case $S=-1$, the operators $\phi(g)^-$ are actually bounded operators, and generate the algebra $\mcal'$ \cite{Lechner:2005}. So, we do not need to consider Weyl operators. We can also restrict to $g \in \dcal_\rbb^\omega(\wcal')$ since this space is dense in $\dcal_\rbb(\wcal')$ which is dense in the space of test functions $g \in \scal(\wcal')$ considered by Lechner. The algebra $\mcal'$ generated by these operators smeared with test functions in $\dcal_\rbb^\omega(\wcal')$ is dense in the corresponding algebra constructed by Lechner \cite{Lechner:2005}; hence it is the same algebra, since it is closed by definition.
It is clear that for fields $\phi(g) \fpno \subset \fpno\subset \domain{A}$, and therefore that they transform $\fpno$ into $\domain A$. Using this instead of \eqref{eq:weyldomaincond}, we can follow a similar, but simpler, computation as for \ref{it:aclosed} without reference to Weyl operators, but only to fields, and show that $A \affiliated \mcal$.
\end{proof}

Considering these results, our strategy in order to construct local operators will be first to show $\omega$-locality, then to prove (independently) that $A$ is closed or closable, and finally to apply Proposition~\ref{proposition:locality}.

\section{Closable operators and summability}  \label{sec:Amnbounds}

In Prop.~\ref{proposition:locality} \ref{it:aclosed} we have some conditions for the affiliation of operators to local algebras, in particular one condition is that the operator is closable; but it is difficult to characterize closability of a quadratic form $A$ in terms of the Araki expansion and to apply this condition directly to examples. Also, even if we can show that $A$ is closable, we can not say much in general about its domain.

For this reason we are aiming only to find a \emph{sufficient} condition, and not a \emph{necessary} one, for the closability of a quadratic form, which would let us to apply Prop.~\ref{proposition:locality} \ref{it:aclosed}. This condition is a summability condition on the norms of the Araki coefficients $\cme{m,n}{A}$, which would imply the absolute convergence of the sum in the Araki expansion on a certain domain.

So, first we present a Proposition which gives a sufficient condition for the closability of $A$ as an operator.

\begin{proposition}\label{proposition:summable1}
 Let $A \in \qf^\omega$. Suppose that for each fixed $n$,
\begin{equation}\label{eq:summable1}
   \sum_{m=0}^\infty \frac{2^{m/2}}{\sqrt{m!}} \Big( \onorm{ \cme{m,n}{A} }{m \times n} + \onorm{ \cme{n,m}{A} }{n \times m} \Big) < \infty.
\end{equation}
Then, $A$ extends to a closed operator $A^-$ with core $\fpno$, and $\fpno \subset \domain  (A^-)\st$.
\end{proposition}

\begin{proof}
Let $k \in \nbb_0$. Using Prop.~\ref{pro:zzdcrossnorm} and the representation \eqref{eq:arakiexp}, we compute
\begin{equation}
\begin{aligned}
 \gnorm{ A Q_k e^{-\omega(H/\mu) } }{}
 & \leq \sum_{m,n =0}^\infty \frac{1}{m!n!} \gnorm{ z^{\dagger m}z^n(\cme{m,n}{A}) Q_k e^{-\omega(H/\mu) }  }{}
\\
 & \leq \sum_{m=0}^\infty \sum_{n=0}^k \frac{2}{m!n!} \frac{\sqrt{k!(k-n+m)!}}{(k-n)!} \onorm{ \cme{m,n}{A} }{m \times n} ;
\end{aligned}
\end{equation}
this means that if the infinite series on the r.h.s.~in the above equation converges, then we can extend $A Q_k e^{-\omega(H/\mu) }$ to a bounded operator.
So, we estimate $k!/n!(k-n)! \leq 2^k$ and $(k-n+m)!/(k-n)!m! \leq 2^{k-n+m}$ using $\frac{(p+q)!}{p!q!}\leq 2^{p+q}$, and we obtain
\begin{equation}\label{eq:aqknorm}
 \gnorm{ A Q_k e^{-\omega(H/\mu) } }{}
  \leq 2^{k+1} \sum_{n=0}^k  \frac{2^{-n/2}}{\sqrt{n!}} \sum_{m=0}^\infty  \frac{2^{m/2}}{\sqrt{m!}} \onorm{ \cme{m,n}{A} }{m \times n} ,
\end{equation}
and this converges since we assumed \eqref{eq:summable1}. Since $k$ was arbitrary, this allows us to define $A$ as an (unbounded) operator on $\fpno$. We can use a similar argument with the roles of $m$ and $n$ exchanged and with $A\st$ in place of $A$ to show that also $\fpno \subset \domain A\st$. Hence we have that both $A$ and $A\st$ are densely defined in $\Hil$ (as $\fpno$ is dense in $\Hil$); this implies by application of \cite[Thm.~VIII.1]{Reed:1972} (see item (b) and the proof of the theorem) that $A^- := A^{\ast\ast}$ is a closed extension of $A$ with core $\fpno$ (by definition of $A^-$), and $(A^-)\st = A^{\ast\ast\ast}=A\st$.
\end{proof}

We can also find a stricter summability condition so that also the extra condition in Prop.~\ref{proposition:locality}\ref{it:aclosed} \eqref{eq:weyldomaincond} is fulfilled.

\begin{proposition}\label{proposition:summable2}
 Let $A \in \qf^\omega$. Suppose that
\begin{equation}\label{eq:summable2}
   \sum_{m,n=0}^\infty \frac{2^{(m+n)/2}}{\sqrt{m!n!}} \onorm{ \cme{m,n}{A} }{m \times n} < \infty.
\end{equation}
Then, $A$ extends to a closed operator $A^-$ with core $\fpno$; one has $\fpno \subset \domain (A^-)\st$; and the condition \eqref{eq:weyldomaincond} is fulfilled by $A^-$.
\end{proposition}

\begin{proof}
Since the hypothesis of this Proposition is stronger than the one in Prop.~\ref{proposition:summable2}, the first part of the statement above has been already proved in the proposition before; so it only remains to show that $\exp (i \phi(g)^-) \fpno \subset \domain{A^-}$.
For this purpose it is useful the following lemma:
\begin{lemma}
With $\phi(g) = \zd(g^+) + z(g^-)$ and for $g \in \dcal^\omega(\rbb^2)$, it follows from \eqref{omegaz} that
\begin{equation}
 \gnorm{ e^{\omega(H/\mu)} \phi(g)^j e^{-\omega(H/\mu)} Q_\ell  }{}
  \leq  \sqrt{\frac{(\ell+j)!}{\ell!}}\,c_g^j , \quad \text{where } c_g := \onorm{g^+}{2} + \onorm{g^-}{2}.\label{normphij}
\end{equation}
\end{lemma}
\begin{proof}
We prove Eq.~\eqref{normphij} using induction in $j$:
\begin{multline}
\gnorm{ e^{\omega(H/\mu)} \phi(g)^j e^{-\omega(H/\mu)} Q_\ell  }{}\\
=\gnorm{ e^{\omega(H/\mu)}\phi(g)e^{-\omega(H/\mu)} Q_{\ell +j-1}e^{\omega(H/\mu)}\phi(g)^{j-1}e^{-\omega(H/\mu)}Q_\ell }{}\\
 =\gnorm{ e^{\omega(H/\mu)}\phi(g)e^{-\omega(H/\mu)} Q_{\ell +j-1}}{}\cdot \gnorm{e^{\omega(H/\mu)}\phi(g)^{j-1}e^{-\omega(H/\mu)}Q_\ell }{}\\
 \leq \sqrt{\ell +j}c_{g}\sqrt{\frac{(\ell +j -1)!}{\ell!}}c_{g}^{j-1}\\
 =c_{g}^{j}\sqrt{\frac{(\ell+j)!}{\ell!}}.
\end{multline}
where in the fourth inequality we made use of \eqref{omegaz} and of \eqref{normphij} with $j$ replaced by $j-1$.
\end{proof}
Since $\fpn$ consists of analytic vectors for $\phi(g)$, we can write  $e^{i \phi(g)^-} $ as a (convergent) exponential series; we then obtain for $k \geq \ell$ and using \eqref{normphij},
\begin{multline}
 \gnorm{ P_k e^{\omega(H/\mu)} e^{i \phi(g)^-} e^{-\omega(H/\mu)} Q_\ell  }{}
  \leq \sum_{j=k-\ell}^\infty \frac{1}{j!}  \gnorm{ e^{\omega(H/\mu)} \phi(g)^j e^{-\omega(H/\mu)} Q_\ell  }{}
\\
\leq \sum_{j=k-\ell}^{\infty}\frac{1}{j!}\sqrt{\frac{(\ell + j)!}{\ell!}} c_{g}^{j}
\\
=\sum_{p=0}^{\infty}\frac{1}{(p+k-\ell)!}\sqrt{\frac{(p+k)!}{\ell!}}c_{g}^{p+k-\ell}
\\
\leq c_g^{k-\ell} \frac{2^{k/2}}{\sqrt{(k-\ell)!}} \sum_{p=0}^\infty \frac{(\sqrt{2} c_g)^p}{\sqrt{p!}}
\leq c_g' \frac{(\sqrt{2}c_g)^k}{\sqrt{(k-\ell)!}}\label{powerphinorm}
\end{multline}
with some constant $c_g'>0$ depending on $g$. (Notice that in the first inequality we can replace $\phi(g)^-$ with $\phi(g)$ since $e^{i\phi(g)^-}$ is applied to finite particle number vectors.)

Now let $\psi \in \Hil_\ell^\omega$. Using the estimates \eqref{eq:aqknorm}, \eqref{powerphinorm} and $P_k \leq Q_k$, we have
\begin{multline}
 \gnorm{ A P_k e^{i \phi(g)^-} \psi  }{}
  \leq \gnorm{ A Q_k e^{-\omega(H/\mu)}  }{}
 \gnorm{ P_k e^{\omega(H/\mu)} e^{i \phi(g)^-} e^{-\omega(H/\mu)} Q_\ell  }{}
 \gnorm{  e^{\omega(H/\mu)} \psi }{}
\\
  \leq
 2 c_g' \gnorm{ e^{\omega(H/\mu)} \psi}{}
\frac{(\sqrt{8}c_g)^k}{\sqrt{(k-\ell)!}}
\sum_{m,n=0}^\infty
\frac{2^{(m-n)/2}}{\sqrt{m!n!}} \onorm{ \cme{m,n}{A} }{m \times n},\label{apweylnorm}
\end{multline}
where the series exists by the hypothesis \eqref{eq:summable2}. This expression is summable over $k$ since $\frac{(\sqrt{8}c_g)^k}{\sqrt{(k-\ell)!}}$ is summable over $k$ by quotient criteria. Since $Q_k \exp(i \phi(g)^-) \psi$ is a vector in $\fpno\subset \domain A^-$ due to \eqref{powerphinorm} and in particular it is a finite particle number vector, we have that $\chi_k := Q_k \exp(i \phi(g)^-) \psi \in \domain A^-$; by \eqref{apweylnorm}, the same is true for $A \chi_k$. This implies that $\chi_k$ and $A \chi_k$ are both Cauchy sequences in $\Hil$. Indeed, one has from \eqref{powerphinorm} that $Q_{k}\exp(i\phi(g)^-)\psi\rightarrow \exp(i\phi(g)^-)\psi$ for $k\rightarrow \infty$, so this is a convergent sequence in $\Hil$ and therefore a Cauchy sequence. Similarly, due to \eqref{apweylnorm}, we have:
\begin{multline}
||AQ_{k}e^{i\phi(g)^-}\psi - AQ_{\ell}e^{i\phi(g)^-}\psi||=||A(Q_{k}-Q_{\ell})e^{i\phi(g)^-}\psi||\\
\leq \sum_{j=\ell +1}^{k}||AP_{k}e^{i\phi(g)^-}\psi||\leq \sum_{j=\ell +1}^{\infty}||AP_{k}e^{i\phi(g)^-}\psi||\rightarrow 0, \quad \ell \rightarrow \infty.
\end{multline}
Hence, this is also a convergent sequence in $\Hil$ and therefore a Cauchy sequence.

This means that $(\chi_k,A \chi_k)$ converges in graph norm in $\Hil \times \Hil$, since both entries are Cauchy sequences. Since $A$ is a closed operator, the limit point of that sequence in the graph is also in the graph; this implies that $\lim \chi_k = \exp(i \phi(g)^-) \psi$ is contained in the domain of $A^-$ by definition of graph. Thus \eqref{eq:weyldomaincond} holds.
\end{proof}

\section{Examples of closable operators} \label{sec:boundsexamples}

We present now some examples where we can apply the results before and obtain closable operators.

The simplest example is where the Araki expansion of $A$ is finite, namely where only a finite number of $\cme{m,n}{A}$ is different from zero. An example of such situation is when $A$ is a polynomial in $\phi(g)$. In this case the conditions of Prop.~\ref{proposition:summable2} are fulfilled in a trivial way. In general the requirement that $A$ is local will force the Araki expansion to be infinite; however, an interesting example of a local operator with a finite Araki expansion can be found in Sec.~\ref{sec:evenexample}.

We present an example of an operator with infinite Araki expansion: the ``normal ordered exponential'' of $\phi(g)$ given by
\begin{equation}\label{eq:normalexp}
  \mathclap{:}\exp c \,\phi(g)\mathclap{:} \; := \sum_{m,n \geq 0} \frac{c^{m+n}}{m!n!} \zd(g^+)^m z(g^-)^n , \quad c \in \cbb, \;g \in \scal(\rbb^2).
\end{equation}
Using \eqref{eq:fmndef} and \eqref{proposition:fmnbasis}, its Araki coefficients are
\begin{equation}
 \cmelong{m,n}{\,\mathclap{:}\exp c \,\phi(g)\mathclap{:} \,}(\thetav,\etav) = c^{m+n}
\operatorname{Sym}_{S,\thetav} \operatorname{Sym}_{S,\etav} g^+(\theta_1)\cdots g^+(\theta_m) g^-(\eta_1)\cdots g^-(\eta_n).
\end{equation}
By \eqref{eq:crossnormcomparison} (right inequality) and using $||g^+||^{m}||g^-||^n \leq (||g^+||+||g^-||)^{m} (||g^+||+ ||g^-||)^{n}$, they fulfil for any $\omega$,
\begin{equation}
  \bigonorm{ \cmelong{m,n}{\,\mathclap{:}\exp c \,\phi(g)\mathclap{:} \,} }{m \times n}
  \leq |c|^{m+n}(\gnorm{g^+}{2}+ \gnorm{g^-}{2})^{m+n},
\end{equation}
This implies that \eqref{eq:normalexp} is a well-defined element of $\qf^\omega$ by Prop.~\ref{proposition:expansionunique}. Moreover, the summability condition \eqref{eq:summable2} is also fulfilled:
\begin{eqnarray}
\sum_{m,n=0}^{\infty}\frac{2^{(m+n)/2}}{\sqrt{m!n!}}\bigonorm{ \cmelong{m,n}{\,\mathclap{:}\exp c \,\phi(g)\mathclap{:} \,} }{m \times n}
&=& \sum_{m,n=0}^{\infty}\frac{2^{(m+n)/2}}{\sqrt{m!n!}}|c|^{m+n}(\gnorm{g^+}{2}+ \gnorm{g^-}{2})^{m+n}\nonumber\\
&\leq& \sum_{m,n=0}^{\infty}\frac{c'^{m+n}}{\sqrt{m!n!}}.
\end{eqnarray}
We would also be able to show using techniques as in Sec.~\ref{sec:ftoa} that this quadratic form is $\omega$-local in $\wcal_r'$ if $\supp g \subset \wcal_r'$. Then, we can apply Proposition~\ref{proposition:summable2} and show that $\mathclap{:}\exp c \,\phi(g)\mathclap{:}$ is a closable operator and moreover, by Proposition~\ref{proposition:locality}, that its closure is affiliated with $\A(\wcal)$.

But our interest is actually in closable operators which are affiliated with the ``strictly local'' algebra $\A(\ocal)$, where $\ocal$ are the double cones. We will investigate examples of these operators in Sec.~\ref{sec:localexamples}. For the moment, we discuss some properties of the class of closable operators that we consider.

In the example above, with the methods discussed before, we have shown that $\mathclap{:}\exp c \,\phi(g)\mathclap{:}$ is closable for any $c \in \cbb$. In the free case, where $S=1$, we know even more: if $c$ is purely imaginary and $g$ is real valued, the operator is actually bounded. However, this is not visible within our methods since the exponential series does not converge \emph{absolutely} in operator norm.
In nontrivial examples we will always need to compute estimates for $\onorm{\cme{m,n}{A}}{m \times n}$, instead with the exact expressions $\cme{m,n}{A}$, so that we do not have control on the convergence of the series for the $\cme{m,n}{A}$ itself, but we can only use summability conditions like \eqref{eq:summable1} or \eqref{eq:summable2}.

We expect that whenever we have local observables in the theory which exist as Wightman fields or more in general as Jaffe fields \cite{Jaffe:1967}, they can be described by our class of closable operators. But our class is actually even more general: Indeed we do not require that our closable local operators $A$ have a common invariant domain, namely that they fulfil the condition common to Wightman fields that $\forall f\; :\; \phi(f)\mathcal{D}\subset \mathcal{D}$, where $\mathcal{D}\subset \Hil$ is dense and $\mathcal{D} \subset \domain \phi(f)$. This means that we do not know whether $n$-point functions of these operators would exist. On the other hand, they would still be meaningful local observables by showing that they are affiliated with local algebras of bounded operators.

%% file: theoremstatement.tex
\chapter{The characterization theorem for local operators} \label{sec:localitythm}

We will present in this chapter the conditions which characterize $\omega$-locality of a quadratic form $A$ in a standard double cone $\ocal_r$ of radius $r$ and center at the origin. These conditions are formulated in different ways. They can be imposed on a quadratic form, and in this case we call them conditions (A); or they can be imposed on the Araki coefficients as analytic functions of one variable, and we call them conditions (F'); they can also be imposed on the Araki coefficients as meromorphic functions of several variables, in that case we call them conditions (F). We will show in the next Chapters~\ref{sec:atofp},\ref{sec:fptof} and \ref{sec:ftoa} that these conditions are equivalent.

Note that the conditions (A), (F) and (F') depend on the indicatrix $\omega$ and on the radius $r$ of the double cone; but we will not indicate this explicitly.

\section{Formulate conditions (A)}

The condition of locality for an operator, or more in general, for a quadratic form $A\in\qf^\omega$ is better formulated using the notion of $\omega$-locality given by Def.~\ref{definition:omegalocal}. So, we have the following definition of the condition (A):

\begin{definition}\label{def:conditionA}
$A \in \qf^\omega$ fulfills condition (A) if it is $\omega$-local in $\ocal_r$.
\end{definition}

\section{Formulate conditions (F')}

Now we expand the quadratic form $A$ into the Araki series; we find that the condition of $\omega$-locality can be expressed for the coefficients $\cme{m,n}{A}$ in terms of analytic continuations of these distributions along certain lines in $\rbb^{k}$. To describe this continuation we need to introduce some notation, that can be found also in more details in Appendix~\ref{sec:graphs}.

We call a \emph{graph} $\gcal$ in $\rbb^k$ a collection of nodes which are points on the lattice $\pi \zbb^k$, connected by edges; the edges are lines parallel to the axis which connect nodes on the lattice that are next neighbors. Namely, an edge is given by $\lambdav(s) = \boldsymbol{\nu} + s \ev^{(j)}$, where $\boldsymbol{\nu}$ and $\boldsymbol{\nu} + \pi \ev^{(j)}$ are nodes of $\gcal$, $\ev^{(j)}$ is a standard basis vector of $\rbb^k$ and $0 < s < \pi$.

The \emph{tube over $\gcal$}, which we denote $\tube(\gcal)$, is the set of all $\zetav = \thetav + i \lambdav$ with $\thetav \in \rbb^k$ and $\lambdav$ on an edge of $\gcal$.

We call a \emph{CR distribution on $\tube(\gcal)$} (where \emph{CR} stands for ``Cauchy-Riemann'') a distribution on $\tube(\gcal)$ which is analytic along the edges; namely, considering the edge $\lambdav(s)$ given above, we have that $F$ is analytic in $\zeta_j$ in the specified domain, and it is still a distribution in the remaining (real) variables.
We also require that the boundary values at the nodes, namely the values in the limits $s \searrow 0$ and $s \nearrow \pi$, exist as distributions, and that if several edges meet in a common node, then the corresponding distributional boundary values are the same. In this way, we can just speak of \emph{the} distributional boundary value at a node, without specifying which is the direction of the limit.

The first graph in $\rbb^k$ that we will consider in our study is denoted with $\gcal^k_+$ and it is given as follows. It has the nodes $\lambdav^{(k,j)} = (0,\ldots,0,\pi\ldots,\pi)$, where there are $j$ entries of $\pi$ with $0 \leq j \leq k$. Its edges are the lines parallel to the axis which connect the nodes $\lambdav^{(k,j)}$ and $\lambdav^{(k,j+1)}$.

Another graph that we will consider is denoted with $\gcal^k_-$. This graph is given by $\gcal^k_+$ shifted of $-i\piv$. Namely, it has the nodes $\lambdav^{(k,-j)} = (-\pi,\ldots,-\pi,0,\ldots,0)$, where there are $j$ entries of $-\pi$, with $0 \leq j \leq k$, and the edges are lines parallel to the axis connecting next neighbors.

We also consider the union of $\gcal^k_+$ and $\gcal^k_-$, that we denote with $\gcal^k_0$, cf.~Fig.~\ref{fig:zerostair}. Since the graphs $\gcal^k_\pm$ have the node $\lambdav^{(k,0)}=0$ in common, the graph $\gcal^k_0$ has $2k+1$ nodes and $2k$ edges.

\begin{figure}
\begin{center}
\input{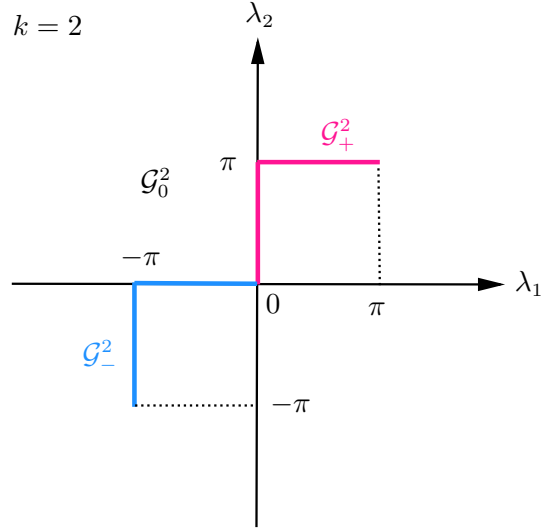}
\caption[The stair $\gcal_{0}^{2}$]{The stair $\gcal_{0}^{2}=\gcal^2_+ \cup \gcal^2_-$. As for $\gcal^2_+$ and $\gcal^2_-$ see the corresponding definitions before Def.~\ref{def:conditionFP}} \label{fig:zerostair}
\end{center}
\end{figure}

Now that we have introduced our notation, we can formulate our locality condition in terms of CR distributions on $\tube(\gcal^k_0)$ as follows:

\begin{definition} \label{def:conditionFP}
A collection $F'=(F'_{k})_{k=0}^\infty$ of distributions on $\tube(\gcal^k_0)$ fulfills condition (F') if the following holds for any $k$, and with $\thetav \in \rbb^k$ arbitrary:
\begin{enumerate}
\renewcommand{\theenumi}{(F\arabic{enumi}')}
\renewcommand{\labelenumi}{\theenumi}

\item \label{it:fpmero} \emph{Analyticity:} $F'_k$ are CR distributions on $\tube(\gcal^k_0)$.

\item \label{it:fpperiod} \emph{Periodicity:}
$
F_k' (\thetav + i \lambdav^{(k,-k)}) =
F_k' (\thetav + i \lambdav^{(k,k)}).
$

\item \label{it:fpsymm} \emph{$S$-symmetry}:
For any $1 \leq j < k$,
\begin{equation*}
F_k'(\theta_1, \ldots, \theta_j, \theta_{j+1}, \ldots , \theta_k)
= S(\theta_{j+1}-\theta_j) F_k'(\theta_1, \ldots, \theta_{j+1}, \theta_{j}, \ldots , \theta_k) .
\end{equation*}

\item \label{it:fprecursion} \emph{Recursion relations:}
For any $0 \leq m \leq k$,
\begin{equation*}
F_{k}'(\thetav +i\lambdav^{(k,-m)})
=\sum_{C\in\ccal_{m,k-m}}(-1)^{|C|}\delta_{C}
S_{C} R_{C}(\thetav) F_{k-2|C|}'(\check \thetav + i \lambdav^{(k-2|C|,m-|C|)}),
\end{equation*}
where $\check\thetav = (\theta_{m+1},\ldots,\widehat{\theta_{l_{1}}},\ldots,\widehat{\theta_{l_{|C|}}},\ldots,\theta_{k},\theta_{1},\ldots,\widehat{\theta_{r_{1}}},\ldots,\widehat{\theta_{r_{|C|}}},\ldots,\theta_{m}).$

\item \label{it:fpboundsreal}
\emph{Bounds at nodes:}
For any $j \in \{0,\ldots,k\}$,
\begin{equation*}
\onorm{ F_k'( \cdotarg + i \lambdav^{(k,j)} )}{(k-j) \times j} < \infty,
 \quad
\onorm{ F_k'( \cdotarg + i \lambdav^{(k,-j)} )}{j \times (k-j)} < \infty.
\end{equation*}

\item \label{it:fpboundsimag}
\emph{Bounds at edges:}
For each fixed $k$ there exists a $c>0$ such that for any $\lambdav$ on an edge of $\mathcal{G}^{k}_{\pm}$,
\begin{equation*}
\gnorm{ e^{\pm i \mu r \sum_j \sinh \zeta_j}e^{-\sum_{j}\oa(\pm \sinh \zeta_{j})}
F_k'( \pmb{\zeta})\big\vert_{\zetav=\cdot +i\lambdav} }{\times} \leq c.
\end{equation*}

\end{enumerate}

\end{definition}

Note that in \ref{it:fpboundsimag} the argument $\pm \sinh \zeta_j$ lies in the upper half complex plane both on $\gcal_+^k$ and $\gcal_-^k$, and thus in the domain of $\oa$. Likewise, the expression $\pm i \mu r \sum_j \sinh \zeta_j$ lies in the lower half complex plane both on $\gcal_+^k$ and $\gcal_-^k$, and therefore $e^{\pm i \mu r \sum_j \sinh \zeta_j}$ is a damping factor at large $\re \zeta_j$, and this damping factor is the stronger the larger $r$ is.

\section{Formulate conditions (F)}

Now we will present the formulation of a locality condition in terms of meromorphic functions $F_k$ on $\cbb^k$. Indeed, we will find that the Araki coefficients of a local operator can be extended as meromorphic functions to the entire multi-variables complex plane. The conditions will involve also the expressions of the residues of $F_k$ and other properties of these functions; some of the notations that we use for the residues in several complex variables can be found in Appendix~\ref{sec:bvlemma}.

For formulating these conditions we need again to introduce some notation. We denote with $\gcal^k_1$ the graph which is defined as a ``periodic extension'' of $\gcal^k_0$. Namely, it has the nodes $\lambdav^{(k,j)}$ with $j \in \zbb$, and we set, as a recursive expression, for any $j$, $\lambdav^{(k,j+2k)} := \lambdav^{(k,j)} + 2 \piv$, where $\piv=(\pi,\ldots,\pi)$.

We introduce also for given $k$ and $0 \leq j \leq k$ the vectors
\begin{equation}
   \nuv^{(k,j)} = (1,2,\ldots,k\!-\!j,\;-j,\ldots,-2,-1) \in \rbb^k.
\end{equation}

Now we formulate the locality condition in terms of properties of the functions $F_k$ as follows:

\begin{definition}\label{def:conditionF}
A collection $F=(F_{k})_{k=0}^\infty$ of functions $\cbb^k \to \bar\cbb$ fulfills conditions (F) if the following holds for any fixed $k$, and with $\zetav \in \cbb^k$ arbitrary:

\begin{enumerate}
\renewcommand{\theenumi}{(F\arabic{enumi})}
\renewcommand{\labelenumi}{\theenumi}

\item \label{it:fmero}
\emph{Analyticity:}
$F_k$ is meromorphic on $\cbb^k$, and analytic where $\im \zeta_1 < \ldots < \im \zeta_k < \im \zeta_1 + \pi$.

\item \label{it:fsymm} \emph{$S$-symmetry:}
\begin{equation*}
F_k(\zeta_1,\ldots,\zeta_j,\zeta_{j+1},\ldots,\zeta_k)
=  S(\zeta_{j+1}-\zeta_j) F_k(\zeta_1,\ldots,\zeta_{j+1},\zeta_{j},\ldots,\zeta_k) .
\end{equation*}

\item \label{it:fperiod} \emph{$S$-periodicity:}
\begin{equation*}
F_k (\zetav + 2i\pi \ev^{(j)} ) =
\Big(\prod_{\substack{i=1 \\ i \neq j}}^k S(\zeta_i-\zeta_j)\Big)  F_k (\zetav ).
\end{equation*}

\item \label{it:frecursion}

\emph{Recursion relations:}
The $F_k$ have first order poles at $\zeta_n-\zeta_m = i \pi$, where $1 \leq m < n \leq k$,
and
\begin{equation*}
\operatorname{Res}_{\zeta_n-\zeta_m = i \pi} F_{k}(\boldsymbol{\zeta})
= - \frac{1}{2\pi i }
\Big(\prod_{j=m}^{n} S_{j,m} \Big)
\Big(1-\prod_{p=1}^{k} S_{m,p} \Big)
F_{k-2}( \boldsymbol{\hat\zeta} ).
\end{equation*}

\item \label{it:fboundsreal}
\emph{Bounds on nodes:}
For each $j \in \{0,\ldots,k\}$ and $\ell \in \zbb$, we have
\begin{equation*}
\| F_k\big( \cdotarg + i \lambdav^{(k,j+k\ell)} + i 0 \nuv^{(k,j)}\big) \|_{(k-j) \times j}^{\omega} < \infty.
\end{equation*}

\item \label{it:fboundsimag}
\emph{Pointwise bounds:}
There exist $c,c'>0$ such that for all $\zetav\in\tube(\ich \gcal^k_\pm)$:
\begin{equation*}
  |F_k(\zetav)| \leq c \frac{\prod_j \exp \big(\mu r  |\im \sinh \zeta_j|+ c' \omega(\cosh \re \zeta_j)\big)}{ \operatorname{dist}(\im \zetav,\partial \ich \gcal_\pm)^{k/2}}.
\end{equation*}
\end{enumerate}
\end{definition}

We notice that the conditions (F) are ``translation invariant'' by $i\piv$; this means that if we have a family of functions $F_k$ which fulfill these conditions, then also $F_k(\cdotarg + i \piv)$ fulfill them as well (see Sec.~\ref{sec:locdoublecone} for details).

\section{Formulate the theorem}

Now the following theorem states that these conditions are equivalent:

\begin{theorem} \label{thm:localequiv}
 Let $r>0$ and an analytic indicatrix $\omega$ be fixed.
\begin{enumerate}
\renewcommand{\theenumi}{(\roman{enumi})}
\renewcommand{\labelenumi}{\theenumi}
 \item If $A \in \qf^\omega$ fulfills (A), then there is a unique set of functions $F'_k$ fulfilling (F') such that
\begin{equation}\label{cmeFprel}
\cme{m,n}{A}(\thetav,\etav) = F'_{m+n}(\thetav,\etav+i\piv),
\quad
\cme{m,n}{J A^\ast J}(\thetav,\etav) = F'_{m+n}(\thetav-i \piv,\etav).
\end{equation}
 \item If $F'_k$ fulfill (F'), then there are unique functions $F_k$ fulfilling (F), such that for $0 \leq j \leq k$ and $\ell \in \{0,k\}$,
\begin{equation}
F_{k}'\big(\thetav + i \lambdav^{(k,j-\ell)}\big)
= F_{k}\big(\thetav + i \lambdav^{(k,j-\ell)} + i0\nuv^{(k,j)} \big).
\end{equation}
 \item \label{thm:converse} If $F_k$ fulfill (F), then the following quadratic form $A$ fulfills (A):
\begin{equation}\label{ArakiF}
A := \sum_{m,n=0}^\infty \int \frac{d^{m}\theta\,d^n\eta}{m!n!}
F_{m+n}(\thetav + i \zerov, \etav+i\piv-i\zerov)  z^{\dagger m}(\thetav) z^n(\etav).
\end{equation}
\end{enumerate}
\end{theorem}

We will prove separately the three parts of the theorem in the following three Chapters~\ref{sec:atofp},\ref{sec:fptof} and \ref{sec:ftoa}.

%% file: from_a_to_fp.tex
\chapter{(A) \texorpdfstring{$\Rightarrow$}{=>} (F')}\label{sec:atofp}

In this chapter we want to show that if we have a quadratic form $A$ which is localized in the standard double cone, then its Araki coefficients $\cme{m,n}{A}$ are boundary values of a common CR distributions on a graph, which fulfil specific symmetry conditions, recursion relations and bounds.

So, our task is to prove the following theorem:

\begin{theorem}\label{theorem:AtoFp}
For any $A\in\mathcal{Q}^{\omega}$ fulfilling $(A)$, there are functions $F_k'$ fulfilling  (F') such that for any $k \in \nbb_0$ and $0 \leq j \leq k$,
\begin{equation}\label{eq:fkpboundary}
F_{k}'(\thetav+i\lambdav^{(k,j)})=\cme{k-j,j}{A}(\thetav),
\quad
F_{k}'(\thetav+i\lambdav^{(k,-k+j)})=\cme{k-j,j}{JA^{*}J}(\thetav).
\end{equation}

\end{theorem}
(These two equations are a rewritten form of equations \eqref{cmeFprel} in the previous chapter.)

\section{Define function on positive simplex} \label{sec:fkppositive}

A first step in the proof of this theorem is to consider the case where the quadratic form $A$ is $\omega$-local in a wedge, and study the properties of analyticity of its Araki coefficients $\cme{m,n}{A}$.

We prove the following lemma which is very similar to \cite[Lemma~4.1]{Lechner:2008}, but it is more general in our case because it is formulated for the quadratic form $A\in \qf^{\omega}$ and to the class of vectors in $\mathcal{H}^{\omega}$.

\begin{lemma}\label{lemma:K}
Let $A\in \qf^{\omega}$ be $\omega$-local in $\rightwedge$,  and $\psi\in P_{n_1}\hcal^{\omega}$, $\chi \in P_{n_2}\hcal^{\omega}$. There exists an analytic function $K:\strip(0,\pi) \to \mathbb{C}$ whose boundary values satisfy, $\theta \in \rbb$,
\begin{equation}\label{kboundary}
K(\theta)= \hscalar{ \psi}{ [\zd(\theta),A]\chi},
\quad
K(\theta +i\pi)= \hscalar{ \psi}{[A,z(\theta)]\chi}
\end{equation}
in the sense of distributions. Moreover, there holds the bound
\begin{equation}\label{K:L2bound}
\Big( \int d\theta\, |K(\theta +i\lambda)|^2\Big)^{1/2} \leq c_{n_1,n_2} \onorm{\psi}{2}\onorm{\chi}{2}\gnorm{A}{n_1 + n_2 +1}^{\omega}, \quad 0\leq \lambda\leq \pi.
\end{equation}
with $c_{n_1,n_2}:= 2(\sqrt{n_1 + 1}+ \sqrt{n_2 +1})$.
\end{lemma}

\begin{proof}
The proof of this Lemma mainly follows the proof of \cite[Lemma 4.1]{Lechner:2008}. But here we are going to repeat the argument again.

We considers the time zero fields $\varphi,\pi$ of $\phi$ \cite[Eq.~(3.18)]{Lechner:2008}, and the corresponding expectation values, $f \in \scal_\rbb(\rbb^2)$,
\begin{equation}\label{def:k+k-}
k_{-}(f):=\langle \psi,[\varphi(f),A]\chi\rangle,\quad k_{+}(f):=\langle \psi,[\pi(f),A]\chi\rangle.
\end{equation}
Since $A$ is $\omega$-local -- in the sense of Lemma~\ref{lemma:localitychar}\ref{it:chartempered} -- in the standard right wedge and $\phi$ is localized in the standard left wedge, we have that these $k_{\pm}$ are Schwartz distributions with support in the right half-line. Then, we apply the classical Paley-Wiener result for tempered distributions \cite[Thm.~IX.16]{ReedSimon:1975-2}, which implies that the Fourier-Laplace transforms $\tilde k_\pm$ of $k_\pm$ are analytic functions on the lower half plane, bounded by a polynomial in $p$ at infinity and by an inverse power of $\im p$ near the real line.

We choose as our function $K$ the following combination of $\tilde k_\pm$:
\begin{equation}
 K(\zeta) := \frac{1}{2}\Big( \mu \cosh(\zeta) \,\tilde k_-(-\mu\sinh\zeta) - i \tilde k_+(-\mu\sinh\zeta) \Big).\label{defK}
\end{equation}
Noticing that the strip $\strip(0,\pi)$ is mapped by $(-\sinh)$ to the lower half plane, this $K$ works out to have the proposed analyticity property; and also one can show that it has the proposed boundary values. Indeed, using the relation between $z, \zd$ and $\varphi, \pi$:
\begin{eqnarray}
\zd(f^+) &=& \frac{1}{2}(\varphi(f)-i\pi(\omega^{-1}f))\\
z(f^+) &=& \frac{1}{2}(\varphi(f_-)+i\pi(\omega^{-1}f_-)),
\end{eqnarray}
where $\omega=\mu \cosh\theta$ and $f_-(x):= f(-x)$, we compute:
\begin{equation}
\begin{aligned}
  \langle\psi,[z(f^+),A]\chi\rangle
  &=
  \frac{1}{2}(k_-(f_-)+i\,k_+(\omega^{-1}f_-))\\
  & =
  \frac{1}{2}\int dp\,
  \Big(\tilde{k}_-(p)+\frac{i\tilde{k}_+(p)}{\sqrt{p^2+m^2}}\Big)\tilde{f}(p)\\
  &=
  \frac{1}{2}\int d\theta \Big( \mu \cosh(-\theta) \; \tilde{k}_-(-\mu \sinh(-\theta)) +i\tilde{k}_+(-\mu \sinh(-\theta)) \Big)f^+(\theta).
\end{aligned}
\end{equation}
Similarly, we obtain:
\begin{equation}
  \langle\psi,[\zd(f^+),A]\chi\rangle =  \frac{1}{2}\int d\theta \Big(   \mu \cosh\theta\; \tilde{k}_{-}(-\mu\sinh\theta)-i\tilde{k}_{+}(-\mu\sinh\theta)\Big)f^+(\theta).
\end{equation}
So, $K$ has boundary values which are defined as distributions and formally we can write:
\begin{equation}
K(\theta)=\langle \psi, [\zd(\theta),A] \chi \rangle.
\end{equation}
and since $\tilde{k}_{\pm}(-\mu\sinh(\theta+i\pi))=\tilde{k}_{\pm}(\mu\sinh\theta)$ and $\cosh(\theta+i\pi)=-\cosh\theta$, we have
\begin{equation}
K(\theta+i\pi)=\langle \psi, [A,z(\theta)] \chi \rangle.
\end{equation}
Regarding the proposed bounds for $K$, Eq.~\eqref{K:L2bound}, we first compute a bound for the boundary values of $K$ given by  Eq.~\eqref{kboundary}. We compute, $f \in  \dcal(\rbb)$,
\begin{equation}
 \begin{aligned}
\left|\int d\theta\, K(\theta + i\pi) f(\theta) \right|&=
|\langle \psi, [A,z(f)] \chi \rangle|\\
&\leq |\langle \psi, z(f)A \chi \rangle|+|\langle \psi, A z(f) \chi \rangle|\\
&=|\langle \psi, z(f)Q_{n_1 +1}A e^{-\omega(H/\mu)}e^{\omega(H/\mu)}Q_{n_2}\chi \rangle|\\
&\quad +|\langle \psi,Q_{n_1}e^{\omega(H/\mu)}e^{-\omega(H/\mu)} A Q_{n_2 -1}z(f) \chi \rangle|\\
&\leq \gnorm{{\zd}(\overline{f})\psi}{} \gnorm{Q_{n_1+1}Ae^{-\omega(H/\mu)}Q_{n_2}}{}\gnorm{e^{\omega(H/\mu)}\chi}{}\\
&\quad + \gnorm{e^{\omega(H/\mu)}\psi}{} \gnorm{Q_{n_1}e^{-\omega(H/\mu)}A Q_{n_2-1}}{}\gnorm{z(f)\chi}{}\\
&\leq 2\sqrt{n_1 +1} \gnorm{f}{}\gnorm{\psi}{2}\gnorm{A}{n_1 + n_2 +1}^{\omega}\onorm{\chi}{2} \\
&\quad + 2\sqrt{n_2}\onorm{\psi}{2}
\gnorm{A}{n_1 +n_2}^{\omega}\gnorm{f}{}\gnorm{\chi}{2}\\
&\leq 2(\sqrt{n_1 +1}+ \sqrt{n_2})\gnorm{f}{} \onorm{\psi}{2}\onorm{A}{n_1 + n_2 +1}\onorm{\chi}{2}.
\end{aligned}
\end{equation}
where we used that $\gnorm{Q_k A e^{-\omega(H/\mu)}Q_k}{}\leq 2 \onorm{A}{k}$ for any $k \in \nbb_0$ and \cite[Lemma 4.1.3]{Lechner:2006}.

Similarly, we find,
\begin{equation}
|K(f)| \leq 2(\sqrt{n_1}+\sqrt{n_2 +1})\gnorm{f}{} \onorm{\psi}{2} \onorm{\chi}{2}\onorm{A}{n_1 + n_2 +1}.
\end{equation}
As a consequence of Riesz' Lemma, we have that the boundary values of $K$ are $L^2$-functions with norms
\begin{equation}\label{Kboundarystima}
\gnorm{K}{2} \leq c_{n_1, n_2}\onorm{\psi_{n_1}}{2} \onorm{\chi_{n_2}}{2}\onorm{A}{n_1 + n_2 +1}.
\end{equation}
where $c_{n_1,n_2}:=2(\sqrt{n_1 +1}+ \sqrt{n_2+1})$.

To prove Eq.~\eqref{K:L2bound}, we consider the function \eqref{defK} shifted by $\exp(-i\mu s \sinh \zeta)$:
$K^{(s)}(\zeta):= e^{-i\mu s \sinh \zeta}K(\zeta)$, $s >0$. We consider its absolute value:
\begin{equation}\label{Ks}
|K^{(s)}_{\lambda}(\theta)|= \frac{1}{2}e^{- \mu s \sin\lambda \cosh\theta}|\mu \cosh(\theta + i\lambda)\tilde{k}_-(-\mu \sinh(\theta+i\lambda))-i\tilde{k}_+(-\mu \sinh(\theta +i\lambda))|.
\end{equation}
Since $\theta \rightarrow \tilde{k}_-(-\mu \sinh(\theta+i\lambda))$ and $\theta \rightarrow \tilde{k}_+(-\mu \sinh(\theta +i\lambda))$ are bounded by polynomials in $\cosh \theta$ for $|\theta|\rightarrow \infty$ and for every fixed $\lambda \in (0,\pi)$ (see remark after Eq.~\eqref{def:k+k-}), we have, due to the damping factor $e^{-i\mu s \sinh \zeta}$, that $K_{\lambda}^{(s)}\in L^{2}(\rbb)$ for every $\lambda \in [0,\pi]$, $s>0$. So, we can apply the three lines theorem, and by the estimates of the boundary values of $K$ given by \eqref{Kboundarystima}, we get
\begin{equation}
 \gnorm{K_{\lambda}^{(s)}}{2} \leq c_{n_1,n_2}\onorm{\psi}{2}\onorm{\chi}{2}\onorm{A}{n_1 + n_2 +1}, \quad 0\leq \lambda \leq \pi.
\end{equation}
Since \eqref{Ks} increases monotonically for $s \rightarrow 0$, then this bound holds in particular for $s=0$,
$K_\lambda = K_{\lambda}^{(0)}$, $0\leq \lambda\leq \pi$.

This concludes the proof of Lemma~\ref{lemma:K}.
\end{proof}

Using the Lemma above we can study analytic continuations of the functions $\cme{m,n}{A}$; this is again very similar to \cite[Lemma 4.3]{Lechner:2008}. Note that in \cite[Lemma 4.3]{Lechner:2008} Lechner proved the analogous result for
$A$ localized in the \emph{right} wedge; in our case we consider $A$ in the \emph{left} wedge: This is consistent with the relation $\cme{m,n}{A}=\langle J A^\ast J\rangle_{m+n,m}^{\mathrm{con}}$.

\begin{lemma}\label{lemma:analbouf}
Let $A\in\qf^{\omega}$ be $\omega$-local in $\leftwedge$. Then,

\begin{enumerate}

\item $\cme{m,n}{A}$ has an analytic continuation in the variable $\theta_{m}$ to the strip $\strip(0,\pi)$, $m\geq 1$. Its distributional boundary value at $\im\theta_{m}=\pi$ is given by
\begin{equation}\label{eq:analbouf}
\cme{m,n}{A}(\theta_1,\ldots,\theta_{m}+i\pi,\ldots,\theta_{m+n})
=
\cme{m-1,n+1}{A}(\theta_1,\ldots,\theta_{m},\ldots,\theta_{m+n}).
\end{equation}
\item There exists a constant $c_{m+n}$ such that the distribution $\cme{m,n}{A}$ fulfils the following bound, $g_1, \ldots, g_{m+n} \in \dcal(\rbb)$, $0\leq \lambda\leq \pi$:
\begin{multline}\label{fmnbound}
 \Big\lvert \int \cme{m,n}{A}(\theta_1,\ldots,\theta_{m-1},\theta_m +i\lambda,\theta_{m+1},\ldots,\theta_{m+n}) \prod_{j=1}^{m+n}g_j(\theta_j)\, d\theta_j\Big\rvert\\
\leq c_{m+n}\onorm{A}{m+n}\prod_{j=1}^{m+n}\onorm{g_j}{2}, \quad 0\leq \lambda\leq \pi.
\end{multline}
\end{enumerate}
\end{lemma}

The idea for the proof of the above lemma is to rewrite the definition of $\cme{m,n}{A}$, or $\langle J A\st J \rangle^\mathrm{con}_{m+n,n}$ in the notation of \cite{Lechner:2008}, in terms of a sum of matrix elements of commutators $[\zd(\theta), J A\st J]$, as in \cite[Lemma 4.2]{Lechner:2008}. Then one can apply Lemma~\ref{lemma:K} and find the analytic continuation of the Araki coefficients. More details on this technique can be found in \cite[Sec.~4]{Lechner:2008}.

To prove Lemma~\ref{lemma:analbouf}, we therefore first need the result of Lechner \cite[Lemma 4.2]{Lechner:2008} which we rewrite here using our notation. This result is proved by using the exchange relations of the Zamolodchikov-Faddeev algebra.

\begin{lemma}
Let $\hat{\ccal}_{m,n}\subset \ccal_{m,n}$ denote the subset of those contractions $C\in\ccal_{m,n}$ which do not contract $m$, i.e. fulfil $m\notin \pmb{\ell_{C}}$. Then
\begin{eqnarray}
\cme{m,n}{A}&=&\sum_{C\in \hat{\ccal}_{m,n}}(-1)^{|C|}\delta_{C}\cdot S_{C}^{(m)}\cdot \langle J\pmb{r}_{C},\left[JA^{*}J,z^{\dagger}_{m}\right]J\pmb{\ell}_{C}\cup\{  m\}\rangle,\label{fcom}\\
f_{m-1,n+1}^{[A]}&=&\sum_{C\in \hat{\ccal}_{m,n}}(-1)^{|C|}\delta_{C}\cdot S_{C}^{(m-1)}\cdot \langle J \pmb{r}_{C},\left[z_{m},JA^{*}J\right]J\pmb{\ell}_{C}\cup\{  m\}\rangle.\label{fcomm}
\end{eqnarray}
\end{lemma}
Note that here and for the rest of this section we will write explicitly the upper index $(m)$ on the factor $S_{C}$ (see Def.~\eqref{eq:sc}) for a matter of convenience.
\begin{proof}
We rewrite in our notation the proof in \cite[Appendix A]{Lechner:2008}.

We denote with $\hat{\ccal}_{m,n}$ the contractions $C\in \ccal_{m,n}$ which do not contract $m$, namely $m\notin J\pmb{\ell_{C}}$ and with $\check{\ccal}_{m,n}$ the contractions with $m\in J\pmb{\ell_{C}}$. We notice that $\ccal_{m,n}=\hat{\ccal}_{m,n} \dot{\cup} \check{\ccal}_{m,n}$.

We consider a contraction $C'\in \check{\ccal}_{m,n}$, which can be written as the union of a contraction $C\in \hat{C}_{m,n}$ and a contraction $\{(m,r)\}$ with $r\notin J\pmb{r_{C}}$. Namely, $C'= C \cup \{ (m,r)  \}$ and we have $|C|=|C'|-1$. Then, recalling the definitions \eqref{eq:deltac} and \eqref{eq:sc}, we have in this case:
\begin{eqnarray}
\delta_{C'} &=& \delta_{m,r} \cdot \delta_{C},\\
S_{C'}^{(m)} &=& \prod_{j=1}^{|C|}\prod_{m_{j}=l_{j}+1}^{r_{j}-1}S_{m_{j},l_{j}}^{(m)}\cdot \prod_{\substack{r_{i}<r_{j} \\ l_{i}<l_{j}}}S^{(m)}_{l_{j},r_{i}}\cdot \prod_{p=m+1}^{r-1}S_{p,m}^{(m)}\cdot \prod_{\substack{r_{i}<r \\ l_{i}<m}}S_{m,r_{i}}^{(m)}\cdot \prod_{\substack{r<r_{j} \\ m<l_{j}}}S_{l_{j},r}^{(m)}\nonumber\\
&=& S_{C}^{(m)}\cdot \prod_{p=m+1}^{r-1}S_{m,p} \cdot \prod_{r_{i}<r}S_{r_{i},m},
\end{eqnarray}
since $l_{1},\ldots, l_{|C|}<m$. By multiplying the two last products of $S$-factors in the formula above, we get:
\begin{equation}
\delta_{C'}\cdot S_{C'}^{(m)}=\delta_{C}\cdot S_{C}^{(m)}\cdot \delta_{m,r}\cdot \prod_{\substack{ p=m+1 \\ p\neq r_{i} \text{ for } r_{i}<r}}^{r-1}S_{m,p}.\label{sdeltaproofcomm}
\end{equation}
Now we consider the matrix element $\langle J\pmb{r}_{C},JA^{*}J \; J\pmb{\ell}_{C}\rangle$. Using the Zamolodchikov's algebra \eqref{zamoloalgebra} repeatedly, we find
\begin{eqnarray}
\langle J\pmb{r}_{C},JA^{*}J\; J\pmb{\ell}_{C}\rangle &=& \langle \zd_{m+1}\ldots \widehat{\zd_{r_{1}-m}}\ldots \widehat{\zd_{r_{|C|}-m}}\ldots \zd_{m+n}\Omega, JA^{*}J \zd_{m}\ldots \widehat{\zd_{l_{1}}}\ldots \widehat{\zd_{l_{|C|}}}\ldots \zd_{1}\Omega \rangle \nonumber\\
&=& \langle J\pmb{r}_{C},\left[JA^{*}J,z^{\dagger}_{m}\right]J\pmb{\ell}_{C}\cup\{  m\}\rangle \nonumber\\
&+& \sum_{\substack{ r=m+1\\ r\notin J\pmb{r_{C}}}}^{m+n}\delta_{m,r}\prod_{\substack{ p=m+1 \\ p\neq r_{i} \text{ for } r_{i}<r}}^{r-1}S_{m,p}\cdot  \langle J\pmb{r}_{C}\cup \{r \},JA^{*}J \;J\pmb{\ell}_{C}\cup\{  m\}\rangle.\label{proofcomm}
\end{eqnarray}
We multiply Eq.~\eqref{proofcomm} with $(-1)^{|C|}\delta_{C}S_{C}^{(m)}$ and we sum over $C\in \hat{\ccal}_{m,n}$. Since $C'=C\cup \{ (m,r) \}$, we have $\sum_{C'\in \check{\ccal}_{m,n}}=\sum_{r=m+1, r\notin J\pmb{r_{C}}}^{m+n}\sum_{C\in \hat{\ccal}_{m,n}}$. The delta distributions and the $S$-factors in \eqref{proofcomm} are equal to the ones in \eqref{sdeltaproofcomm}. Hence, we have:
\begin{multline}
\sum_{C\in \hat{\ccal}_{m,n}}(-1)^{|C|}\delta_{C}S_{C}^{(m)}\langle J\pmb{r}_{C},JA^{*}J\; J\pmb{\ell}_{C}\rangle \\
 =\sum_{C\in \hat{\ccal}_{m,n}}(-1)^{|C|}\delta_{C}S_{C}^{(m)}\langle J\pmb{r}_{C},\left[JA^{*}J,z^{\dagger}_{m}\right]J\pmb{\ell}_{C}\cup\{  m\}\rangle\\
- \sum_{C'\in \check{\ccal}_{m,n}}(-1)^{|C'|}\delta_{C'}S_{C'}^{(m)}\langle J\pmb{r}_{C'},JA^{*}J\; J\pmb{\ell}_{C'}\rangle.
\end{multline}
where we used that $(-1)^{|C'|}=-(-1)^{|C|}$. Since $\ccal_{m,n}=\hat{\ccal}_{m,n}\dot{\cup}\check{\ccal}_{m,n}$, we finally find
\begin{multline}
\sum_{C\in \hat{\ccal}_{m,n}}(-1)^{|C|}\delta_{C}S_{C}^{(m)}\langle J\pmb{r}_{C},\left[JA^{*}J,z^{\dagger}_{m}\right]J\pmb{\ell}_{C}\cup\{  m\}\rangle\\
=\sum_{C\in \ccal_{m,n}}(-1)^{|C|}\delta_{C}S_{C}^{(m)}\langle J\pmb{r}_{C},JA^{*}J\; J\pmb{\ell}_{C}\rangle.
\end{multline}
The right hand side of the formula above is by definition $\cme{m,n}{A}$. Hence, we proved \eqref{fcom}.

The proof of \eqref{fcomm} is analogous.
\end{proof}

Now, we can prove Lemma \ref{lemma:analbouf} following closely the proof of \cite[Lemma 4.3]{Lechner:2008}.

\begin{proof}
\begin{enumerate}

\item We consider the distributions $\cme{m,n}{A}$ \eqref{eq:fmndef} rewritten in the following way:
\begin{equation}
\cme{m,n}{A}(\thetav,\pmb{\eta}):=\sum_{C\in \ccal_{m,n}}(-1)^{|C|}\delta_{C}\cdot S_{C}^{(m)}\cdot \langle J\pmb{r}_{C},JA^{*}J\pmb{\ell}_{C}\rangle.\label{fmnjaj}
\end{equation}
Using \eqref{fcom}, we rewrite $\cme{m,n}{A}$ as:
\begin{equation}
f_{m,n}^{[A]}(\thetav,\pmb{\eta})=\sum_{C\in \hat{\ccal}_{m,n}}(-1)^{|C|}\delta_{C}\cdot S_{C}^{(m)}\cdot \langle J\pmb{r}_{C},\left[JA^{*}J,z^{\dagger}_{m}\right]J\pmb{\ell}_{C}\cup\{  m\}\rangle.
\end{equation}
In \eqref{fmnjaj} we have that $S_C^{(m)}$ \eqref{eq:sc} has an analytic continuation in the variable $\theta_m$ to the strip $\strip(0,\pi)$ with boundary value $S_{C}^{(m-1)}$ at $\mathbb{R}+i\pi$. Indeed, the factors $S_{m,r_{j}}^{(m)}=S_{m,r_{j}}$ in $S_{C}^{(m)}$ can be analytically continued in $\theta_{m}$ to the strip $\strip(0,\pi)$ with boundary value $S(\theta_{m}+i\pi-\theta_{r_{j}})=S(\theta_{r_{j}}-\theta_{m})=S^{(m-1)}(\theta_{m}-\theta_{r_{j}})$ since the scattering function $S$ is analytic in $\strip(0,\pi)$ and fulfils crossing-symmetry relation. The other factors $S_{m_{j},r_{j}}^{(m)}$, with $m_{j}\neq m$, do not depend on $\theta_{m}$ (in particular $l_{i},r_{j}\neq m$) since $m$ is not contracted in $C\in \hat{\ccal}_{m,n}$ and therefore they fulfil $S_{a,b}^{(m)}=S^{(m-1)}_{a,b}$ ($a,b\neq m$).

We have that the delta distribution $\delta_C$ does not depend on the variable $\theta_m$ since $m$ is not contracted in $C$.

We know from Lemma \ref{lemma:K} that $\langle J\pmb{r}_{C},\left[JA^{*}J,z^{\dagger}_{m}\right]J\pmb{\ell}_{C}\cup\{  m\}\rangle$, integrated in $\theta_j, j\neq m$, with test functions in $\mathcal{D}(\mathbb{R})$, can be analytically continued in $\theta_{m}$ to the strip $\strip(0,\pi)$. Also, due to Lemma \ref{lemma:K} its boundary value at $\im \zeta_{m}=i\pi$ is given by $\langle J\pmb{r}_{C},\left[z_{m},JA^{*}J\right]J\pmb{\ell}_{C}\cup\{  m\}\rangle$.

From the considerations above we can conclude that $\cme{m,n}{A}$ can be analytically continued in the variable $\theta_m$ to the strip $\strip(0,\pi)$ with distributional boundary value at the other side of the strip $\im \zeta_{m}=i\pi$ given by:
\begin{equation}
\cme{m,n}{A}(\theta_{1},\ldots,\theta_{m}+i\pi,\ldots,\theta_{m+n})=\sum_{C\in \hat{\ccal}_{m,n}}(-1)^{|C|}\delta_{C}S^{(m-1)}_{C}\langle J\pmb{r}_{C},\left[z_{m},JA^{*}J\right]J\pmb{\ell}_{C}\cup\{  m\}\rangle.\label{proofanalyfmn}
\end{equation}
Using \eqref{fcomm} the right hand side of \eqref{proofanalyfmn} is exactly $f_{m-1,n+1}^{[A]}$.

\item Using \eqref{fcom}, we rewrite the left hand side of \eqref{fmnbound} as follows:
\begin{multline}\label{intfmmncomm}
 \Big\lvert \int \cme{m,n}{A}(\theta_1,\ldots,\theta_{m-1},\theta_m +i\lambda,\theta_{m+1},\ldots,\theta_{m+n}) \prod_{j=1}^{m+n}g_j(\theta_j)\, d\theta_j\Big\rvert\\
=\Big \lvert \int \sum_{C \in \hat{\ccal}_{m,n}}(-1)^{|C|} \delta_C(\thetav) S_C^{(m)}(\thetav) \times \\
\times \langle J\pmb{r_C},[JA^{*}J,{\zd}(\theta_m +i\lambda)]J\pmb{\ell_C}\cup \{m\}\rangle \prod_{j=1}^{m+n}g_j(\theta_j)\, d\theta_j\Big \rvert.
\end{multline}
where $\pmb{r_C}(\thetav)={\zd}(\theta_{m+n})\ldots \widehat{{\zd}(\theta_{r_1})}\ldots \widehat{{\zd}(\theta_{r_{|C|}})} \ldots {\zd}(\theta_{m+1})\Omega$.

We denote $\pmb{\theta_r} := (\theta_{r_1}, \ldots, \theta_{r_{|C|}})$, $\boldsymbol{\theta_r} \in \rbb^{|C|}$. Following Lechner, we split $\delta_C S_C^{(m)}$ into the product of three factors: $\delta_{C}S_{C}^{(m)}= \delta_{C} S_{C}^L S_C ^M S_C^R$, where $S_C^L$ depends on $\{\theta_1,\ldots, \theta_{m-1} \} \backslash \{ \theta_{l_1},\ldots, \theta_{l_{|C|}} \}$ and $\pmb{\theta_r}$, where $S_C^R$ depends on $\{\theta_{m+1},\ldots, \theta_{m+n} \}$, and where $S_C^M = \prod_{j=1}^{|C|}S_{m,l_j}$ depends on $\theta_m$. We also define, $g_1,\ldots, g_{m+n} \in \dcal(\rbb)$,
\begin{eqnarray}
F_{\pmb{\theta_r}}^L &:=& S_{C}^{L} \cdot (g_{m-1} \otimes \ldots \otimes \widehat{g_{l_1}}\otimes \ldots \otimes \widehat{g_{l_{|C|}}}\otimes \ldots \otimes g_1),\\
F_{\pmb{\theta_r}}^R &:=& S_{C}^{R} \cdot (g_{m+1} \otimes \ldots \otimes \widehat{g_{r_1}}\otimes \ldots \otimes \widehat{g_{r_{|C|}}}\otimes \ldots \otimes g_{m+n}).
\end{eqnarray}
$F_{\pmb{\theta_r}}^L$ and $F_{\pmb{\theta_r}}^R$ are functions of $m-1-|C|$ variables, $\{\theta_1,\ldots, \theta_{m-1} \} \backslash \{ \theta_{l_1},\ldots, \theta_{l_{|C|}} \}$, and of $n-|C|$ variables, $\{\theta_{m+1},\ldots, \theta_{m+n} \}\backslash \{ \theta_{r_1},\ldots, \theta_{r_{|C|}} \}$, respectively, and they depend parametrically on $\pmb{\theta_r}$. Since $|S(\theta)|=1$ for $\theta \in \rbb$, they have norms:
\begin{eqnarray}
 \onorm{F^L _{\pmb{\theta_r}}}{2} &\leq& \prod_{\substack{ j=1 \\ j \neq \pmb{\ell_C} }}^{m-1}\onorm{g_j}{2},\\
 \onorm{F^R _{\pmb{\theta_r}}}{2} &\leq& \prod_{\substack{ j=m+1 \\ j \neq \pmb{r_C} }}^{m+n}\onorm{g_j}{2},
\end{eqnarray}
where we made use of the sublinearity of $\omega$, see \ref{it:sublinear}.

After integrating over the delta distributions in \eqref{intfmmncomm}, we find
\begin{multline}
 \text{l.h.s.}\eqref{intfmmncomm}\leq \Big\lvert   \sum_{C \in \hat{\ccal}_{m,n}}(-1)^{|C|} \int d^{|C|}\boldsymbol{\theta_r}\, \prod_{j=1}^{|C|}g_{l_j}(\theta_{r_j})g_{r_j}(\theta_{r_{j}})\times \\
\times \int d\theta_m \, g_m(\theta_m) \prod_{j=1}^{|C|} S(\theta_m -\theta_{r_j }+i\lambda)\langle J\pmb{r_C}, [JA^{*}J, \zd(\theta_m +i\lambda)] J\pmb{\ell_C}\cup\{ m \} \rangle (F_{\boldsymbol{\theta_r}}^R \otimes F_{\boldsymbol{\theta_r}}^{L})\Big\rvert.
\end{multline}
We move the absolute value inside the integral in $\boldsymbol{\theta_r}$ and we apply the Cauchy-Schwarz inequality with respect to the variable $\theta_m$, we find:
\begin{multline}
\text{l.h.s.}\eqref{intfmmncomm}\leq
  \sum_{C \in \hat{\ccal}_{m,n}} \int d^{|C|}\boldsymbol{\theta_r}\, \prod_{j=1}^{|C|} | g_{l_j}(\theta_{r_j})g_{r_j}(\theta_{r_{j}})| \times \\
\times \Big\lvert \int d\theta_m \, g_m(\theta_m) \prod_{j=1}^{|C|} S(\theta_m -\theta_{r_j }+i\lambda)\langle J\pmb{r_C}, [JA^{*}J, \zd(\theta_m +i\lambda)] J\pmb{\ell_C}\cup\{ m \} \rangle (F_{\boldsymbol{\theta_r}}^R \otimes F_{\boldsymbol{\theta_r}}^{L})\Big\rvert\\
\leq  \sum_{C \in \hat{\ccal}_{m,n}} \int d^{|C|}\boldsymbol{\theta_r}\, \prod_{j=1}^{|C|} | g_{l_j}(\theta_{r_j})g_{r_j}(\theta_{r_{j}})| \times \\
\times \gnorm{g_m}{2}\cdot \Big( \int d\theta_m\, |\langle J\pmb{r_C}, [JA^{*}J, \zd(\theta_m +i\lambda)] J\pmb{\ell_C}\cup\{ m \} \rangle (F_{\boldsymbol{\theta_r}}^R \otimes F_{\boldsymbol{\theta_r}}^{L})|^2 \Big)^{1/2},
\end{multline}
where we used that $|S(\zeta)| \leq 1$, $\zeta \in S(0,\pi)$. Putting \eqref{K:L2bound}, we find
\begin{multline}
 \text{l.h.s.}\eqref{intfmmncomm}\leq \sum_{C \in \hat{\ccal}_{m,n}} \int d^{|C|}\boldsymbol{\theta_r}\, \prod_{j=1}^{|C|} | g_{l_j}(\theta_{r_j})g_{r_j}(\theta_{r_{j}})| \times \\
\times \gnorm{g_m}{2} c_{m-1-|C|,n-|C|}\sqrt{(m-1-|C|)!}\sqrt{(n-|C|)!} \onorm{JA^{*}J}{m+n-2|C|} \onorm{F_{\boldsymbol{\theta_r}}^L}{2}\gnorm{F_{\boldsymbol{\theta_r}}^R}{2}\\
\leq \sum_{C \in \hat{\ccal}_{m,n}} \int d^{|C|}\boldsymbol{\theta_r}\, \prod_{j=1}^{|C|} | g_{l_j}(\theta_{r_j})g_{r_j}(\theta_{r_{j}})| \times \\
\times \gnorm{g_m}{2} c_{m-1-|C|,n-|C|}\sqrt{(m-1-|C|)!}\sqrt{(n-|C|)!} \onorm{A}{m+n} \prod_{\substack{ j=1 \\ j \neq \pmb{\ell_C} }}^{m-1}\onorm{g_j}{2}\prod_{\substack{ j=m+1 \\ j \neq \pmb{r_C} }}^{m+n}\onorm{g_j}{2}.
\end{multline}
Now, we can apply the Cauchy-Schwarz inequality with respect to the variable $\boldsymbol{\theta_r}$, we find
\begin{multline}\label{intfmncomput}
\text{l.h.s.}\eqref{intfmmncomm}\leq
 \sum_{C \in \hat{\ccal}_{m,n}}
\prod_{j=1}^{|C|}\gnorm{g_{l_j}}{2}\gnorm{g_{r_j}}{2}\times \\
\times \gnorm{g_m}{2} c_{m-1-|C|,n-|C|}\sqrt{(m-1-|C|)!}\sqrt{(n-|C|)!} \onorm{A}{m+n} \prod_{\substack{ j=1 \\ j \neq \pmb{\ell_C} }}^{m-1}\onorm{g_j}{2}\prod_{\substack{ j=m+1 \\ j \neq \pmb{r_C} }}^{m+n}\onorm{g_j}{2}.
\end{multline}
As for the constants depending on $m,n$ in the expression above, we compute:
\begin{equation}
 c_{m-1-|C|,n-|C|}\sqrt{(m-1-|C|)!}\sqrt{(n-|C|)!} \leq \frac{4\sqrt{(m+n)!}}{|C|!}.
\end{equation}
Inserting in \eqref{intfmncomput}, the dependence on $C$ of this equation is only in the term $1/|C|!$, hence we can compute
\begin{equation}
\sum_{C\in \hat{\ccal}_{m,n}}\frac{1}{|C|!}= 2^{m+n-1},
\end{equation}
see also \cite[page 20]{Lechner:2008}.

Finally, we arrive at
\begin{equation}
\text{l.h.s.}\eqref{intfmmncomm}\leq\sqrt{(m+n)!}2^{m+n+1}\prod_{j=1}^{m+n}\onorm{g_j}{2} \onorm{A}{m+n}.
\end{equation}
This concludes the proof of Lemma~\ref{lemma:analbouf}.
\end{enumerate}
\end{proof}

Using the results above, now we can define $F_k'$ on the ``positive simplex'' $\tube(\gcal_+^k)$, cf.~Fig.~\ref{fig:zerostair}, as follows.

\begin{proposition}\label{proposition:analpositivesimplex}

If $A\in \mathcal{Q}^{\omega}$ is $\omega$-local in the wedge $\wcal_r'$, then there are CR distributions $F_k'$ on $\tube(\gcal^k_+)$ such that
\begin{equation}\label{eq:fkpboundaryvalue}
F_{k}'(\thetav+i\lambdav^{(k,j)})=\cme{k-j,j}{A}(\thetav), \quad
0\leq j \leq k.
\end{equation}
\end{proposition}

\begin{proof}
First we consider the case where $r=0$. By Lemma~\ref{lemma:analbouf} we have that $\cme{k-j,j}{A}$ can be analytically continued to the edge $\im \zetav = \lambdav^{(k,j)} + \lambda \ev^{(k-j)}$, $0 < \lambda < \pi$; we define $F_k'$ on that edge as the analytic continuation of $\cme{k-j,j}{A}$ due to that Lemma. The boundary values of $F_k'$ are \eqref{eq:fkpboundaryvalue}, independent of the direction, due to Eq.~\eqref{eq:analbouf}.

Now we consider the case where $r \neq 0$. We consider the translated $A$, $B := U(-r\ev^{(1)}) A U(-r\ev^{(1)})^\ast$, so that the resulting quadratic form is $\omega$-local in the standard wedge $\leftwedge$. Hence, we can apply the result before for $r=0$ and conclude that $\cme{k-j,j}{B}$ has an analytic function ${F_k'}{}^{[B]}$ on $\tube(\gcal^k_+)$ whose boundary values are $\cme{k-j,j}{B}$.

We set
\begin{equation}\label{eq:shiftedfp}
 F_k'{}^{[A]}(\zetav) := F_k'{}^{[B]}(\zetav) \exp \big(i\mu r \textstyle{\sum_j} \sinh \zeta_j \big).
\end{equation}
One can see that this function has the proposed analyticity property and, using the relation between $\cme{k-j,j}{B}$ and $\cme{k-j,j}{A}$ given by Eq.~\eqref{eq:fmnpoincare} and the boundary values of $B$, one can also see that it has the proposed boundary values at the nodes.
\end{proof}

\section{Define function on negative simplex} \label{sec:fkpnegative}

Now, if an operator, or rather a quadratic form, $A$ is localized in a double cone $\ocal_r$, then it is in particular localized in the wedge $\wcal_r'$. Also, since $A$ is localized in the double cone, $JA\st J$ is localized in the \emph{same} wedge. Hence, we can apply Prop.~\ref{proposition:analpositivesimplex} both to $A$ and to $J A\st J$. This implies an extension of the domain of analyticity of the functions $F_k'$ discussed in Prop.~\ref{proposition:analpositivesimplex} to $\tube(\gcal^k_0)$, cf.~Fig.~\ref{fig:zerostair}.

In the following proposition we define the functions $F_k'$ for an operator localized in a double cone $\ocal_r$:

\begin{proposition}\label{proposition:ffjcont}
If $A$ is $\omega$-local in $\ocal_r$, then there exist CR distributions $F_k'$ which fulfil \ref{it:fpmero} and \ref{it:fpperiod}, and have the boundary values \eqref{eq:fkpboundary}.
\end{proposition}

\begin{proof}

Since $A$ is $\omega$-local in a double cone $\ocal_r$, then $A$ is, in particular, localized in the wedge $\wcal_r'$. Therefore, we can apply Prop.~\ref{proposition:analpositivesimplex} and define a family of functions $F_k'$ analytic on the tube $\tube(\gcal_+^k)$, the boundary of which are the required boundary values \eqref{eq:fkpboundaryvalue}.

Since $A$ is $\omega$-local in the double cone, then $JA^{*}J$ is $\omega$-local in the \emph{same} wedge $\wcal_r'$. Therefore, by Prop.~\ref{proposition:analpositivesimplex} there is another set of functions $F_{k}'[JA^{*}J]$ analytic on the same domain $\tube(\gcal^k_+)$.

We use this to define $F_k'$ on $\tube(\gcal^k_+ - \piv)$, cf.~Fig.~\ref{fig:zerostair}, by
\begin{equation}
   F_{k}'(\zetav- i \piv) := F_{k}'{}^{[JA^{*}J]} (\zetav).\label{Fpnegativesimplex}
\end{equation}
This has the required boundary values
\begin{equation}\label{eq:fkpotherbv}
 F_{k}'(\thetav+i\lambdav^{(k,-k+j)})
 = F_{k}'{}^{[JA^{*}J]} (\thetav+i\lambdav^{(k,j)})
 = f_{k-j,j}^{[JA^{*}J]}(\thetav)
\end{equation}

To show that the functions $F'_k$ are actually analytic on $\tube(\gcal_0^k)$, it remains to show that the two boundary values (coming from $\gcal_+^k$ and $\gcal_-^k$) of $F_k'$ at $\lambdav^{(k,0)}$ agree. By Eq.~\eqref{eq:fkpboundaryvalue}, we know that for real $\thetav$,
\begin{equation}
F_{k}'(\thetav+i\lambdav^{(k,0)})\Big\vert_{\gcal_+^k}=\cme{k,0}{A}(\thetav).
\end{equation}
Likewise, by Eq.~\eqref{eq:fkpotherbv}, we have
\begin{equation}
F_{k}'(\thetav+i\lambdav^{(k,0)})\Big\vert_{\gcal_-^k}
=
\cme{0,k}{J A\st J}(\thetav).
\end{equation}
However, a short computation shows that these are equal:
\begin{equation}
\cme{k,0}{A}(\thetav)
=\hscalar{z^{\dagger}(\theta_{1})\ldots z^{\dagger}(\theta_{k})\Omega}{A\Omega}
=\hscalar{\Omega}{J A\st J z^{\dagger}(\theta_{k})\ldots z^{\dagger}(\theta_{1})\Omega}
=\cme{0,k}{JA^{*}J}(\thetav).
\end{equation}
This proves \ref{it:fpmero}. Similarly, we can show that also the boundary values $F_{k}'(\thetav+i\lambdav^{(k,-k)})$ and $F_{k}'(\thetav+i\lambdav^{(k,k)})$ agree, and this gives the periodicity condition \ref{it:fpperiod}.
\end{proof}

\section{Cross-norm bounds} \label{sec:fkpbasic}

To prove the bounds \ref{it:fpboundsreal} and \ref{it:fpboundsimag}, we use the \emph{maximum modulus principle}; this is contained in the proof of the following proposition:

\begin{proposition} \label{proposition:fpbounds}
If $A$ is $\omega$-local in $\ocal_r$, then the distributions $F_k'$ fulfil the bounds  \ref{it:fpboundsreal} and \ref{it:fpboundsimag} with this parameter $r$.
\end{proposition}

\begin{proof}
We notice that on the nodes of the graph $\gcal_+^k$ we have, by Eq.~\eqref{eq:fkpboundary} and as a consequence of Prop.~\ref{proposition:fmnbound},
\begin{equation}\label{eq:fkpnodebounds}
  \gnorm{ F_k' (\cdotarg + i \lambdav^{(k,j)})}{(k-j) \times j}^{\omega}
 = \gnorm{ \cme{k-j,j}{A} }{(k-j) \times j}^{\omega} < \infty
\end{equation}
for any $0 \leq j \leq k$.

Also on the nodes of the graph $\gcal_-^k$, $\lambdav^{(k,-j)}$, there holds a similar inequality, for any $0 \leq j \leq k$,
\begin{equation}
\gnorm{F'_{k}(\cdotarg + i\lambda^{(k,-j)})}{j \times (k-j)}^{\omega}= \gnorm{ \cme{j,k-j}{JA^{*}J} }{j \times (k-j)}^{\omega} < \infty.
\end{equation}
due to Eq.~\eqref{eq:fkpboundary} and Prop.~\ref{proposition:fmnbound}. This proves \ref{it:fpboundsreal}.

Now we consider the distribution
\begin{equation}\label{eq:fkhatdef}
   \hat F_k (\zetav) := F_k' (\zetav) \exp(-i \mu r \textstyle{\sum_j} \sinh \zeta_j)\exp(-\sum_{j}\oa(\sinh\zeta_{j})).
\end{equation}
This is a CR distribution on $\tube(\gcal_+^k)$, since $F_k'$ is a CR distribution on $\tube(\gcal_+^k)$ by Prop.~\ref{proposition:ffjcont} and because the exponential functions are entire analytic.

On the nodes of the graph we have that the first exponential gives a factorizing phase factor, and the second exponential fulfils the estimate
\begin{equation}\label{expomegastima}
|e^{-\oa(\sinh\zeta_{j})}|\leq e^{-\omega(\cosh\theta_{j})}e^{\omega(1)},
\end{equation}
where we used \ref{it:omegaestimate}, and the relations $\sinh\theta = \cosh\theta -e^{-\theta}$ (which implies $|\sinh\theta|\geq \cosh\theta -1$) and $\omega(\cosh\theta)\leq \omega(\cosh\theta -1)+\omega(1)$ (which implies that $\omega(\cosh\theta)-\omega(1)\leq \omega(\cosh\theta -1)\leq \omega(|\sinh\theta|)$), together with the sublinearity and monotonicity of $\omega$.

Therefore, from \eqref{eq:fkpnodebounds} and \eqref{eq:crossnormcomparison}(left inequality), $\hat F_k$ fulfils the bound:
\begin{multline}
  \gnorm{ \hat F_k (\cdotarg + i \lambdav^{(k,j)})}{ \times }
  =\gnorm{F'_k (\cdotarg + i \lambdav^{(k,j)})\exp(-\sum_{j}\oa(\sinh\zeta_{j}))\exp(-i\mu r \sum_{j}\sinh \zeta_{j}) }{ \times }\\
  \leq e^{\sum_{j=1}^{k}\omega(1)}\gnorm{F'_k (\cdotarg + i \lambdav^{(k,j)})\exp(-\sum_{j}\oa(\cosh\theta_{j}))\exp(-i\mu r \sum_{j}\sinh \zeta_{j}) }{ \times }\\
  \leq e^{k\cdot \omega(1)}\gnorm{  F'_k (\cdotarg + i \lambdav^{(k,j)})\exp(-i\mu r \sum_{j}\sinh \zeta_{j})}{(k-j) \times j}^{\omega} \\
  =e^{k\cdot \omega(1)}\gnorm{\cme{k-j,j}{A}}{(k-j) \times j}^{\omega}< \infty.
\end{multline}
Moreover, by \eqref{eq:shiftedfp} we have for $\im\zetav$ on an edge of $\gcal_+^k$:
\begin{multline}
\hat{F}_k(\zetav) =\\
 F_k'^{[U(-re^{(1)})AU(-re^{(1)})^{*}]} (\zetav)\exp(i \mu r \textstyle{\sum_j} \sinh \zeta_j) \exp(-i \mu r \textstyle{\sum_j} \sinh \zeta_j)\exp(-\sum_{j}\oa(\sinh\zeta_{j}))\\
=F_k'^{[U(-re^{(1)})AU(-re^{(1)})^{*}]} (\zetav)\exp(-\sum_{j}\oa(\sinh\zeta_{j})).
\end{multline}
Now we compute a bound for $(\hat{F}_k \ast (g_1 \otimes \ldots \otimes g_k))(\pmb{t} +i\lambdav)$, with $g_1, \ldots, g_k \in \dcal(\rbb)$ and $\lambdav$ on an edge of $\gcal^k_+$ in the direction $\pmb{e}^{(j)}$. By Lemma~\ref{lemma:analbouf} we have
\begin{multline}\label{stimahatF}
|(\hat{F}_k \ast (g_1 \otimes \ldots \otimes g_k))(\pmb{t} +i\lambdav)|=\\
 \Big\lvert \int d\thetav\, F_k'^{[U(-re^{(1)})AU(-re^{(1)})^{*}]} (\thetav +i\lambdav)e^{-\sum_{j}\oa(\sinh(\theta_{j}+i\lambda_j))} g_1(t_1 -\theta_1)\ldots g_k(t_k - \theta_k)  \Big\rvert=\\
\Big\lvert \int d\thetav\, \cme{j,k-j}{A}(\theta_1, \ldots,\theta_{j-1}, \theta_j +i\lambda_j, \theta_{j+1},\ldots,\theta_k)e^{-\sum_{\ell}\oa(\sinh(\theta_{\ell}+i\lambda_\ell))} g_1(t_1 -\theta_1)\ldots g_k(t_k - \theta_k)  \Big\rvert\\
\leq c_k \onorm{A}{k} \prod_{\ell=1}^k \onorm{e^{-\oa(\sinh(\cdotarg +i\lambda_\ell))}g_\ell(t_\ell-\cdotarg)}{2}.
\end{multline}
Using \eqref{expomegastima}, we find that
\begin{eqnarray}
\onorm{e^{-\oa(\sinh(\cdotarg +i\lambda_\ell))}g_\ell(t_\ell-\cdotarg)}{2} &=& \Big( \int |g_\ell(t_\ell -\theta_\ell)|^2 e^{-2\oa(\sinh(\theta_\ell +i\lambda_\ell))}e^{2\omega(\cosh\theta_\ell)}\, d\theta_\ell  \Big)^{1/2} \nonumber\\
&\leq& e^{\omega(1)} \gnorm{g_\ell}{2}.
\end{eqnarray}
Inserting in \eqref{stimahatF}, we arrive at
\begin{equation}
 |(\hat{F}_k \ast (g_1 \otimes \ldots \otimes g_k))(\pmb{t} +i\lambdav)|\leq c'_k \onorm{A}{k} \prod_{\ell=1}^k \gnorm{g_\ell}{2}.
\end{equation}
where $c'_k :=e^{\omega(1)} c_k$.

This shows that the convoluted function $(\hat{F}_k \ast (g_1 \otimes \ldots \otimes g_k))(\pmb{t} +i\lambdav)$ is bounded in $\pmb{t}$ and therefore we can apply the maximum modulus principle for the norm $\gnorm{\cdotarg}{\times}$ (see Lemma~\ref{lemma:maxmodcross}), we find
\begin{equation}
  \sup_{\lambdav} \gnorm{\hat F_k(\cdotarg + i \lambdav)}{\times}
  \leq \sup_{0 \leq j \leq k} \gnorm{\hat F_k(\cdotarg + i \lambdav^{(k,j)})}{\times} < \infty.
\end{equation}
Similarly, we can apply the arguments above to the function $F_k'$ on $\tube(\gcal_+^k-i \piv)$, with $J A^\ast J$ in the place of $A$; the only differences with respect to the result above is a shift of $\zetav$ by $i\piv$ due to \eqref{Fpnegativesimplex}, and hence a resulting minus sign in the exponent in \eqref{eq:fkhatdef}. Namely, the argument applies in a similar way to the distribution:
\begin{equation}
\hat F_k^{[JA^{*}J]}(\zetav) := F_k'^{[JA^{*}J]} (\zetav) \exp(i \mu r \textstyle{\sum_j} \sinh \zeta_j)\exp(-\sum_{j}\oa(-\sinh\zeta_{j})).
\end{equation}
This gives the bounds \ref{it:fpboundsimag}.
\end{proof}

\section{Recursion relations} \label{sec:fkprecursion}

It remains to show the property of $S$-symmetry \ref{it:fpsymm} and the recursion relations \ref{it:fprecursion}. The proof of this is a direct consequence of the properties of the Araki coefficients, as we can see in the proof of the following proposition.
\begin{proposition} \label{proposition:fprecursion}
For fixed $A \in \qf^{\omega}$, any set of distributions $F_k'$ with boundary values \eqref{eq:fkpboundary} fulfils \ref{it:fpsymm} and \ref{it:fprecursion}.
\end{proposition}

\begin{proof}

By \eqref{eq:fkpboundary}, the property of S-symmetry \ref{it:fpsymm} of the functions $F_k'$ is a direct consequence of the property of S-symmetry of the distributions $\cme{k,0}{A}$, discussed in Prop.~\ref{proposition:fmnsymm}.

To prove \ref{it:fprecursion}, we use Prop.~\ref{proposition:fmnreflected} and we insert there $\cme{m,n}{A}$ as the boundary value of $F'_{m+n}$ given by \eqref{eq:fkpboundary}. We find
\begin{equation}
F'_{m+n}(\thetav -i\pmb{\pi},\etav)
=\sum_{C\in\ccal_{m,n}}(-1)^{|C|}\delta_{C}
 S_{C} R_C(\thetav,\etav) F'_{m+n-2|C|}(\hat \etav, \hat\thetav + i \piv).\label{recursione5}
\end{equation}
This is exactly \ref{it:fprecursion}.
\end{proof}

The proof of Theorem~\ref{theorem:AtoFp} is a consequence of Propositions~\ref{proposition:analpositivesimplex}, \ref{proposition:ffjcont}, \ref{proposition:fpbounds} and \ref{proposition:fprecursion}.

%% file: from_fp_to_f.tex
\chapter{(F') \texorpdfstring{$\Rightarrow$}{=>} (F)}\label{sec:fptof}

In the following Chapter we show that we can extend the CR distributions $F_k'$ to meromorphic functions $F_k$ which fulfil the properties (F). More precisely, we want to prove the following Theorem:

\begin{theorem} \label{theorem:FptoF}
If $F_k'$ are distributions fulfilling (F'), then there exist functions $F_k$ fulfilling (F) such that their distributional boundary values are
\begin{equation}\label{eq:ffpboundary}
\begin{aligned}
F_{m+n}(\thetav+i\zerov,\etav+i\piv-i\zerov) &= F_{m+n}'(\thetav,\etav+i\piv), \\
F_{m+n}(\thetav-i\piv+i\zerov,\etav-i\zerov) &= F_{m+n}'(\thetav-i\piv,\etav),
\end{aligned}
\end{equation}
where $\thetav\in\rbb^m$, $\etav\in \rbb^n$.
\end{theorem}

\section{Continuation along graphs} \label{sec:fkalonggraph}

We consider a family of CR distributions $F_{k}'$ on $\tube(\gcal^k_{0})$ which fulfils the conditions (F'). Using the symmetry relations in conditions (F'), we want to extend the functions $F_{k}'$ to certain larger graphs.

\subsection{Continue to one stair} \label{sec:fkonestair}

The first graph that we consider is $\gcal_{1}^{k}:=\gcal_{0}^{k}+2 \piv \mathbb{Z}$. This graph has the edges and the nodes of $\gcal^k_0$ translated simultaneously in all coordinates by integer multiples of $2\pi$, cf.~Fig.~\ref{fig:onestair}.

\begin{figure}
\begin{center}
\input{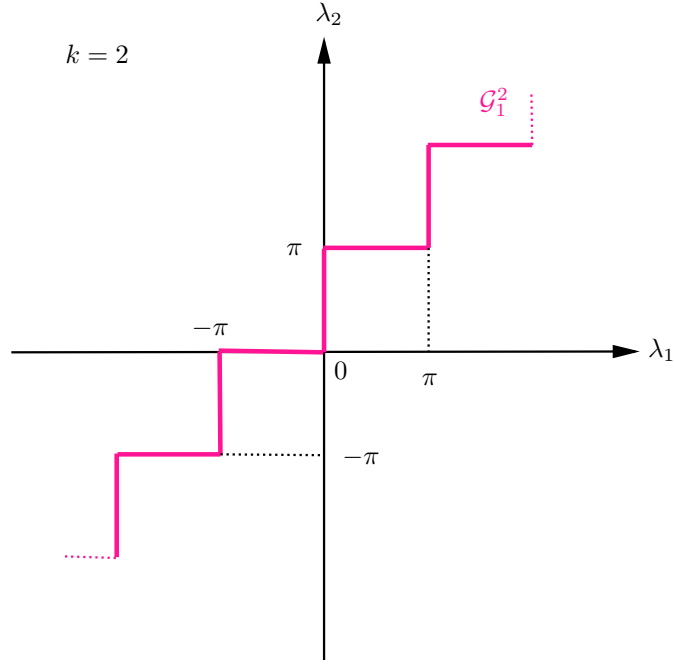}
\caption[The stair $\gcal_{1}^{2}$]{The stair $\gcal_{1}^{2}$, which arises from $\gcal_{0}^{2}$ by $2\pi$ periodic continuation.} \label{fig:onestair}
\end{center}
\end{figure}

We continue $F_{k}'$ to $\tube(\gcal_1^k)$ by defining for $n\in \mathbb{Z}$ and $\zetav\in \tube(\gcal_{0}^{k})$,
\begin{equation}\label{defonestair}
F_{k}'(\zetav+2 i n \piv):=F_{k}'(\zetav).
\end{equation}
This is indeed a CR distribution on the graph, since we can show, using \ref{it:fpperiod}, that the boundary values $F_{k}(\thetav+i\lambdav^{(k,k+2nk)}+i0\ev^{(k)})$ and $F_{k}(\thetav+i\lambdav^{(k,k+2nk)}-i0\ev^{(1)})$ agree for all $n \in \mathbb{Z}$ at real $\thetav$.

This can be see from the following direct computation:
\begin{eqnarray}
F_{k}(\thetav+i\lambdav^{(k,k+2nk)}+i0\ev^{(k)})&=&F_{k}(\thetav+i\lambdav^{(k,k+2nk)}+i0\ev^{(k)}-2in'\piv )\\
&=& F_{k}(\thetav+i\lambdav^{(k,-k)}+i0\ev^{(k)})\nonumber\\
&=&F_{k}(\thetav+i\lambdav^{(k,k)}-i0\ev^{(1)})\nonumber\\
&=&F_{k}(\thetav+i\lambdav^{(k,k+2nk)}-i0\ev^{(1)}),\nonumber
\end{eqnarray}
where in the first and in the last equality we made use of the definition \eqref{defonestair}, where in the second equality we put $n'=n+1$ since $(\lambdav^{(k,k+2nk)}+i0\ev^{(k)}-2in'\pi)$ must be a point in the graph $\gcal_{0}^{k}$, and where in the third equality we made use of the condition \ref{it:fpperiod}.

We notice that the graphs $\gcal_0^k$ and $\gcal_1^k$ can also be written in the following way:
\begin{align}
   \label{setg0}
   \gcal_0^k &= \{  \text{edges} : -\pi \leq \lambda_1 \leq \ldots \leq \lambda_k \leq \lambda_1 +  \pi\leq 2\pi \},
 \\ \label{setg1}
   \gcal_1^k &= \{  \text{edges} : \lambda_1 \leq \ldots \leq \lambda_k \leq \lambda_1 +  \pi \}.
\end{align}
We will use the notation $\{\text{edges} : C(\lambdav) \}$ to indicate the graph given by all the next-neighbour edges on the lattice $\pi\zbb^k$ such that
the condition $C(\lambdav)$ is true for all $\lambdav$ on the edge; the nodes of the graph are the endpoints of these edges.

We can prove the equality in \eqref{setg0} as follows:
Clearly $\gcal^{k}_{0}$ is a subset of \eqref{setg0}.
Since any edge in the set \eqref{setg0} fulfils $\lambda_{m}\in [l\pi, (l+1)\pi]$, with $l=-1,0$, then it is either an edge in $\gcal_{+}^{k} = \{  \text{edges} \,|\, 0\leq  \lambda_1 \leq \ldots \leq \lambda_m \leq \ldots \leq \lambda_k \leq \pi \}$ or in $\gcal_{-}^{k} = \{  \text{edges} \,|\, -\pi\leq  \lambda_1 \leq \ldots \leq \lambda_m \leq \ldots \leq \lambda_k \leq 0 \}$. Since $\gcal_{0}^{k}=\gcal_{+}^{k}\cup \gcal_{-}^{k}$, it is then an edge of $\gcal_{0}^{k}$.

We can prove the equality in \eqref{setg1} as follows:
Since any edge in $\gcal_{0}^{k}$ fulfils the condition $-\pi\leq \lambda_{1}\leq \lambda_{2}\leq \ldots \leq \lambda_{k}\leq \lambda_{1}+\pi\leq 2\pi$, using the definition of $\gcal_{1}^{k}$ we have that a generic edge in $\gcal_{1}^{k}$ is determined by the condition $-\pi+2n\pi\leq \lambda_{1}\leq \ldots \leq \lambda_{k}\leq \lambda_{1}+\pi\leq 2\pi+2n\pi$ for some $n \in \mathbb{Z}$. Since $n$ was arbitrary, this is equivalent to the condition that $\lambda_{1}\leq \ldots \leq \lambda_{k}\leq \lambda_{1}+\pi$.

\subsection{Continue to all stairs} \label{sec:fkallstairs}

The next graph that we consider is for $0 \leq m \leq k-1$ the graph $\gcal^{k}_{1,m}:=\gcal^{k}_{1}+\lambdav^{(k,-m)}=\gcal^{k}_{1}+(-\pi,\ldots,-\pi,0,\ldots,0)$, with $m$ entries of $-\pi$, cf.~Fig.~\ref{fig:multistairs} and Fig.~\ref{fig:multistairs3}; notice that $\gcal^k_{1,0}=\gcal^k_{1}$.

\begin{figure}
\begin{center}
\input{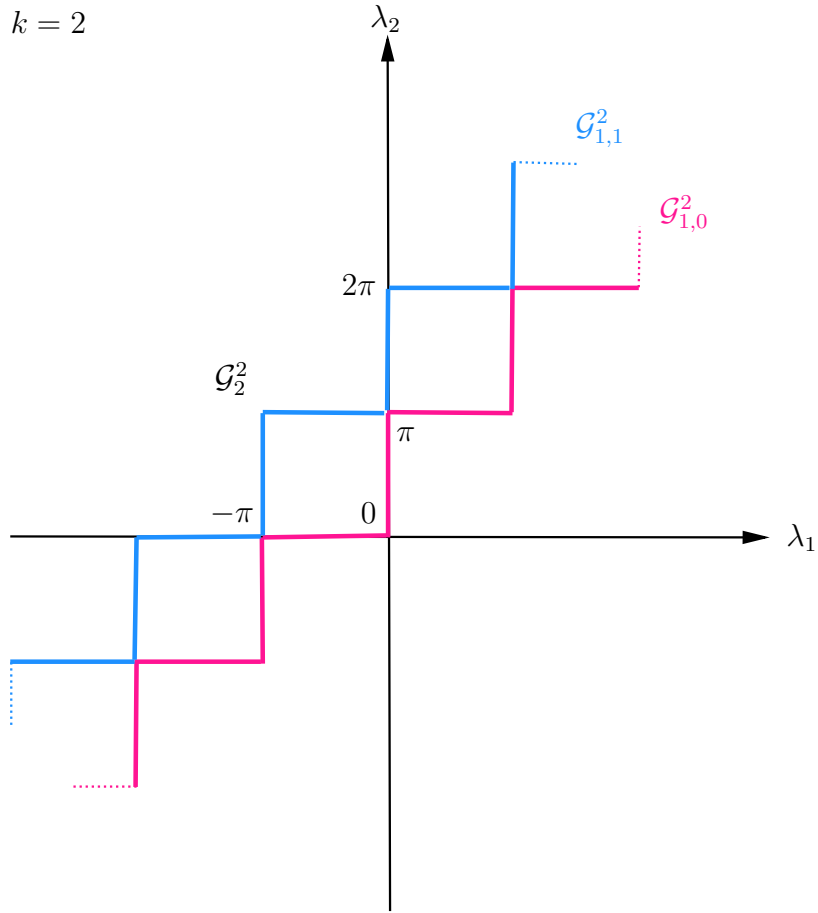}
\caption[The stairs $\gcal_{1,m}^{2}$]{The stairs $\gcal_{1,m}^{2}$, where the stair $\gcal_{1,1}^{2}$ arises from the stair $\gcal_{1,0}^{2}$ by shifting by $-\pi$ in $\lambda_1$ direction.} \label{fig:multistairs}
\end{center}
\end{figure}

\begin{figure}
\begin{center}
\input{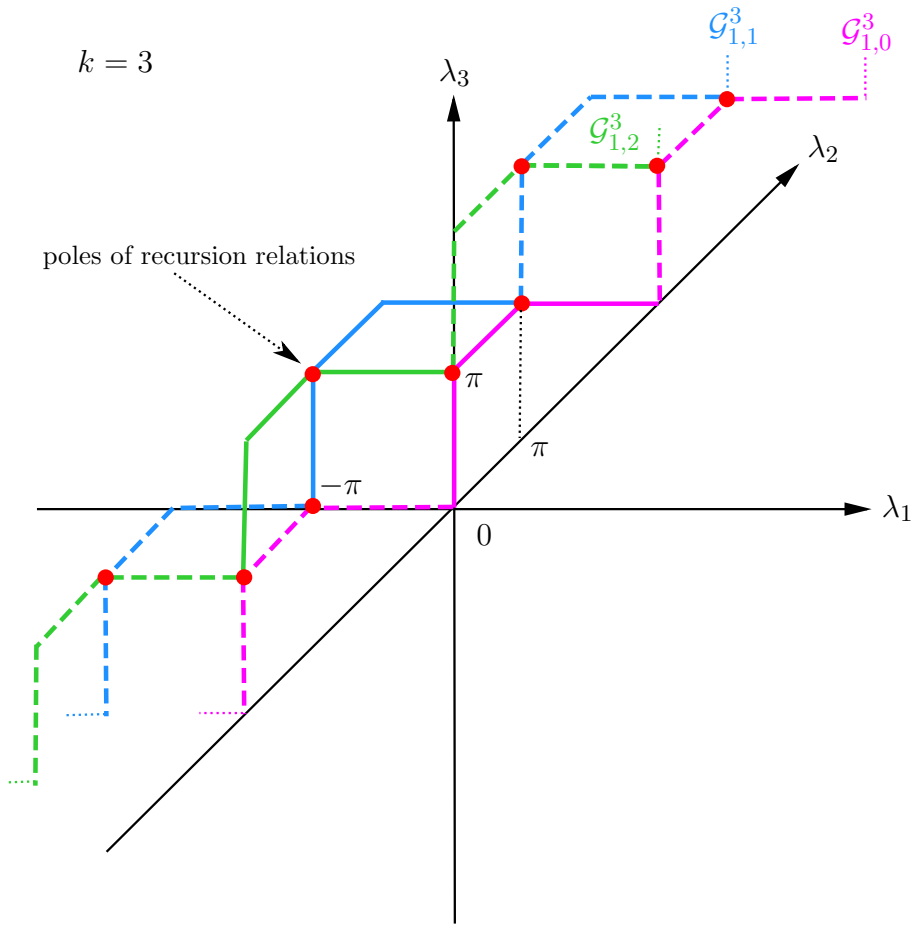}
\caption[The stairs $\gcal_{1,m}^{3}$]{The stairs $\gcal_{1,m}^{3}$, where the stair $\gcal_{1,1}^{3}$ arises from the stair $\gcal_{1,0}^{3}$ by shifting by $-\pi$ in $\lambda_1$ direction and the stair $\gcal_{1,2}^{3}$ arises from $\gcal_{1,0}^{3}$ by shifting by $-\pi$ in $\lambda_2$ direction.} \label{fig:multistairs3}
\end{center}
\end{figure}

We define $F_{k}'$ on $\tube(\gcal^{k}_{1,m})$ by
\begin{equation}\label{stairs}
F_{k}'(\zetav):=F_{k}'(\zeta_{m+1},\ldots,\zeta_{k},\zeta_{1}+2i\pi,\ldots,\zeta_{m}+2i\pi),
\end{equation}
with $\pmb{\zeta}\in \tube(\gcal^{k}_{1,m})$.

We note that this definition is consistent with the domain $\tube(\gcal^{k}_{1})$ of definition of $F_k$ in the following sense: We can show that for $\zetav\in \tube(\gcal^{k}_{1,m})$, the argument $\hat{\zetav}:=(\zeta_{m+1},\ldots,\zeta_{k},\zeta_{1}+2i\pi,\ldots,\zeta_{m}+2i\pi)$ of the function $F_k$ on the r.h.s.~is in $\tube(\gcal^{k}_{1})$.

Indeed, consider a point $\zetav\in \mathcal{T}(\gcal^{k}_{1,m})$,  $\zetav=\check{\zetav}+i(\underbrace{-\pi,\ldots,-\pi}_{m},0,\ldots,0)$, with $\check{\zetav}\in \mathcal{T}(\gcal^{k}_{1})$.

We have $\hat{\zetav}=(\check{\zeta}_{m+1},\ldots,\check{\zeta}_{k},\check{\zeta}_{1}+2i\pi-i\pi,\ldots,\check{\zeta}_{m}+2i\pi-i\pi)= (\check{\zeta}_{m+1},\ldots,\check{\zeta}_{k},\check{\zeta}_{1}+i\pi,\ldots,\check{\zeta}_{m}+i\pi)$.

Since $\im \check{\zetav}\in \gcal_{1}^{k}=\left\{ \text{edges} \,|\, \im\zeta_{1}\leq \ldots \leq \im\zeta_{k}\leq \im\zeta_{1}+\pi\right\}$ (see Eq.~\eqref{setg1}), we have that the following sets of inequalities hold: $\im\check{\zeta}_{m+1}\leq \ldots \leq \im\check{\zeta}_{k}$, $\im\check{\zeta}_{1}+\pi\leq \ldots \leq \im\check{\zeta}_{m}+\pi$, $\im \check{\zeta}_{k}\leq \im \check{\zeta}_{1}+\pi$ and $\im\check{\zeta}_{m}+\pi\leq \im\check{\zeta}_{m+1}+\pi$ since $\im \check{\zeta}_{m}<\im \check{\zeta}_{m+1}$.

Hence, $\hat{\zetav}$ is a point in $\mathcal{T}(\gcal^{k}_{1})$.

Definition \eqref{stairs} yields $F_{k}'$ as a distribution on the tube over the graph
\begin{equation}\label{eq:g2def}
 \gcal^{k}_2 := \bigcup_{0 \leq m \leq k-1} \gcal^{k}_{1,m} =
\{  \text{edges} : \lambda_1 \leq \ldots \leq \lambda_k \leq \lambda_1 + 2 \pi \}.
\end{equation}
To prove the second equality in \eqref{eq:g2def}, consider a generic edge in the set $\{ \text{edges} : \lambda_1 \leq \lambda_2 \leq \ldots \leq \lambda_k \leq \lambda_1 + 2 \pi \}$:
\begin{equation}
\lambdav=(0,\ldots,0,\rho,\pi,\ldots,\pi,\underbrace{2\pi,\ldots,2\pi}_{j})+n\piv ,\quad n\in \mathbb{Z}, \quad \rho \in [0,\pi],\label{firstlambda}
\end{equation}
or
\begin{equation}
\lambdav = (\underbrace{0,\ldots,0}_{j},\pi,\ldots,\pi,\rho,2\pi,\ldots,2\pi)+n\piv ,\quad n\in \mathbb{Z},\quad \rho=[\pi,2\pi].\label{secondlambda}
\end{equation}
In the case of \eqref{firstlambda}, we compute
\begin{equation}
\lambdav+(\pi,\ldots,\pi,\underbrace{0,\ldots,0}_{j})=(\pi,\ldots,\pi,\rho+\pi,2\pi,\ldots,2\pi)+n\piv  \in \gcal^{1}_{k}.
\end{equation}
We find $\lambdav\in \gcal^{k}_{1,k-j}$ (because $\lambdav\in \gcal_{1}^{k}+(-\pi,\ldots,-\pi,\underbrace{0,\ldots,0}_{j})$).

In the case of \eqref{secondlambda}, we compute
\begin{equation}
\lambdav+(\underbrace{\pi,\ldots,\pi}_{j},0,\ldots,0)=(\pi,\ldots,\pi,\rho,2\pi,\ldots,2\pi)+n\piv  \in \gcal^{k}_{1}.
\end{equation}
We find $\lambdav\in \gcal^{k}_{1,j}$ (because $\lambdav\in \gcal_{1}^{k}+(\underbrace{-\pi,\ldots,-\pi}_{j},0,\ldots,0)$).

This concludes the proof of the second equality in \eqref{eq:g2def}.

\subsection{Difference of boundary values} \label{sec:fkdifference}

It is important to note that $F_k'$ are CR distributions on all $\tube(\gcal^{k}_{1,m})$, but they are \emph{not} CR distributions on $\tube(\gcal^k_2)$. In other words, we have that on the nodes that two graphs $\gcal^{k}_{1,m}$ and $\gcal^{k}_{1,m'}$ ($m > m')$ have in common, the boundary values of $F_k'$ from different edges do not need to agree. We can show that these common nodes are given by
\begin{equation}\label{eq:stairsmeet}
\underbrace{(\overbrace{-\pi \ldots -\pi}^{m'}, \overbrace{0 \ldots 0}^{m-m'}, \overbrace{\pi \ldots \pi}^{k-m})}_{=:\lambdav^{m \cap m'}}+\ell\piv , \quad \ell\in \mathbb{Z}.
\end{equation}

Indeed, let $\zetav$ be a point in $\gcal^{k}_{m}$ and $\gcal^{k}_{m'}$, with $m'<m$ and let $\lambdav:=\im\zetav$.

Since $\lambdav\in \gcal^{k}_{m}$, we have $\lambdav+(\underbrace{\pi,\ldots,\pi}_{m},0,\ldots,0)\in \gcal_{1}^{k}$, namely $(\lambda_{1}+\pi,\ldots,\lambda_{m}+\pi,\lambda_{m+1},\ldots,\lambda_{k})\in \gcal^{k}_{1}$.

This implies
\begin{equation}
\lambda_{1}+\pi\leq \lambda_{m'}+\pi\leq \lambda_{m}+\pi\leq \lambda_{m+1}\leq \ldots \leq \lambda_{k}\leq \lambda_{1}+2\pi,
\end{equation}
which gives in particular the condition
\begin{equation}
\lambda_{m'+1}\leq \ldots \leq \lambda_{m}\leq \lambda_{m+1}-\pi\leq \ldots \leq \lambda_{k}-\pi\leq \lambda_{1}+\pi.\label{condone}
\end{equation}
On the other hand, since $\lambdav\in \gcal^{k}_{m'}$, we have also $\lambdav+(\underbrace{\pi,\ldots,\pi}_{m'},0,\ldots,0)\in \gcal_{1}^{k}$, namely $(\lambda_{1}+\pi,\ldots,\lambda_{m'}+\pi,\lambda_{m'+1},\ldots,\lambda_{k})\in \gcal^{k}_{1}$.

This implies
\begin{equation}
\lambda_{1}+\pi\leq\ldots \leq \lambda_{m'}+\pi\leq \lambda_{m'+1}\leq \ldots \leq \lambda_{k}\leq \lambda_{1}+2\pi,
\end{equation}
which gives in particular
\begin{equation}
\lambda_{1}+\pi\leq \ldots \leq \lambda_{m'}+\pi\leq \lambda_{m'+1}\leq \ldots \leq \lambda_{m}.\label{condtwo}
\end{equation}
From the conditions \eqref{condone} and \eqref{condtwo}, we find
\begin{equation}
\lambda_{1}+\pi\leq \ldots \leq \lambda_{m'}+\pi\leq \lambda_{m'+1}\leq \ldots \leq \lambda_{m}\leq \lambda_{m+1}-\pi \leq \ldots \leq \lambda_{k}-\pi \leq \lambda_{1}+\pi.
\end{equation}
This gives
\begin{equation}
\lambda_{1}+\pi= \ldots = \lambda_{m'}+\pi= \lambda_{m'+1}= \ldots = \lambda_{m}= \lambda_{m+1}-\pi = \ldots = \lambda_{k}-\pi =\lambda_{1}+\pi.
\end{equation}
Hence, we have
\begin{equation}
\lambdav=(\underbrace{\lambda_{1} \ldots \lambda_{1}}_{m'},\underbrace{\lambda_{1}+\pi \ldots \lambda_{1}+\pi}_{m-m'},\underbrace{\lambda_{1}+2\pi \ldots \lambda_{1}+2\pi}_{k-m})=(0\ldots 0,\pi \ldots \pi,2\pi \ldots 2\pi)+\underbrace{\lambda_{1}\pmb{1}}_{\in \piv \mathbb{Z}}.
\end{equation}
Therefore, the points in common between the two stairs $\gcal^{k}_{1,m}$ and $\gcal^{k}_{1,m'}$, $m'<m$, are given by
\begin{equation}
\lambdav=(\underbrace{-\pi \ldots -\pi}_{m'}, \underbrace{0 \ldots 0}_{m-m'}, \underbrace{\pi \ldots \pi}_{k-m})+\ell\piv , \quad \ell\in \mathbb{Z}.
\end{equation}
Now, we compute the difference of the boundary values of $F_k$ at the point $\zetav=\thetav+i\lambdav^{m \cap m'}$, with $\ell=0$.

On $\tube(\gcal_{1,m}^{k})$ (i.e. $\zetav \in \tube(\gcal_{1,m}^{k})$), we have
\begin{multline}\label{eq:fgm}
F_{k}'(\zetav)\big\vert_{\gcal^{k}_{1,m}}=F_{k}'(\zeta_{m+1},\ldots,\zeta_{k},\zeta_{1}+2i\pi,\ldots,\zeta_{m}+2i\pi)\\
=F_{k}'(\theta_{m+1}+i\pi,\ldots,\theta_{k}+i\pi,\theta_{1}+2i\pi-i\pi,\ldots,\theta_{m'}+2i\pi-i\pi,\theta_{m'+1}+2i\pi,\ldots,\theta_{m}+2i\pi)\\
=F_{k}'(\theta_{m+1}-i\pi,\ldots,\theta_{k}-i\pi,\theta_{1}-i\pi,\ldots,\theta_{m'}-i\pi,\theta_{m'+1},\ldots,\theta_{m}),
\end{multline}
where in the first equality we made use of definition \eqref{stairs}, where in the second equality we used Eq.~\eqref{eq:stairsmeet} with $\ell =0$ and in the third equality we made use of condition \ref{it:fpperiod}. Note that in the last line of \eqref{eq:fgm} the function $F_{k}$ is evaluated at the node $\pmb{\lambda}^{-(m'+k-m)}$ of $\gcal_{0}^{k}$.

Analogously, we find for $\pmb{\zeta}\in \tube(\gcal_{1,m'}^{k})$:
\begin{multline}\label{eq:fgmp}
F_{k}'(\zetav)\big\vert_{\gcal^{k}_{1,m'}}
=F_{k}'(\zeta_{m'+1},\ldots,\zeta_{k},\zeta_{1}+2i\pi,\ldots,\zeta_{m'}+2i\pi)\\
=F_{k}'(\theta_{m'+1},\ldots,\theta_{m},\theta_{m+1}+i\pi,\ldots,\theta_{k}+i\pi,\theta_{1}+2i\pi-i\pi,\ldots,\theta_{m'}+2i\pi-i\pi)\\
=F_{k}'(\theta_{m'+1},\ldots,\theta_{m},\theta_{m+1}+i\pi,\ldots,\theta_{k}+i\pi,\theta_{1}+i\pi,\ldots,\theta_{m'}+i\pi),
\end{multline}
where in the first equality we made use of definition \eqref{stairs} and in the second equality we used Eq.~\eqref{eq:stairsmeet} with $\ell=0$. Note that in last line of \eqref{eq:fgmp} the function $F_{k}$ is evaluated at the node $\pmb{\lambda}^{(k-m+m')}$ of $\gcal_{0}^{k}$.

We can compute the difference of the boundary values \eqref{eq:fgm} and \eqref{eq:fgmp} using condition \ref{it:fprecursion}; this difference is in general non-zero and also quite complicated to write down. Below we will simplify it by multiplying the functions $F_k'$ with certain linear factors (see Prop.~\ref{proposition:extendinterior}).

We did the computation above considering the point \eqref{eq:stairsmeet} with $\ell=0$; we can compute the case $\ell=1$ in a similar way, and then we can obtain the corresponding result for general $\ell$ by periodicity.

\section{Extend to the interior} \label{sec:fkinterior}

Now we will use the results of the previous section in order to construct the meromorphic functions $F_k$ (fulfilling the properties (F)); first we construct it on the tube over the open set:
\begin{equation}\label{eq:ich2}
\ical_2^k := \ich \gcal_2^k =\{ \lambdav\in\rbb^k : \lambda_{1}<\ldots < \lambda_{k}< \lambda_{1}+2\pi \}.
\end{equation}
We can prove the second equality in \eqref{eq:ich2} following the argument schematized in the four points below:
\begin{enumerate}
\item $\ich{\gcal_{2}^{k}}$ is translation invariant (by $k\piv $ obvious, by other translation due to convex combination: if $\lambdav\in \ich{\gcal_{2}^{k}}$, then $\lambdav+k\piv \in \ich{\gcal_{2}^{k}}$, therefore for every \\ $\mu \in [0,1)$: $\underbrace{\mu(\lambdav+k\piv )+(1-\mu)\lambdav}_{\lambdav+\mu k \piv }\in \ich{\gcal_{2}^{k}}$).

\item $ \ich{\gcal_{++}^{k}}=\ich(\{ \text{edges} \;|\; 0\leq \lambda_{1}\leq \ldots \leq \lambda_{k}\leq 2\pi \})=\{ \lambdav \; |\; 0< \lambda_{1}< \ldots <\lambda_k < 2\pi (< 2\pi +\lambda_1)\}$, cf.~\cite[Corollary 5.2.6]{Lechner:2006}.

\item Due to item 1, we have for all $c\in \rbb$ that $\ich{\gcal_{2}^{k}}\supset  \ich{\gcal_{++}^{k}}+c$. This implies $\bigcup_{c} (\ich{\gcal_{++}^{k}}+c)\subset \ich{\gcal_{2}^{k}}$. We call $M_c := \ich{\gcal_{2}^{k}} + c $. This set can alternatively be written as $M_c = \{ \lambdav \, | \, 0 < \lambda_1-c < \ldots < \lambda_k -c < 2\pi  \}$. Set $M := \{ \lambdav \, | \, \lambda_1 < \ldots < \lambda_k < \lambda_1 +2\pi \}$. We want to show that $M \subset \bigcup_{c} M_c$ (which implies $M \subset \ich{\gcal_{2}^{k}}$). We prove this as follows: Given $\lambdav \in M$, choose $c= \lambda_1 - \frac{\lambda_1 +2\pi -\lambda_k}{2}= \frac{\lambda_1}{2} + \frac{\lambda_k}{2} -\pi$. Then we show that $\lambdav \in M_c$: Indeed, we have $\lambda_1 -c = \frac{\lambda_1}{2}-\frac{\lambda_k}{2}+\frac{2\pi}{2}=\frac{1}{2}(\lambda_1 -\lambda_k +2\pi)>0$, $\lambda_k -c = \frac{\lambda_k}{2}-\frac{\lambda_1}{2}+\pi= \frac{1}{2}(\lambda_k - \lambda_1 + 2\pi)<2\pi$, and moreover we have that $\lambda_i < \lambda_j$ if and only if $\lambda_i - c < \lambda_j - c$. This concludes the proof that $M \subset \ich{\gcal_{2}^{k}}$.

\item On the other hand, $\overline{\gcal_{2}^{k}}= \{ \text{edges} \; |\; \lambda_{1}\leq  \ldots \leq \lambda_{k}\leq \lambda_{1}+2\pi \}\subset \overline{M}$. Since $M$ is convex, we have $\ich \gcal_{2}^{k}\subset M$. Hence, we have shown that $\ich{\gcal_{2}^{k}}= M$.
\end{enumerate}

Now, we prove the following proposition which shows that we can extend meromorphically the functions $F_k$ to $\tube(\ical_2^k)$.
\begin{proposition} \label{proposition:extendinterior}
 Let $F_k'$ be distributions fulfilling (F'). Then there exist meromorphic functions $F_k$ on $\tube(\ical_2^k)$ which have the boundary values \eqref{eq:ffpboundary}. They are analytic except for possible first-order poles at $\zeta_j-\zeta_m=i\pi$, $j>m$.
\end{proposition}

\begin{proof}
Using the functions $F_k'$ defined in the previous section, we can define distributions $G_k$ on $\mathcal{T}(\gcal_{2}^{k})$ by
\begin{equation}\label{gdef}
G_k(\zetav) := F_k'(\zetav) \cdot \prod_{j > j'} \frac{ \zeta_{j}-\zeta_{j'} - i \pi}{\zeta_{j}-\zeta_{j'}+i\pi}.
\end{equation}
We can show that these functions $G_k$ are CR distributions on the graph $\gcal_{2}^{k}$. We already know that the functions $F_k'$ are CR distributions on all $\tube(\gcal^{k}_{1,m})$, but they are possibly not CR distributions on $\tube(\gcal^k_2)$.

We note that the rational factor $\prod_{j > j'} \frac{ \zeta_{j}-\zeta_{j'} - i \pi}{\zeta_{j}-\zeta_{j'}+i\pi}$ is bounded at real infinity and it has a pole at $\zeta_{j}-\zeta_{j'}+i\pi=0$, $(i>j)$.

\emph{Remark}: the pole of the rational factor will not effect our argument in Section \ref{sec:fkpermute} since we will never consider $G_k$ in the region $\{ \zetav \; |\; |\im \zeta_j - \im \zeta_{j'}|< 2\pi \}$. Indeed, by the time we will extend the function to there in Section \ref{sec:fkpermute}, we have gone back to $F_k$.

So, it remains to show that the boundary values of the functions $G_k$ agree at the nodes $\im \zetav=\lambdav^{m \cap m'}$, given by Eq.~\eqref{eq:stairsmeet}. For this, we consider the product $F_k'(\zetav) \cdot \prod_{j > j'} \frac{ \zeta_{j}-\zeta_{j'} - i \pi}{\zeta_{j}-\zeta_{j'}+i\pi}$ at the points $\zetav=\thetav +i \lambdav^{m \cap m'}$. Using \ref{it:fprecursion}, we can compute the difference of the boundary values \eqref{eq:fgm} and \eqref{eq:fgmp}. Inserting this into the product above, we find in particular the product $\delta_C \cdot \prod_{j > j'} \frac{ \zeta_{j}-\zeta_{j'} - i \pi}{\zeta_{j}-\zeta_{j'}+i\pi}$. In $\delta_C$, factors $\delta(\theta_\ell -\theta_r)$ of two types can occur: either $\ell \in \{m'+1,\ldots,m\}$ and $r \in \{m+1,\ldots, k\}$ or $\ell \in \{m'+1,\ldots,m\}$ and $r \in \{1,\ldots, m'\}$. On the other hand, the factor $\prod_{j > j'} \frac{ \zeta_{j}-\zeta_{j'} - i \pi}{\zeta_{j}-\zeta_{j'}+i\pi}$, evaluated at the nodes $\im \zetav=\lambdav^{m \cap m'}$, contains also terms where $j' \in \{m'+1,\ldots, m \}$ and $j \in \{ m+1,\ldots, k \}$ or where $j'\in \{ 1, \ldots, m'\}$ and $j\in \{m' +1, \ldots, m\}$. So, one sees that the support of these delta functions always coincide with the zero of the rational factor in \eqref{gdef}. It is indeed $\zeta_\ell -\zeta_r = \pm i\pi$ which cancel the delta functions, except for the term corresponding to the contraction $C=(m,n,\emptyset)$.

One follows a similar argument for the case $\ell = 1$. In this case, we have to show that the boundary values of the functions $G_k$ agree at the nodes $\im \zetav=\lambdav^{m \cap m'} +i\piv$. In place of \eqref{eq:fgm} and \eqref{eq:fgmp} we find, respectively,
\begin{eqnarray}
F_{k}'(\zetav)\big\vert_{\gcal^{k}_{1,m}} &=& F_{k}'(\theta_{m+1},\ldots,\theta_k,\theta_1,\ldots,\theta_{m'},\theta_{m'+1}+i\pi,\ldots,\theta_m +i\pi),\label{bovalell1}\\
F_{k}'(\zetav)\big\vert_{\gcal^{k}_{1,m'}} &=& F_{k}'(\theta_{m'+1}-i\pi,\ldots,\theta_m -i\pi, \theta_{m+1},\ldots, \theta_k,\theta_1,\ldots,\theta_{m'}).\label{bouvalell1}
\end{eqnarray}
As before, we consider the product $F_k'(\zetav) \cdot \prod_{j > j'} \frac{ \zeta_{j}-\zeta_{j'} - i \pi}{\zeta_{j}-\zeta_{j'}+i\pi}$, but at the points $\zetav=\thetav +i\lambdav^{m \cap m'}+i\piv$. We compute the difference of the boundary values \eqref{bovalell1} and \eqref{bouvalell1} using \ref{it:fprecursion}. Inserting this into the product above, we find in particular the product $\delta_C \cdot \prod_{j > j'} \frac{ \zeta_{j}-\zeta_{j'} - i \pi}{\zeta_{j}-\zeta_{j'}+i\pi}$. In $\delta_C$, we can have factors $\delta(\theta_\ell -\theta_r)$ where $\ell \in \{ m+1,\ldots, k\}$ and $r\in \{ m'+1, \ldots, m \}$ or where $\ell \in \{1,\ldots, m' \}$ and $r\in \{ m' +1, \ldots, m \}$. On the other hand, the factor $\prod_{j > j'} \frac{ \zeta_{j}-\zeta_{j'} - i \pi}{\zeta_{j}-\zeta_{j'}+i\pi}$, now evaluated at the nodes $\im \zetav=\lambdav^{m \cap m'} +i\piv$, contains also terms where $j' \in \{ m'+1,\ldots,m \}$ and $j \in \{m+1,\ldots, k \}$ or where $j' \in \{ 1, \ldots, m' \}$ and $j \in \{ m'+1,\ldots, m \}$. So, also in this case, one sees that the support of these delta functions always coincide with the zero of the rational factor in \eqref{gdef}. Therefore, the zero $\zeta_\ell -\zeta_r = \pm i\pi$ cancels the delta functions, except for the term corresponding to the contraction $C=(m,n,\emptyset)$.

One can then obtain the corresponding result for general $\ell$ by periodicity.

That gives:
\begin{equation}
G_{k}(\zetav)\big\vert_{\gcal^{k}_{1,m}}=G_{k}(\zetav)\big\vert_{\gcal^{k}_{1,m'}},
\end{equation}
where $G_{k}(\zetav)\big\vert_{\gcal^{k}_{1,m}}$ indicates that we approach the points \eqref{eq:stairsmeet} from the edges of $\gcal^{k}_{1,m}$. Hence, the $G_k$ are CR distributions on $\gcal_2^k$.

Now we can apply Lemma \ref{lem:graphtube}, which gives an extension of $G_k$ to an analytic function on $\tube(\ich \gcal_{2}^{k})$, with distributional boundary values on $\tube(\ach \gcal_{2}^{k})$. Then, we define $F_k$ as
\begin{equation}
F_k(\zetav) := G_k(\zetav) \cdot \prod_{j > j'} \frac{ \zeta_{j}-\zeta_{j'} + i \pi}{\zeta_{j}-\zeta_{j'}-i\pi},\label{def:FfromG}
\end{equation}
This function is clearly analytic in the same domain, but it has possible poles at $\zeta_{j}-\zeta_{j'}=i\pi$. Considering the limit in the sense of distributions of $F_k$ to the boundary of $\tube(\gcal^k_0)$, we find that it coincides with $F_k'$ by construction. This because the function $F_k$ is analytic on the edges of $\gcal^k_0$ and therefore there is no need of the rational factor in \eqref{gdef}. So, on the edges of $\gcal^k_0$ the functions $F_k$ and $F'_k$ agree, and therefore they agree on the nodes of the graph. This gives \eqref{eq:ffpboundary}.
\end{proof}

\section{Permuted stairs} \label{sec:fkpermute}

We can extend meromorphically the functions $F_{k}$ to a larger graph by using the property of $S$-symmetry of the functions $F_k'$ cf.~Fig.~\ref{fig:multistairsS}, as the following proposition shows.

\begin{proposition}\label{proposition:extendpermuted}
  The functions $F_k$ of Prop.~\ref{proposition:extendinterior} extend meromorphically to $\tube(\ical_3^k)$, where
\begin{equation}\label{eq:i3k}
\ical_3^k :=\{ \lambdav \in \rbb^k: |\lambda_{j}-\lambda_{j'}|< 2\pi \text{ for all $j,j'$}\}.
\end{equation}
They fulfil the S-symmetry condition \ref{it:fsymm}.
\end{proposition}

\begin{figure}
\begin{center}
\input{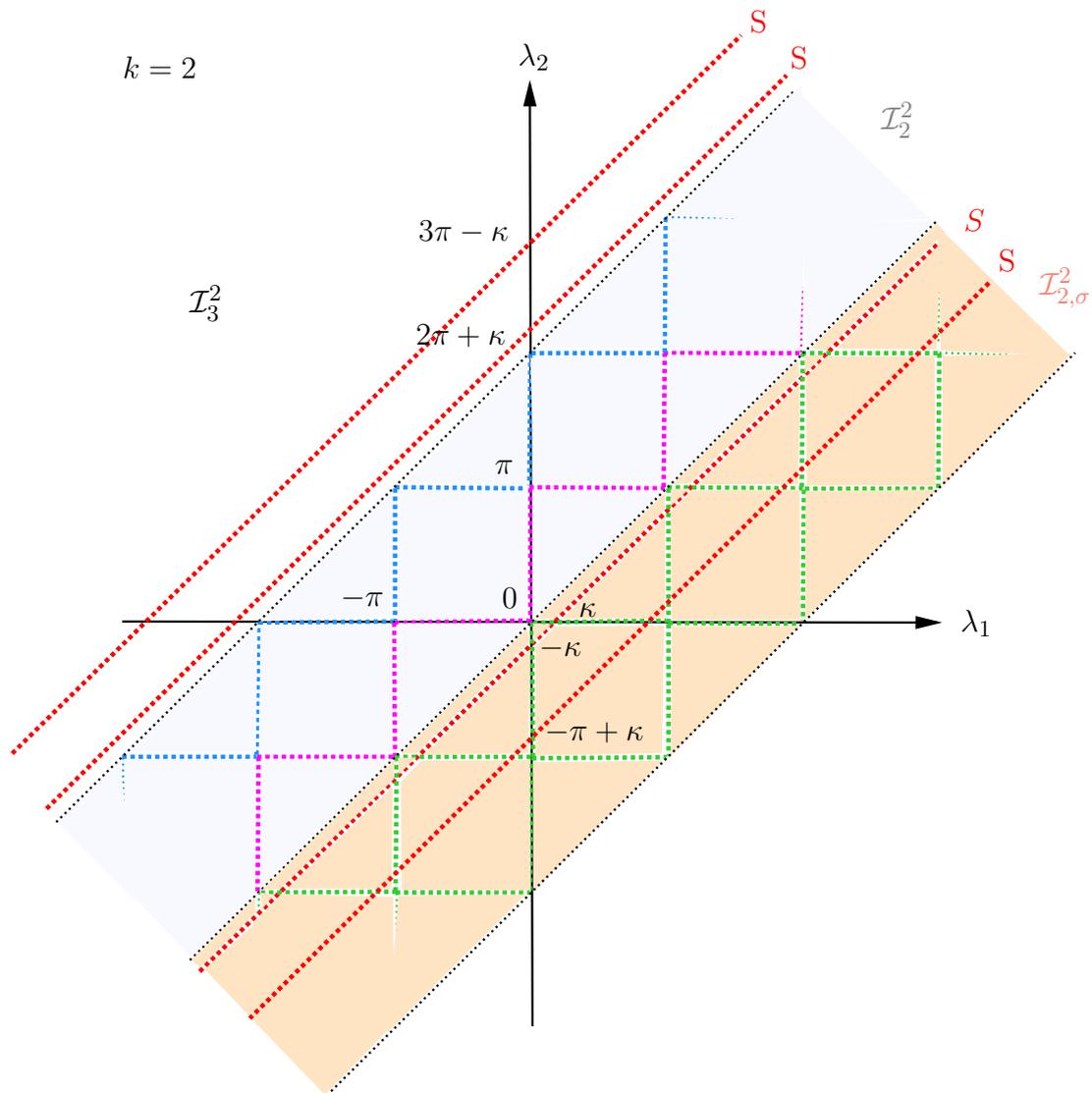}
\caption[The regions $\ical_{2}^{2}$, $\ical_{2,\sigma}^{2}$ and $\ical_{3}^{2}$]{The regions $\ical_{2}^{2}$, $\ical_{2,\sigma}^{2}$ and $\ical_{3}^{2}$, where $\ical_{3}^{2}$ is the convex hull of the union of $\ical_{2}^{2}$ and $\ical_{2,\sigma}^{2}$, and where $\ical_{2}^{2}$, $\ical_{2,\sigma}^{2}$ are defined by Eq.~\eqref{eq:ich2} and Eq.~\eqref{i2sigma}, respectively. The red dotted lines show poles due to the function $S$ and for simplicity we restrict to the case where $S$ has only one pole in the strip $-\pi/2<\im \zeta<0$.} \label{fig:multistairsS}
\end{center}
\end{figure}

\begin{proof}
Given a fixed permutation $\sigma \in \perms{k}$, we consider the ``permuted region''
\begin{equation}\label{i2sigma}
  \ical_{2,\sigma}^k :=\{ \lambdav  \in \rbb^k: \lambda_{\sigma(1)}<\ldots<\lambda_{\sigma(k)}<\lambda_{\sigma(1)}+2\pi \}.
\end{equation}
We define the function $F_k$ on the tube based on this permuted region $\tube(\ical_{2,\sigma}^k)$ by
\begin{equation}\label{permf}
F_{k}(\zetav):=F_{k}(\zeta_{\sigma(1)},\ldots,\zeta_{\sigma(k)}) \,S^{\sigma}(\zetav),
\end{equation}
where $S^\sigma$ is given by Eq.~\eqref{eq:Sperm}. Since $S$ is a meromorphic function for all arguments, $S^\sigma$ is also a meromorphic function for all arguments; hence, \eqref{permf} defines $F_k$ as a meromorphic function on each of the disjoint regions $\tube(\ical_{2,\sigma}^k)$.
But since $S$ has no poles on the real line, we can find a complex neighbourhood $\ncal$ of $\rbb^k$ (not necessarily tubular) where all $S^\sigma$ are analytic, cf.~Fig.~\ref{fig:neighperm}; hence $F_k$ is analytic in $\ncal \cap \tube(\ical_{2,\sigma}^k)$ for all $\sigma$. Due to the property of $S$-symmetry \ref{it:fpsymm}, the boundary values of $F_k$ at $\rbb^k$ from within all these domains coincide in the sense of distributions. So, we can apply the edge-of-the-wedge theorem (for example, see \cite{Eps:edge_of_wedge}) around each real point, and find that $F_k$ has an analytic continuation to a possibly smaller complex neighbourhood $\ncal'\subset \ncal$ of $\rbb^k$. This implies that $F_k$ is meromorphic on the connected domain
\begin{equation}
   \rcal := \ncal' \cup \bigcup_{\sigma \in \perms{k}} \tube(\ical_{2,\sigma}^k).
\end{equation}
By \cite{Eps:edge_of_wedge} we have just shown that $F_k$ is meromorphic in the tubes $\tube(\ical_{2,\sigma}^k)$ and that the boundary values at $\rbb$ coincide; in the case where the function was analytic, we could apply the tubular edge-of-the-wedge theorem \cite{Bros:1977} and extend this function to the envelope of holomorphy of $\rcal$ given by $\ich(\rcal)$. To extend this result to the case of meromorphic functions, we use \cite[Theorem~3.6.6]{JarnickiPflug:2000}, which says that the envelope of meromorphy is the same as the envelope of holomorphy. So, we can extend the function meromorphically to $\ich(\rcal) = \tube(\ical_{3}^k)$.
\end{proof}

\begin{figure}
\begin{center}
\input{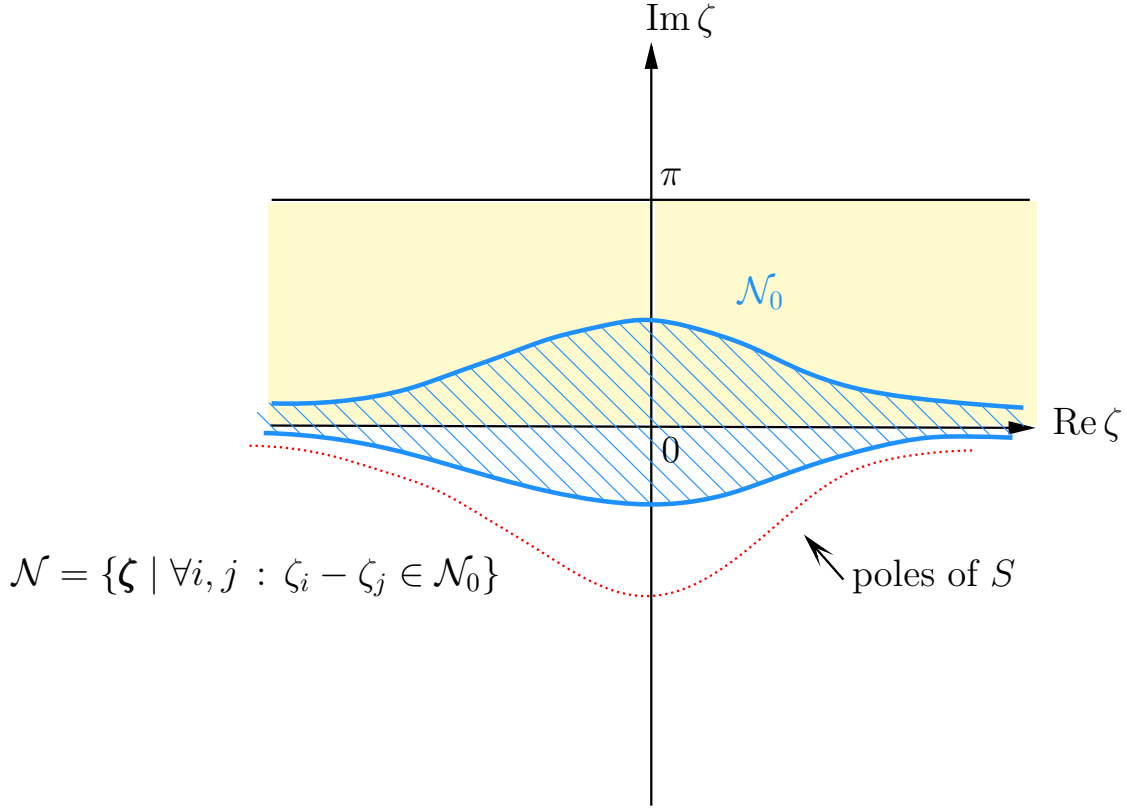}
\caption[The neighbourhood $\ncal$]{The neighbourhood $\ncal$. If the poles of the function $S$ approach the real axis as $\re \zeta \rightarrow \pm \infty$ as depicted above, then we can choose a complex neighbourhood $\ncal_0$ of the real axis as indicated, on which $S$ is still analytic.} \label{fig:neighperm}
\end{center}
\end{figure}

\section{Extension to entire plane} \label{sec:fkeverywhere}

Now we use the properties of periodicity and $S$-symmetry of $F_k'$ to extend meromorphically $F_k$ to the entire multi-variable complex plane, cf.~Fig.~\ref{fig:entire2}.

\begin{proposition}\label{proposition:extendeverywhere}
  The functions $F_k$ of Prop.~\ref{proposition:extendinterior} extend meromorphically to $\cbb^k$. They fulfil \ref{it:fmero}, \ref{it:fsymm} and \ref{it:fperiod}.
\end{proposition}

\begin{figure}
\begin{center}
\input{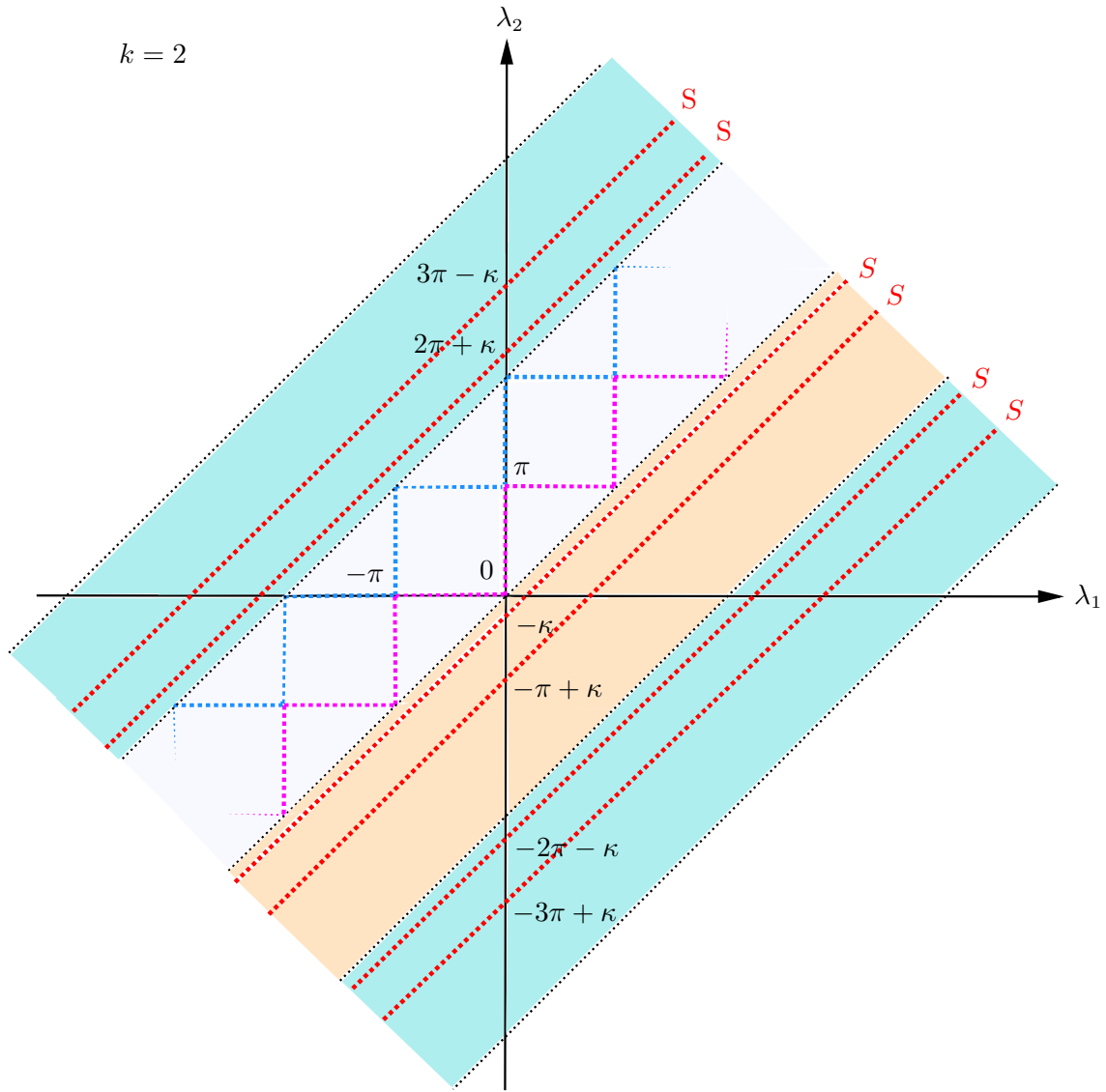}
\caption[The extension of $F_2$ to the entire rapidity multi-variables complex plane]{The extension of $F_2$ to the entire rapidity multi-variables complex plane by $2\pi$ periodic continuation of $\ical_{3}^{2}$.} \label{fig:entire2}
\end{center}
\end{figure}

\begin{proof}
We define $F_{k}$ on $\cbb^k$ by
\begin{equation}\label{fperiodica}
F_{k}(\zetav) :=
\Bigg( \prod_{\ell=1}^{k}\Big( \prod_{\substack{j=1\\ j \neq \ell}}^{k} S(\zeta_\ell-\zeta_j)
   \Big)^{ n_\ell } \Bigg)
F_{k}(\zetav + 2 i \pi \nv),
\end{equation}
where we choose $\nv\in\zbb^k$ such that $\zetav + 2 i \pi\nv \in \tube(\ical_3^k)$.
To show that this is well-defined, we need first to prove that it is possible to choose such $\nv$ for any $\zetav$: Given $\lambdav\in \mathbb{R}^{k}$, choose $n_{j}\in \mathbb{Z}$ such that $\tilde{\lambda}_{j}:=\lambda_{j}+2\pi n_{j}\in [0,2\pi)$. Then $|\tilde{\lambda}_{j}-\tilde{\lambda}_{i}|< 2\pi$ for all $i,j$, and therefore $\tilde{\lambdav}\in \ical_3^k$.
However, we might have several of such choices for $\nv$. Suppose that, for fixed $\zetav$, there exist $\nv\neq\nv'$ such that $\im\zetav+2\pi\nv\in \ical_3^k$ and $\im\zetav+2\pi\nv'\in \ical_3^k$. In this case, we have (a priori) two definitions of $F_{k}(\zetav)$, namely one built with $\nv$ and one built with $\nv'$. What we have to prove is that these two definitions actually give the same value $F_{k}(\zetav)$. Namely, we need to show that
\begin{equation}\label{periodrel}
\underbrace{\prod_{\ell=1}^{k}\Big( \prod_{\substack{j=1\\ j \neq \ell}}^{k} S(\zeta_\ell-\zeta_j)
   \Big)^{ n_\ell }}_{=:S_{\nv}(\zetav)}
F_{k}(\zetav +2i\pi \nv)
=\underbrace{\prod_{\ell=1}^{k}\Big( \prod_{\substack{j=1\\ j \neq \ell}}^{k} S(\zeta_\ell-\zeta_j)
   \Big)^{ n'_\ell }}_{=:S_{\nv'}(\zetav)}
F_{k}(\zetav+2i\pi \nv').
\end{equation}
In order to simplify the above expression, we call $\hat{\zetav}=\zetav+2i\pi \nv'$. Dividing by $S_{\nv'}(\zetav)$, and using $2 i\pi$-periodicity of the $S$-factors, $S(\hat{\zeta}_{\ell}-\hat{\zeta}_{i})=S(\zeta_{\ell}-\zeta_{i}+2\pi i)=S(\zeta_{\ell}-\zeta_{i})$, formula \eqref{periodrel} becomes:
\begin{equation}
F_{k}(\hat{\zetav}+2i\pi (\nv-\nv'))
=\prod_{\ell=1}^{k}\Big( \prod_{\substack{i=1\\ i \neq \ell}}^{k} S(\zeta_\ell-\zeta_i)
   \Big)^{n_\ell- n'_\ell }F_{k}(\hat{\zetav})\label{mminusnp}
\end{equation}
Calling for all $j$, $n_{j}-n'_{j}=: n''_{j}$ we get from \eqref{mminusnp} the same expression as in \eqref{fperiodica}. So, we can assume without loss of generality that $\nv'=0$ and $\im\zetav \in \ical_3^k$.

Hence, for $\zetav\in \ical_3^k$ and $\zetav+2i\pi\nv\in \ical_3^k$, we want to show:
\begin{equation}
F_{k}(\zetav)=\prod_{\ell=1}^{k}\Big( \prod_{\substack{i=1\\ i \neq \ell}}^{k} S(\zeta_\ell-\zeta_i)
   \Big)^{ n_\ell }
F_{k}(\zetav +2i\pi \nv).\label{nnpcond}
\end{equation}
We can check that the factor $S_{\nv}(\zetav)$ defined above has the property that $S_{\nv}(\zetav)=S_{\nv^\rho}(\zetav^\rho)$ for any permutation $\rho$. Indeed, we have by definition:
\begin{eqnarray}
S_{\nv^{\rho}}(\zetav^{\rho}) &=& \prod_{\ell=1}^{k}\Big( \prod_{\substack{j=1\\ j \neq \ell}}^{k} S(\zeta_{\rho(\ell)}-\zeta_\rho{j})\Big)^{n_{\rho(\ell)}}\nonumber\\
&=& \prod_{\rho^{-1}(\ell)=1}^{k}\Big( \prod_{\substack{\rho^{-1}(j)=1\\ \rho^{-1}(j) \neq \rho^{-1}(\ell)}}^{k} S(\zeta_{\ell}-\zeta_j)\Big)^{n_{\ell}} \nonumber\\
&=& \prod_{\ell=1}^{k}\Big( \prod_{\substack{j=1\\ j \neq \ell}}^{k} S(\zeta_{\ell}-\zeta_j)\Big)^{n_{\ell}}.\label{Snrho}
\end{eqnarray}
since the products over $\ell$ and $j$ run over all the indices $1,\ldots, k$, and since $\rho^{-1}(j)= \rho^{-1}(\ell)$ if and only if $j=\ell$.

We can see that relation \eqref{periodrel} is invariant under the permutation of the components of $\zetav$ and $\nv$ by $\rho$:
\begin{equation}
S_{\nv^{\rho}}(\zetav^{\rho})F_{k}(\zetav^{\rho} +2i\pi \nv^{\rho}) =S_{{\nv'}^{\rho}}(\zetav^{\rho})F_{k}(\zetav^{\rho}+2i\pi {\nv'}^{\rho}).
\end{equation}
The invariance of this relation is due to \eqref{Snrho} and the fact that $F_k$ is S-symmetric by Prop.~\ref{proposition:extendpermuted}:
\begin{equation}
S_{\nv}(\zetav)S^{\rho}(\zetav +2i\pi \nv)F_{k}(\zetav +2i\pi \nv) =S_{\nv'}(\zetav)S^{\rho}(\zetav+2i\pi \nv')F_{k}(\zetav+2i\pi \nv').
\end{equation}
Using that $S$ is $2\pi i$-periodic, we find
\begin{equation}
S_{\nv}(\zetav)S^{\rho}(\zetav)F_{k}(\zetav +2i\pi \nv) =S_{\nv'}(\zetav)S^{\rho}(\zetav)F_{k}(\zetav+2i\pi \nv'),
\end{equation}
and therefore:
\begin{equation}
S_{\nv}(\zetav)F_{k}(\zetav +2i\pi \nv) =S_{\nv'}(\zetav)F_{k}(\zetav+2i\pi \nv').
\end{equation}
Hence we can assume that $n_1 \leq \ldots \leq n_k$.

Now, with $\lambdav:=\im\zetav$, the conditions $\lambdav\in \ical_3^k$ and $\lambdav+2 \pi \nv \in \ical_3^k$ imply (cf.~\eqref{eq:i3k})
\begin{equation}
\forall j,k: \quad |\lambda_{j}-\lambda_{k}| < 2\pi, \quad |\lambda_{j}-\lambda_{k}+2\pi (n_{j}-n_{k})| < 2\pi.
\end{equation}
We can show that $n_j\in\{N,N+1\}$ for all $j$ with some fixed $N \in \zbb$: We have $2\pi > |(\lambda_{i}-\lambda_{j})+2\pi (n_{i}-n_{j})|\geq | |\lambda_{i}-\lambda_{j}|-2\pi |n_{i}-n_{j}|| \geq 2\pi |n_{i}-n_{j}|- |\lambda_{i}-\lambda_{j}|\geq 2\pi |n_{i}-n_{j}|-2\pi$. This implies $2 > |n_{i}-n_{j}|$, and therefore we have for all $i=1,\ldots, k$ : $|n_{i}-n_{j}|=1$. We choose $N=\min_j n_j$.

In the following we consider only the case $N=0$; indeed, we can handle the other case $N=-1$ with similar arguments, and to prove the case for all other $N$ we use the $2i\piv$-periodicity of $F_k$.

We want to show the identity \eqref{periodrel} (where $\nv=(0,\ldots,0,1,\ldots,1)$ with $m$ entries of $0$, and $\nv'=0$) between meromorphic functions, hence it suffices to check it on a real open set, possibly on the boundary of the domain. Therefore, we can choose $\im \zetav=0$ and $\im\zetav+2\pi \nv\in \bar\ical_3^k$.
Inserting $F_k'$ as the boundary value of $F_k$, it remains to show that for real $\thetav$, in the sense of distributions, we have:
\begin{equation}\label{neq01form}
F_{k}'(\thetav)=\Big(\prod_{\ell>m}\prod_{j \leq m}S(\theta_{\ell}-\theta_{j})\Big)
F_{k}'(\theta_1,\ldots,\theta_m,\theta_{m+1} + 2 \pi i, \ldots, \theta_k+2\pi i).
\end{equation}
This uses
\begin{equation}
\prod_{\ell=1}^{k}\Big( \prod_{\substack{i=1\\ i \neq \ell}}^{k} S(\zeta_\ell-\zeta_i)
   \Big)^{ n_\ell }=\prod_{\ell\;:\; n_{\ell}=1}\prod_{i\neq \ell}S(\theta_{\ell}-\theta_{i})=\prod_{\ell \;:\; n_{\ell}=1}\Big(\prod_{i\; :\; n_{i}=0}S(\theta_{\ell}-\theta_{i})\Big).
\end{equation}

On the right hand side of \eqref{neq01form}, $F_k'$ is evaluated on a point of $\tube(\gcal^{k}_{1,m})$ (this point is $\thetav +i(\underbrace{0,\ldots,0}_{m},\underbrace{2\pi,\ldots,2\pi}_{k-m})$). Using Eq.~\eqref{stairs}, we find
\begin{equation}
F_{k}'(\thetav)=\Big(\prod_{\ell>m}\prod_{j \leq m}S(\theta_{\ell}-\theta_{j})\Big)
F_{k}'(\theta_{m+1}+ 2 \pi i,\ldots,\theta_{k}+ 2 \pi i,\theta_{1}+ 2 \pi i,\ldots,\theta_{m}+ 2 \pi i).
\end{equation}
Then, using the $2 i \piv$-periodicity of $F_k'$ \eqref{defonestair}, we find
\begin{equation}
\rhs{neq01form} = \Big(\prod_{\ell>m}\prod_{j \leq m}S(\theta_{\ell}-\theta_{j})\Big) F_{k}'(\theta_{m+1},\ldots,\theta_{k},\theta_{1},\ldots,\theta_{m}).
\end{equation}
We consider a permutation $\sigma$ given by
\begin{equation}
\sigma=\begin{pmatrix} 1 & \ldots & m & m+1 & \ldots & k \\ m+1 & \ldots & k & 1 &\ldots & m \end{pmatrix}.
\end{equation}
We compute $S^{\sigma}$ from the definition \eqref{eq:Sperm}: $\sigma$ swaps $1\ldots m$ with $m+1\ldots k$, but leaves the order of indices unchanged otherwise. Thus the $S$ factors in the product are of the form $S_{\ell,j}$ where $\ell>m$ and $j\leq m$. Hence, we find
\begin{equation}\label{stairform}
\rhs{neq01form} = \Big(\prod_{\ell>m}\prod_{j \leq m}S(\theta_{\ell}-\theta_{j})\Big) F_{k}'(\theta_{m+1},\ldots,\theta_{k},\theta_{1},\ldots,\theta_{m}) = S^\sigma(\thetav) F_k'(\thetav^\sigma).
\end{equation}
Since $F_{k}'$ is S-symmetric by \ref{it:fpsymm}, this proves \eqref{neq01form}.

By the argument above we have just showed that $F_k$ is well-defined on $\cbb^k$. It is meromorphic on $\cbb^k$ due to the poles of the $S$-matrix and the possible poles of the recursion relations, and it is analytic on $\ich(\gcal_1^k)$; hence \ref{it:fmero} is fulfilled. We have already showed \ref{it:fsymm} in Prop.~\ref{proposition:extendpermuted} (there we actually proved $S$-symmetry of $F_k$ on a smaller domain than $\cbb^k$; but since $F_k$ is meromorphic, this property holds also on the larger domain). Regarding \ref{it:fperiod}, we notice that this is a special case of \eqref{fperiodica}, where the shift by $2i\pi$ involves all the complex arguments $\zetav$ of $F_k$.
\end{proof}

\section{Residua}

Now, we want to compute the residua of $F_{k}$ and prove formula \ref{it:frecursion}.
\begin{proposition}\label{proposition:residua}
  The first-order poles of $F_k$ at $\zeta_m-\zeta_n=i\pi$, $m>n$, have residua as given by \ref{it:frecursion}.
\end{proposition}

\begin{proof}
It suffices to prove \ref{it:frecursion} for $m=1$, $n=k$; indeed, the general case follows from that particular case by using $S$-symmetry (see below).
Since the residues are meromorphic functions on the pole hypersurfaces, it suffices to verify formula \ref{it:frecursion} on a real open set.
Therefore we compute the difference of the boundary values of $F_k$ at the points $\zetav_\pm = \thetav + i (0,\ldots, 0, \pi \pm 0)$, where we assume that $\theta_j \neq \theta_{j'}$ for $j \neq j'$ (except for $j=1,j'=k$).  We note that $\zetav_-\in\ich(\gcal_1^k)$ but $\zetav_+\in\ich(\gcal_{1,k-1}^k)$. Hence, using \eqref{stairs}, the $2\pi i$-periodicity of $F_k$, and the boundary values of $F_k$ as in \eqref{eq:ffpboundary}, we find
\begin{multline}
  F_{k}(\theta_{1}, \ldots, \theta_{k-1},\theta_{k}+i\pi -i0)- F_{k}(\theta_{1}, \ldots, \theta_{k-1},\theta_{k}+i\pi+i0)\\
=F'_{k}(\theta_{1} \ldots \theta_{k-1},\theta_{k}+i\pi)-F'_{k}(\theta_{k}-i\pi,\theta_{1}, \ldots, \theta_{k-1})\\
=\delta(\theta_k-\theta_1)\Big( 1- \prod_{p=1}^{k}S(\theta_p-\theta_k)\Big)F_{k-2}(\theta_{2}, \ldots, \theta_{k-1}),
\end{multline}
where in the second equality we made use of \ref{it:fprecursion} in the case $m=1$. So, we can read from the formula above the residue of the pole:
\begin{equation}
\res_{\zeta_{k}-\zeta_{1}=i\pi}F_{k}(\zetav)=\frac{1}{2 \pi i}\Big( 1-\prod_{p=1}^{k}S(\zeta_{p}-\zeta_{1})\Big)F_{k-2}(\hat{\zetav}).
\end{equation}
This is exactly \ref{it:frecursion} in the case $m=1,n=k$.
\end{proof}
Now, using $S$-symmetry on the residua, we compute the residua of $F_{k}$ for generic $m,n$, with $m<n$:
\begin{multline}
\res_{\zeta_{n}-\zeta_{m}=i\pi}F_{k}(\zetav)
=(\zeta_{n}-\zeta_{m}-i\pi)
F_{k}(\zetav)\big\vert_{\zeta_{n}-\zeta_{m}=i\pi}\\
=(\zeta_{n}-\zeta_{m}-i\pi)F_{k}(\zeta_{m},\hat{\zetav},\zeta_{n})\cdot \prod_{j=1}^{m-1}S(\zeta_{m}-\zeta_{j})\prod_{i=n+1}^{k}S(\zeta_{i}-\zeta_{n})\big \vert_{\zeta_{n}-\zeta_{m}=i\pi}\\
=[\res_{\zeta_{k}-\zeta_{1}=i\pi}F_{k}(\zeta_{m},\hat{\zetav},\zeta_{n})]\cdot \prod_{j=1}^{m-1}S(\zeta_{m}-\zeta_{j})\prod_{i=n+1}^{k}S(\zeta_{i}-\zeta_{n})\\
=\frac{1}{2i\pi}\Big( 1-\prod_{p=1}^{k}S(\zeta_{p}-\zeta_{1})\Big)\prod_{j=1}^{m-1}S(\zeta_{m}-\zeta_{j})\prod_{i=n+1}^{k}S(\zeta_{i}-\zeta_{n})\cdot F_{k-2}(\hat{\zetav})\\
=\frac{1}{2\pi i}\Big( 1-\prod_{p=1}^{k}S(\zeta_{p}-\zeta_{m})\Big)\prod_{j=1}^{m-1}S(\zeta_{m}-\zeta_{j})\prod_{i=n+1}^{k}S(\zeta_{i}-\zeta_{n})F_{k-2}(\hat{\zetav})\\
=-\prod_{q=1}^{k}S(\zeta_{q}-\zeta_{m})\Big(  1-\prod_{p=1}^{k}S(\zeta_{m}-\zeta_{p})\Big)\prod_{j=1}^{m-1}S(\zeta_{m}-\zeta_{j})\prod_{i=n+1}^{k}S(\zeta_{i}-\zeta_{n})\\
=-\prod_{q=m}^{n}S(\zeta_{q}-\zeta_{m})\Big(  1-\prod_{p=1}^{k}S(\zeta_{m}-\zeta_{p})\Big),
\end{multline}
where in the third equality we made use of \ref{it:frecursion} for $l=1,r=k$. So, we find, $l<r$,
\begin{equation}
\res_{\zeta_{n}-\zeta_{m}=i\pi}F_{k}(\zetav)=-\frac{1}{2i\pi}\Big( \prod_{j=m}^{n}S_{j,m}\Big)\Big( 1-\prod_{p=1}^{k}S_{m,p}\Big)F_{k-2}(\hat{\zetav}).
\end{equation}

\section{Pointwise bounds} \label{sec:fkbounds}

Now we discuss the boundedness properties of $F_k$, namely \ref{it:fboundsreal} and \ref{it:fboundsimag}.

\begin{proposition}\label{proposition:fpointwise}
The functions $F_k$ fulfil the bounds \ref{it:fboundsreal} and \ref{it:fboundsimag}.
\end{proposition}

\begin{proof}

We notice that \ref{it:fboundsreal} is invariant under the transformation $\ell \rightarrow \ell \pm 2$ due to the property of periodicity \eqref{defonestair}. Hence, \ref{it:fboundsreal} follows directly from \ref{it:fpboundsreal}. To prove \ref{it:fboundsimag}, we consider the function
\begin{equation}
H_\pm(\zetav):= \exp\big(\pm i \mu r \sum_j \sinh \zeta_j -\sum_{j}\oa(\pm \sinh \zeta_{j})\big)\label{fhdef}
F_k( \zetav ).
\end{equation}
From condition \ref{it:fpboundsimag}, we know that
\begin{equation}
   \gnorm{H_\pm(\cdotarg + i \lambdav)}{\times} \leq c \quad \text{for $\lambdav$ on an edge of $\gcal_\pm^k$. }
\end{equation}
By the maximum modulus principle, Lemma~\ref{lemma:maxmodcross}, we have that the same bound holds for all $\lambdav\in\ich\gcal_\pm^k$. Then, applying Prop.~\ref{proposition:pointwise}, we find
\begin{equation}
   H_\pm(\thetav + i \lambdav) \leq c'  \operatorname{dist}(\lambdav,\partial \ich \gcal_\pm)^{-k/2}.\label{hboundinterior}
\end{equation}
By computing $\re\oa(\pm\sinh (\theta_j+i\lambda_j)) \leq a_\omega \omega(\cosh \theta_j) + b_\omega$, we find from \eqref{hboundinterior} and \eqref{fhdef} the bound \ref{it:fboundsimag} for $F_k$.
\end{proof}

The proof of Theorem \ref{theorem:FptoF} is a consequence of Propositions~\ref{proposition:extendeverywhere}, \ref{proposition:residua} and \ref{proposition:fpointwise}.

%% file: from_f_to_a.tex
\chapter{(F) \texorpdfstring{$\Rightarrow$}{=>} (A)}\label{sec:ftoa}

In this chapter we want to prove that given a family of meromorphic functions $F_k$ with the properties (F), then the Araki expansion defines a quadratic form which is $\omega$-local in a double cone. In other words, we want to prove the following theorem:
\begin{theorem}\label{theorem:FtoA}
If $(F_k)$ is a sequence of functions fulfilling (F), then
\begin{equation}\label{eq:afromf}
A := \sum_{m,n=0}^\infty \int  \frac{d^{m}\theta d^n \eta}{m!n!}F_{m+n}(\thetav+i\zerov,\etav+i\piv-i\zerov) z^{\dagger m}(\thetav) z^n(\etav)
\end{equation}
defines a quadratic form fulfilling (A).
\end{theorem}

\section{Well-definedness} \label{sec:definesq}

We can show that \eqref{eq:afromf} is well-defined. Indeed, set $g_{mn}(\thetav,\etav):= F_{m+n}(\thetav+i\zerov,\etav+i\piv-i\zerov)$. We have from \ref{it:fboundsreal} that $\onorm{g_{mn}}{m\times n}<\infty$. Hence, by applying Prop.~\ref{proposition:expansionunique}, we have that the series in \eqref{eq:afromf} is a well-defined quadratic form $A\in\qf^\omega$.

\section{Commutator for creators-annihilators} \label{sec:zdzcommute}	

In order to show that $A$ is $\omega$-local in a double cone, we need to compute the commutators of $A$ with the wedge-local fields $\phi(x),\phi'(x)$, and to show that they vanish if $x$ is in certain regions in Minkowski space.

We compute these commutators from the Araki expansion and we express them in terms of the Araki coefficients.

To that end, first we compute the commutator $\lbrack z^{\dagger m}(\thetav) z^n(\etav), z^{\# \prime} (\betav) \rbrack$ in operator form, where $z^{\# \prime} = z^{\prime} ,z^{\dagger \prime} $; an expression for this commutator generalizes the commutation relations \eqref{comzpz} and involves the multiplication operator $B^{g,\theta}$, which is defined in Eq.~\eqref{multop}.

It is useful the following lemma:
\begin{lemma}
Let $g\in\mathcal{H}_{1}$. The following exchange relations hold on $\fpn$ (in the sense of operator-valued distributions):
\begin{equation}
B^{g,\theta'}\zd(\theta)=S(\theta'-\theta)\zd(\theta)B^{g,\theta'}.\label{combz}
\end{equation}
\end{lemma}
\begin{proof}
\eqref{combz} can be computed directly from the definitions. We apply $B^{g,\theta'}\zd(\theta)$ to an arbitrary $n$-particle vector $\Psi_{n} \in\fpn$. We have
\begin{equation}
B^{g,\theta'}z^{\dagger}(\theta)\Psi_{n}=g(\theta')\prod_{j=1}^{n}S(\theta'-\theta_{j})S(\theta'-\theta)z^{\dagger}(\theta)\Psi_{n}=S(\theta'-\theta)z^{\dagger}(\theta)B^{g,\theta'}\Psi_{n}.
\end{equation}
\end{proof}

Now, we can prove the following lemma:
\begin{lemma}
Let $g\in\mathcal{H}_{1}$. The following commutation relations hold in the sense of operator-valued distributions on $\fpn$:
\begin{align}
[z(\overline{g})',\zd(\theta_{1})\ldots \zd(\theta_{m})]&=\sum_{j=1}^{m} \Big(\prod_{l=j+1}^{m}S(\theta_{j}-\theta_{l})\Big)\zd(\theta_{1})\ldots\widehat{\zd(\theta_{j})}\ldots \zd(\theta_{m})B^{g,\theta_{j}},\label{comm2}
\\
[\zd(\overline{g})',z(\theta_{1})\ldots z(\theta_{m})]&=-\sum_{j=1}^{m}\Big(\prod_{l=1}^{j-1}S(\theta_{l}-\theta_{j})\Big)(B^{\bar g,\theta_{j}})^{*}z(\theta_{1})\ldots\widehat{z(\theta_{j})}\ldots z(\theta_{m}).\label{comm4}
\end{align}
\end{lemma}

\begin{proof}
Our proof of Eq.~\eqref{comm2} is based on induction on $m$.
For $m=1$, Eq.~\eqref{comm2} reduces to \eqref{comzpz}, and is proven as in \cite[Lemma 4.2.5]{Lechner:2006}.

So, assume that Eq.~\eqref{comm2} holds for $m-1$ in place of $m$. We have
\begin{multline}\label{eq:computecommutator}
[z(\overline{g})',\zd(\theta_{1})\ldots \zd(\theta_{m-1})\zd(\theta_{m})]\\
=[z(\overline{g})',\zd(\theta_{1})\ldots \zd(\theta_{m-1})]\zd(\theta_{m})
+\zd(\theta_{1})\ldots \zd(\theta_{m-1})[z(\overline{g})',\zd(\theta_{m})]\\
=\sum_{j=1}^{m-1}\Big(\prod_{l=j+1}^{m-1}S(\theta_{j}-\theta_{l})\Big) \zd(\theta_{1})\ldots \widehat{\zd(\theta_{j})}\ldots \zd(\theta_{m-1})B^{g,\theta_{j}}\zd(\theta_{m})\\
+\zd(\theta_{1})\ldots \zd(\theta_{m-1})B^{g,\theta_{m}},
\end{multline}
where in the second equality we made use of Eq.~\eqref{comzpz} and Eq.~\eqref{comm2} in the case $m-1$.

Now, using the exchange relation \eqref{combz}, we bring $B^{g,\theta_{j}}$ to the last position in the third line of \eqref{eq:computecommutator}, we find
\begin{multline}
[z(\overline{g})',\zd(\theta_{1})\ldots \zd(\theta_{m})]=\sum_{j=1}^{m-1}\prod_{l=j+1}^{m}S(\theta_{j}-\theta_{l})\zd(\theta_{1})\ldots\widehat{\zd(\theta_{j})}\ldots \zd(\theta_{m})B^{g,\theta_{j}}\\
+\zd(\theta_{1})\ldots \zd(\theta_{m-1})B^{g,\theta_{m}}.
\end{multline}
Hence, we have
\begin{equation}
[z(\overline{g})',\zd(\theta_{1})\ldots \zd(\theta_{m})]=\sum_{j=1}^{m}\prod_{l=j+1}^{m}S(\theta_{j}-\theta_{l})\zd(\theta_{1})\ldots\widehat{\zd(\theta_{j})}\ldots \zd(\theta_{m})B^{g,\theta_{j}},
\end{equation}
which is \eqref{comm2}.

Then one can obtain \eqref{comm4} from \eqref{comm2} by taking the adjoint of \eqref{comm2}:
\begin{multline}
[z^{\dagger}(\bar{g})',z(\theta_{1})\ldots z(\theta_{m})]=([z^{\dagger}(\theta_{m})\ldots z^{\dagger}(\theta_{1}),z(g)'])^{*}\\
=-\sum_{j=1}^{m}\prod_{l=1}^{j-1}S(\theta_{l}-\theta_{j})(B^{\bar g,\theta_{j}})^{*}z(\theta_{1})\ldots\widehat{z(\theta_{j})}\ldots z(\theta_{m}),
\end{multline}
where in the second equality we made use of Eq.~\eqref{comm2}.
\end{proof}

\section{Commutator for Araki expansion} \label{sec:Afieldcomm}

Now, given a generic $A\in\qf^\omega$, we compute the commutator $[A,\phi'(x)]$ in terms of its Araki expansion.
\begin{proposition}\label{proposition:aphicomm}
Let $A\in Q^{\omega}$, $g\in \mathcal{D}(\mathbb{R}^{2})$. It holds that
\begin{multline}\label{maincommutator}
[A,\phi'(g)]=\sum_{m,n\geq 0}\int \frac{d^{m}\thetav d^{n}\etav}{m!n!}\int d\xi\;
 \Big( \cme{m,n+1}{A} (\thetav, \xi , \etav) z^{\dagger m}(\thetav) (B^{\overline{g^{+}},\xi})^{*}  z^{n}(\etav) \\
- \cme{m+1,n}{A} (\thetav,\xi,\etav) z^{\dagger m}(\thetav) B^{g^{-},\xi}  z^{n}(\etav)\Big).
\end{multline}
\end{proposition}
\begin{proof}
We need to compute the commutator:
\begin{equation}
[A,\phi'(g)]=\Big\lbrack \sum_{m,n=0}^{\infty}\int \frac{d^{m}\thetav d^{n}\etav}{m!n!}\;\cme{m,n}{A}(\thetav,\etav) z^{\dagger m}(\thetav)z^{n}(\etav),z^{\dagger}(\overline{g^{+}})'+z(\overline{g^{-}})'\Big\rbrack.
\end{equation}
Using \eqref{zzp}, we find
\begin{multline}
[A,\phi'(g)]=\sum_{m,n=0}^{\infty}\int \frac{d^{m}\thetav d^{n}\etav}{m!n!}\;
\cme{m,n}{A}(\thetav,\etav) z^{\dagger m} (\thetav) [z^{n}(\etav),\zd(\overline{g^{+}})']\\
+\sum_{m,n=0}^{\infty}\int \frac{d^{m}\thetav d^{n}\etav}{m!n!}\;
\cme{m,n}{A}(\thetav,\etav) [z^{\dagger m}(\thetav),z(\overline{g^{-}})']z^{n}(\etav).
\end{multline}
Applying \eqref{comm2} and \eqref{comm4} to the formula above, we find
\begin{equation}
\begin{aligned}
\lbrack A,\phi'(g) \rbrack &=
\sum_{m\geq 0, n\geq 1}\int \frac{d^{m}\thetav d^{n}\etav}{m!n!}\;\cme{m,n}{A}(\thetav,\etav) \sum_{j=1}^{n}\ \Big(\prod_{l=1}^{j-1}S(\eta_{l}-\eta_{j})\Big) \times \\
&\qquad\qquad \times z^{\dagger m}(\thetav)(B^{\overline{g^{+}},\eta_{j}})^{*}z(\eta_{1})\ldots\widehat{z(\eta_{j})}\ldots z(\eta_{n})\\
&\quad -\sum_{m\geq 1, n\geq 0}\int \frac{d^{m}\thetav d^{n}\etav}{m!n!}\;\cme{m,n}{A}(\thetav,\etav)
\sum_{j=1}^{m}\Big(\prod_{l=j+1}^{m}S(\theta_{j}-\theta_{l})\Big) \times \\
&\qquad\qquad \times z^{\dagger}(\theta_{1})\ldots\widehat{\zd(\theta_{j})}\ldots \zd(\theta_{m})B^{g^{-},\theta_{j}}z^{n}(\etav).
\end{aligned}
\end{equation}
We call $\eta_{j}=:\xi$ in the first sum and $\theta_{j}=:\xi$ in the second sum; we permute the argument of $\cme{m,n}{A}$ so that they become $(\hat\thetav,\xi,\etav)$ and $(\thetav,\xi,\hat\etav)$, respectively; we notice that this cancels the $S$-factors in the sums. (Here we use Prop.~\ref{proposition:fmnsymm}.) Hence, we find:
\begin{multline}
[A,\phi'(g)]=\sum_{m\geq 0, n\geq 1}\int \frac{d^{m}\thetav d^{n-1}\hat{\etav}}{m!(n-1)!}\int d\xi\;
 \cme{m,n}{A}(\thetav,\xi,\hat\etav)
z^{\dagger m}(\thetav)(B^{\overline{g^{+}},\xi})^{*}z^{n-1}(\hat{\etav})\\
-\sum_{m\geq 1, n\geq 0}\int \frac{d^{m-1}\hat{\thetav} d^{n}\etav}{(m-1)!n!}\int d\xi\;
 \cme{m,n}{A}(\hat{\thetav},\xi,\etav)z^{\dagger m-1}(\hat{\thetav})B^{g^{-},\xi}z^{n}(\etav).
\end{multline}
Now, we relabel the summation indices, we find
\begin{multline}
[A,\phi'(g)]=\sum_{m,n\geq 0}\int \frac{d^{m}\thetav d^{n}\etav}{m!n!}\int d\xi\;
 \Big( \cme{m,n+1}{A} (\thetav, \xi , \etav) z^{\dagger m}(\thetav) (B^{\overline{g^{+}},\xi})^{*}  z^{n}(\etav) \\
- \cme{m+1,n}{A} (\thetav,\xi,\etav) z^{\dagger m}(\thetav) B^{g^{-},\xi}  z^{n}(\etav)\Big),
\end{multline}
which is \eqref{maincommutator}.
\end{proof}

\section{Localization in a left wedge} \label{sec:awedgelocal}

Using that we have a family of functions $F_k$ fulfilling the conditions (F), we want to prove that the operator $A$, given by \eqref{eq:afromf}, is localized in a shifted left wedge. This means that we have to prove that the commutator $[A,\phi'(g)]$ vanishes if $g$ has support in the corresponding right wedge. To show this, we use Prop.~\ref{proposition:aphicomm}, and the idea for the proof is as follows. Due to properties \ref{it:fmero} and \ref{it:fsymm} (see Prop.~\ref{proposition:expansionunique}), we have $\cme{m,n}{A}(\thetav,\etav)=F_{m+n}(\thetav+i\zerov,\etav+i\piv-i\zerov)$; then, it follows that
$\cme{m+1,n}{A}(\thetav,\xi+i\pi,\etav)=\cme{m,n+1}{A}(\thetav,\xi,\etav)$. In matrix elements between vectors of finite particle number we can compute directly $B^{g^-,\xi+i \pi} = (B^{\overline{g^+},\xi})\st$, using also that $g$ has compact support. Inserting this into Eq.~\eqref{maincommutator}, we see that $[A,\phi'(g)]$ vanishes if it is possible to shift the integration contour in $\xi$ from $\rbb$ to $\rbb+i\pi$.

The fact that we can shift the integration contours depends on the growth behaviour of the analytic functions involved, namely $F_k$ and $g$, and therefore it depends on the localization regions of $A$ and $g$. So, first we study the growth behaviour of these functions.

Hence, we define for fixed $m,n\in\nbb_0$, $f \in \dcal(\rbb^{m+n})$, $q \in \nbb_0$, and $\nuv \in \rbb^q$, the function
\begin{equation}\label{eq:kxi}
   K(\xi) := e^{-i\mu r \sinh \xi} \Big(\prod_{j=1}^q S(\xi-\nu_j)\Big)\int d^m\theta d^n\eta \, f(\thetav,\etav)  F_{m+n+1}(\thetav+i\zerov,\xi,\etav+i\piv-i\zerov).
\end{equation}
Since $F_k$ fulfils the conditions (F), this $K$ is analytic for $\xi\in\strip(0,\pi)$, with boundary values which are distributions. We also define for fixed $g \in \dcal(\rbb^2)$ the function,
\begin{equation}\label{eq:hxi}
  h(\xi) := e^{i\mu r \sinh \xi} g^-(\xi) .
\end{equation}
Since the support of $g$ is compact, this function is entire analytic.

Now, we want to prove the following lemma:

\begin{lemma} \label{lemma:Kbounds}
If the $F_{k}$ fulfil \ref{it:fmero} and \ref{it:fboundsimag}, then there exist $c,c'>0$ such that
\begin{equation}
    \lvert K(\xi+i\lambda) \rvert
  \leq
   \frac{c \, e^{c' \omega(\cosh \xi)}}{(\lambda(\pi-\lambda))^{(m+n)/2}}.\label{stimaK}
\end{equation}
\end{lemma}

\begin{figure}
\begin{center}
\input{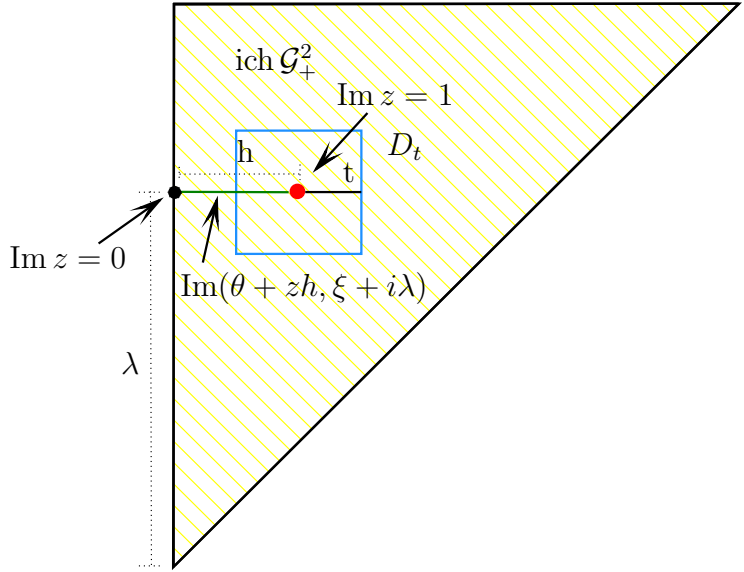}
\caption[The computation of the pointwise bound \eqref{stimaK}]{The computation of the pointwise bound \eqref{stimaK}. $D_t$ refers to the disc introduced in the proof of Prop.~\ref{proposition:pointwise}} \label{fig:pointwise}
\end{center}
\end{figure}

\begin{proof}
We introduce $h:=\min(\lambda,\pi-\lambda)/(m+n+1)$ and $\nuv_L := (1,2,\ldots,m)$, $\nuv_R := (n,\ldots,2,1)$. For fixed $\xi,\lambda,\thetav,\etav$, we define
\begin{equation}
   G(z):=e^{-i\mu r \sinh \xi} F_{m+n+1}(\thetav+zh\nuv_L,\xi+i\lambdav,\etav +i\piv -zh\nuv_R).\label{defGz}
\end{equation}
We can show that when $z\in \rbb +i(0,1)$, then $(\thetav+zh\nuv_L,\xi+i\lambdav,\etav +i\piv -zh\nuv_R) \in \ich \gcal_{+}^{k}$, cf.~Fig.~\ref{fig:pointwise}. That is, we have to show:
\begin{equation}
0< \im z h < \ldots < m h \im z < \lambda < \pi - n h \im z < \ldots < \pi-h\im z < \pi.
\end{equation}
Obviously, we have $\im z h >0$, $\pi-h \im z < \pi$, $\im z h < \ldots < m h \im z $ and $\pi - n h \im z < \ldots < \pi-h\im z$. It remains to show that $ m h \im z < \lambda < \pi - n h \im z$. By definition $h \leq \frac{\lambda}{m+n+1}$, hence $m h \im z \leq  m \frac{\lambda}{m+n+1}\cdot 1 = \frac{m}{m+n+1}\lambda < \lambda$. By definition we have also $h \leq \frac{\pi - \lambda}{m+n+1}$, hence $\pi -n h \im z \geq \pi - \frac{(\pi-\lambda)n}{m+n+1}\im z \geq \pi - (\pi - \lambda)\im z \geq \pi - (\pi - \lambda)=\lambda$.

Since the function $F_{m+n+1}$ is analytic in $\ich \gcal_{+}^{k}$, we have that the function $G$ is analytic in $z \in \rbb+i(0,1)$; for the imaginary part of the argument of $F_{m+n+1}$, we can show
\begin{equation}
 \operatorname{dist}\big( (h \im z, \ldots, mh\im z, \lambda, \pi-nh\im z, \ldots, \pi-h \im z)   ,\partial\ich\gcal^k_+\big) \geq \frac{1}{\sqrt{2}} h \im z.\label{distargFG}
\end{equation}
The proof of \eqref{distargFG} works as follows. Note that $\partial\ich\gcal^k_+ = \{  \lambda_i = \lambda_{i+1} \text{ for some } i \}\vee \{ \lambda_1=0 \} \vee \{ \lambda_k = \pi \}$. We compute \eqref{distargFG} directly, noting that $\operatorname{dist}(\cdot)$ denotes the euclidean distance in $\rbb^{k}$. Obviously, we have:
\begin{multline}
\operatorname{dist}\big( (h \im z, \ldots, mh\im z, \lambda, \pi-nh\im z, \ldots, \pi-h \im z) ,\{\lambda_1=0  \} \big)\\
 = \sqrt{  (h \im z - 0)^2 + (2 h \im z - \lambda_2)^2 + \ldots }\geq  h \im z,
\end{multline}
and,
\begin{multline}
\operatorname{dist}\big( (h \im z, \ldots, mh\im z, \lambda, \pi-nh\im z, \ldots, \pi-h \im z) ,\{\lambda_k=\pi  \} \big) \\
= \sqrt{ \ldots + (\pi - 2 h \im z -\lambda_{k-1})^2 + (\pi - h \im z -\pi)^2 }\geq  h \im z.
\end{multline}
So, it remains to show
\begin{multline}
\operatorname{dist}\big( (h \im z, \ldots, mh\im z, \lambda, \pi-nh\im z, \ldots, \pi-h \im z)   , \{  \lambda_i = \lambda_{i+1} \text{ for some } i \} \big) \\
\geq \frac{1}{\sqrt{2}} h \im z.
\end{multline}
By direct computation, we find
\begin{multline}
\operatorname{dist}\big( (h \im z, \ldots, mh\im z, \lambda, \pi-nh\im z, \ldots, \pi-h \im z)   , \{  \lambda_i = \lambda_{i+1} \text{ for some } i \} \big)\\
=\sqrt{ \ldots + (h i \im z - \lambda_i )^2 + (h (i + 1) \im z - \lambda_i)^2 + \ldots }.
\end{multline}
We call $\mu := h i \im z - \lambda_i$, $h(i +1)\im z - \lambda_i = \mu + h \im z$, $a:= h \im z$. We compute the minima of the function $f(\mu):= \mu^2 + (\mu + a)^2$: from $\frac{d}{d\mu}(\mu^2 + (\mu + a)^2)=2\mu + 2(\mu + a)\stackrel{!}{=}0$, we find $\mu = -a/2$. We have $f(-a/2)=( -a/2)^2 + (  -a/2 +a )^2 = 2 ( a/2 )^2 = a^2/2$. Hence, we have $\mu^2 + (\mu + a)^2 \geq a^2 /2$; namely:
\begin{equation}
\sqrt{ \ldots + (h i \im z - \lambda_i )^2 + (h (i + 1) \im z - \lambda_i)^2 + \ldots } \geq \frac{1}{\sqrt{2}}h \im z.
\end{equation}
This concludes the proof of \eqref{distargFG}.

Then, since the argument of $F_{m+n+1}$ is a point in the interior of $\tube(\gcal_+^k)$, we can apply condition \ref{it:fboundsimag}, which gives, for any $R>0$, a constants $c_R$ (dependent on $R$, but not on $\xi,\lambda$) and $c'$ such that
\begin{equation}\label{boundGz}
   |G(z)| \leq c_R  e^{c' \omega(\cosh \xi)} (h|\im z|)^{-k/2}
\quad \text{for all }  z \in (-R,R)+i(0,1), \; \|\thetav\|<R, \; \|\etav\|<R.
\end{equation}
Following, for example, \cite[Prop.~4.2]{BostelmannFewster:2009}, we compute a bound for the boundary distribution $\Big\lvert
\int G(x+i0) g(x)\,dx\Big\rvert$ as follows: let $z=x+iy$, we have
\begin{eqnarray}
\Big\lvert\int G(x+i0) g(x)\,dx\Big\rvert &=& \Big\lvert\lim_{y\rightarrow 0}\int G(x+iy) g(x)\,dx\Big\rvert \nonumber\\
&=& \Big\lvert\lim_{y\rightarrow 0}\int G^{(-\ell)}(x+iy) g^{(\ell)}(x)\,dx\Big\rvert.\label{boundboundaryG}
\end{eqnarray}
We can show that
\begin{equation}
\lvert  G^{(-\ell)}(z)\rvert \leq c^{(\ell)}_R e^{c'\omega(\cosh\xi)}h^{-k/2}( |\im z|^{-k/2 + \ell -1/4} +1 ), \quad (\ell > k/2 +1/4).\label{Gellebound}
\end{equation}
We prove \eqref{Gellebound} using induction on $\ell$. For $\ell =0$ it follows directly from \eqref{boundGz}. Now assume that \eqref{Gellebound} is true for $\ell -1$ in place of $\ell$; we prove it for $\ell$:
\begin{multline}
\int_{i/2}^{z}\lvert G^{(-\ell +1)}(z')\, dz'\rvert \leq c_R^{\ell -1}e^{c'\omega(\cosh\xi)}h^{-k/2}\int_{i/2}^{z}( \lvert \im z'\rvert ^{-k/2 + \ell -1 -1/4}+1)\, dz'.\label{intGellinduction}
\end{multline}
We choose in the strip $z\in \rbb +i(0,1)$ the integration path: $\gamma:= \{ 0\leq \re z' \leq \re z ,\; \im z' = 1/2  \}\cup \{ \re z' = \re z,\; 1/2 \leq \im z' \leq \im z  \}$. We compute the integral above along this curve, we find:
\begin{multline}
\text{ r.h.s.\eqref{intGellinduction} }= c_R^{\ell -1}e^{c'\omega(\cosh\xi)}h^{-k/2}\int_{0}^{\re z} d\re z' \\
+ c_R^{\ell -1}e^{c'\omega(\cosh\xi)}h^{-k/2}\int_{1/2}^{\im z} d \im z'\; ( \lvert \im z'\rvert ^{-k/2 + \ell -1 -1/4}+1)\\
= c_R^{\ell -1}e^{c'\omega(\cosh\xi)}h^{-k/2}\Big( \re z + \frac{\lvert \im z'\rvert^{-k/2+\ell-1-1/4+1}}{-k/2 +\ell -1 -1/4 +1}\big\rvert_{1/2}^{\im z} + \im z'\big\rvert_{1/2}^{\im z}\Big)\\
\leq c^{(\ell)}_R e^{c'\omega(\cosh\xi)}h^{-k/2}( |\im z|^{-k/2 + \ell -1/4} +1 ),
\end{multline}
where we used that $\re z < R$.

Choosing $\ell > k/2 +1/4$ and inserting \eqref{Gellebound} in \eqref{boundboundaryG}, we find
\begin{equation}
\text{ r.h.s. }\eqref{boundboundaryG} \leq c'^{(\ell)}_R e^{c' \omega(\cosh \xi)} \gnorm{g^{(\ell)}}{1} h^{-k/2}.
\end{equation}
Hence, denoting $c_{g,R}:=c'^{(\ell)}_R \gnorm{g^{(\ell)}}{1}$, and inserting the definition of $h$ where we estimated $\min(\lambda,\pi-\lambda)\geq \frac{1}{3}(\pi - \lambda)\lambda$, we obtain the following estimate for the boundary distribution:
\begin{equation}
 \Big\lvert
\int G(x+i0) g(x)\,dx
\Big\rvert
\leq c_{g,R} h^{-k/2}
   e^{c' \omega(\cosh \xi)}
\leq \frac{ (3(m+n+1))^{k/2}c_{g,R} \, e^{c' \omega(\cosh \xi)}} {( \lambda(\pi-\lambda) )^{k/2}},\label{boundaryGbound}
\end{equation}
where the constant $c_{g,R}$ might depend on the test function $g\in\dcal(-R,R)$ and the cutoff $R$, but \emph{not} on $G$ (and hence not on $\xi,\lambda,\thetav,\etav$).

Now, we recall \eqref{eq:kxi} and \eqref{defGz}, and we compute
\begin{multline}
\lvert   K(\xi+i\lambda) \rvert \\
= \Big\lvert e^{-i\mu r \sinh (\xi +i\lambda)} \Big(\prod_{j=1}^q S(\xi-\nu_j +i\lambda)\Big)\int d^m\theta d^n\eta \, f(\thetav,\etav)  F_{m+n+1}(\thetav+i\zerov,\xi +i\lambda,\etav+i\piv-i\zerov) \Big\rvert\\
= \Big\lvert  \Big(\prod_{j=1}^q S(\xi-\nu_j +i\lambda)\Big)\int d^m\theta d^n\eta \, f(\thetav,\etav) e^{-i\mu r \sinh (\xi +i\lambda)} F_{m+n+1}(\thetav+i\zerov,\xi +i\lambda,\etav+i\piv-i\zerov)   \Big\rvert\\
=\big \lvert  \Big(\prod_{j=1}^q S(\xi-\nu_j +i\lambda)\Big)\Big\rvert \cdot \Big\lvert \int d^m\theta d^n\eta \, f(\thetav,\etav) e^{-i\mu r \sinh (\xi +i\lambda)} F_{m+n+1}(\thetav+i\zerov,\xi+i\lambda,\etav+i\piv-i\zerov) \Big\rvert\\
\leq \Big\lvert \int d^m\theta d^n\eta \, f(\thetav,\etav) e^{-i\mu r \sinh (\xi +i\lambda)} F_{m+n+1}(\thetav+i\zerov,\xi +i\lambda,\etav+i\piv-i\zerov) \Big\rvert.
\end{multline}
where we have estimated the factors $S(\xi-\nu_j +i\lambda)$ by $1$, since they are bounded functions on the strip $\strip(0,\pi)$.

In order to apply \eqref{boundaryGbound}, we perform in the last line of the equation above a coordinate transformation $t:\cbb \times \rbb^{m+n-1} \rightarrow \cbb^m$ and $u:\cbb \times \rbb^{m+n-1} \rightarrow \cbb^n$, such that we can rewrite that integral as
\begin{multline}
\lvert   K(\xi+i\lambda) \rvert \leq \Big\lvert  \int d^m\theta d^n\eta \, f(\thetav,\etav) e^{-i\mu r \sinh (\xi +i\lambda)} F_{m+n+1}(\thetav+i\zerov,\xi +i\lambda,\etav+i\piv-i\zerov)\Big\rvert\\
=\Big\lvert  \int dx d\rho \, f(t(x,\rho),u(x,\rho)) e^{-i\mu r \sinh (\xi +i\lambda)} F_{m+n+1}(t(x+i0,\rho),\xi +i\lambda,u(x+i0,\rho)+i\piv)\Big\rvert\\
= \Big\lvert \int dx d\rho\, \tilde{f}(x,\rho) e^{-i\mu r \sinh (\xi +i\lambda)}\tilde{F}(x+i0,\rho,\xi +i\lambda)\Big\rvert.
\end{multline}
where we denoted $ \tilde{f}(x,\rho):=f(t(x,\rho),u(x,\rho))$ and $\tilde{F}(x+i0,\rho,\xi +i\lambda):=F_{m+n+1}(t(x+i0,\rho),\xi +i\lambda,u(x+i0,\rho)+i\piv)$.

We have that $G(x+i0)=e^{-i\mu r \sinh (\xi +i\lambda)}\tilde{F}(x+i0,\rho,\xi +i\lambda)$. We can now apply \eqref{boundaryGbound}, we find
\begin{eqnarray}
\lvert   K(\xi +i\lambda) \rvert &\leq& \int d\rho\, \Big\lvert  \int dx\,  \tilde{f}(x,\rho)G(x+i0)  \Big\rvert \nonumber\\
&\leq&  \int d\rho \frac{ (3(m+n+1))^{k/2}c_{\tilde{f},\rho,R} \, e^{c' \omega(\cosh \xi)}} {( \lambda(\pi-\lambda) )^{(m+n)/2}}\nonumber\\
&\leq & \frac{ (3(m+n+1))^{k/2} \, e^{c' \omega(\cosh \xi)}} {( \lambda(\pi-\lambda) )^{(m+n)/2}} \int_{\operatorname{supp}\tilde{f}} d\rho\, c_{\tilde{f},\rho,R},
\end{eqnarray}
where we used that $f$ has compact support in a ball of radius $R$, and where the last integral is finite; we call $c:=(3(m+n+1))^{k/2}\int_{\operatorname{supp}\tilde{f}} d\rho\, c_{\tilde{f},\rho,R}$. This gives the result in Lemma~\ref{lemma:Kbounds}.
\end{proof}

\begin{lemma} \label{lemma:integralshift}
Let $m,n\in\nbb_0$. If $F_{m+n}$ fulfils \ref{it:fmero} and \ref{it:fboundsimag}, then there exists an analytic indicatrix $\omega'\geq \omega$ such that for all $f \in \dcal(\rbb^{m+n})$ and $g \in \dcal^{\omega'}(\wcal_r)$,
\begin{equation}
  \int K(\xi+i0) h(\xi) d\xi = \int K(\xi+i\pi-i0) h(\xi+i\pi) d\xi,\label{shiftcontoureq}
\end{equation}
where $f,g$ enter the definitions of $K,h$ given by Eq.~\eqref{eq:kxi} and Eq.~\eqref{eq:hxi}.
\end{lemma}
\begin{proof}
We set $\omega'(p):=a_\omega (c'\omega(p)+ \frac{m+n+6}{2} \log(1+p))$, with $c'$ as in Lemma~\ref{lemma:Kbounds}, and $a_{\omega'}=a_\omega \geq 1$. This is a valid indicatrix, see Example 1 in Sec.~\ref{sec:jaffee}.
We will show below that
\begin{equation}
\int K(\xi+i0)h(\xi)\;d\xi=\lim_{\epsilon \searrow 0}\int K(\xi+i\epsilon) h(\xi+i\epsilon)\;d\xi,\label{contour}
\end{equation}
and similarly for the upper boundary. We can show that for fixed $\epsilon$ and with some $c_\epsilon>0$,
\begin{equation}
\forall \lambda \in [\epsilon, \pi-\epsilon]:|h(\xi+i\lambda)K(\xi+i\lambda)|\leq \frac{c_\epsilon}{(1+\cosh\xi)^{(m+n+4)/2}}.
\end{equation}
Indeed, on the strip $\lambda \in [\epsilon, \pi-\epsilon]$ we can estimate the denominator in \eqref{stimaK} by $(\lambda(\pi-\lambda))^{(m+n)/2}\geq \epsilon^{m+n}$, and we can set $c_\epsilon := 1/\epsilon^{m+n}$; moreover, using Lemma~\ref{lemma:Kbounds} and Prop.~\ref{prop:omegapw}, we have
\begin{multline}
|h(\xi+i\lambda)K(\xi+i\lambda)|\leq c_\epsilon e^{c'\omega(\cosh\xi)}e^{-\omega'(\cosh\xi)/a_{\omega'}}\\
\leq c_\epsilon \exp\Big( c'\omega(\cosh\xi) - \frac{ a_{\omega}(c'\omega(\cosh\xi) + \frac{m+n+4}{2}\log (1+ \cosh\xi))}{a_{\omega'}} \Big)\\
\leq c_\epsilon \exp \Big( -\frac{m+n+4}{2}\log(1+ \cosh\xi) \Big)\\
= c_\epsilon (1+ \cosh\xi)^{-(m+n+4)/2},
\end{multline}
where we have set $a_{\omega}=a_{\omega'}$.

So, on the strip $\lambda \in [\epsilon, \pi-\epsilon]$, where we have shown above that the function $h\cdot K$ is analytic and decays fast in real direction, we can apply Cauchy's formula and we get,
\begin{equation}
\forall \epsilon>0: \int d\xi\; K(\xi+i\epsilon)h(\xi+i\epsilon)=\int d\xi\; K(\xi+i\pi-i\epsilon)h(\xi+i\pi-i\epsilon)\label{prop}
\end{equation}
Now, to conclude the result \eqref{shiftcontoureq} from this, it remains to show Eq.~\eqref{contour}. Namely, we need to show that
\begin{equation}
\lim_{\epsilon \searrow 0} \int d\xi\; K(\xi+i\epsilon) \big( h(\xi)-h(\xi+i\epsilon) \big)=0.
\end{equation}
We denote $K^{(-\ell)}$ the $\ell$-th antiderivative of $K$, with $\Big( \frac{\partial}{\partial \zeta}\Big)^j K^{(-\ell)}(i\frac{\pi}{2})=0$, $0\leq j<\ell$. We will show
\begin{equation}
\lvert K^{(-\ell)}(\xi +i \lambda)\rvert \leq c\, e^{c'\omega(\cosh\xi)}(1+ \lvert \xi \rvert)^{\ell}\left[ \lambda (\pi-\lambda)\right]^{(m+n)/2-\ell -1/4},\; \ell<(m+n)/2.\label{boundantideK}
\end{equation}
We prove \eqref{boundantideK} using induction on $\ell$. For $\ell =0$, it follows directly from the bound of Lemma~\ref{lemma:Kbounds}. Now assume that \eqref{boundantideK} is true for $\ell -1$ in place of $\ell$, we prove it for $\ell$.

We integrate $K^{(-\ell+1)}$ along the lines from $i\pi/2$ to $\xi + i \pi/2$ and then from $\xi + i \pi/2$ to $\xi + i \lambda$. By induction hypothesis, we find
\begin{equation}
\begin{aligned}
\int_{i\pi/2}^{\xi + i\lambda} \lvert &K^{(-\ell +1)}(\xi' + i\lambda') d(\xi' +i\lambda')\rvert\\
&\leq \int_{i\pi/2}^{\xi +i\lambda}c\, e^{c'\omega(\cosh\xi')}(1+ \lvert \xi' \rvert)^{\ell-1}\left[ \lambda' (\pi-\lambda')\right]^{(m+n)/2-\ell +1-1/4}\\
&=c\,e^{c'\omega(\cosh\xi)}(1+ \lvert \xi \rvert)^{\ell-1}\int_{\pi/2}^{\lambda}d\lambda'\,\left[ \lambda' (\pi-\lambda')\right]^{(m+n)/2-\ell +1-1/4}\\
&\quad+c\,\left[ (\pi/2) (\pi-\pi/2)\right]^{(m+n)/2-\ell +1-1/4}\int_{0}^{\xi}d\xi'\,e^{c'\omega(\cosh\xi')}(1+ \lvert \xi' \rvert)^{\ell-1}\\
&=c\,e^{c'\omega(\cosh\xi)}(1+ \lvert \xi \rvert)^{\ell-1}\Big((\pi/2)^{(m+n)/2-\ell +1-1/4}\int_{\pi/2}^{\lambda>\pi/2}d\lambda'\, (\pi-\lambda')^{(m+n)/2-\ell +1-1/4}\\
&\quad+ (\pi/2)^{(m+n)/2-\ell +1-1/4}\int_{\pi/2}^{\lambda<\pi/2}d\lambda'\,{\lambda'}^{(m+n)/2-\ell +1-1/4}\Big)\\
&\quad+c\,\left[ (\pi/2) (\pi-\pi/2)\right]^{(m+n)/2-\ell +1-1/4}e^{c'\omega(\cosh\xi)}\int_{0}^{\xi}d\xi'\,(1+ \lvert \xi' \rvert)^{\ell-1}\\
&=c\,e^{c'\omega(\cosh\xi)}(1+ \lvert \xi \rvert)^{\ell-1}\pi/2)^{(m+n)/2-\ell +1-1/4}\times\\
&\quad \times \Big( (\pi-\lambda')^{(m+n)/2-\ell -1/4}\big\rvert_{\pi/2}^{\lambda}+ {\lambda'}^{(m+n)/2-\ell -1/4}\big\rvert_{\pi/2}^{\lambda}\Big)\\
&\quad+c\,\left[ (\pi/2) (\pi-\pi/2)\right]^{(m+n)/2-\ell +1-1/4}e^{c'\omega(\cosh\xi)}(1+ \lvert \xi \rvert)^{\ell}\\
&\leq c''\, e^{c'\omega(\cosh\xi)}(1+ \lvert \xi \rvert)^{\ell}\left[ \lambda (\pi-\lambda)\right]^{(m+n)/2-\ell -1/4},
\end{aligned}
\end{equation}
where in the third equality we used that $e^{c'\omega(\cosh\xi')}\leq e^{c'\omega(\cosh\xi)}$ and the monotonicity of $\omega$, i.e. \ref{it:omegamonoton}. This concludes the proof of \eqref{boundantideK}.

For $\ell>(m+n)/2$, we find by repeated integration, with some $c''>0$ and for all $\lambda$,
\begin{equation}
|K^{(-\ell)}(\xi+i\lambda)|\leq c'' (1+|\xi|)^\ell e^{c' \omega(\cosh \xi)}.\label{antiderivKbound}
\end{equation}
Using integration by parts and the bound \eqref{antiderivKbound}, we find
\begin{multline}
\lim_{\epsilon \searrow 0}\left|  \int d\xi\; K(\xi+i\epsilon) \big( h(\xi)-h(\xi+i\epsilon) \big) \right|\\
=\lim_{\epsilon\searrow 0}\left|  \int d\xi\; K^{(-\ell)}(\xi+i\epsilon)\left(  h^{(\ell)}(\xi)-h^{(\ell)}(\xi+i\epsilon)\right)\right|\\
\leq c'' \lim_{\epsilon\searrow 0} \epsilon \int d\xi\; (1+|\xi|)^\ell e^{c' \omega(\cosh \xi)} \sup_{0 < \lambda< \pi } |h^{(\ell+1)}(\xi+i\lambda)| = 0\label{integcontourlimit}
\end{multline}
if we can show that the integral in the last line is finite. To that end, using the bounds on $g^-$ from Prop.~\ref{prop:omegapw} and the definition of $\omega'$ after Eq.~\eqref{shiftcontoureq}, we have that for all $\lambda\in(0,\pi)$,
\begin{eqnarray}
 |h^{(\ell+1)}(\xi+i\lambda)| &\leq&
 c''' (\cosh \xi)^{\ell+1} e^{-\omega'(\cosh\xi)/a_{\omega'}}\nonumber\\
 &=& c'''(\cosh\xi)^{\ell +1}e^{-c'\omega(\cosh\xi)}e^{-\frac{m+n+6}{2}\log(1+\cosh\xi)}  \nonumber\\
 &\leq& c'''e^{-c'\omega(\cosh\xi)}(\cosh\xi)^{\ell +1} (\cosh\xi)^{-\frac{m+n}{2}-3}\nonumber\\
  &=&  c''' (\cosh \xi)^{\ell-(m+n)/2-2}  e^{-c'\omega(\cosh \xi)},
\end{eqnarray}
where we used $c':= 1/ a_{\omega'}$.

Inserting in \eqref{integcontourlimit}, we find
\begin{multline}
\int d\xi\; (1+|\xi|)^\ell e^{c' \omega(\cosh \xi)} \sup_{0 < \lambda< \pi } |h^{(\ell+1)}(\xi+i\lambda)|\\
\leq \int d\xi\; (1+|\xi|)^\ell e^{c'\omega(\cosh\xi)}e^{-c'\omega(\cosh\xi)}(1+\cosh\xi)^{\ell -(m+n)/2 -2}\\
=\int d\xi\; (1+|\xi|)^\ell (1+\cosh\xi)^{\ell -(m+n)/2 -2}.
\end{multline}
This integral is finite if we choose $ m+n < 2 \ell \leq m+n+2$ (where the left inequality is due to the condition $\ell > (m+n)/2$ before \eqref{antiderivKbound} and the right inequality is due to the requirement that $\ell -(m+n)/2 -2<0$).
\end{proof}

Using this we can prove wedge-locality of $A$, as discussed in the beginning of this section.

\begin{proposition}\label{proposition:wedgelocal}
Let $F_k$ be a sequence of functions fulfilling \ref{it:fmero}, \ref{it:fsymm}, \ref{it:fboundsreal} and \ref{it:fboundsimag}, and let $A$ be as in \eqref{eq:afromf}. Then $A$ is $\omega$-local in the wedge $\wcal_r'$.
\end{proposition}
\begin{proof}

We have that $A$, given by \eqref{eq:afromf}, is well-defined by properties \ref{it:fmero}, \ref{it:fboundsreal} and we have $\cme{m,n}{A}(\thetav,\etav)=F_{m+n}(\thetav+i\zerov,\etav+i\piv-i\zerov)$ by \ref{it:fsymm}. This is due to an application of Prop.~\ref{proposition:expansionunique}.

Now, by application of Lemma~\ref{lemma:localitychar}\ref{it:charomegavar}, it suffices to show that $\hscalar{\psi}{ [A,\phi'(g)]\chi}=0$ for fixed $\psi,\chi\in\fpno$, and for all $g \in \dcal^{\omega'}(\wcal_r)$, where the indicatrix $\omega'$ is chosen in a suitable way.

We can assume that $\psi,\chi$ have fixed particle number and compact support in rapidity space. This is possible because for any $\psi \in \hcal^{\omega,f}$, more specifically $\psi\in \hcal^{\omega}\cap \hcal_{n}$, there exists $\psi_{n}\in \hcal_{n}\cap \dcal(\rbb^{n})$, such that $\gnorm{\exp(\omega(H/\mu))(\psi-\psi_n)}{}\rightarrow 0$ for $n\rightarrow \infty$. Moreover, let $\psi_n \rightarrow \psi$, $\chi_n \rightarrow \chi$ for $n\rightarrow \infty$ in the sense of the above norm, then $\hscalar{\psi_n}{ [A,\phi'(g)]\chi_n}$ converges to $\hscalar{\psi}{ [A,\phi'(g)]\chi}$, since $[A,\phi'(g)]$ is a well-defined element of $\qf^\omega$.

We use Prop.~\ref{proposition:aphicomm}, and we consider a summand in Eq.~\eqref{maincommutator} with fixed $m,n$, since we can compare in this equation only terms with the same number of creators $\zd$ and annihilators $z$. It suffices to show that if $g\in \mathcal{D}^{\omega'}(\wcal_r)$, for fixed $q \in \nbb_0$, we have
\begin{multline}
\int d^m\theta d^n\eta \int d\xi \; f(\thetav,\etav) \Big(F_{m+n+1}(\thetav+i\zerov, \xi +i\pi-i0, \etav +i\piv-i\zerov)(B_{q}^{\overline{g^{+}},\xi})^{*}\\
-F_{m+n+1}(\thetav+i\zerov,\xi+i0,\etav +i\piv-i\zerov)B_{q}^{g^{-},\xi}\Big)=0.\label{vanishing}
\end{multline}
We use the definitions \eqref{multop}, \eqref{eq:kxi} and \eqref{eq:hxi}, and we rewrite \eqref{vanishing} as follows:
\begin{equation}
  \int K(\xi+i0) h(\xi) d\xi = \int K(\xi+i\pi-i0) h(\xi+i\pi) d\xi.
\end{equation}
This is given by Lemma~\ref{lemma:integralshift}; moreover, here we used the following fact: By Proposition \ref{prop:omegapw} the Fourier transform $g^{-}$ of a function $g\in \mathcal{D}^{\omega'}(\wcal_r)$ extends to an analytic function on the strip $\strip(0,\pi)$ with boundary value $g^{-}(\xi +i\pi)=g^{+}(\xi)$. Recalling that the scattering function $S$ is analytic in $\strip(0,\pi)$, we have
\begin{equation}
\begin{aligned}
B^{g^{-},\xi+i\pi}_{q}&=g^{-}(\xi+i\pi)\prod_{j=1}^{q}S(\xi-\theta_{j}+i\pi)\\
&= g^{+}(\xi)\prod_{j=1}^{q}S(\xi-\theta_{j})^{*} \\
&=\Big( \overline{g^{+}(\xi)}\prod_{j=1}^{q}S(\xi-\theta_{j})\Big)^{*} \\
&= (B^{\overline{g^{+}},\xi}_{q})^{*}.
\end{aligned}
\end{equation}
\end{proof}

\section{Generalized recursion relations} \label{sec:genrecursion}

Now, we start to pass from the wedge locality to the \emph{double cone} locality of $A$. So, as first step, we compute the residua of $F_k$ in several dimensions, using \ref{it:frecursion}. We have the following lemma:

\begin{lemma}\label{lemma:resgenrecursion}
There holds
\begin{equation}\label{resgenrec}
\res_{\eta_{r_{1}}-\theta_{\ell_{1}}=0}\ldots \res_{\eta_{r_{|C|}}-\theta_{\ell_{|C|}}=0}F_{m+n}(\thetav,\etav +i\piv)=\\
\frac{(-1)^{|C|}}{(2i\pi)^{{|C|}}}S_{C} R_{C}(\thetav,\etav) F_{m+n-2{|C|}}(\hat{\thetav},\hat{\etav}+i\piv),
\end{equation}
where $C$ is the contraction $(m,n,\{(\ell_1,r_1+m),\ldots,(\ell_{|C|},r_{|C|}+m)\})$.
\end{lemma}
\begin{proof}
Our proof is based on induction on $\ell$.
We first note that \ref{it:frecursion} in our specific situation simplifies to
\begin{equation}\label{eq:oneres}
\res_{\eta_{r}-\theta_{\ell}=0}F_{m+n}(\thetav+i\zerov,\etav+i\piv-i\zerov)
=-\frac{1}{2\pi i }S_{C_{1}} R_{C_{1}} (\thetav,\etav)
F_{m+n-2}( \boldsymbol{\hat\theta}+i\zerov,\boldsymbol{\hat\eta} +i\piv-i\zerov),
\end{equation}
where $C_1=(m,n,\{(\ell,r+m)\})$.

Also, notice that in the case $\zetav=(\thetav,\etav +i\piv)$, $m=\ell, n=r$, the S-factors in  \ref{it:frecursion} simplify:
\begin{equation}
\Big( 1-\prod_{p=1}^{m+n}S_{\ell,p}(\zetav)\Big)\cdot \Big(  \prod_{q=\ell}^{r}S_{q,\ell}(\zetav)\Big)
= \Big( 1-\prod_{p=1}^{m+n}S_{\ell,p}^{(m)}(\thetav,\etav)\Big)\cdot \Big(  \prod_{q=\ell}^{r}S_{q,\ell}^{(m)}(\thetav,\etav)\Big)
=R_{C_{1}}S_{C_{1}}.
\end{equation}
This is just Eq.~\eqref{resgenrec} in the case $|C|=1$.

Now assume that Eq.~\eqref{resgenrec} holds for $|C|-1$ in place of $|C|$.
We consider $C=C'\dot\cup C_{1}$, with $C'=(m,n,\{ (\ell_{2},r_{2}+m),\ldots, (\ell_{|C|},r_{|C|}+m)  \})$ and $C_1 \in \ccal_{m-|C'|,n-|C'|}$, $|C_1|=1$. We have
\begin{equation}
\begin{aligned}
\res_{\eta_{r_{1}}-\theta_{\ell_{1}}=0}&\Big(\res_{\eta_{r_{2}}-\theta_{\ell_{2}}=0}\ldots\res_{\eta_{r_{|C|}}-\theta_{\ell_{|C|}}=0} F(\thetav ,\etav +i\piv) \Big)\\
&=\frac{(-1)^{|C'|}}{(2i\pi)^{|C'|}}R_{C'} S_{C'} (\thetav,\etav)
\res_{\eta_{r_{1}}-\theta_{\ell_{1}}=0} F_{m+n-2|C|+2}(\hat{\thetav},\hat{\etav}+i\piv)\\
&=\frac{(-1)^{|C|}}{(2i\pi)^{|C|}}R_{C'} S_{C'} (\thetav,\etav)
 R_{C_{1}} S_{C_{1}} (\hat\thetav,\hat\etav) F_{m+n-2|C|}( \Hat{\Hat{\thetav}},\Hat{\Hat{\etav}} +i\piv)\\
&=\frac{(-1)^{|C|}}{(2i\pi)^{|C|}}R_{C} S_{C} (\thetav,\etav)F_{m+n-2|C|}( \Hat{\Hat{\thetav}},\Hat{\Hat{\etav}} +i\piv),
\end{aligned}
\end{equation}
where in the first equality we used Eq.~\eqref{resgenrec} in the case $|C|-1$, where in the second equality we made use of Eq.~\eqref{eq:oneres} and in the third equality we used Lemma~\ref{lemma:contractcompose}. To obtain $R_{C'}R_{C_{1}}=R_C$ in the third equality, we used \eqref{rc} and the fact that $\delta_{C'}R_{C_{1}}(\hat{\thetav},\hat{\etav})= \delta_{C'}R_{C_{1}}(\thetav,\etav)$.
\end{proof}

\section{Coefficients of the reflected operator} \label{sec:fshifted}

Now, we prove the following proposition:
\begin{proposition}\label{proposition:fshifted}
Given a family of functions $F_{k}$ with properties (F), the quadratic form $A$ in \eqref{eq:afromf} fulfils
\begin{equation}\label{eq:Fshifted}
JA^{*}J=\sum_{m,n=0}^{\infty}\int \frac{d^{m}\thetav d^{n}\etav}{m!n!}\;F^{\pi}_{m+n}(\thetav+i\zerov,\etav+i\piv-i\zerov)z^{\dagger m}(\thetav)z^{n}(\etav),
\end{equation}
where $F_{k}^{\pi}=F_{k}(\cdot +i\piv)$.
\end{proposition}

\begin{proof}
By Prop.~\ref{proposition:expansionunique}, this is a well-defined element of $\qf^\omega$. So, we need only to show that the coefficients $\cme{m,n}{\cdotarg}$ of the right and of the left hand sides of \eqref{eq:Fshifted} are equal. We can rewrite the coefficients of the right hand side as follows:
\begin{multline}
F^{\pi}_{m+n}(\thetav+i\zerov,\etav+i\piv-i\zerov)
=F_{m+n}(\thetav+i\piv+i\zerov,\etav+2i\piv-i\zerov)\\
= \prod_{k=1}^{n}\Big(  \prod_{\ell=1}^{m}S(\eta_{k}-\theta_{\ell})\prod_{\substack{q=1 \\ q\neq k}}^{n}S(\eta_{q}-\eta_{k})\Big)F_{m+n}(\thetav+i\piv+i\zerov,\etav-i\zerov)\\
= \prod_{k=1}^{n}\Big(  \prod_{\ell=1}^{m}S(\eta_{k}-\theta_{\ell})\prod_{\substack{ q=1 \\ q\neq k}}^{n}S(\eta_{q}-\eta_{k})\Big)\Big(  \prod_{t=1}^{n}\prod_{u=1}^{m}S(\theta_{u}-\eta_{t})\Big)F_{m+n}(\etav-i\zerov,\thetav+i\piv+i\zerov),\label{Fpicomput}
\end{multline}
where in the first equality we made use of \ref{it:fperiod} and in the second equality we used \ref{it:fsymm}.

One can check that $\prod_{k=m+1}^{m+n}\prod_{\substack{ q=m+1 \\ q\neq k}}^{m+n}S(\eta_{q}-\eta_{k})=1$. Hence, we find:
\begin{equation}
F^{\pi}_{m+n}(\thetav+i\zerov,\etav+i\piv-i\zerov)=F_{m+n}(\etav-i\zerov,\thetav+i\piv+i\zerov).
\end{equation}
As for the left hand side of \eqref{eq:Fshifted}, we know that the coefficients of $JA\st J$ are given by Prop.~\ref{proposition:fmnreflected}; inserting $\cme{m,n}{A}$ as the boundary value of $F_k$, we need to show:
\begin{equation}
F_{m+n}(\etav-i\zerov,\thetav+i\piv+i\zerov)
=
\sum_{C\in\ccal_{m,n}}(-1)^{|C|}\delta_{C} S_C R_C (\thetav,\etav) F_{m+n}(\hat{\etav}+i\zerov,\hat\thetav+i\piv-i\zerov).\label{substcoeff}
\end{equation}
Now, we prove the following lemma:
\begin{lemma}
There holds the following equality:
\begin{equation}
F_{m+n}(\etav-i\zerov,\thetav+i\piv+i\zerov)
=
\sum_{C\in\ccal_{n,m}}\delta_{C} S_C R_C (\etav,\thetav) F_{m+n}(\hat{\etav}+i\zerov,\hat\thetav+i\piv-i\zerov).
\end{equation}
\end{lemma}
\begin{proof}
The proof of this lemma is an application of Prop.~\ref{proposition:multivarres} with the substitution $\zv = (\etav,\thetav)$,
with the indices $p$ in Prop.~\ref{proposition:multivarres} labelling the pairs $(\ell,r)$, $1 \leq \ell \leq n$, $n+1 \leq r \leq n+m$, with the contractions $C \in \ccal_{n,m}$ in place of $M \subset \{1,\ldots,p\}$, and with the following vectors in $\rbb^{n+m}$,
\begin{align}
\av_{(\ell,r)} &:= (0,\ldots,0,\underbrace{1}_{\ell},0,\ldots,0,\underbrace{-1}_{r},0,\ldots,0),\\
\bv_C &:= (1,\ldots,\underbrace{0}_{\ell_j},\ldots,\underbrace{1}_n,\underbrace{-1}_{n+1},\ldots,\underbrace{0}_{r_j},\ldots,-1),
  \quad \text{with } C=(n,m,\{(\ell_j,r_j)\}), \\
\cv &:= (\underbrace{-1,\ldots,-1}_{n},\underbrace{1,\ldots,1}_m).
\end{align}
We note that $\av_{(\ell,r)} \cdot \bv_C \geq 0$, and that this $=0$ when $(\ell,r)$ is contracted in $C$; we also note that $\av_{(\ell,r)} \cdot \cv < 0$; so we can apply Prop.~\ref{proposition:multivarres}.
We insert the residues of $F_{m+n}$ given by Lemma~\ref{lemma:resgenrecursion} into Prop.~\ref{proposition:multivarres}; however we note that the orientation of the hyperplanes $\zv\cdot\av_{\ell,r}=0$ is opposite to the pole hyperplanes in \eqref{resgenrec}, and this gives an additional factor $(-1)^{|C|}$. Hence, we obtain
\begin{equation}
F_{m+n}(\etav-i\zerov,\thetav+i\piv+i\zerov)
=
\sum_{C\in\ccal_{n,m}}\delta_{C} S_C R_C (\etav,\thetav) F_{m+n}(\hat{\etav}+i\zerov,\hat\thetav+i\piv-i\zerov).
\end{equation}
We note that there are sets of pairs $(\ell,r)$ that cannot form a valid contraction, but these cases correspond to residua that vanish. For example for $m=2$ and $n=2$, let us consider the set of pairs $\{ (1,3), (2,3) \}$. This set cannot form a contraction since the indices in the pairs are not pairwise different. On the other hand, in Prop.~\ref{proposition:multivarres} the residua on the hyperplanes $\zeta_3 -\zeta_1 =i\pi$ and $\zeta_3 -\zeta_2=i\pi$ is zero since the function $F$ does not depend on the variable $\zeta_3$ after computing the first residue. A second example is given by the set of pairs $\{ (1,2),(3,4)  \}$; this case also does not form a valid contraction since $\ell \leq 2$ and $r\geq 3$ are violated. On the other hand, the residua on the poles $\zeta_1 -\zeta_2 =0$ and $\zeta_3 -\zeta_4 =0$ is also zero since two right or two left variables do not differ by $i\pi$.
\end{proof}
Using Eq.~\eqref{eq:RSJ}, we find:
\begin{equation}
F_{m+n}(\etav-i\zerov,\thetav+i\piv+i\zerov)
=
\sum_{C\in\ccal_{n,m}}(-1)^{|C|}\delta_{C^{J}} S_{C^{J}} R_{C^{J}} (\thetav,\etav) F_{m+n}(\hat{\etav}+i\zerov,\hat\thetav+i\piv-i\zerov).
\end{equation}
Now, relabelling the summation index from $C$ to $C^J$, we finally find
\begin{equation}
F_{m+n}(\etav-i\zerov,\thetav+i\piv+i\zerov)
=
\sum_{C\in\ccal_{m,n}} (-1)^{|C|}\delta_{C} S_C R_C(\thetav,\etav) F_{m+n}(\hat{\etav}+i\zerov,\hat\thetav+i\piv-i\zerov).
\end{equation}
This proves \eqref{substcoeff} and therefore concludes the proof of Prop.~\ref{proposition:fshifted}.
\end{proof}

\section{Locality in a double cone}\label{sec:locdoublecone}

The last step that we need in order to conclude that $A$ fulfils the condition (A), namely that it is $\omega$-local in the double cone $\ocal_r$, is the following lemma:

\begin{lemma}\label{lemma:shiftpi}
If some functions $F_k$ fulfil the conditions (F), then $F_k(\cdotarg + i \piv)$ fulfil them as well.
\end{lemma}
\begin{proof}
Condition \ref{it:fmero} is clearly invariant under the translation $\zetav\rightarrow \zetav+i\piv$. The same is true for conditions \ref{it:fsymm} and \ref{it:fperiod} since the $S$-factors depend only on difference of rapidities. The poles of the recursion relations depend on differences of rapidities, too, so condition \ref{it:frecursion} is also invariant under $\zetav\rightarrow \zetav+i\piv$. In \ref{it:fboundsreal}, the shift of $i\piv$ implies $l \rightarrow l+1$. Since $l\in \mathbb{Z}$ is arbitrary, this means only a relabelling of the nodes $\pmb{\lambda}^{(k,j+kl)}$, $j\in \{0,\ldots,k \}$. So, condition \ref{it:fboundsreal} is invariant. Since $\sinh(\zetav + i \piv)=-\sinh \zetav$, the exponential factor in \ref{it:fboundsimag} is invariant under $\zetav\rightarrow \zetav+i\piv$. Moreover, the shift of $i\piv$ implies that the argument of $F_{k}$ is shifted from $\mathcal{G}_{-}$ to $\mathcal{G}_{+}$ and from $\mathcal{G}_{+}$ to $\mathcal{G}_{+} + \piv$, so by \ref{it:fperiod}, condition \ref{it:fboundsimag} is also invariant under the shift $\zetav\rightarrow \zetav+i\piv$.
\end{proof}
\begin{proof}[Proof of Theorem~\ref{theorem:FtoA}]

We already showed in Sec.~\ref{sec:definesq} that $A$ given by \eqref{eq:afromf} is a well-defined element of $\qf^\omega$. Moreover, we showed in Prop.~\ref{proposition:wedgelocal} that $A$ is $\omega$-local in $\wcal_r'$.

Now, from Prop.~\ref{proposition:fshifted}, we know that $J A^\ast J$ is given in the same form as $A$, just with the coefficient functions $F_k^\pi$ instead of $F_k$. We showed in Lemma \ref{lemma:shiftpi} that in the case where $F_k$ fulfils the conditions (F), then $(F_k^\pi)$ also fulfil the same conditions (F).
Then, by applying Prop.~\ref{proposition:wedgelocal} to $J A^\ast J$, we find that $A$ is $\omega$-local in $\wcal_{-r}$. Thus we have that $A$ is $\omega$-local in $\wcal_r' \cap \wcal_{-r} = \ocal_r$, and hence it fulfils (A).
\end{proof}

%% file: localexamples.tex
\chapter{Examples of local operators}\label{sec:localexamples}

We will now discuss some concrete examples for the case $S=-1$, which fulfil the conditions (F) introduced in Chapter~\ref{sec:localitythm}. In one of these example we admit only a finite number of coefficient functions $F_k$ for even $k$, which -- because of the recursion relation -- is possible
only if $S = -1$; the other example contains a infinite family of coefficient functions for odd $k$.

For $\zetav \in \cbb^k$, we set
\begin{equation}
   E(\zetav) = \sum_{j=1}^k \cosh \zeta_j .
\end{equation}
(the dimensionless energy function).

\section{Buchholz-Summers type} \label{sec:evenexample}

We discuss an example of a local operator similar to the one given by Buchholz and Summers in \cite{BuchholzSummers:2007}.
\begin{proposition}
We put $F_{k}=0$ for $k\neq 2$ and we set
\begin{equation}
F_{2}(\zeta_{1},\zeta_{2})=\sinh \frac{\zeta_{1}-\zeta_{2}}{2}\tilde{g}(\mu E(\pmb{\zeta})),
\end{equation}
where $\tilde{g}$ denotes the Fourier transform of a function $g \in \mathcal{D}(-r,r)$ for some $r>0$.

The family of functions $F_k$ fulfil the properties (F) with respect to $\omega(p):= \ell \log (1+p)$ with $\ell$ sufficiently large.
\end{proposition}

We can show that in this case the recursion relations are trivial. Note also that Eq.~\eqref{ArakiF} becomes
\begin{multline}
A=\frac{1}{2}\int d\theta_{1}d\theta_{2} \sinh\Big( \frac{\theta_{1}-\theta_{2}}{2}\Big)\tilde{g}(\mu E(\pmb{\theta}))z^{+}(\theta_{1})z^{+}(\theta_{2})\\
-\frac{1}{2}\int d\theta_{1}d\theta_{2} \sinh\Big( \frac{\theta_{1}-\theta_{2}}{2}\Big)\tilde{g}(-\mu E(\pmb{\theta}))z(\theta_{1})z(\theta_{2})\\
+i\int d\theta_{1} d\theta_{2}\cosh \Big( \frac{\theta_{1}-\theta_{2}}{2}\Big)\tilde{g}(\mu \cosh \theta_{1}-\mu\cosh\theta_{2})z^{+}(\theta_{1})z(\theta_{2}).
\end{multline}
so, the Araki sum is finite.

\begin{proof}
We verify conditions (F):
\begin{enumerate}
\renewcommand{\theenumi}{F\arabic{enumi}}
\renewcommand{\labelenumi}{(\theenumi)}

\item \emph{Analyticity}: Since $g$ has compact support, its Fourier transform $\tilde{g}$ is entire analytic. $\sinh$, $E(\pmb{\zeta})$ are entire, too, and the composition of analytic functions is also analytic. So, the function $F_{2}$ is entire analytic.

\item \emph{$S$-symmetry}: Since $E(\zeta_{1},\zeta_{2})=E(\zeta_{2},\zeta_{1})$, we have
\begin{equation}
F_{2}(\zeta_{2},\zeta_{1})=\sinh \Big( \frac{\zeta_{2}-\zeta_{1}}{2}\Big)\tilde{g}(\mu E(\zeta_{2},\zeta_{1}))
=-\sinh\Big( \frac{\zeta_{1}-\zeta_{2}}{2}\Big)\tilde{g}(\mu E(\pmb{\zeta}))=-F_{2}(\zeta_{1},\zeta_{2}).
\end{equation}

\item \emph{$S$-periodicity}: Due to the property of $S$-symmetry, it suffices to show the property of periodicity with respect to the variable $\zeta_{1}$. Since $E(\zeta_{1}+2i\pi,\zeta_{2})=E(\zeta_{1},\zeta_{2})$, then:
\begin{multline}
F_{2}(\zeta_{1}+2i\pi,\zeta_{2})=\sinh \Big( \frac{\zeta_{1}-\zeta_{2}+2i\pi}{2}\Big)\tilde{g}(\mu E(\zeta_{1}+2i\pi,\zeta_{2}))\\
=-\sinh\Big( \frac{\zeta_{1}-\zeta_{2}}{2}\Big)\tilde{g}(\mu E(\zeta_{1},\zeta_{2}))=-F_{2}(\zeta_{1},\zeta_{2}).
\end{multline}

\item \emph{Recursion relations}: we can show that both sides of \ref{it:frecursion} are zero for all $k$. To see this, we consider the left hand side of \ref{it:frecursion}. Since the function $F_{k}$ is entire analytic, its residua are zero, so the left hand side of \ref{it:frecursion} vanishes. Now we consider the right hand side of \ref{it:frecursion}. In the case $S=-1$ and $k$ even, $\Big( 1- \prod_{p=1}^{k}S_{p,m}\Big)=0$. If $k$ is odd, all the functions $F_{k}$ are zero by hypothesis. So, in both cases the right hand side of \ref{it:frecursion} also vanishes.

\item \emph{Bounds on nodes}: We show that the $||\cdot||_{m\times n}^{\omega}$-norm (see Eq.~\eqref{eq:crossnorm}) of $F_{2}$ on the nodes of $\mathcal{G}_{0}$ is finite. For any $j\in \{0,1,2  \}$, the nodes of $\mathcal{G}_{0}$ are the points $\pmb{\lambda}^{(2,j)}$ and $\pmb{\lambda}^{(2,-j)}$. We consider only the nodes $\pmb{\lambda}^{(2,j)}$, since the proof is analogous for the other points.

\begin{enumerate}
\item We use Formula \eqref{eq:crossnormcomparison} and we compute on the node $\pmb{\lambda}^{(2,0)}$ the norm $||f_{20}||_{2}^{2}$:
\begin{equation}
\int d^{2}\pmb{\theta}\;|f_{20}(\theta_{1},\theta_{2})|^{2}\leq  \int d\theta_{1}d\theta_{2}\;\left|\; \sinh\Big( \frac{\theta_{1}-\theta_{2}}{2}\Big)\tilde{g}(\mu E(\pmb{\theta}))\;\right|^{2}.\label{node20}
\end{equation}
Since $g$ is a Schwartz function, then there is a constant $c$ such that $|\tilde{g}(p)|\leq \frac{c}{1+p^{2}}$. Hence
\begin{equation}
\int d^{2}\pmb{\theta}\;|f_{20}(\theta_{1},\theta_{2})|^{2}\leq \int d^{2}\pmb{\theta}\;\sinh^{2} \Big( \frac{\theta_{1}-\theta_{2}}{2}\Big)\Big( \frac{c}{1+(\cosh\theta_{1}+\cosh\theta_{2})^{2}\mu^{2}}\Big)^{2}.\label{node20int}
\end{equation}
Since $ \sinh^{2} \Big( \frac{\theta_{1}-\theta_{2}}{2}\Big)\leq \cosh^{2}\Big( \frac{\theta_{1}-\theta_{2}}{2}\Big)\leq (\cosh\theta_{1})^{2}+(\cosh\theta_{2})^{2}$, the integrand function behaves for large $\pmb{\theta}$ as
\begin{equation}
\sinh^{2} \Big( \frac{\theta_{1}-\theta_{2}}{2}\Big)\Big( \frac{c}{1+(\cosh\theta_{1}+\cosh\theta_{2})^{2}\mu^{2}}\Big)^{2}\approx
\frac{1}{(\cosh\theta_{1})^{2}+(\cosh\theta_{2})^{2}}
\end{equation}
and, therefore, the integral on the right hand side of \eqref{node20int} is finite.

By \eqref{eq:crossnormcomparison}, $||f_{20}||_{2}<\infty$ implies $||f_{20}||_{2\times 0}^{\omega}<\infty$.

\item By a similar computation in \eqref{node20int}, we have $||f_{02}||^{\omega}_{0\times 2}<\infty$ on the node $\pmb{\lambda}^{(2,2)}$.

\item We use Formula \eqref{eq:crossnormcomparison} with $\omega(p)=\ell \log(1+p)$ and we compute on the node $\pmb{\lambda}^{(2,1)}$ the norm $||f_{11}e^{-\omega(\theta_{2})}||^{2}_{2}$:
\begin{multline}\label{eq:f11int}
\int d\theta_{1}d\theta_{2}\;|f_{11}(\theta_{1},\theta_{2})(1+E(\theta_{2}))^{-\ell}|^{2}\\
=\int d\theta_{1}d\theta_{2}\;\left|\; \cosh\Big( \frac{\theta_{1}-\theta_{2}}{2}\Big)\tilde{g}(\mu \cosh \theta_{1}-\mu \cosh \theta_{2})(1+E(\theta_{2}))^{-\ell}\;\right|^{2}.
\end{multline}
Using the substitution $x_{j}=\sinh \frac{\theta_{j}}{2}$, $d\theta_{j}=\frac{2dx_{j}}{\sqrt{1+x_{j}^{2}}}$, $j=1,2$, we find
\begin{multline}
\lhs{eq:f11int}\\
=
\int \frac{2 dx_{1}}{\sqrt{1+x_{1}^{2}}}\frac{2 dx_{2}}{\sqrt{1+x_{2}^{2}}}\;
\left|\Big(\prod_{j=1}^{2}\sqrt{1+x_{j}^{2}}-x_{1}x_{2}\Big)\tilde{g}(2x_{1}^{2}-2x_{2}^{2})2(1+x_{2}^{2})^{-\ell}\right|^{2},
\end{multline}
where we used that $ \cosh( \frac{\theta_{1}-\theta_{2}}{2})=\cosh\frac{\theta_{1}}{2}\cosh\frac{\theta_{2}}{2}-\sinh\frac{\theta_{1}}{2}\sinh\frac{\theta_{2}}{2}$ and $\cosh\frac{\theta_{j}}{2}=\sqrt{1+\sinh^2 \frac{\theta_{j}}{2}}$, $j=1,2$.

Using $\frac{1}{\sqrt{1+x_{j}^{2}}}\leq 1,\; j=1,2$, and since $\prod_{j=1}^{2}\sqrt{1+x_{j}^{2}}-x_{1}x_{2}\leq (1+x_{1}^{2})(1+x_{2}^{2})+|x_{1}x_{2}|\leq 2(1+x_{1}^{2})(1+x_{2}^{2})$ (since $|x|\leq 1+x^{2}$ for $x\in\mathbb{R}$), we have:
\begin{equation}
\text{l.h.s (8.8)}
\leq 16\int dx_{1}dx_{2}\;(1+x_{1}^{2})^{2}(1+x_{2}^{2})^{2-2\ell}|\tilde{g}(2x_{1}^{2}-2x_{2}^{2})|^{2}.
\end{equation}
Using that $|\tilde{g}(p)|\leq c_{\ell}(1+p^{2})^{-\ell}$ with $\ell$ as above, we find
\begin{equation}
\text{l.h.s (8.8)}\leq 16 \int dx_{1}dx_{2}\; (1+x_{1}^{2})^{2}(1+x_{2}^{2})^{2-2\ell}\frac{c_{\ell}^{2}}{(1+4(x_{1}^{2}-x_{2}^{2})^{2})^{2\ell}}.\label{substexamp}
\end{equation}
Now we perform in \eqref{substexamp} the substitution $x_{1}=\rho\sin\varphi$, $x_{2}=\rho\cos\varphi$. As for the term $\Big( \frac{1}{1+x_{2}^{2}}\frac{1}{1+4(x_{1}^{2}-x_{2}^{2})^{2}}\Big)^{2\ell}$, we have:
\begin{eqnarray}
\frac{1}{1+x_{2}^{2}}\frac{1}{1+4(x_{1}^{2}-x_{2}^{2})^{2}}&=& \frac{1}{1+\rho^{2}\cos^{2}\varphi}\frac{1}{1+4\rho^{4}(\cos^{2}\varphi-\sin^{2}\varphi)^{2}}\nonumber\\
&=&\frac{1}{1+\rho^{2}\cos^{2}\varphi +4\rho^{4}(\cos^{2}\varphi-\sin^{2}\varphi)^{2}+\ldots}\nonumber\\
&\leq& \frac{1}{1+\rho^{2}(\cos^{2}\varphi +4\rho^{2}(\cos^{2}\varphi -\sin^{2}\varphi)^{2})}.
\end{eqnarray}
There holds:
\begin{equation}
\frac{1}{1+\rho^{2}(\cos^{2}\varphi+4\rho^{2}(\cos^{2}\varphi - \sin^{2}\varphi)^{2})}\leq
\begin{cases} \frac{1}{1+c \rho^{2}}, \quad &\rho\geq 1,
\\ 1, \quad &\rho\leq 1,
\end{cases}
\end{equation}
where in the upper inequality we made use of the fact that $\cos^{2}\varphi +4\rho^{2}(\cos^{2}\varphi - \sin^{2}\varphi )^{2}\geq \cos^{2}\varphi + 4(\cos^{2}\varphi - \sin^{2}\varphi )^{2}\geq c >0$.

Hence, we have from \eqref{substexamp}:
\begin{equation}
 \begin{aligned}
\lhs{eq:f11int}
 &\leq 16 \Big(\int_{0}^{1} \rho d\rho d\varphi\; (1+\rho^{2}\sin^{2}\varphi)^{2}(1+\rho^{2}\cos^{2}\varphi)^{2} \times \\
 &\qquad \times \frac{c_{\ell}^{2}}{(1+\rho^{2}(\cos^{2}\varphi +4\rho^{2}(\cos^{2}\varphi -\sin^{2}\varphi)^{2}))^{2\ell}}\\
 &\quad +\int_{1}^{\infty} \rho d\rho d\varphi\; (1+\rho^{2}\sin^{2}\varphi)^{2}(1+\rho^{2}\cos^{2}\varphi)^{2} \times \\
 &\qquad\times \frac{c_{\ell}^{2}}{(1+\rho^{2}(\cos^{2}\varphi +4r^{2}(\cos^{2}\varphi -\sin^{2}\varphi)^{2}))^{2\ell}}\Big)\\
&\leq 16  \Big( c_{\ell}^{2}\int_{0}^{1} \rho d\rho d\varphi\; (1+\rho^{2}\sin^{2}\varphi)^{2}(1+\rho^{2}\cos^{2}\varphi)^{2}\\
&\quad +\int_{1}^{\infty} \rho d\rho d\varphi\; (1+\rho^{2}\sin^{2}\varphi)^{2}(1+\rho^{2}\cos^{2}\varphi)^{2}\frac{c_{\ell}^{2}}{(1+c \rho^{2})^{2\ell}}\Big).
\end{aligned}
\end{equation}
The integral in the fourth line of the formula above is clearly finite. As for the integral in the last line, we notice that the integrand function behaves like $\frac{c}{\rho^{4\ell-9}}$ for $\rho>>1$, so this integral is finite if we choose $\ell$ sufficiently large.
\end{enumerate}

\item \emph{Pointwise bounds}: For $g\in \mathcal{D}(-r,r)$ we have
\begin{equation}
|\tilde{g}(\mu E(\pmb{\zeta}))|\leq \frac{||g||_{1}}{\sqrt{2\pi}}e^{\mu r |\im (\cosh\zeta_{1}+\cosh\zeta_{2})|},
\end{equation}
and $\sinh z$ for $z\in \mathbb{C}$ is bounded by
\begin{equation}
|\sinh z|=\left|  \frac{e^{z}-e^{-z}}{2i}\right|\leq \left| \frac{e^{z}}{2i}\right| + \left|\frac{e^{-z}}{2i} \right|\leq e^{|\operatorname{Re}z|}.
\end{equation}
Hence, we have
\begin{eqnarray}
|F_{2}(\pmb{\zeta})| &=& \left|  \sinh\Big( \frac{\zeta_{1}-\zeta_{2}}{2}\Big)\tilde{g}(\mu E(\pmb{\zeta}))\right|\nonumber\\
&\leq& e^{| \operatorname{Re}( \frac{\zeta_{1}-\zeta_{2}}{2} ) |}c e^{\mu r |\im (\cosh\zeta_{1} + \cosh\zeta_{2})|}\nonumber\\
&\leq& (\cosh \operatorname{Re}\zeta_{1})(\cosh\operatorname{Re}\zeta_{2}) c e^{\mu r (|\im \sinh \zeta_{1}| + |\im \sinh \zeta_{2}|)},\label{pointestim}
\end{eqnarray}
where in the last inequality we made use of the equalities $\im (\sinh z) = \cosh (\re z) \sin (\im z)$ and $\im (\cosh z) = \sinh (\re z) \sin( \im z)$ which imply the estimate $|\im (\cosh z)|\leq |\im (\sinh z)|$.

Since $e^{\omega(\cosh z)}= (1+\cosh z)^{3}$, $\ell=3$, we have from \eqref{pointestim}:
\begin{equation}
|F_{2}(\pmb{\zeta})|\leq c \prod_{j=1}^{2}\Big(e^{\omega(\cosh \operatorname{Re}\zeta_{j})}e^{\mu r |\im \sinh \zeta_{j}|}\Big).
\end{equation}
This is in agreement with condition \ref{it:fboundsimag} because $\operatorname{dist}(\im \pmb{\zeta}, \partial \ich \mathcal{G}_{+})$ is bounded from above.
\end{enumerate}
\end{proof}

\section{Schroer-Truong type}\label{sec:oddexample}

Now we discuss an example which is similar to the one given by Schroer and Truong in \cite{SchroerTruong:1978}, who give examples of local operators in the Ising model as formal series in terms of annihilators and creators.

The test function $g$ which appears in Prop.~\ref{pro:schroertruong} and Prop.~\ref{pro:kprototype} need to be of a certain class of test functions similar to Jaffe class, more precisely in the space of functions $\mathcal{S}_\omega$ defined in \cite[Definition~1.8.1]{Bjoerck:1965}. We do not state the technical definition here, but we will remark that this space of test function $\mathcal{S}_\omega$ is non-trivial since $\dcal^\omega$ is dense in $\mathcal{S}_\omega$ due to \cite[Proposition~1.8.6]{Bjoerck:1965} and \cite[Theorem~1.8.7]{Bjoerck:1965}, and since it was already shown in \cite[Lemma~1.3.9]{Bjoerck:1965} that $\dcal_\omega$ is non-trivial. Further, the functions in $\mathcal{S}_\omega$ fulfil specific bounds in momentum space, which are part of the definition \cite[Definition~1.8.1]{Bjoerck:1965}, and are given by
\begin{equation}
\sup_{p \in \rbb} e^{\lambda \omega(p)}| \partial^j \tilde{g}(p)| < \infty
\end{equation}
for each $j\in \nbb_0$, for each non-negative constant $\lambda$ and for an indicatrix $\omega$ of a certain class, see \cite[Definition~1.3.23]{Bjoerck:1965}. Note that the indicatrix $\omega(p)$ in Prop.~\ref{pro:schroertruong}, given by $\omega(p)=p^\alpha$, with $1/3<\alpha<1$, belongs to this class. We will make use of these bounds in Prop.~\ref{pro:kprototype}.

Most of the material in this section is due to H. Bostelmann. Our task is to prove the following proposition:

\begin{proposition}\label{pro:schroertruong}
Let $g \in \dcal(\rbb)$, $g \in \mathcal{S}_\omega$ with the indicatrix $\omega(p)$ given by $\omega(p)=p^\alpha$, with $1/3<\alpha<1$. We consider the set of meromorphic functions
\begin{equation}\label{eq:oddexamplePerm}
  F_{2k+1} (\zetav) =  \frac{1}{(4 \pi i)^{k} k!} \tilde g( \mu E(\zetav) ) \sum_{\sigma \in \perms{2k+1}} \sign \sigma \prod_{j=1}^k \tanh \frac{ \zeta_{\sigma(2j-1)} - \zeta_{\sigma(2j)}}{2}, \quad k \in \nbb_0,
\end{equation}
with $F_{2k} = 0$ for any $k$.

The family of functions $F_k$ fulfil the properties (F).
\end{proposition}
\begin{proof}
The proof of this proposition is given in the following Sec.~\ref{sec:elemproperties} and Sec.~\ref{sec:operatorbounddom}.
\end{proof}

\subsection{Elementary properties}\label{sec:elemproperties}

We consider the following building block of the functions $F_{2k+1}$:
\begin{equation}
   T_m(\zetav) :=  \frac{1}{2^{k} k !}\sum_{\sigma \in \perms{m}} \sign \sigma \prod_{j=1}^{k} \tanh \frac{ \zeta_{\sigma(2j-1)} - \zeta_{\sigma(2j)}}{2}, \quad \text{where } k = \lfloor m/2 \rfloor, \; m \in \nbb_0.
\end{equation}
It is useful to rewrite these in the following way. We call a \emph{pairing} of $m$ indices a set of pairs, $P = \{ (\ell_1,r_1),\ldots,(\ell_k,r_k)\}$, with $k = \lfloor m/2\rfloor$, where $\ell_j,r_j \in \{1,\ldots,m\}$ are pairwise different and $\ell_j < r_j$. This implies that $\{ \ell_j\} \cup \{r_j\}$ is the whole set $\{1,\ldots,m\}$ if $m$ even, except for one element in the case where $m$ is odd, namely $\{1,\ldots,m\}\setminus \{ \hat{m}\}$. We denote the set of all pairings of $m$ indices by $\pcal_m$. The \emph{signum} of a pairing $P$ is defined as
\begin{equation}\label{eq:signPdef}
   \sign P =  \left\{\begin{array}{l} \sign \begin{pmatrix} 1 & 2 & 3 & 4 & \cdots & 2k\!-\!1 & 2k  \\
         \ell_1 & r_1 & \ell_2 & r_2 & \cdots & \ell_k & r_k \end{pmatrix} , \quad \text{ if } m \text{ even } \\ \sign \begin{pmatrix} 1 & 2 & 3 & 4 & \cdots & 2k\!-\!1 & 2k &  m \\
         \ell_1 & r_1 & \ell_2 & r_2 & \cdots & \ell_k & r_k &\hat{m} \end{pmatrix} , \quad \text{ if } m \text{ odd } \end{array} \right.
\end{equation}
Note that this definition does not depend on the order of pairs $(\ell_j,r_j)$.

Using these definitions, we can express the function $T_m$ as
\begin{equation}\label{eq:TmWithPairs}
   T_m(\zetav) =  \sum_{P \in \pcal_m} \sign P \prod_{(\ell,r)\in P} \tanh \frac{ \zeta_{\ell} - \zeta_{r}}{2}.
\end{equation}
Note that here $P$ corresponds to $2^k k!$ permutations $\sigma$ in the sum \eqref{eq:oddexamplePerm}: these come from permuting the pairs among each other ($k!$ possibilities) and exchanging the two elements in each pair ($2^k$ possibilities).

For real arguments $T_m$ is evidently bounded. We conjecture a specific bound which plays some role in our argument. It will become useful in Sec.~\ref{sec:operatorbounddom}, when we will need to discuss condition \ref{it:fboundsreal}.
\begin{conjecture} \label{conj:Tmbounds}
 We have $|T_m(\thetav)| \leq 1$ for all $\thetav \in \rbb^m$.
\end{conjecture}

While we have not been able to prove this in full, there is strong evidence that the conjecture is true.

First, we have verified numerically up to $m=11$ by evaluating $T_m$ at a large number of randomly chosen points.

Second, we sketch an argument for a rigorous proof. With the substitution $x_j := \tanh \frac{\theta_j}{2}$, the function becomes
\begin{equation}
   T_m(\xv) =  \sum_{P \in \pcal_m} \sign P \prod_{(\ell,r)\in P} \frac{x_\ell -x_r}{1-x_\ell x_r}.
\end{equation}
This is now defined on the cube $\xv \in (-1,1)^m$. On the boundary of the cube, that is when one component of $\xv$ is equal to $\pm 1$, it is easy to see that
\begin{equation}
T_m(\xv)= \pm T_{m-1}(\hat \xv),
\end{equation}
where $\hat\xv:=(x_1, \ldots, \widehat{\pm 1},\ldots, x_m)$.

Therefore, if we can show that $T_m$ takes its extrema at the boundary, then the conjectured bound follows by induction. In fact, we verified for $m\leq 4$ that the following formula holds
\begin{equation}\label{logderiv}
\frac{d}{d\alpha}\log |T_m(\alpha \xv)|= \frac{1}{\alpha} \sum_{i<j}\frac{1+x_i x_j \alpha^2}{1-x_i x_j \alpha^2}.
\end{equation}
The right hand side is clearly positive for $\alpha \in (0,1)$ and $\xv \in (-1,1)^m$. Therefore, if Eq.~\eqref{logderiv} holds for all $m$, then since $T_m(0)=0$, the function does take its extrema at the boundary and therefore the conjecture would be proven.

We prove some further results about $T_m$.
\begin{lemma} \label{lem:Tmresidues}
 For any $1 \leq i < j \leq m$,
 \begin{equation}
  \res_{\zeta_j-\zeta_i=i\pi} T_m = 2 (-1)^{i+j}  T_{m-2} (\hat\zetav),\label{resTm}
 \end{equation}
where $\hat\zetav = (\zeta_1,\ldots,\hat \zeta_i,\ldots,\hat\zeta_j,\ldots\, \zeta_m)$.
\end{lemma}
\begin{proof}
Writing $T_m$ as in \eqref{eq:TmWithPairs}, it is clear that only pairings $P$ with $(i,j)\in P$ contribute to the residuum. Writing therefore $P = P' \dot\cup \{(i,j)\}$, we have
\begin{equation}
  \res_{\zeta_j-\zeta_i=i\pi} T_m
=   \sum_{\substack{P \in \pcal_m \\ (i,j)\in P}}   \sign P
\res_{\zeta_j-\zeta_i=i\pi} \tanh \frac{\zeta_i-\zeta_j}{2}
\prod_{(\ell,r) \in P'} \tanh \frac{ \zeta_{\ell} - \zeta_{r}}{2}.\label{resTmproof}
\end{equation}
One notes that $\res_{\zeta_j-\zeta_i=i\pi} \tanh \frac{\zeta_i-\zeta_j}{2} = -2$ and, from Eq.~\eqref{eq:signPdef}, we have
$\sign P = (-1)^{i+j+1} \sign P'$: Indeed, the passage
\begin{multline}
P=\begin{pmatrix} 1 & 2 & \cdots & i & \cdots & j & \cdots &  2k\!-\!1 & 2k  \\
         \ell_1 & r_1 & \cdots & \ell_i & \cdots & r_j & \cdots & \ell_k & r_k \end{pmatrix}
         \\
         \rightarrow P'=\begin{pmatrix} 1 & 2 & \cdots & i & i+1 & \cdots & 2k\!-\!1 & 2k  \\
         \ell_1 & r_1 & \cdots & \ell_i & r_j & \cdots & \ell_k & r_k \end{pmatrix}
\end{multline}
corresponds to permuting the $i$-column to the right (or, equivalently, the $j$-column to the left) and the sign of this permutation is $(-1)^{i-j+1}$.

Inserting this in \eqref{resTmproof}, we find immediately \eqref{resTm}.
\end{proof}

Using \eqref{eq:TmWithPairs}, we can rewrite the meromorphic functions $F_{2k+1}$ as
\begin{equation}\label{eq:FoddWithT}
 F_{2k+1} (\zetav) = \frac{1}{(2 \pi i)^{k}} \tilde g( \mu E(\zetav) ) T_{2k+1}(\zetav).
\end{equation}
Using \eqref{eq:TmWithPairs} and the result in Lemma~\ref{lem:Tmresidues}, now it is easy to check that the functions $F_k$ above fulfil the conditions \ref{it:fmero}, \ref{it:fsymm}, \ref{it:fperiod}, \ref{it:frecursion}, \ref{it:fboundsimag}.

Indeed, these functions $F_k$ are analytic if $\im \zeta_i < \im\zeta_j < \im \zeta_i + \pi$ ($i < j$) because the factors $\tanh$ are analytic in this domain (they have a pole at $\im\zeta_j -\im\zeta_i = \pi$); (note that it is enough to check this for $\im \zeta_i < \im\zeta_j$, $i < j$, since \ref{it:fmero} requires so). Hence they satisfy property \ref{it:fmero}.

From \eqref{eq:oddexamplePerm} we can see that they are antisymmetric \ref{it:fsymm}, due to the fact that $\tanh\Big(\frac{\zeta_r -\zeta_\ell}{2}\Big)=-\tanh \Big(\frac{\zeta_\ell -\zeta_r}{2}\Big)$.

Since $\tanh(z+i\pi)=\tanh(z)$, they are also $(2 \pi i)$-periodic in each variable \ref{it:fperiod}; we note that in this case the $S$-factor in \ref{it:fperiod} is equal to $1$.

Now we compute
\begin{multline}
  \res_{\zeta_j-\zeta_i=i\pi} F_{2k+1} = \frac{1}{(2 \pi i)^{k}} \tilde g( \mu E(\zetav) ) \Big\vert_{\zeta_j-\zeta_i=i\pi} \res_{\zeta_j-\zeta_i=i\pi} T_{2k+1} (\zetav)
\\
= \frac{2(-1)^{i+j} }{(2 \pi i)^{k}}  \tilde g( \mu E(\hat \zetav) )\, T_{2k-1}(\hat\zetav)
= \frac{2 (-1)^{i+j}}{2\pi i} F_{2k-1}(\hat\zetav),
\end{multline}
where in the second equality we used Lemma~\ref{lem:Tmresidues}. Note that on the pole $\zeta_j-\zeta_i=i\pi$, $\cosh \zeta_i = - \cosh \zeta_j$; this implies that $E(\zetav)$ on the pole becomes $E(\hat \zetav)$.

This gives the condition \ref{it:frecursion}, where we note that in the case $S=-1$, the factor $(1-\prod_p S_{m,p})$ becomes $2$ and the factor $\prod_{j=m}^{n}S_{j,m}$ becomes $(-1)^{m-n+1}$.

In the following proposition, we show the bound \ref{it:fboundsimag}:
\begin{proposition}
 If $r>0$, and $\supp g \subset (-r,r)$, then the $F_{2k+1}$ fulfil \ref{it:fboundsimag} with this constant $r$.
\end{proposition}

\begin{proof}
From \eqref{eq:FoddWithT} we want to compute bounds for
\begin{equation}
\lvert F_{2k+1} (\zetav) \rvert = \frac{1}{(2 \pi i)^{k}} \lvert \tilde g( \mu E(\zetav) ) \rvert \cdot \lvert T_{2k+1}(\zetav) \rvert
\end{equation}
for all $\zetav\in\tube(\ich \gcal^k_\pm)$.

As for the bound on $g$, we can show that the support properties of $g$ imply
\begin{equation}
   \lvert \tilde g (p + i q) \rvert \leq
       \frac{\lVert g \rVert_1}{\sqrt{2\pi}} e^{r|q|}.\label{gstima}
\end{equation}
 Indeed, we have
 \begin{eqnarray}
 \lvert \tilde g (p + i q) \rvert &=& \frac{1}{\sqrt{2\pi}}\Big \lvert \int dx\, e^{i(p+iq)\cdot x}g(x)\Big\rvert \nonumber\\
 &\leq& \frac{1}{\sqrt{2\pi}}\int dx\, \lvert    e^{i(p+iq)\cdot x}\rvert \cdot \lvert  g(x)\rvert \nonumber\\
 &\leq& \frac{1}{\sqrt{2\pi}}\sup_{x\in [-r,r]}\lvert  e^{i(p+iq)\cdot x}\rvert \int dx\, \lvert  g(x)\rvert \nonumber\\
 &=& \frac{\lVert g \rVert_1}{\sqrt{2\pi}} e^{r|q|}.
 \end{eqnarray}
Now if we take $\zetav$ of the form $\zetav^0 + i \lambda e^{(j)}$, where $\im \zetav^0$ is a node of the graph $\gcal_1$, then we have
\begin{equation}
\lvert\im E(\zetav)\rvert = \lvert\im \cosh (\zetav^0_j + i \lambda) \rvert
 \leq  \sinh \lvert \re\zeta_j \rvert \sin \lambda
 \leq  \cosh (\re \zeta_j) \sin \lambda.
\end{equation}
Hence, we have from \eqref{gstima}:
\begin{equation}\label{eq:gtildebound}
   \lvert \tilde g (\mu E(\zetav)) \rvert \leq
       \frac{\lVert g \rVert_1}{\sqrt{2\pi}} e^{\mu r \cosh (\re\zeta_j )\sin \lambda}.
\end{equation}
Now it remains to find bounds on $T_{2k+1}(\zeta)$. To find these bounds, we consider the function
\begin{equation}
 z \mapsto \frac{z+i\pi}{z-i} \tanh\frac{z}{2}.
\end{equation}
This function has poles at $z=i$, $z=\pm i\pi$; so it is clearly analytic on $-\pi\leq \im z\leq 0$. We can show that this function is bounded on the strip $-\pi \leq \im z \leq 0$. Indeed, for $\re z \rightarrow \infty$ one has $\lvert \tanh\frac{z}{2} \rvert \leq 1$ and $\lvert \frac{z+i\pi}{z-i\pi} \rvert\leq 1$; on every compact set on the real line and for $\im z \in [-\pi,0]$, the only crucial point is the behaviour of the function in the limit $z\rightarrow -i\pi$ (since the other pole at $z=i$ is excluded in the interval $\im z\in [-\pi,0]$); for $z\rightarrow -i\pi$ one has $\tanh \frac{z}{2} \approx \frac{2}{z+i\pi}+\ldots$, and hence $\frac{z+i\pi}{z-i} \tanh\frac{z}{2}\approx \frac{z+i\pi}{z-i}\frac{2}{z+i\pi}+\ldots$; so the function is continuous on the compact set, and therefore it is bounded.

Hence we can find a constant $c'_{1}>0$ such that
\begin{equation}
\lvert \frac{z+i\pi}{z-i} \tanh\frac{z}{2} \rvert \leq c'_{1}.
\end{equation}
This implies:
\begin{eqnarray}
  \big\lvert (z+i\pi) \tanh\frac{z}{2} \big\rvert &\leq& c'_{1} \lvert  z-i \rvert \nonumber\\
  &\leq& c'_{1} \sqrt{ (\re z)^2 + (\pi +1)^2} \nonumber\\
  &\leq& c''_{1}( \lvert \re z \rvert + \pi +1 )\nonumber \\
  &\leq& c_{1}\Big( \frac{\re z}{\pi +1} +1\Big)\nonumber\\
  &\leq& c_1( \lvert \re z \rvert + 1) \quad
  \text{whenever } -\pi  \leq \im z \leq 0,\label{ztanh}
\end{eqnarray}
since $\frac{1}{\pi +1}\leq 1$.

Also, since $\lvert \tanh\frac{z}{2} \rvert \leq 1$, we have from the equation above
\begin{equation}
  \big\lvert z+i\pi \big\rvert \leq c_1( \lvert \re z \rvert + 1) \quad
  \text{whenever } -\pi  \leq \im z \leq 0.\label{zpi}
\end{equation}
Now we can compute for any $\zetav$ with $\im \zeta_1  < \ldots< \im \zeta_{2k+1} < \im \zeta_1+\pi$ the following estimate,
\begin{multline}
 \Big\lvert T_{2k+1}(\zetav) \prod_{i<j} (\zeta_i-\zeta_j+i \pi) \Big\rvert
 \leq \Big\lvert  \sum_{P \in \pcal_{2k+1}}
  \prod_{(\ell,r)\in P} \tanh \frac{\zeta_\ell-\zeta_r}{2}   \prod_{i<j} (\zeta_i-\zeta_j+i \pi)\Big\rvert\\
 \leq
 \sum_{P \in \pcal_{2k+1}}
  \prod_{(\ell,r)\in P}
     \Big\lvert (\zeta_\ell-\zeta_r+i\pi) \tanh \frac{\zeta_\ell-\zeta_r}{2} \Big\rvert
  \prod_{(i,j)\not\in P}
     \Big\lvert \zeta_i-\zeta_j+i\pi \Big\rvert,\label{Tstima}
\end{multline}
where in the first equality we used \eqref{eq:TmWithPairs} and where in the second inequality we split the product $ \prod_{i<j} \lvert \zeta_i-\zeta_j+i \pi\rvert$ into $\prod_{(\ell,r)\in P}\lvert \zeta_\ell-\zeta_r+i\pi\rvert \prod_{(i,j)\not\in P}\lvert \zeta_i-\zeta_j+i\pi \rvert$.

We estimate the first product in the second line of \eqref{Tstima} using \eqref{ztanh} and the second product using \eqref{zpi}, we find
\begin{multline}
\Big\lvert T_{2k+1}(\zetav) \prod_{i<j} (\zeta_i-\zeta_j+i \pi) \Big\rvert\\
\leq
\sum_{P \in \pcal_{2k+1}}
  \prod_{(\ell,r)\in P}c_1(|\re \zeta_\ell - \re \zeta_r| + 1)
  \prod_{(i,j)\not\in P}c_1(|\re \zeta_i - \re \zeta_j| + 1)\\
\leq \sum_{P \in \pcal_{2k+1}}  c_1^{k(k-1)/2}\prod_{i<j} (|\re \zeta_i - \re \zeta_j| + 1)\\
\leq
  |\pcal_{2k+1}| c_1^{k(k-1)/2} \prod_{i<j} (|\re \zeta_i - \re \zeta_j| + 1),\label{Tstima2}
\end{multline}
where in the second inequality we used that the number of pairs $(i,j)$ with $i<j$ is $k(k-1)/2$, and that $\prod_{(\ell,r)\in P}(|\re \zeta_\ell - \re \zeta_r| + 1)
  \prod_{(i,j)\not\in P}(|\re \zeta_i - \re \zeta_j| + 1)= \prod_{i<j} (|\re \zeta_i - \re \zeta_j| + 1)$.
In the third inequality we denoted the sum of the pairings $P\in \pcal_{2k+1}$ by $|\pcal_{2k+1}|$.

At fixed $k$, the right-hand side is bounded by a polynomial in the $\re\zeta_j$; more precisely, by multiplying out the sum $(|\re \zeta_i | + |\re \zeta_j| + 1)$ in the first line and estimating each term of the sum by $\prod_{j=1}^{2k+1}(|\re\zeta_j|^{(2k+1)k}+1)$, we get
\begin{eqnarray}
\prod_{i<j} (|\re \zeta_i - \re \zeta_j| + 1) &\leq& \prod_{i<j} (|\re \zeta_i | + |\re \zeta_j| + 1)\nonumber\\
&\leq& 3^{(2k+1)k}\prod_{j=1}^{2k+1}(|\re\zeta_j|^{(2k+1)k}+1)\nonumber\\
&\leq& 3^{(2k+1)k}\prod_{j=1}^{2k+1}c_k e^{|\re\zeta_j|}\nonumber\\
&\leq& c'_k e^{\sum_{j=1}^{2k+1}|\re \zeta_j|}.
\end{eqnarray}
Hence we can find a constant $c_2>0$ such that we have from \eqref{Tstima2}
\begin{equation}\label{eq:Tkbound}
\lvert T_{2k+1}(\zetav) \rvert
 \leq c_2  \Big(\prod_{i<j} \lvert\zeta_i-\zeta_j+i \pi \rvert^{-1}\Big)
 \exp \Big(\sum_{j=1}^{2k+1} \lvert \re \zeta_j \rvert\Big) .
\end{equation}
Since with $\omega(p)\sim p^{\alpha}$, $0<\alpha<1$, we have $e^{|\operatorname{Re}\zeta_{j}|}\leq e^{\omega(\cosh(\operatorname{Re}\zeta_{j}))}$,
Equations~\eqref{eq:FoddWithT}, \eqref{eq:gtildebound} and \eqref{eq:Tkbound} together imply the estimate needed for \ref{it:fboundsimag}.
\end{proof}

We have shown that the family of functions $F_{2k+1}$ given by \eqref{eq:FoddWithT} fulfil the conditions \ref{it:fmero}, \ref{it:fsymm}, \ref{it:fperiod}, \ref{it:frecursion} and \ref{it:fboundsimag} of the theorem, but we have not yet shown that they fulfil \ref{it:fboundsreal}. This will be discussed in the next subsection.

\subsection{Operator bounds and domain}\label{sec:operatorbounddom}

In order to show that the family of functions $F_{2k+1}$ are the Araki coefficients of a local operator, we need to show that, setting $f_{mn}(\thetav,\etav) = F_{m+n}(\thetav,\etav+i\pi-i0)$, the norm $\lVert  f_{mn} \rVert_{m \times n}^{\omega}$ is finite, as required by \ref{it:fboundsreal}. Moreover, we compute estimates for $\lVert  f_{mn} \rVert_{m \times n}^{\omega}$ which imply that the Araki series
\begin{equation}
  A = \sum_{m,n} \int \frac{d^m\theta d^n \eta}{m!n!} f_{mn}(\thetav,\etav) z^{\dagger m}(\thetav) z^n(\etav),
\end{equation}
is summable when we apply it to vectors of a certain class; we use this to show that $A$ is a closable operator on a dense domain (see Prop.~\ref{proposition:summable1}). After lots of preparations, the main result of this section will be Theorem~\ref{theorem:oddoperator}.

In order to compute these estimates, we start by introducing some notation.
Given a smooth function of two real variables, that we call $K$, we define the following quantity:
\begin{equation} \label{eq:deltadef}
  \delta K(\theta,\eta) :=  \frac{K(\theta,\eta)-K(\eta,\eta)}{\theta-\eta}.
\end{equation}
We can extend this definition to the diagonal $\theta=\eta$ by taking the limit of \eqref{eq:deltadef} for $\theta \rightarrow \eta$. We will now show that this $\delta K$ is a smooth function (this will become useful later on, for example in Eq.~\eqref{meanhatK}). First of all, we know that the limit
\begin{equation}
\lim_{\theta\rightarrow \eta}\frac{K(\theta,\eta)-K(\eta,\eta)}{\theta-\eta}=\frac{\partial}{\partial \theta}K(\theta,\eta)\big\vert_{\theta=\eta}
\end{equation}
exist because by hypothesis $K$ is a smooth function. Then we need to show that the function defined for $\theta \neq \eta$ by $\frac{K(\theta,\eta)-K(\eta,\eta)}{\theta-\eta}$ and for $\theta=\eta$ by $\frac{\partial}{\partial \theta}K(\theta,\eta)\big\vert_{\theta=\eta}$ is smooth. We use the fact that a function is smooth if and only if it has a Taylor expansion of all orders. Since $K$ is smooth we can write its Taylor expansion around $\theta=\eta$:
\begin{equation}
K(\theta,\eta)-K(\eta,\eta)=(\theta-\eta)K'(\eta,\eta)+\frac{(\theta-\eta)^2}{2}K''(\eta,\eta)+\ldots+\frac{(\theta-\eta)^{k-1}}{(k-1)!}K^{(k-1)}(\eta,\eta)+R
\end{equation}
with $\lvert R \rvert \leq c \lvert  \theta-\eta \rvert^k$. This implies that the function $(K(\theta,\eta)-K(\eta,\eta))/(\theta-\eta)$ has a Taylor expansion around $\theta = \eta$ given by:
\begin{equation}
\frac{K(\theta,\eta)-K(\eta,\eta)}{\theta-\eta}=K'(\eta,\eta)+\frac{\theta-\eta}{2}K''(\eta,\eta)+\ldots+\frac{(\theta-\eta)^{k-2}}{(k-1)!}K^{(k-1)}(\eta,\eta)+R'
\end{equation}
with $R':= \frac{R}{\theta-\eta}$. We have $\lvert R'\rvert\leq c \lvert \theta-\eta \rvert^{k-1}$. So, this function has also a Taylor expansion, but of order $k-2$, and therefore it is $(k-2)$-times differentiable. Since this holds for all $k$, then $\delta K$ is smooth.

We can rewrite \eqref{eq:deltadef} as a Taylor expansion where the function $\delta K$ is used as a Taylor remainder term:
\begin{equation} \label{eq:deltaremain1}
  K(\theta,\eta) =  K(\eta,\eta) + (\theta-\eta) \delta K(\theta,\eta).
\end{equation}
In the case where $K$ depends on several variables, $K: \rbb^m \times \rbb^n \to \cbb$, then similarly to \eqref{eq:deltadef} we can define $\delta_j K(\thetav,\etav)$ as:
\begin{equation}\label{deltajKdef}
 \delta_j K(\thetav,\etav) :=  \frac{K(\thetav,\etav)-K(\theta_1, \ldots, \theta_{j-1},\eta_j,\theta_{j+1},\ldots,\theta_m)}{\theta_j-\eta_j}.
\end{equation}
We can also define for $k \leq \max(m,n)$ and $J = \{j_1,\ldots  j_\ell \} \subset \{1,\ldots, k\}$, the quantity $\delta_J K := \delta_{j_1} \ldots \delta_{j_l} K$; we notice that this does not depend on the order of the indices $j_1,\ldots,j_\ell$. We can then generalize Eq.~\eqref{eq:deltaremain1}, and we have
\begin{equation} \label{eq:deltaremainn}
\begin{aligned}
  K(\thetav,\etav) &=  \sum_{J \subset \{1,\ldots,k\}} \Big(\prod_{j \in J} (\theta_j-\eta_j) \Big) \delta_J K(\thetav^J,\etav),
\\
 &\text{where} \quad \theta^J_j = \begin{cases}
                              \theta_j &\quad  \text{if } j \in J \text{ or } j>k, \\
                              \eta_j &\quad  \text{otherwise, }
                            \end{cases}
\end{aligned}
\end{equation}
where $k$ is some integer with $k<m$.

We can prove \eqref{eq:deltaremainn} using induction on $k$. Using\\ $\sum_{J \subset \{1,\ldots,k\}}=\sum_{\substack{J \subset \{1,\ldots,k\} \\ 1\in J}}+\sum_{\substack{J \subset \{1,\ldots,k\} \\ 1\not\in J}}$, we can write
\begin{equation}
\begin{aligned}
K(\thetav,\etav) &=  \sum_{J \subset \{1,\ldots,k\}} \Big(\prod_{j \in J} (\theta_j-\eta_j) \Big) \delta_J K(\thetav^J,\etav)\\
&=\sum_{\substack{J \subset \{1,\ldots,k\} \\ 1\in J}}\Big(\prod_{j \in J} (\theta_j-\eta_j) \Big) \delta_J K(\thetav^J,\etav) 
+\sum_{\substack{J \subset \{1,\ldots,k\} \\ 1\not\in J}}\Big(\prod_{j \in J} (\theta_j-\eta_j) \Big) \delta_J K(\thetav^J,\etav)\\
&=\sum_{J \subset \{2,\ldots,k\}} \Big(\prod_{j \in J} (\theta_j-\eta_j) \Big)(\theta_1-\eta_1) \delta_{j_2}\ldots \delta_{j_k}\delta_1 K(\thetav^{J\cup \{ 1 \}},\etav)\\
&\qquad +\sum_{ J \subset \{2,\ldots,k\}}\Big(\prod_{j \in  J} (\theta_j-\eta_j) \Big) \delta_{ J} K(\eta_1, \widehat{\thetav^{ J}},\etav)\\
&=(\theta_1 -\eta_1)\delta_1 \sum_{J \subset \{2,\ldots,k\}} \Big(\prod_{j \in J} (\theta_j-\eta_j) \Big) \delta_{j_2}\ldots \delta_{j_k} K(\thetav^{J\cup \{ 1 \}},\etav)
+K(\eta_1, \hat{\thetav},\etav)\\
&=(\theta_1-\eta_1)\delta_1 K(\theta_1,\hat \thetav,\etav)-K(\eta_1, \hat \thetav,\etav)=K(\thetav,\etav),
\end{aligned}
\end{equation}
where in the last equality we made use of \eqref{eq:deltaremain1}.

We use these definitions to obtain estimates of certain integral kernels, as we can see in the following.
\begin{lemma} \label{lem:kerneloneover}
 Let $K : \rbb^m \times \rbb^{n} \to \cbb$ be smooth, and $k \leq \max(m,n)$. Then,
\begin{equation}
   \Biggnorm{
   K(\thetav,\etav)  \prod_{j=1}^k \frac{1}{\theta_j-\eta_j \pm i0}
  }{m \times n}
\leq
  (2\pi)^{k}
  \sum_{J \subset \{1,\ldots,k \} } \gnorm{\delta_J K}{J},
\end{equation}
where
\begin{equation}
 \gnorm{\delta_J K}{J}^2= \int d\thetav_J d\etav_J\, \sup_{\etav_{J^c}}|\delta_J K(\thetav^J,\etav)|^2.
\end{equation}
\end{lemma}

\noindent
(We denote with $\thetav_J$ the components $\theta_j$ with $j \in J$ or $j>k$, and $\etav_J$, $\etav_{J^c}$ analogously.)

\begin{proof}
We consider the special case where $m=n=k$. We consider $f,g \in L^2(\rbb^k)$. Inserting \eqref{eq:deltaremainn}, we find
\begin{multline}\label{eq:ijintro}
  \int d^k\theta \, d^k\eta \,
 f(\thetav) g(\etav) K(\thetav,\etav) \prod_{j=1}^k \frac{1}{\theta_j-\eta_j \pm i0}
\\
=\sum_{J \subset \{1,\ldots,k\}}
  \int d^k\theta \, d^k\eta \,
     f(\thetav) g(\etav)  \delta_J K(\thetav^J,\etav)\prod_{j\in J}(\theta_j - \eta_j)
 \prod_{j=1}^{k} \frac{1}{\theta_j-\eta_j \pm i0}
\\
=
\sum_{J \subset \{1,\ldots,k\}}
\underbrace{
  \int d^k\theta \, d^k\eta \,
     f(\thetav) g(\etav)  \delta_J K(\thetav^J,\etav)
 \prod_{j \not\in J} \frac{1}{\theta_j-\eta_j \pm i0}
}_{=:I_J}.
\end{multline}
We estimate the integral $I_J$ in the following way. We first split $d^k\theta=(\prod_{j \in J} d\theta_j)(\prod_{i \not\in J} d\theta_i)$ (the same for $d^k\eta$), we have
\begin{equation}\label{Khat}
\begin{aligned}
|I_J| &\leq
  \int (\prod_{j \in J} d\theta_j)
 \Big\lvert \int
  (\prod_{j \not\in J} \frac{d\theta_j d\eta_j}{\theta_j-\eta_j\pm i0 })
  f(\thetav) \hat K (\hat \thetav, \hat \etav)
\Big\rvert ,
\\
& \text{where } \hat K(\hat \thetav, \hat \etav) = \int (\prod_{j \in J} d\eta_j)  g(\etav) \delta_J K(\thetav^J,\etav).
\end{aligned}
\end{equation}
We notice that $\hat K$ depends only on $\theta_j$, $j \in J$, and on $\eta_j$, $j \not\in J$. Then we realize that the inner integral can be seen as the convolution of $f$ with $(\eta_j-\theta_j\pm i0)^{-1}$.
\begin{equation}\label{eqtruncated}
\begin{aligned}
\Big\lvert \int
  (\prod_{j \not\in J} \frac{d\theta_j d\eta_j}{\theta_j-\eta_j\pm i0 })
  & f(\thetav) \hat K (\hat \thetav, \hat \etav)
\Big\rvert
\\
&=\Big\lvert \int(\prod_{j \not\in J} d\eta_j)\Big(\int
  (\prod_{j \not\in J} \frac{d\theta_j }{\theta_j-\eta_j\pm i0 })
  f(\thetav)\Big) \hat K (\hat \thetav, \hat \etav)
\Big\rvert\\
&=\Big\lvert \int(\prod_{j \not\in J} d\eta_j)\Big( \frac{1}{\cdot-\eta_j\pm i0 }\ast
  f\Big)(\eta_j) \hat K (\hat \thetav, \hat \etav)
\Big\rvert.
\end{aligned}
\end{equation}
The expression above can be estimated in Fourier space considering that by Fourier transform that convolution becomes a multiplication with $\sqrt{2 \pi} \Theta(\pm p)$.
\begin{equation}
\begin{aligned}
\text{l.h.s.}\eqref{eqtruncated} &= \Big\lvert \int(\prod_{j \not\in J} dp_j) \Big(\frac{1}{\cdot-\eta_j\pm i0 }\ast
  f\Big)^{\widetilde{}}(p_j) \widetilde{\hat K} (\hat \thetav, \pmb{p})
\Big\rvert\\
&=\Big\lvert \int(\prod_{j \not\in J} dp_j) (2\pi)^{|J^c|/2}\Big(\frac{1}{\cdot-\eta_j\pm i0 }\Big)^{\widetilde{}}(p_j)\cdot
  \tilde{f}(p_j) \widetilde{\hat K} (\hat \thetav, \pmb{p})
\Big\rvert\\
&=\Big\lvert \int(\prod_{j \not\in J} dp_j) (2\pi)^{|J^c|} \Theta(\pm p_j)
  \tilde{f}(p_j) \widetilde{\hat K} (\hat \thetav, \pmb{p})
\Big\rvert\\
&\leq \Big(\int (\prod_{j \not\in J} dp_j)\lvert (2\pi)^{|J^c|} \Theta(\pm p_j)\tilde{f}(p_j)\rvert^{2}\Big)^{1/2}\Big(\int (\prod_{j \not\in J} dp_j)\lvert \widetilde{\hat K} (\hat \thetav, \pmb{p}) \rvert^{2}\Big)^{1/2}\\
  &=(2\pi)^{|J^c|}\Big(\int (\prod_{j \not\in J} dp_j)\lvert \tilde{f}(p_j)\rvert^{2}\Big)^{1/2}\Big(\int (\prod_{j \not\in J} dp_j)\lvert \widetilde{\hat K} (\hat \thetav, \pmb{p}) \rvert^{2}\Big)^{1/2}\\
  &=(2\pi)^{|J^c|}\Big(\int (\prod_{j \not\in J} d\theta_j)\lvert f(\theta_j)\rvert^{2}\Big)^{1/2}\Big(\int (\prod_{j \not\in J} d\eta_j)\lvert \hat K (\hat \thetav, \hat\etav) \rvert^{2}\Big)^{1/2}.
\end{aligned}
\end{equation}
Inserting into \eqref{Khat}, and since $|J^c|\leq k$, we obtain the estimate:
\begin{equation}
|I_J| \leq (2\pi)^{k}
  \int (\prod_{j \in J} d\theta_j)
  \Big(\int (\prod_{j \not\in J} d\theta_j) |f(\thetav)|^2 \Big)^{1/2}
  \Big(\int (\prod_{j \not\in J} d\eta_j) |\hat K(\hat \thetav, \hat\etav)|^2 \Big)^{1/2}.
\end{equation}
Applying the Cauchy-Schwarz inequality with respect to the $\theta_j$ ($j \in J$) integrals, we find
\begin{eqnarray}
|I_J| &\leq& (2\pi)^{k}
  \Big(\int (\prod_{j=1}^{k}d\theta_j) |f(\thetav)|^2 \Big)^{1/2}
  \Big(\int (\prod_{j \in J} d\theta_j)(\prod_{j \not\in J} d\eta_j) |\hat K(\hat \thetav, \hat\etav)|^2 \Big)^{1/2}\nonumber\\
  &=&(2\pi)^{k}   \gnorm{f}{2} \, \gnorm{\hat K}{2}.
\end{eqnarray}
We insert the definition of $\hat K$ \eqref{Khat} and we apply the Cauchy-Schwarz inequality to the $\eta_j$ ($j\in J$) integrals, we arrive at
\begin{equation}
\begin{aligned}
|I_J|&\leq (2\pi)^{k}   \gnorm{f}{2} \Big(\int (\prod_{j \in J} d\theta_j)(\prod_{j \not\in J} d\eta_j) \Big\lvert \int (\prod_{j \in J} d\eta_j)  g(\etav) \delta_J K(\thetav^J,\etav)  \Big\rvert^2 \Big)^{1/2}\\
&\leq (2\pi)^{k}   \gnorm{f}{2} \Big(\int (\prod_{j \in J} d\theta_j)(\prod_{j \not\in J} d\eta_j) \Big(\int (\prod_{j \in J} d\eta_j)  \lvert g(\etav)\rvert^2 \Big) \times \\
  &\qquad \qquad \times \Big(\int (\prod_{j \in J} d\eta_{j})\lvert \delta_J K(\thetav^J,\etav) \rvert^2\Big)\Big)^{1/2}\\
&\leq (2\pi)^k \gnorm{f}{2} \gnorm{g}{2}\cdot \Big( \int \Big( \prod_{j\in J}d\theta_j d\eta_j\Big) \sup_{\eta_{j}, j \notin J} |\delta_J K(\thetav^J,\etav)|^2\Big)^{1/2}\\
 &\leq (2\pi)^{k} \gnorm{f}{2} \, \gnorm{g}{2} \,
  \gnorm{\delta_J K}{J} .
\end{aligned}
\end{equation}
Inserting the estimate above in \eqref{eq:ijintro}, we find the result in Lemma~\ref{lem:kerneloneover} in the case $m=n=k$.

The more general statement (also for $m, n \neq k$) can then be obtained using the Cauchy-Schwarz inequality in the following way. We denote with $\tilde{\thetav}$ and $\tilde{\etav}$ the first $k$ variables in $\thetav$ and in $\etav$, respectively. We can write
\begin{multline}
\Biggnorm{
   K(\thetav,\etav)  \prod_{j=1}^k \frac{1}{\theta_j-\eta_j \pm i0}
  }{m \times n}\\
  = \sup_{\substack{ \gnorm{f}{2} \leq 1\\ \gnorm{g}{2} \leq 1}}\Big\lvert \int d\tilde{\thetav} d\hat \thetav d \tilde{\etav} d\hat \etav\; K(\tilde\thetav,\hat\thetav,\tilde\etav,\hat\etav)\prod_{j=1}^k \frac{1}{\tilde{\theta}_j- \tilde{\eta}_j \pm i0}f(\tilde\thetav,\hat\thetav)g(\tilde\etav,\hat\etav)  \Big\rvert\\
   =\sup_{\substack{ \gnorm{f}{2} \leq 1\\ \gnorm{g}{2} \leq 1}}\int  d\hat \thetav  d\hat \etav\, \Big\lvert \int d\tilde{\thetav}d \tilde{\etav}\, K(\tilde\thetav,\hat\thetav,\tilde\etav,\hat\etav)\prod_{j=1}^k \frac{1}{\tilde{\theta}_j- \tilde{\eta}_j \pm i0}f(\tilde\thetav,\hat\thetav)g(\tilde\etav,\hat\etav)  \Big\rvert.
\end{multline}
We apply the result for $m=n=k$, discussed above, and we find
\begin{multline}
\Biggnorm{K(\thetav,\etav)  \prod_{j=1}^k \frac{1}{\theta_j-\eta_j \pm i0}
  }{m \times n}\\
\leq (2\pi)^{k}
  \sum_{J \subset \{1,\ldots,k \} } \sup_{\substack{ \gnorm{f}{2} \leq 1\\ \gnorm{g}{2} \leq 1}}\int d\hat \thetav  d\hat \etav\, \sup_{\eta_j, j\notin J}\gnorm{\delta_J K((\tilde \thetav,\hat \thetav)^{J},(\tilde \etav,\hat\etav))}{2}\gnorm{f(\tilde \thetav,\hat \thetav)}{2}\gnorm{g(\tilde \etav,\hat \etav)}{2},
\end{multline}
where the $L^2$ norms in the equation above are taken with respect to the $\tilde \thetav$ , $\tilde \etav$ variables. Using the Cauchy-Schwarz inequality, we find
\begin{eqnarray}
\Biggnorm{K(\thetav,\etav)  \prod_{j=1}^k \frac{1}{\theta_j-\eta_j \pm i0}
  }{m \times n}
  &\leq& (2\pi)^{k}
  \sum_{J \subset \{1,\ldots,k \} } \sup_{\substack{ \gnorm{f}{2} \leq 1\\ \gnorm{g}{2} \leq 1}}\gnorm{\delta_J K}{J}\gnorm{f}{2}\gnorm{g}{2}\nonumber\\
  \leq (2\pi)^{k}
  \sum_{J \subset \{1,\ldots,k \} } \gnorm{\delta_J K}{J},
\end{eqnarray}
which is the desired result in Lemma~\ref{lem:kerneloneover}.
\end{proof}

We use this result to obtain estimates for more general integral kernels.
\begin{lemma} \label{lem:kernelcoth}
Let $K : \rbb^m \times \rbb^{n} \to \cbb$ be smooth, and $k \leq \min(m,n)$. Then,
\begin{equation}
   \Biggnorm{
   K(\thetav,\etav)  \prod_{j=1}^k \coth \frac{\theta_j-\eta_j\pm i0}{2}
  }{m \times n}
\leq
  (8\pi)^{k}
  \sum_{J \subset \{1,\ldots,k\}} \gnorm{\delta_J K}{J} .
\end{equation}
\end{lemma}

\begin{proof}
We consider the function $L(x):=\coth x - x^{-1}$; this function is real analytic; there is actually no pole at $x=0$ since $\lim_{x\rightarrow 0}(\coth x - x^{-1})\approx x^{-1}-x^{-1}=0$. Moreover, by computing the extrema of this function, we can show that $|\coth x - x^{-1}|\leq 1$ for all $x \in \rbb$.

We have
\begin{equation}\label{LKstima}
\begin{split}
  \Biggnorm{K(\thetav,\etav)  \prod_{j=1}^k \coth \frac{\theta_j-\eta_j\pm i0}{2}}{m \times n}
 = \Biggnorm{ K(\thetav,\etav)  \prod_{j=1}^k \Big(L(\frac{\theta_j-\eta_j}{2}) + \frac{2}{\theta_j-\eta_j\pm i0} \Big) }{m \times n}\\
 \leq \sum_{J \subset \{1,\ldots,k\}}
   \Biggnorm{
   K(\thetav,\etav) \Big(\prod_{j \in J^c} L(\frac{\theta_j-\eta_j}{2}) \Big)
 \Big(\prod_{j \in J}  \frac{2}{\theta_j-\eta_j\pm i0} \Big)
  }{m\times n}
\\
 =
 \sum_{J \subset \{1,\ldots,k\}} 2^{|J|}
   \Biggnorm{
  \underbrace{ K(\thetav,\etav) \Big(\prod_{j \in J^c} L(\frac{\theta_j-\eta_j}{2}) \Big)}_{=:\hat K(\thetav,\etav)}
 \Big(\prod_{j \in J}  \frac{1}{\theta_j-\eta_j\pm i0} \Big)
  }{m\times n},
\end{split}
\end{equation}
where in the first equality we wrote $\coth x = L(x) + x^{-1}$ and in the second inequality we used the distributional law.

Now we can apply Lemma~\ref{lem:kerneloneover} with $\hat K$ in place of $K$, we have
\begin{equation}
\Biggnorm{
  \hat K(\thetav,\etav)
 \Big(\prod_{j \in J}  \frac{1}{\theta_j-\eta_j\pm i0} \Big)}{m\times n}\leq (2\pi)^{|J|}\sum_{I\subset J}\gnorm{\delta_{I}\hat K}{I}.
\end{equation}
Hence, we find from \eqref{LKstima},
\begin{multline}\label{KhatKstima}
\Biggnorm{K(\thetav,\etav)  \prod_{j=1}^n \coth \frac{\theta_j-\eta_j\pm i0}{2}}{m \times n}
\leq \sum_{J\subset \{ 1,\ldots,k  \}} 2^{|J|}(2\pi)^{|J|}\sum_{I\subset J}\gnorm{\delta_{I}\hat K}{I}\\
=\sum_{J\subset \{ 1,\ldots,k  \}} (4\pi)^{|J|}\sum_{I\subset J}\gnorm{\delta_{I}\hat K}{I}.
\end{multline}
Note that for any $I \subset J$, we find
\begin{eqnarray}
\gnorm{\delta_I \hat K}{I}^2 &=& \int \Big( \prod_{i \in I}d\theta_i d\eta_i \Big)\sup_{\eta_i, i\in I^c}
|\delta_{I} K (\thetav^I,\etav)|^2
\cdot\Big(  \prod_{j\in J^c}L\Big(\frac{\theta_j -\eta_j}{2}\Big)\Big)^2 \nonumber\\
&\leq& \prod_{j\in J^c}\sup_{\theta_j-\eta_j \in \rbb }\big\lvert L\Big(\frac{\theta_j-\eta_j}{2}\Big)\big\rvert^2 \cdot \gnorm{\delta_{I} K}{I}^2\nonumber\\
&\leq& (\gnorm{L}{\infty }^k)^2 \gnorm{\delta_I K}{I}^2.
\end{eqnarray}
Inserting this into \eqref{KhatKstima} and noting that $\gnorm{L}{\infty }\leq 1$, we obtain
\begin{equation}
\begin{split}
  \Biggnorm{K(\thetav,\etav)  \prod_{j=1}^n \coth \frac{\theta_j-\eta_j\pm i0}{2}}{m \times n}
 &\leq
 \sum_{J \subset \{1,\ldots,k\}} (4 \pi)^{|J|}
 \sum_{I \subset J } \gnorm{ \delta_I K}{I}
\\
&\leq
  (8 \pi)^{k}
 \sum_{I \subset \{1,\ldots,k\} } \gnorm{ \delta_I K}{I},
\end{split}
\end{equation}
where in the second inequality we used that there are at most $2^k$ sets $J$ that we can insert between $I$ and $\{1,\ldots,k  \}$: $I \subset J \subset \{1, \ldots,k   \}$, and that $|J|\leq k$. This gives the result in Lemma~\ref{lem:kernelcoth}.
\end{proof}

We computed the estimates above on the remainder terms $\delta_J K(\thetav^J,\etav)$. But it can be hard to compute these estimates in the examples. So, we try instead to relate them to the partial derivatives of $K$ near the diagonal, which are more simple to compute. We present the relation between the estimates on $\delta_J K(\thetav^J,\etav)$ and on the partial derivatives of $K$ in the following proposition.

\begin{proposition} \label{pro:pderivremainder}
 Let $K: \rbb^m \times \rbb^{n} \to \cbb$ be smooth, and $k \leq \min(m,n)$. Let $L:\rbb^m \times \rbb^{n} \to \rbb_+$ and $\epsilon: \rbb^{n} \to (0,1]$ be continuous such that\footnote{We write $\partial \theta_N = \prod_{j \in N} \partial \theta_j$.}
\begin{equation}\label{eq:derivcond}
  \forall N \subset \{1,\ldots, k\}:
 \; \Big\lvert \frac{\partial^{|N|}K}{\partial \theta_N} (\thetav,\etav) \Big\rvert
  \leq L(\thetav^{N^c},\etav)
  \quad \text{whenever } |\theta_j-\eta_j| < \epsilon(\etav)\text{ for all } j \in N.
\end{equation}
Then, for all $J \subset \{1,\ldots,k \}$,
\begin{equation}\label{L2deltaK}
    \gnorm{\delta_J K}{J}
  \leq 8^{|J|} \sum_{M \subset J} \gnorm{\epsilon(\etav)^{-|J|} L(\thetav^M,\etav) }{J}.
\end{equation}
\end{proposition}
Note: In the norm $\gnorm{\epsilon(\etav)^{-|J|}L(\thetav^M,\etav)}{J}$, the integrations are not performed with respect to the variables $\theta_j$ ``removed'' in $\thetav^M$.

\begin{proof}
 First we want to prove that for all $k,m,n,L$ with the restrictions above, and all $J \subset \{1,\ldots,k\}$, there holds
\begin{equation}\label{eq:deltaKinduction}
  \int \sup_{\eta_j, j \notin J^c}| \delta_J K(\thetav^J,\etav) |^2 d \theta_J
  \leq 16^{|J|} \sum_{M \subset J} \epsilon(\etav)^{-2|J|} \int d\theta_M \sup_{\eta_j, j\notin J^c}L(\thetav^M,\etav)^2,
\end{equation}
where $d\theta_M = \prod_{j \in M} d\theta_j$.

We prove this using induction on $|J|$.

For $|J|=0$, namely for $J=\emptyset$, formula \eqref{eq:deltaKinduction} reads
\begin{equation}
  \sup_{\eta_1,\ldots,\eta_k}| K( \eta_1,\ldots,\eta_k,\theta_{k+1},\ldots,\theta_{m},\etav) |^2
 \leq
  \sup_{\eta_1,\ldots,\eta_k} L( \eta_1,\ldots,\eta_k,\theta_{k+1},\ldots,\theta_{m},\etav)^2
\end{equation}
by definition of $\thetav^{J}$ and $\thetav^{M}$ given in \eqref{eq:deltaremainn}, and since $M=0$.

This is just \eqref{eq:derivcond} with $N=\emptyset$, $\theta_j=\eta_j$ ($j \leq k$) and by the definition of $\thetav^{N^c}$ given in \eqref{eq:deltaremainn}.

Now consider the case $J\neq \emptyset$. By renumbering of the variables, we can assume that $1\in J$ and write $J=\{1\} \dot\cup \hat J$.

Then we split the integral in \eqref{eq:deltaKinduction} into the sum of two integrals, one with $|\theta_1-\eta_1|< \epsilon (\etav)$ and one with $|\theta_1-\eta_1| > \epsilon (\etav)$.

First we consider the integral \eqref{eq:deltaKinduction} with $|\theta_1-\eta_1|< \epsilon (\etav)$. We write $\thetav=(\theta_1,\hat \thetav)$, $\etav=(\eta_1,\hat \etav)$, and we define (similarly to \eqref{eq:deltadef}):
\begin{equation}\label{hatKdeltaK}
   \hat K(\hat\thetav,\hat\etav) :=\frac{K(\thetav,\etav)-K(\eta_1,\hat\thetav,\etav)}{\theta_1-\eta_1},
\end{equation}
where $\hat K(\hat\thetav,\hat\etav)$ depends parametrically on $\theta_1,\eta_1$. We can show that $\delta_{\hat J}\hat K = \delta_J K$. Indeed, using \eqref{eq:deltadef}, we can rewrite \eqref{hatKdeltaK} as
\begin{equation}
\hat K(\hat\thetav,\hat\etav) = \delta_1 K(\thetav,\etav).
\end{equation}
Setting $J=\{  1\}\dot{\cup}\hat J$, with $\hat J = \{ j_2,\ldots,j_l \}\subset \{ 2,\ldots,k \}$, and recalling the definitions of $\delta_J K$ and $J$ after \eqref{deltajKdef}, we have
\begin{equation}
\delta_{\hat J}\hat K = \delta_{j_2}\ldots \delta_{j_l}\hat K= \delta_{j_2}\ldots \delta_{j_l}\delta_1 K = \delta_J K.
\end{equation}
By definition \eqref{hatKdeltaK}, and since $\hat K$ is a continuous and differentiable function on the interval $|\theta_1-\eta_1|< \epsilon (\etav)$, we can apply the mean value theorem, and we have
\begin{equation}\label{meanhatK}
    \hat K(\hat\thetav,\hat\etav) = \frac{\partial K}{\partial \theta_1} (\xi,\hat\thetav,\eta_1,\hat\etav)
\end{equation}
where $\xi$ depends on the variables $\hat\thetav,\etav$ but $|\xi-\eta_1|<\epsilon(\etav)$; hence, we have that
\begin{multline}
    \forall N \subset \{2,\ldots, k\}:
 \; \Big\lvert \frac{\partial^{|N|}\hat K}{\partial \theta_N} (\hat\thetav,\hat\etav) \Big\rvert
  =\Big\lvert \frac{\partial^{|N|+1} K}{\partial \theta_1 \partial \theta_N} (\xi,\hat\thetav,\etav)
\Big\rvert
  \leq \sup_{|\xi-\eta_1|<\epsilon(\etav)}
\; \Big\lvert \frac{\partial^{|N|+1} K}{\partial \theta_1 \partial \theta_N} (\xi,\hat\thetav,\etav)
\Big\rvert\\
\leq
L(\thetav^{N^c \backslash\{1\}},\etav)
=L(\eta_1,\thetav^{N^c \backslash\{1\}},\etav)
= \hat L(\hat\thetav^{N^c}, \hat \etav),
\end{multline}
where in the first equality we made use of \eqref{meanhatK}, where in the third inequality we used \eqref{eq:derivcond}, and in the last equality we defined $\hat L(\hat\thetav,\hat\etav):=L(\eta_1,\hat\thetav,\etav)$, which depends on $\eta_1$ as a parameter. So, we have shown that $\hat K$ fulfils the hypothesis \eqref{eq:derivcond} of the Prop.~\ref{pro:pderivremainder} with respect to $\hat L$. Hence, we can apply Prop.~\ref{pro:pderivremainder} with respect to $\hat K, \hat L, \hat J$ (induction hypothesis). We get
\begin{multline}\label{eq:intLE}
   \int_{ |\theta_1-\eta_1|< \epsilon(\etav) } d \theta_J
   \sup_{\eta_j, j\in J^c}| \delta_J K(\thetav^J,\etav) |^2
  =
   \int_{ |\theta_1-\eta_1|< \epsilon(\etav) } d \theta_1 \int d\theta_{\hat J}
   \sup_{\eta_j, j\in \hat{J}^c}| \delta_{\hat J} \hat K(\hat \thetav^{\hat J},\hat \etav) |^2
\\
  \leq 2 \epsilon(\etav)  16^{|\hat J|}
  \sum_{M \subset \hat J} \epsilon(\etav)^{-2|\hat J|} \int d\hat\theta_{M} \sup_{\eta_j, j\in \hat{J}^c}\hat L(\hat\thetav^{M},\hat\etav)^2\\
  \leq \frac{1}{8} 16^{|J|} \sum_{M \subset J} \epsilon(\etav)^{-2|J|} \int d\theta_M \sup_{\eta_j, j\in J^c}L(\thetav^M,\etav)^2,
\end{multline}
where in first equality we used that $\delta_{\hat J}\hat K = \delta_J K$, and $J=\{1\} \dot\cup \hat J$ so that we can split the integral in $\theta_J$ into the integrals in $\theta_1$ and $\theta_{\hat J}$. In the second inequality we used that  $\int_{ |\theta_1-\eta_1|< \epsilon(\etav) } d \theta_1 = 2\epsilon(\etav)$ and we applied \eqref{L2deltaK}, with $\hat K$ and $\hat J$ in place of $K$ and $J$, to $\int d\theta_{\hat J}| \delta_{\hat J} \hat K(\hat \thetav^{\hat J},\hat \etav) |^2$; in the last inequality we used that $|\hat J|>|J|$ and therefore that if $M \subset \hat J$, then $M\subset J$; also, since $|\hat J| >|J|$, then $\epsilon(\etav)^{-2|\hat J|}\leq \epsilon(\etav)^{-2|J|}$. We also used that $\epsilon(\etav)\leq 1$, so that the factor $\epsilon(\etav)$ is estimated by $1$ in the last inequality. Moreover, using the definition $\hat L(\hat\thetav,\hat\etav):=L(\eta_1,\hat\thetav,\etav)$, we can write $\hat L(\hat\thetav^{M},\hat\etav)=L(\eta_1,\hat \thetav^{M},\etav)$ with $M \in \hat J$; on the other hand, since $1 \not\in M$ the right hand side is the same as $L(\thetav^M,\etav)$, with $M \in J$. Finally, since $|\hat J|=|J|-1$, then $2 \cdot 16^{|\hat J|}=2\cdot 16^{|J|}\cdot 1/16 =16^{|J|}\cdot 1/8$.

Now we consider the part of the integral \eqref{eq:deltaKinduction} with $|\theta_1-\eta_1| > \epsilon (\etav)$. We use that $J= \{ 1 \}\dot \cup \hat J$ and, by definition of $\delta_1 K (\thetav,\etav)=(K(\thetav,\etav)-K(\eta_1,\hat \thetav,\etav))/(\theta_1-\eta_1)$, we have that
\begin{equation}
 \delta_J K(\thetav^J,\etav)
 =\delta_1 \delta_{\hat J} K(\thetav^J,\etav)
 =\frac{\delta_{\hat J}K(\thetav^J,\etav)- \delta_{\hat J}K(\eta_1,\hat \thetav^{J},\etav)}{\theta_1 - \eta_1}
 = \frac{\delta_{\hat J}K(\thetav^J,\etav)- \delta_{\hat J}K(\thetav^{\hat J},\etav)}{\theta_1 - \eta_1},
\end{equation}
where in the last equality we used that $1 \not \in \hat J$.

From the equation above, we compute the estimate
\begin{equation}\label{eq:ksplit}
 |\delta_J K(\thetav^J,\etav) |^2
 \leq \frac{2}{(\theta_1-\eta_1)^2}
 \big( |\delta_{\hat J}K(\thetav^J,\etav)|^2 + |\delta_{\hat J}K(\thetav^{\hat J},\etav)|^2 \big).
\end{equation}
We consider the first summand on the r.h.s.. We consider $K$ as a function of $m+n-2$ variables, regarding the dependence on $\theta_1,\eta_1$ as parameters. We know that it fulfils \eqref{eq:derivcond}, so we can apply \eqref{L2deltaK} with $\hat J$ in place of $J$ (induction hypothesis) (see the second inequality of the equation below).
\begin{multline}\label{eq:intGE1}
   \int_{ |\theta_1-\eta_1| > \epsilon(\etav) } d \theta_J
  \frac{2}{(\theta_1-\eta_1)^2} \sup_{\eta_j,j \in J^c} |\delta_{\hat J}K(\thetav^J,\etav)|^2
\\
= \Big(\int_{ |\theta_1-\eta_1| > \epsilon(\etav) } d\theta_1 \frac{2}{(\theta_1-\eta_1)^2}\Big) \int d\theta_{\hat J}\sup_{\eta_j,j\in J^c}|\delta_{\hat J}K(\thetav^{\{ 1\}\dot \cup \hat J},\etav)|^2
\\
\leq
 4 \epsilon(\etav)^{-1}  16^{|\hat J|}
\sum_{\hat M \subset \hat J} \epsilon(\etav)^{-2|\hat J|} \int d\theta_{\hat M}  \sup_{\eta_j,j\in J^c}L(\thetav^{\{1\}\cup \hat M },\etav)^2
\\
\leq \frac{1}{4} 16^{|J|} \epsilon(\etav)\epsilon(\etav)^{-2|\hat J|-2} \sum_{ M \subset  J} \int d\theta_{ M} \sup_{\eta_j,j\in J^c}L(\thetav^{ M },\etav)^2
\\
\leq \frac{1}{4} 16^{|J|} \epsilon(\etav)^{-2|J|} \sum_{ M \subset  J} \int d\theta_{ M} \sup_{\eta_j,j\in J^c}L(\thetav^{ M },\etav)^2,
\end{multline}
where in the second inequality we computed the integral
\begin{equation}
\int_{ |\theta_1-\eta_1| > \epsilon(\etav) } d\theta_1 \frac{2}{(\theta_1-\eta_1)^2}=\int_{\eta_1 +\epsilon}^{\infty}\frac{2}{(\theta_1-\eta_1)^2}+\int_{\infty}^{\eta_1 -\epsilon}\frac{2}{(\theta_1-\eta_1)^2}=\frac{4}{\epsilon}.
\end{equation}
In the last inequality in \eqref{eq:intGE1} we used the substitution $M = \hat M \cup \{1\}$. Since $\hat M \subset \hat J$, then $M \subset J$; we also replaced the integral in $\theta_{ \hat M}$ with the integral in $\theta_M$, noting that in the case where $1\in M$, we are integrating over more terms (than before) which are however positive. We also used $\epsilon(\etav)\leq 1$.

Now we consider the integral over the second summand in \eqref{eq:ksplit}, we apply \eqref{L2deltaK} to $K(\eta_1,\theta_2,\ldots,\theta_m,\etav)$ and $L(\eta_1,\theta_2,\ldots,\theta_m,\etav)$ as functions of $m+n-2$ variables, parametrically dependent on $\eta_1$, and with $\hat J$ in place of $J$ (induction hypothesis)(see the second inequality below), and we obtain
\begin{multline} \label{eq:intGE2}
   \int_{ |\theta_1-\eta_1| > \epsilon(\etav) } d \theta_J
  \frac{2}{(\theta_1-\eta_1)^2}  \sup_{\eta_j, j\in J^c}|\delta_{\hat J} K(\thetav^{\hat J},\etav)|^2
\\
=\int_{ |\theta_1-\eta_1| > \epsilon(\etav) } d \theta_1
  \frac{2}{(\theta_1-\eta_1)^2}  \int d \theta_{\hat J} \sup_{\eta_j, j\in J^c}|\delta_{\hat J} K(\thetav^{\hat J},\etav)|^2
\\
\leq
 \Big( \int_{ |\theta_1-\eta_1| > \epsilon(\etav) } d\theta_1 \frac{2}{(\theta_1-\eta_1)^2} \Big)
 16^{|\hat J|} \sum_{ M \subset \hat J} \epsilon(\etav)^{-2|\hat J|}
 \int d\theta_{ M} \sup_{\eta_j, j\in J^c} L(\thetav^{ M },\etav)^2
\\
\\
\leq \frac{1}{4} \epsilon(\etav)^{-2|J|}  16^{|J|} \sum_{ M \subset J} \int d\theta_{M}\sup_{\eta_j, j\in J^c} L(\thetav^{ M },\etav)^2,
\end{multline}
where we have used $\epsilon(\etav)\leq 1$.

Now we can sum the three estimates \eqref{eq:intLE}, \eqref{eq:intGE1}, \eqref{eq:intGE2}, and we find \eqref{eq:deltaKinduction}.

Then we integrate \eqref{eq:deltaKinduction} over $\etav$ and we get
\begin{equation}\label{squareresult}
  \gnorm{ \delta_J K(\thetav^J, \etav)}{J}^2 \leq 16^{|J|} \sum_{M \subset J}
  \gnorm{\epsilon(\etav)^{-2|J|}  L(\thetav^M,\etav)}{J}^2,
\end{equation}
To pass from \eqref{squareresult} to \eqref{L2deltaK} we need to take the square root of both sides in \eqref{squareresult}:
\begin{equation}
  \gnorm{ \delta_J K(\thetav^J, \etav)}{2} \leq 4^{|J|} \sqrt{\sum_{M \subset J}
  \gnorm{\epsilon(\etav)^{-2|J|}  L(\thetav^M,\etav)}{2}^2}.
\end{equation}
Now we use the usual estimate of the $\ell^2$ norm against the $\ell^1$ norm: for a vector $\pmb{a}\in \rbb^n$, we have $\sqrt{\sum_j |a_j|^2} \leq \sum_j |a_j|$; noting that the sum on the r.h.s. has $2^{|J|}$ summands, we finally find
\begin{equation}
\gnorm{ \delta_J K(\thetav^J, \etav)}{2} \leq 8^{|J|} \sum_{M \subset J}
  \gnorm{\epsilon(\etav)^{-2|J|}  L(\thetav^M,\etav)}{2}.
  \end{equation}
\end{proof}

\emph{Remark}: To prove the results in the following Lemma~\ref{lem:expquotient} and Prop.~\ref{pro:kprototype}, we need the following estimates.

For fixed $0 < \alpha < 1$, we can show that there is a constant $c_\alpha>0$ such that for all $a,b \in \rbb$,
\begin{equation}\label{eq:alphaineq1}
   c_\alpha^{-1} (|a|^\alpha + |b|^\alpha) \geq (|a| + |b|)^\alpha
  \geq c_\alpha (|a|^\alpha + |b|^\alpha) .
\end{equation}
This follows from the fact that the function $(|a|^\alpha + |b|^\alpha)/(|a| + |b|)^\alpha $ is homogeneous of order 0, and continuous and non-zero on the unit circle. See Sec~\ref{sec:jaffee} for details on the argument. In particular we find from there that
\begin{equation}
m (|a|+|b|)^\alpha \leq |a|^\alpha + |b|^\alpha \leq M (|a|+|b|)^{\alpha}\leq \frac{M}{m}(|a|^\alpha+|b|^\alpha).
\end{equation}
By setting $c_\alpha ^{-1}:= 1/m$ and $c_\alpha := 1/m$, we find \eqref{eq:alphaineq1}.

Another estimate that one can obtain from \eqref{eq:alphaineq1} is
\begin{equation}\label{eq:alphaineq2}
   c_\alpha^{-1} (|a|^\alpha + |b|^\alpha) \geq |a \pm b|^\alpha
  \geq c_\alpha |a|^\alpha - |b|^\alpha,
\end{equation}
and the same with $a$ and $b$ exchanged.

To show the right inequality in \eqref{eq:alphaineq2}, we consider the left inequality in \eqref{eq:alphaineq1} and we replace $a$ with $a \pm b$, we find:
\begin{equation}
c_\alpha ^{-1}( |a\pm b|^\alpha + |b|^\alpha )\geq (|a\pm b| + |b|)^\alpha.
\end{equation}
By moving $c_\alpha$ and $|b|^\alpha$ to the right hand side of the inequality, we get
\begin{equation}
|a\pm b|^\alpha \geq c_\alpha (|a\pm b|+ |b|)^\alpha - |b|^\alpha.
\end{equation}
Since $|a \pm b|\geq ||a|-|b||$, we find
\begin{equation}
|a\pm b|^\alpha \geq c_\alpha (|a\pm b|+ |b|)^\alpha - |b|^\alpha \geq c_\alpha (|a| - |b| + |b|)^\alpha -|b|^\alpha = c_\alpha |a|^\alpha - |b|^\alpha.
\end{equation}
Therefore,
\begin{equation}
|a\pm b|^\alpha \geq c_\alpha |a|^\alpha - |b|^\alpha.
\end{equation}
To show the left inequality in \eqref{eq:alphaineq2}, we consider again the left inequality in \eqref{eq:alphaineq1}. Since $|a|+|b|\geq |a \pm b|$, we have
\begin{equation}
c_\alpha^{-1} (|a|^\alpha + |b|^\alpha) \geq (|a|+|b|)^{\alpha}\geq |a \pm b|^\alpha,
\end{equation}
and therefore:
\begin{equation}
c_\alpha^{-1} (|a|^\alpha + |b|^\alpha) \geq |a \pm b|^\alpha.
\end{equation}

Another estimate that one also can have by choosing $c_\alpha<1$, is the following:
\begin{equation}\label{eq:alphaineq3}
   |a-b|^\alpha + |b|^\alpha \geq \frac{c_\alpha}{2} (|a|^\alpha + |b|^\alpha).
\end{equation}
To prove this inequality, we consider \eqref{eq:alphaineq2}; it follows from this formula that
\begin{equation}
|a-b|^\alpha + |b|^\alpha \geq c_\alpha |a-b+b|^\alpha = c_\alpha |a|^\alpha.
\end{equation}
Choosing $c_\alpha <1$, we have also
\begin{equation}
|a-b|^\alpha + |b|^\alpha \geq |b|^\alpha \geq c_\alpha |b|^\alpha.
\end{equation}
Taking the sum member by member of the two equations above, we have
\begin{equation}
2(|a-b|^\alpha + |b|^\alpha) \geq c_\alpha ( |a|^\alpha + |b|^\alpha )
\end{equation}
and therefore we find \eqref{eq:alphaineq3}.

Now we want to compute the $\gnorm{\cdotarg}{m \times n}$ of certain concrete functions; in view of Lemma~\ref{lem:kerneloneover}, this will mean to compute the $L^2$ norm of certain exponentials in the simplest case. We have the following lemma:
\begin{lemma} \label{lem:expquotient}
Let $0 < \alpha < 1$, $\beta > 0$, $m,n \in \nbb_0$. Let $K:\rbb^m \times \rbb^n \to \cbb$ be given by
\begin{equation}\label{KexpE}
  K(\thetav,\etav) = \exp( - \beta E(\etav)^\alpha ) \exp( - \beta |E(\thetav) - E(\etav)|^\alpha ).
\end{equation}
Then, there is $d_\alpha>0$ (depending on $\alpha$ only) such that
\begin{equation}\label{L2Kgamma}
  \gnorm{K}{2}^2 \leq d_\alpha^{m+n} \beta^{-(m+n)/2\alpha} \gamma_\alpha(m) \gamma_\alpha(n),
\end{equation}
where
\begin{equation}
   \gamma_\alpha(0) := 1, \quad \gamma_\alpha(k):= \frac{\Gamma(k/2\alpha)}{ \alpha \Gamma(k/2)} \text{ for } k \in \nbb.
\end{equation}
\end{lemma}

\begin{proof}
  We consider the integral
\begin{equation}\label{L2normaK}
  \gnorm{K}{2}^2  = \int d\thetav \, d\etav \, \exp \big(
     - 2 \beta |E(\thetav)-E(\etav)|^\alpha - 2\beta |E(\etav)|^\alpha
 \big),
\end{equation}
and we perform the substitution
\begin{equation}
  p_i = \sinh \frac{\theta_i}{2},
 \quad \frac{2 dp_i}{\sqrt{p_i^2 + 1}} = d\theta_i,
 \quad E(\thetav) = 2 \pv^2 + m,
\end{equation}
and $q_i = \sinh \eta_i/2$ in analogous way. We find
\begin{equation}
\begin{split}
  \gnorm{K}{2}^2 = 2^{m+n} \int d\pv\, d\qv\,
\Big(\prod_{i=1}^m (1+p_i^2)^{-1/2}\Big)
\Big(\prod_{j=1}^n (1+q_j^2)^{-1/2}\Big) \times \\
\times \exp\big(
  - 2\beta |2\pv^2 - 2\qv^2 + m - n |^\alpha - 2\beta |2\qv^2 + n|^\alpha
\big).
\end{split}
\end{equation}
We use the inequality \eqref{eq:alphaineq3}, and the estimates $1+p_i^2 \geq 1$, $1+q_i^2 \geq 1$ (hence we estimate $\prod_{i=1}^m (1+p_i^2)^{-1/2} \leq 1$ and $\prod_{j=1}^n (1+q_j^2)^{-1/2} \leq 1$), and we get
\begin{equation}\label{eq:splitintegral}
  \gnorm{K}{2}^2 \leq 2^{m+n} \int d\pv\, d\qv\,
\exp\big(
  - \beta c_\alpha |2\pv^2 + m |^\alpha -  \beta c_\alpha |2\qv^2 + n|^\alpha
\big).
\end{equation}
We notice that in the integral above the dependence on the variables $p,q$ factorizes; first we consider the integral on $d\pv$, in the case $m \geq 1$.

We perform the substitution with polar coordinates $\pv = p \ev(\Omega)$, $d\pv = p^{m-1}dpd\Omega$,
where $\int d\Omega = 2 \pi^{m/2} / \Gamma(\frac{m}{2})$. Since $m \geq 1$, we can drop the dependence on $m$ in the exponent, we get
\begin{multline} \label{eq:gammaint}
 \int d\pv \exp( - \beta c_\alpha |2 \pv^2 + m |^\alpha )
  \leq  \frac{ 2 \pi^{m/2} }{ \Gamma(\frac{m}{2}) } \int_0^\infty dp\, p^{m-1} \exp(-\beta c_\alpha 2^\alpha p^{2 \alpha})
\\
  =  \Big( \frac{\sqrt{\pi/2}}{(\beta c_\alpha)^{1/2\alpha}} \Big)^m \frac{1}{\alpha}
   \frac{\Gamma(m/2\alpha)}{ \Gamma(m/2)}
 \leq d_\alpha^m \beta^{-m/2\alpha} \gamma_\alpha(m)
\end{multline}
where $d_\alpha>0$ is a suitable constant also involving $\sqrt{\pi/2}/ c_{\alpha}^{1/2\alpha}$.

To prove the second equality in the above equation we used the representation of the gamma function: for $z \in \cbb$, $\Gamma(z)=\int_{0}^{\infty}e^{-t}t^{z-1}dt$. By the substitution \\ $t=\beta c_\alpha 2^\alpha p^{2\alpha}$, $dp=dt (1/t)(1/(2\alpha))( t/(\beta c_\alpha 2^\alpha))^{1/2\alpha}$, $p^{m-1}=(t/(\beta c_\alpha 2^\alpha) )^{(m-1)/2\alpha}$, we find
\begin{multline}
 \frac{ 2 \pi^{m/2} }{ \Gamma(\frac{m}{2}) } \int_0^\infty dp\, p^{m-1} \exp(-\beta c_\alpha 2^\alpha p^{2 \alpha})\\
= \frac{ 2 \pi^{m/2} }{ \Gamma(\frac{m}{2}) } \frac{1}{2\alpha}\Big( \frac{1}{\beta c_\alpha 2^\alpha} \Big)^{1/2\alpha}\Big( \frac{1}{2\alpha} \Big)^{(m-1)/2\alpha}\int_0^\infty dt\,t^{\frac{m}{2\alpha}-1}e^{-t}\\
=\Big( \frac{\sqrt{\pi/2}}{(\beta c_\alpha)^{1/2\alpha}} \Big)^m \frac{1}{\alpha}
   \frac{\Gamma(m/2\alpha)}{ \Gamma(m/2)},
\end{multline}
where in the last equality we used that $\int_0^\infty dt\,t^{\frac{m}{2\alpha}-1}e^{-t} = \Gamma(m/2\alpha)$.

We can find an analogous result for the integral on $d\qv$ in the case $n \geq 1$. Multiplying the two results in \eqref{eq:splitintegral} and redefining the constant $d_\alpha$, we find \eqref{L2Kgamma}.

We will compute separately in the cases $m=0$ or $n=0$ the norm $\gnorm{K}{2}$ using the definition \eqref{KexpE}. For example, in the case $n=0$ and $m \neq 0$, we have
\begin{equation}
\int d\thetav \lvert K(\thetav,0)  \rvert^2 = \int d\thetav \exp(-2\beta |E(\thetav)|^\alpha)
\leq 2^m \int d\pv \exp(-2\beta |2\pv^2 + m|^\alpha)
\end{equation}
where in the second inequality we performed the change of variables after \eqref{L2normaK} and we used that $\prod_{i=1}^m (1+p_i^2)^{-1/2} \leq 1$. Using the same argument as before in \eqref{eq:gammaint}, we find
\begin{equation}
\gnorm{K}{2}^2 \leq d_\alpha^{m} \beta^{-m/2\alpha} \gamma_\alpha(m).
\end{equation}
Here it enters the requirement that $\gamma_\alpha(0) := 1$ in \eqref{L2Kgamma}: This is in order to match the result of our direct computation above with equation \eqref{L2Kgamma} in the case $n=0$ and $m\neq 0$.
\end{proof}

Combining the results above we compute estimates on the integral kernels which are more interesting for the example.

\begin{proposition}\label{pro:kprototype}
Let $0 < \alpha < 1$, $m,n \in \nbb_0$, and $k \leq \min(m,n)$. Let $g : \rbb \to \cbb$ be smooth and $c>0$ such that
\begin{equation}\label{eq:gderivest}
   \Big\lvert \frac{\partial^j g}{\partial p^j}(p) \Big\rvert
   \leq c \, \exp(-|p|^\alpha) \quad \text{for all $p \in \rbb$, $0 \leq j \leq k$.}
\end{equation}
Let $K: \rbb^m \times \rbb^n \to \cbb$ be defined by
\begin{equation}
  K(\thetav,\etav) := \frac{g\big(E(\thetav)-E(\etav)\big)}{  \exp(E(\etav)^\alpha)}
 \prod_{j=1}^k \coth\frac{\theta_j-\eta_j \pm i0}{2}.
\end{equation}
Then, there are constants $c'>0$, $d>0$ (depending on $c$, $\alpha$, and $k$, but not on $m$ or $n$) such that
\begin{equation}\label{Knormmtimesn}
   \gnorm{K}{m \times n} \leq c' d^{m+n} \sqrt{\gamma_\alpha(m)\gamma_\alpha(n)}.
\end{equation}
\end{proposition}

\begin{proof}
We have from Lemma \ref{lem:kernelcoth} that
\begin{equation}\label{eq:kandkprime}
\gnorm{K}{m \times n}
=\Biggnorm{\frac{g\big(E(\thetav)-E(\etav)\big)}{  \exp(E(\etav)^\alpha)}
 \prod_{j=1}^k \coth\frac{\theta_j-\eta_j \pm i0}{2}}{m\times n}
\leq
  (8\pi)^{k}
  \sum_{J \subset \{1,\ldots,k\}} \gnorm{\delta_J K'(\thetav^{J} ,\etav)}{J},
\end{equation}
where $K'(\thetav,\etav) = g\big(E(\thetav)-E(\etav)\big) \exp( - E(\etav)^\alpha)$.

Using Proposition~\ref{pro:pderivremainder} we want to estimate $\delta_J K'$. So, we need to check that $K'$ fulfils the hypothesis \eqref{eq:derivcond} of the proposition. So, we need to compute an estimate for the following quantity
\begin{equation}
 \frac{\partial^{|N|}K'}{\partial \theta_N} (\thetav,\etav)
= \frac{\partial^{|N|}g}{\partial \theta_N}\big(E(\thetav)-E(\etav)\big)
  \exp(-E(\etav)^\alpha).
\end{equation}
After computing one by one these derivatives, applying repeatedly the chain rule, $p=E(\thetav)-E(\etav)$, $j \in N$,
\begin{equation}
\frac{\partial K'}{\partial \theta_j} (\thetav,\etav)
= \frac{\partial p}{\partial \theta_j}\frac{\partial g}{\partial p}\big(E(\thetav)-E(\etav)\big)
  \exp(-E(\etav)^\alpha),
\end{equation}
we find that the partial derivatives of $K'$ are
\begin{equation}\label{eq:kprimederiv}
 \frac{\partial^{|N|}K'}{\partial \theta_N} (\thetav,\etav)
= \Big(\prod_{j \in N} \sinh \theta_j\Big)
  \frac{\partial^{|N|}g }{\partial p^{|N|}} \big(E(\thetav)-E(\etav)\big)
  \exp(-E(\etav)^\alpha).
\end{equation}
Using the hypothesis \eqref{eq:gderivest}, we compute the following estimate
\begin{equation}\label{gEstima}
 \begin{aligned}
    \Big\lvert \frac{\partial^{|N|}g }{\partial p^{|N|}} \big(E(\thetav)-E(\etav)\big) \Big\rvert
  &\leq
    c \exp(-|E(\thetav)-E(\etav)|^\alpha)
 \\ &\leq
    c \exp(|E(\thetav)-E(\thetav^{N^c})|^\alpha)  \exp(-c_\alpha|E(\thetav^{N^c})-E(\etav)|^\alpha),
 \end{aligned}
\end{equation}
where in the second inequality we made use of \eqref{eq:alphaineq2} (right inequality): $|E(\thetav)-E(\thetav^{N^c})+E(\thetav^{N^c})-E(\etav)|\geq c_\alpha |E(\thetav^{N^c})-E(\etav)|^\alpha - |E(\thetav)-E(\thetav^{N^c})|^\alpha$.

Now we compute an estimate for $|E(\thetav)-E(\thetav^{N^c})|$; we have
\begin{equation}\label{eq:edifference1}
  |E(\thetav)-E(\thetav^{N^c})| = \lvert \sum_{j \in N}\cosh \theta_j - \sum_{j \in N}\cosh \eta_j  \rvert  \leq \sum_{j \in N} |\cosh \theta_j - \cosh \eta_j|.
\end{equation}
To estimate $\cosh \theta_j - \cosh \eta_j$, we set $\epsilon(\etav):=(4 \prod_{j=1}^k \cosh \eta_j)^{-1}$. If we consider the closed interval $|\theta_j-\eta_j|<\epsilon(\etav)$, then we can apply the mean value theorem, and we have
\begin{equation}
 \begin{aligned}
    \lvert \cosh \theta_j - \cosh \eta_j \rvert
    &=|\sinh\xi|\cdot |\theta_j - \eta_j| \\
    &\leq \epsilon(\etav) \sup \{ \sinh|\xi| : |\xi - \eta_j| < \epsilon(\etav) \} \\
  &\leq \frac{1}{4 \cosh \eta_j} \sup \{ \sinh|\xi| : |\xi - \eta_j| < \epsilon(\etav) \} \\
  &\leq
   \frac{1}{4 \cosh \eta_j} \sinh( |\eta_j| + \epsilon(\etav) )
\\
   &= \frac{1}{4 \cosh \eta_j} (\sinh |\eta_j| \cosh \epsilon(\etav) + \cosh \eta_j \sinh \epsilon(\etav) )
\\
&=\frac{1}{4}\Big( \frac{\sinh |\eta_j|}{\cosh \eta_j}\cosh \epsilon(\etav) + \sinh \epsilon(\etav)  \Big)
\\
  &\leq \frac{1}{4}2 \cosh \epsilon(\etav)= \frac{1}{2} \cosh \epsilon(\etav)\approx \frac{1}{2} \leq 1
 \end{aligned}
\end{equation}
where in the third inequality we used that $\prod_{i=1}^{k}\cosh\eta_i \geq \cosh \eta_j$ and therefore $\epsilon(\etav)=(4 \prod_{i=1}^k \cosh\eta_i)^{-1}\leq (4\cosh\eta_j)^{-1}$. In the fourth inequality we used that $|\xi|-|\eta_j|< |\xi - \eta_j|< \epsilon(\etav)$, which implies $|\xi| < |\eta_j|+\epsilon(\etav)$. In the seventh inequality we used that $\tanh |\eta_j|\leq 1$ and that for small $\epsilon$, $\sinh \epsilon(\etav)\leq  \cosh\epsilon(\etav)$. In the ninth approximation we used that for small $\epsilon$,  $\cosh\epsilon \approx 1$.

As a consequence of the equation above and of \eqref{eq:edifference1}, we have for $N \subset\{1,\ldots,k\}$, and if $|\theta_j-\eta_j|<\epsilon(\etav)$ for all $j\in N$, that
\begin{equation}\label{eq:edifference}
  |E(\thetav)-E(\thetav^{N^c})| \leq \sum_{j \in N} |\cosh \theta_j - \cosh \eta_j| \leq \sum_{j\in N} 1
  \leq k,
\end{equation}
where we used that $|N|\leq k$.

Inserting \eqref{eq:edifference} in \eqref{gEstima}, we get
\begin{equation}
 \begin{aligned}
    \Big\lvert \frac{\partial^{|N|}g }{\partial p^{|N|}} \big(E(\thetav)-E(\etav)\big) \Big\rvert
 &\leq
    c \exp(|E(\thetav)-E(\thetav^{N^c})|^\alpha)  \exp(-c_\alpha|E(\thetav^{N^c})-E(\etav)|^\alpha)
  \\ &\leq
    c \exp(k^\alpha) \exp(-c_\alpha|E(\thetav^{N^c})-E(\etav)|^\alpha) .
 \end{aligned}
\end{equation}
Inserting into \eqref{eq:kprimederiv}, it follows that
\begin{multline}\label{K'Leq}
    \Big\lvert\frac{\partial^{|N|}K' }{\partial \theta_N} (\thetav,\etav) \Big\rvert
    \leq c\, e^k \Big\lvert\prod_{j \in N} \sinh \theta_j\Big\rvert
      \exp(-c_\alpha|E(\thetav^{N^c})-E(\etav)|^\alpha - E(\etav)^\alpha)\\
  \leq c\,(4e)^k \Big(\prod_{j =1}^k \cosh \eta_j\Big)
   \exp(-c_\alpha|E(\thetav)-E(\etav)|^\alpha - E(\etav)^\alpha)
\end{multline}
where in the first inequality we used that for $0<\alpha<1$, $e^{k^\alpha}\leq e^k$, and in the second inequality we  used that $\lvert \prod_{j\in N}\sinh\theta_j\rvert\leq 4^k \prod_{j=1}^{k}\cosh\eta_j$. This last inequality follows from a short computation: if $|\theta-\eta|\leq \epsilon$, then $|\sinh\theta|\leq |\sinh(\eta +\epsilon)| \leq |\sinh \eta \cosh\epsilon| + |\cosh\eta \sinh\epsilon|$; since $|\sinh\eta|\leq \cosh\eta$ and since for small $\epsilon$, $\cosh\epsilon \leq \cosh 1$ and $\sinh \epsilon \leq \sinh 1$, then we have $|\sinh \eta \cosh\epsilon| + |\cosh\eta \sinh\epsilon| \leq \cosh \eta ( \cosh 1 + \sinh 1)$. Since $\cosh 1 + \sinh 1 = e^1 \leq 4$, then $|\sinh\theta|\leq  \cosh \eta ( \cosh 1 + \sinh 1) \leq 4 \cosh \eta$.

In view of Prop.~\ref{pro:pderivremainder} we call
\begin{equation}
   L(\thetav,\etav) := c\,(4e)^k \Big(\prod_{j =1}^k \cosh \eta_j\Big)
   \exp(-c_\alpha|E(\thetav)-E(\etav)|^\alpha - E(\etav)^\alpha),
\end{equation}
and from \eqref{K'Leq} we write
\begin{equation}
\Big\lvert\frac{\partial^{|N|}K' }{\partial \theta_N} (\thetav,\etav) \Big\rvert \leq L(\thetav^{N^c},\etav).
\end{equation}
With this $L$, we can fulfil the hypothesis \eqref{eq:derivcond} of Proposition~\ref{pro:pderivremainder}.

Hence, we can apply Prop.~\ref{pro:pderivremainder}, and using \eqref{L2deltaK}, we have from \eqref{eq:kandkprime},
\begin{multline}\label{eq:mjsum}
 \gnorm{K}{m \times n}
 \leq
 (8\pi)^{k}\sum_{J \subset \{1,\ldots,k\}}8^{|J|}
  \sum_{M \subset J}
  \gnorm{\epsilon(\etav)^{-|J|} L(\thetav^{M} ,\etav)}{J}
 \\
\leq
  (64\pi)^{k}
  \sum_{M \subset J \subset \{1,\ldots,k\}}
  \gnorm{\epsilon(\etav)^{-|J|} L(\thetav^{M} ,\etav)}{J}
\end{multline}
where in the last inequality we used that $|J|\leq k$.

On the right hand side of the above formula we have
\begin{multline}\label{epsilonL}
 \epsilon(\etav)^{-|J|} L(\thetav^{M} ,\etav)
 =4c(4e)^k \prod_{j=1}^{k}(\cosh\eta_j)^{|J|+1}\exp(-c_\alpha |E(\thetav^M)-E(\etav)|^\alpha - E(\etav)^{\alpha})
 \\
 \leq  \underbrace{4c (4 e)^k \prod_{j=1}^k(\cosh\eta_j)^{|J|+1} \exp\Big( -\frac{c_{\alpha}^k}{2} (\cosh\eta_j)^\alpha\Big)}_{\leq c'}
 \exp\big(- \frac{c_\alpha}{2} |E(\thetav^M)-E(\etav)|^\alpha -\frac{1}{2} E(\etav)^{\alpha} \big).
\end{multline}
In the second inequality, we have used that $\frac{1}{2}E(\etav)^{\alpha}\geq \frac{c_\alpha^k}{2}\sum_{j=1}^k\cosh\eta_j$, which can be obtained by repeated application of \eqref{eq:alphaineq1}.

Also, we split $\exp(-E(\etav)^{\alpha})=\exp(-\frac{1}{2} E(\etav)^{\alpha})\exp(-\frac{1}{2} E(\etav)^{\alpha})$; finally, we used that $E(\thetav^M)-E(\etav)=E(\hat \thetav) - E(\hat \etav)$, where $\hat\thetav,\hat\etav$ denote the variables $\theta_j,\eta_j$ with $j \in M$ or $j>k$.

Clearly, $4 e \cosh^{|J|+1}\eta_j \exp(-\frac{c_\alpha^k}{2}\cosh^\alpha \eta_j)$ is bounded by a constant, that we call $c'$; the constant $c'$ might depend on $k$; we have $c'=1$ for $k=0$.

We compute the following estimate for the last factor $\exp(-\frac{1}{2} E(\etav)^{\alpha})$ in the last line of \eqref{epsilonL}. We split $E(\etav)=E(\hat \etav)+E(\check \etav)$. Using \eqref{eq:alphaineq1}, we find
\begin{equation}
  E(\etav)^\alpha \geq c_\alpha E(\hat\etav)^\alpha + c_\alpha E(\check \etav)^\alpha,
\end{equation}
where $\check \etav$ are the remaining components of $\etav$. Therefore, we have
\begin{equation}
\exp(-\frac{1}{2} E(\etav)^{\alpha})\leq -\frac{c_\alpha}{2}E(\hat\etav)^{\alpha}-\frac{c_\alpha}{2}E(\check \etav)^{\alpha}.
\end{equation}
Inserting this into \eqref{epsilonL}, we find
\begin{equation}
 \epsilon(\etav)^{-|J|} L(\thetav^{M} ,\etav)
 \leq  c' \exp\big(- \frac{c_\alpha}{2} |E(\hat\thetav)-E(\hat\etav)|^\alpha -\frac{c_\alpha}{2} E(\hat\etav)^{\alpha} \big)
  \exp (-\frac{c_\alpha}{2} E(\check\etav)^{\alpha} ).
\end{equation}
Taking the supremum over $\eta_j$, $j\in J^c$, affects only the last factor; we obtain
\begin{equation}
\sup_{\eta_j,j\in J^c}\epsilon(\etav)^{-|J|}L(\thetav^M,\etav)\leq c'\exp\Big( -\frac{c_\alpha}{2}|E(\hat\thetav)-E(\hat\etav)|^\alpha - \frac{c_\alpha}{2}E(\hat\etav)^\alpha \Big)\exp\Big( -\frac{c_\alpha}{2}E(\check{\check{\etav}})^\alpha\Big),
\end{equation}
where $\check{\check{\etav}}$ denotes the variables $\eta_j$, $j\in J$ or $j>k$.

Now we apply Lemma~\ref{lem:expquotient} once in the $\hat\thetav,\hat \etav$ variables, and once in $\check{\check{\etav}}$ (with no corresponding $\theta$'s).
Considering that we have $m-k + |M|$ variables $\hat \thetav$, $n-k+|M|$ variables $\hat \etav$, $n-(n-(k-|M|))=k-|M|$ variables $\check{\etav}$ and $k-|M|-(k-|J|)=|J|-|M|$ variables $\check{\check{\etav}}$, we find
\begin{equation}
 \gnorm{\epsilon(\etav)^{-|J|} L(\thetav^{M} ,\etav)}{2}^2
 \leq (c'')^{m+n} \gamma_\alpha(m-k+|M|) \gamma_\alpha(n-k+|M|) \gamma_\alpha(|J|-|M|),
\end{equation}
where $c''$ is some constant which also include the constants $c', c_\alpha, d_\alpha$ to some power $m + n$ (from Lemma~\ref{lem:expquotient}).

We use monotonicity of $\gamma_\alpha$: $\gamma_\alpha(m-(k-|M|))\leq \gamma_\alpha (m)$, $\gamma_{\alpha}(n-(|J|-|M|))\leq \gamma_\alpha(n)$, and
$\gamma_\alpha(|J|-|M|) \leq \gamma_\alpha(k)$. We absorb this $\gamma_\alpha(k)$ into the constant $c''$ (which might depend on $k$):
\begin{equation}
 \gnorm{\epsilon(\etav)^{-|J|} L(\thetav^{M} ,\etav)}{2}^2
 \leq (c''')^{m+n} \gamma_\alpha(m) \gamma_\alpha(n).
\end{equation}
Inserting this into \eqref{eq:mjsum}, we find
\begin{eqnarray}
\gnorm{K}{m\times n} &\leq& (64\pi)^k \sum_{M\subset J \subset \{ 1,\ldots,k \}}(c''')^{(m+n)/2}\sqrt{\gamma_\alpha(m)\gamma_{\alpha}(n)}\nonumber\\
&=& (64\pi)^k(c''')^{(m+n)/2}\sqrt{\gamma_\alpha(m)\gamma_{\alpha}(n)} \Big(\sum_{M\subset J \subset \{ 1,\ldots,k \}} 1\Big) \nonumber\\
&=& (64\pi)^k(c''')^{(m+n)/2}\sqrt{\gamma_\alpha(m)\gamma_{\alpha}(n)} 2^k 2^{|J|}\nonumber\\
&\leq& (64\pi)^k(c''')^{(m+n)/2}\sqrt{\gamma_\alpha(m)\gamma_{\alpha}(n)} 2^k 2^k \nonumber\\
&\leq& c' d^{m+n} \sqrt{\gamma_\alpha(m)\gamma_\alpha(n)},
\end{eqnarray}
where in the second equality we used that $\sum_{M\subset J \subset \{ 1,\ldots,k \}} 1 = 2^k 2^{|J|}$, and in the fourth inequality that $|J|\leq k$. The constants $c',d$ might depend on $k$, but not on $m,n$. Hence, we find our result \eqref{Knormmtimesn}.
\end{proof}

Now we come back to our example of local operators, given by the family of functions $F_{2k+1}$  defined in \eqref{eq:oddexamplePerm}. We define
\begin{equation} \label{eq:fmnex}
 f_{mn}(\thetav,\etav) := F_{m+n} (\thetav, \etav + i \piv - i \zerov)
\end{equation}
with $\thetav\in\rbb^m$, $\etav\in\rbb^n$, $m+n$ odd.

As for the function $g$ that appears in the definition of $F_{2k+1}$, we assume that its Fourier transform $\tilde g(\mu \cdot)$ fulfils the bounds in Eq.~\eqref{eq:gderivest}.

\begin{proposition}\label{proposition:oddfmnnorm}
 Let $f_{mn}$ be as in \eqref{eq:fmnex}, and $\omega(p)=p^\alpha$. Then, $\onorm{f_{mn}}{m \times n}<\infty$ for all $m,n$. If further Conjecture~\ref{conj:Tmbounds} is true, then there are constants $c,d>0$, depending on $n,\alpha$ but not on $m$, such that
\begin{equation}
   \onorm{f_{mn}(\thetav,\etav)}{m \times n} + \onorm{f_{nm}(\thetav,\etav)}{n \times m} \leq c \, d^m \, \sqrt{\gamma_\alpha(m)}.
\end{equation}
\end{proposition}

\begin{proof}
Recall from \eqref{eq:FoddWithT} and \eqref{eq:TmWithPairs} that
\begin{equation}
   f_{mn}(\thetav,\etav) = \tilde g(\mu E(\thetav)-\mu E (\etav)) \sum_{P \in \pcal_{m+n}}
  \sign P \prod_{(p,q)\in P} \tanh \frac{\zeta_{p}-\zeta_q}{2},
\end{equation}
where $\zeta_j = \theta_j$ for $j \leq m$, and $\zeta_{j}=\eta_{j-m}+i\pi$ for $j>m$. We reorganize the sum over pairings $P$ in the following way:
Fist we consider the number $\ell$ of pairs which contain one $\theta_i$ and one $\eta_j$, and we sum over $0\leq \ell\leq \min(m,n)$; then we sum over all the possible corresponding pairs at fixed $\ell$; finally we sum over all other pairings of $\theta_i$ with $\theta_j$ and $\eta_i$ with $\eta_j$. In this way, for every choice of such pairs $(p_1,q_1),\ldots, (p_\ell, q_\ell)$, there is a $\beta_{(p_1,q_1),\ldots,(p_\ell,q_\ell)} \in \{+1,-1 \}$, such that:
\begin{multline}\label{eq:fmnothersum}
   f_{mn}(\thetav,\etav) = \tilde g(\mu E(\thetav)-\mu E (\etav)) \sum_{\ell = 0}^{\min(m,n)}
  \sum_{\substack{ (p_1,q_1),\ldots,(p_\ell,q_\ell) \\ 1 \leq p_i \leq m < q_i \leq m+n}}
  \beta_{(p_1,q_1),\ldots,(p_\ell,q_\ell)}
\\ \times
  \Big(\prod_{j=1}^\ell \coth \frac{\theta_{p_j}-\eta_{q_j}+i0}{2}\Big)
  T_{m-\ell}(\hat \thetav) T_{n-\ell}(\hat \etav).
\end{multline}
Here, $\hat\thetav$ denotes all $\theta_j$ with $j$ not in the list $p_1,\ldots,p_\ell$, and $\hat\etav$ denotes all $\eta_j$ with $j$ not in the list $q_1,\ldots,q_\ell$.

To find such $\beta$, we can write the permutation in Eq.~\eqref{eq:signPdef} as the composition of three permutations given for $k$ even by:
\begin{align}
   \sigma_1 &=  \begin{pmatrix} 1 & \ldots & 2k  \\
         1 \hat{\ldots}  m  & m+1 \hat{\ldots} m+n & p_1,q_1\ldots p_\ell,q_\ell \end{pmatrix},
\\
   \sigma_2 &=  \begin{pmatrix}   1 \hat{\ldots}  m  & m+1 \hat{\ldots} m+n & p_1,q_1\ldots p_\ell,q_\ell\\
         p'_1,p''_1 \ldots  p'_{m-\ell},p''_{m-\ell}  & m+1 \hat{\ldots} m+n & p_1,q_1\ldots p_\ell,q_\ell \end{pmatrix},
\\
\label{eq:sigma3def}
   \sigma_3 &=  \begin{pmatrix}   p'_1,p''_1 \ldots  p'_{m-\ell},p''_{m-\ell}  & m+1 \hat{\ldots} m+n & p_1,q_1\ldots p_\ell,q_\ell\\
         p'_1,p''_1 \ldots  p'_{m-\ell},p''_{m-\ell}  & q'_1,q''_1 \ldots q'_{n-\ell},q''_{n-\ell} & p_1,q_1\ldots p_\ell,q_\ell \end{pmatrix}.
\end{align}
Note that the pairing $P$ is given by the last line of \eqref{eq:sigma3def}. Then $\sign P$ is therefore the product
\begin{equation}
 \sign P = \sign  \sigma_1 \cdot \sign\sigma_2 \cdot \sign\sigma_3
\end{equation}
where $\sign \sigma_1 = \beta$, $\sign \sigma_2$ and $\sign \sigma_3$ form the respective signum terms in $T_{m-\ell}$ and $T_{n-\ell}$, see \eqref{eq:TmWithPairs}. The case $k$ odd can be treated similarly.

We consider one term in the sum \eqref{eq:fmnothersum}, multiplied with the exponential $\exp(-E(\etav)^\alpha)$:
\begin{equation}\label{eq:oneterm}
 f(\thetav,\etav) :=  \tilde g(\mu E(\thetav)-\mu E (\etav)) \Big(\prod_{j=1}^\ell \coth \frac{\theta_{j}-\eta_{j}+i0}{2}\Big)
  T_{m-\ell}(\hat \thetav) T_{n-\ell}(\hat \etav) \exp(-E(\etav)^\alpha).
\end{equation}
We want to estimate the norm $\gnorm{f}{m \times n}$. The functions $T_{m-\ell}$ and $T_{n-\ell}$ are bounded and therefore Prop.~\ref{pro:kprototype} yields that $\gnorm{f}{m \times n}<\infty$ and hence $\onorm{f_{mn}}{m \times n}<\infty$. If the Conjecture~\ref{conj:Tmbounds} is true, then we have that  $\gnorm{T_{m-\ell}}{\infty} \leq 1$ and $\gnorm{T_{n-\ell}}{\infty} \leq 1$. Applying Proposition~\ref{pro:kprototype}, we have that $\gnorm{f}{m \times n} \leq c\,d^m\,\gamma_\alpha(m)^{1/2}$, where $d^n$ and $\gamma_\alpha(n)^{1/2}$ are absorbed into the constants $c,d$, that therefore might depend on $n$, $k$ and $\alpha$.

We note that every term of the sum \eqref{eq:fmnothersum} is of the form \eqref{eq:oneterm}, except for renumbering of the variables and $\pm$ signs. So we can apply the same estimate to each of this terms. We need only to count the number of summands in \eqref{eq:fmnothersum}; this can be estimated in the following way: we need to multiply the number $\binom{m}{\ell}$ of possibilities of setting $\ell$ indices within $m$, with the number $\binom{n}{\ell}$ of possibilities of setting $\ell$ indices within $n$, and with the number $\ell !$ of possible exchanges of the pairs. So,
there are at most $\binom{m}{\ell} \binom{n}{\ell} \ell! \leq 2^{m+n}\ell!$ different choices of pairs $(p_1,q_1),\ldots,(p_\ell,q_\ell)$; and $\ell$ takes at most $\ell+1\leq \min(m,n)+1 \leq n+1$ values.

Therefore, we have to multiply the estimate discussed above with such number. We absorb $2\cdot 2^{n}k!$ into some other constants $c',d'$ that hence might depend on $n$, $k$ and $\alpha$ (but not on $m$), we find
\begin{equation}
   \gnorm{f_{mn}(\thetav,\etav) \exp(-E(\etav)^\alpha)}{m \times n} \leq c' \, (d')^m \, \sqrt{\gamma_\alpha(m)}.
\end{equation}
In analogous way, we get
\begin{equation}\label{gnormfgamma}
 \gnorm{f_{mn}(\thetav,\etav) \exp(-E(\thetav)^\alpha)}{m \times n} \leq c'' \, (d'')^m \, \sqrt{\gamma_\alpha(m)}.
\end{equation}
Recalling the definition \eqref{eq:crossnorm2}, this gives the result in Prop.~\ref{proposition:oddfmnnorm}.
\end{proof}

Using our results of Sec.~\ref{sec:elemproperties} and Proposition~\ref{proposition:oddfmnnorm}, we can now prove the following corollary:
\begin{corollary}
$A$ is $\omega$-local in $\ocal_r$.
\end{corollary}
\begin{proof}
By Proposition~\ref{proposition:oddfmnnorm} we have that the functions $F_{2k+1}$ satisfy the condition \ref{it:fboundsreal} with the indicatrix $\omega(p)$ given by\footnote{In the case where $\omega$ grows slower, for example polynomially, this result does not seems to be true any more.} $\omega(p)=p^\alpha$, $0<\alpha<1$.

We have also shown in Sec.~\ref{sec:elemproperties} that the functions $F_{2k+1}$ satisfy the conditions \ref{it:fmero}, \ref{it:fsymm}, \ref{it:fperiod}, \ref{it:frecursion} and \ref{it:fboundsimag} of the theorem. Hence, by Theorem~\ref{thm:localequiv}\ref{thm:converse}, $A$ is $\omega$-local in $\ocal_r$.
\end{proof}

Due to Prop.~\ref{proposition:expansionunique} we have from Prop.~\ref{proposition:oddfmnnorm} that $A = \sum \int \frac{d\thetav d\etav}{m!n!} f_{mn}(\thetav,\etav) z^{\dagger m}(\thetav) z^n(\etav)$ is a well defined quadratic form in $\mathcal{Q}^{\omega}$. Moreover, applying Prop.~\ref{proposition:summable1} we can also show that $A$ extends to a closable operator affiliated with $\A(\ocal_r)$, as the following theorem shows:

\begin{theorem}\label{theorem:oddoperator}
  Let $1/3 < \alpha < 1$. Suppose that Conjecture~\ref{conj:Tmbounds} is true. With $f_{mn}$ as above, the quadratic form
\begin{equation}\label{eq:asumqf}
  A = \sum_{m,n=0}^\infty \int \frac{d\thetav d\etav}{m!n!} f_{mn}(\thetav,\etav) z^{\dagger m}(\thetav) z^n(\etav)
\end{equation}
extends to a closed operator which is affiliated with $\A(\ocal_r)$.
\end{theorem}

\begin{proof}
We want to show that the following series converges by using the estimate that we computed in Prop.~\ref{proposition:oddfmnnorm}:
\begin{equation}\label{eq:sumtoconverge}
   \sum_{m=0}^\infty \frac{2^{m/2}}{\sqrt{m!}} \Big( \onorm{ f_{mn} }{m \times n} + \onorm{ f_{nm}}{n \times m} \Big)
  \leq c \sum_{m=0}^\infty (\sqrt{2} d)^m \Big(\frac{\gamma_\alpha(m)}{m!}\Big)^{1/2}.
\end{equation}
Using that ($n\in \nbb$) $\Gamma(n)=(n-1)!$ and the Stirling's approximation, $n! \sim \sqrt{2\pi n}\Big( \frac{n}{e}\Big)^n$,  we find for large $m$,
\begin{multline}
  \frac{\gamma_\alpha(m)}{m!} = \frac{\Gamma(m/2 \alpha)}{\alpha m!\, \Gamma(m/2)}
  =\frac{\Big(  \frac{m}{2\alpha}-1\Big)!}{\alpha m! \Big( \frac{m}{2}-1 \Big)!}\\
  \sim \frac{\sqrt{2\pi \Big( \frac{m}{2\alpha}-1 \Big)}\Big(  \frac{\frac{m}{2\alpha}-1}{e} \Big)^{\frac{m}{2\alpha}-1}}{\alpha \sqrt{2\pi m}\Big( \frac{m}{e}\Big)^m \sqrt{2\pi( m/2 -1)}\Big( \frac{m/2 -1}{e} \Big)^{m/2 -1} }\\
 \sim \frac{1}{\sqrt{2\pi m \alpha}} \Big(\sqrt{2} e^{3/2} (2\alpha e)^{-1/2\alpha}\Big)^m  m^{\frac{m}{2}(1/\alpha - 3)}.
\end{multline}
In the case $\alpha > 1/3$, we have that the right hand side of  \eqref{eq:sumtoconverge} converges by application of the quotient criterion.
Therefore we can apply Prop.~\ref{proposition:summable1} and find that $A$ is closable; since moreover we have shown that $A$ is $\omega$-local (see the remark after Eq.~\eqref{gnormfgamma}), we can apply Prop.~\ref{proposition:locality}\ref{it:aising}, and conclude that the closure is affiliated with $\A(\ocal_r)$.
\end{proof}

\section{Local observables for general $S$}

We will indicate in this section how the construction of local operators discussed in Sec.~\ref{sec:evenexample} and Sec.~\ref{sec:oddexample} for $S=-1$, see Eq.~\eqref{eq:oddexamplePerm}, can be generalized to general $S$ matrices. In order to maintain the comparison with the case $S=-1$, we will restrict to the case $S(0)=-1$ rather than $S(0)=+1$.

Note that in the general case we do not expect ``Buchholz-Summers'' type of operators because the recursion relations force the family of functions $F_k$ to be infinite.

The main building block of \eqref{eq:oddexamplePerm} is the function
\begin{equation}
 H_{-1}(\zeta):=\tanh\frac{\zeta}{2},
\end{equation}
which has the properties
\begin{equation}
 H_{-1}(-\zeta)=-H_{-1}(-\zeta),\quad H_{-1}(\zeta +2i\pi)=H_{-1}(\zeta),\quad \res_{\zeta=-i\pi}H_{-1}(\zeta)=2.
\end{equation}
We propose to replace this with the $S$-dependent variant
\begin{equation}
H_{S}(\zeta)=\frac{e^{\zeta/2}+S(-\zeta)e^{-\zeta/2}}{e^{\zeta/2}+e^{-\zeta/2}},
\end{equation}
which has the properties ($S(i\pi)=S(0)=-1$)
\begin{equation}\label{HSproperties}
\begin{gathered}
H_S(-\zeta)=S(+\zeta)H_S(\zeta),\quad H_S(\zeta +2\pi i)=H_S(\zeta), \\ \res_{\zeta=-i\pi}H_S(\zeta)=(e^{-i\pi/2} +S(+i\pi)e^{+i\pi/2})\res_{\zeta= -i\pi}\frac{1}{e^{\zeta/2} + e^{-\zeta/2}}=2.
\end{gathered}
\end{equation}
Then we consider
\begin{equation}
T_{S,2k+1}(\zetav):=\frac{1}{2^k k!}\sum_{\sigma \in \perms{2k+1}}S^\sigma(\zetav)H_S(\zeta_{\sigma(2j-1)}-\zeta_{\sigma(2j)}).
\end{equation}
This function, according to Eq.~\eqref{HSproperties}, is $2\pi i$-periodic in each variable and $S$-symmetric (by construction). Similarly to the case $S=-1$, see Sec.~\ref{sec:elemproperties}, we can rewrite $T_{S,2k+1}$ in the following way:
\begin{equation}
T_{S,2k+1}(\zetav)= \frac{1}{k!}\sum_{P\in \pcal_{2k+1}^{\text{ordered}}}S^P(\zetav)\prod_{(\ell,r)\in P}H_S(\zeta_\ell -\zeta_r).
\end{equation}
Here $\pcal_{2k+1}^{\text{ordered}}$ denotes the set of all \emph{ordered} pairings of $2k+1$ indices and $S^P$ is the factor $S^\sigma$ for the permutation $\sigma$ corresponding to $P$ by Eq.~\eqref{eq:signPdef}.
We have the following lemma:
\begin{lemma}
$T_{S,2k+1}$ has the residua
\begin{equation}
\res_{\zeta_n -\zeta_m=i\pi}T_{S,2k+1}(\zetav)=-2\Big(\prod_{q=m}^{n}S(\zeta_q - \zeta_m)\Big) T_{S,2k-1}(\hat \zetav).
\end{equation}
\end{lemma}
\begin{proof}
To prove this, we first compute the residue of $T_{S,2k+1}$ at $\zeta_2 -\zeta_1=i\pi$:
\begin{equation}
\res_{\zeta_2 -\zeta_1=i\pi}T_{S,2k+1}(\zetav)= \frac{1}{k!} \sum_{\substack{P \in \pcal_{2k+1}^{\text{ordered}}\\ (1,2)\in P}}S^{P}(\zetav)\res_{\zeta_2- \zeta_1 = i\pi}H_S(\zeta_1 -\zeta_2)\prod_{(\ell,r)\in P'}H_S(\zeta_\ell -\zeta_r),
\end{equation}
where $P'$ denotes the pairing $P$ with the pair $(1,2)$ left out.

We consider the permutation
\begin{equation}
\sigma_P =  \begin{pmatrix} 1 & 2 & & & &  \ldots & & & & 2k +1  \\
         1 & 2 & \ell_1 & r_1 & \ell_2 & r_2 & \ldots & \widehat{1 \;\; 2} & \ldots & \hat m \\
         \ell_1 & r_1 & \ell_2 & r_2 & & \ldots & 1 \;\; 2 & & \ldots & \hat m
         \end{pmatrix}.
\end{equation}
 We call the permutation from the first to the second line $\sigma_{P'}$, and the permutation from the second to the third line $\tau$. We have $\sigma_P = \sigma_{P'}\circ \tau$ and, correspondingly, $S^{\sigma_P}(\zetav)= S^{\sigma_{P'}}(\zetav)S^{\tau}(\zetav^{\sigma_{P'}})=S^{\sigma_{P'}}(\hat \zetav)S^{\tau}(\zetav^{\sigma_{P'}})$. Note that $S^{\tau}(\zetav^{\sigma_{P'}})=1$ on the hypersurface $\zeta_2 -\zeta_1=i\pi$. We have $S^P(\zetav)=S^{\sigma_P}(\zetav)$ and $S^{P'}(\hat \zetav) =S^{\sigma_{P'}}(\hat \zetav)$. Further, we note that the sum over $P$ contains every $P'$ exactly $k$ times. Therefore after renumbering the components of $\zetav$, we can change the sum to $\smash{\sum_{\substack{P \in \pcal_{2k+1}^{\text{ordered}}\\ (1,2)\in P}}=k \sum_{P' \in \pcal_{2k-1}^{\text{ordered}}}}$. Hence, we arrive at
 \begin{equation}
 \res_{\zeta_2 -\zeta_1=i\pi}T_{S,2k+1}(\zetav)
= -\frac{2}{(k-1)!}\sum_{P' \in \pcal_{2k-1}^{\text{ordered}}}S^{P'}(\hat \zetav)\prod_{(\ell,r)\in P'}H_S(\hat{\zeta}_\ell -\hat{\zeta}_r)=-2 T_{S,2k-1}(\hat \zetav).
\end{equation} 

For the residue at $\zeta_n -\zeta_m =i\pi$ ($m<n$) we find, using $S$-symmetry,
\begin{multline}
\res_{\zeta_n -\zeta_m=i\pi}T_{S,2k+1}(\zetav)=\res_{\zeta_n -\zeta_m=i\pi} \prod_{i=1}^{m-1}S(\zeta_m - \zeta_i)\prod_{\substack{q=1\\q \neq m}}^{n-1}S(\zeta_n -\zeta_q)T_{S,2k+1}(\zeta_m, \zeta_n,\hat \zetav)\\
=\Big(\prod_{i=1}^{m-1}S(\zeta_m - \zeta_i)\prod_{ \substack{q=1 \\ q \neq m}}^{n-1}S(\zeta_n -\zeta_q)\Big)\big\vert_{\zeta_n -\zeta_m=i\pi}\cdot \res_{\zeta_n -\zeta_m=i\pi} T_{S,2k+1}(\zeta_m, \zeta_n,\hat \zetav)\\
=-2\Big(\prod_{q=m}^{n}S(\zeta_q - \zeta_m)\Big) T_{S,2k-1}(\hat \zetav).
\end{multline}
\end{proof}
$T_{S,2k+1}$ is however not $S$-periodic. We propose to fix this problem with an extra factor $M_{S,2k+1}$ which fulfils the following properties:
\begin{align}\label{MScondition1}
M_{S,2k+1}(\zetav^\sigma) &= M_{S,2k+1}(\zetav),\\
\label{MScondition2}
M_{S,2k+1}(\zetav +2\pi i \ev^{(j)}) &= \Big( \prod_{i\neq j}S(\zeta_i -\zeta_j) \Big)M_{S,2k+1}(\zetav),\\
\label{MScondition3}
M_{S,2k+1}(\zetav)\big\vert_{\zeta_n -\zeta_m =i\pi}&= M_{S,2k-1}(\hat\zetav)\cdot \frac{1}{2}\Big( 1-\prod_{j=1}^{2k+1}S(\zeta_m -\zeta_j) \Big)
\end{align}
for all $\sigma \in \perms{2k+1}$ and all $\zetav$.

We will make a remark on the existence of such functions below. Given $M_{S,2k+1}(\zetav)$ with the above properties, we can set with a suitable localized test function $g$,
\begin{equation}
F_{2k+1}(\zetav)=\frac{1}{(2\pi i)^k}M_{S,2k+1}(\zetav)T_{S,2k+1}(\zetav)\tilde{g}(\mu E(\zetav)),
\end{equation}
and this will fulfil all conditions \ref{it:fmero}, \ref{it:fsymm}, \ref{it:fperiod}, \ref{it:frecursion}.
In particular, we find
\begin{multline}
\res_{\zeta_n -\zeta_m =i\pi}F_{2k+1}(\zetav)=\frac{1}{(2\pi i)^k}M_{S,2k-1}(\hat\zetav)\frac{1}{2}\Big(1-\prod_{p=1}^{2k+1}S(\zeta_m -\zeta_p)\Big)\tilde{g}(\mu E(\hat\zetav))\times\\
\times (-2)\Big(\prod_{q=m}^{n}S(\zeta_q - \zeta_m)\Big)T_{S,2k-1}(\hat \zetav)\\
= -\frac{1}{2\pi i}\Big(\prod_{q=m}^{n}S(\zeta_q - \zeta_m)\Big)\Big(1-\prod_{p=1}^{2k+1}S(\zeta_m -\zeta_p)\Big)F_{2k-1}(\hat\zetav).
\end{multline}
Regarding the functions $M_{S,2k+1}$, it would of course be important to construct them explicitly, but we have not found an explicit solution yet. Nevertheless, we can show that such functions exist and therefore that the conditions \eqref{MScondition1}--\eqref{MScondition3} are compatible. Namely, using the results of Lechner \cite{Lechner:2008} and our characterization theorem, Thm.~\ref{thm:localequiv}, we know for a large class of functions $S$ that there exist functions $F_{2k+1}^{\text{Lechner}}$ which fulfil all conditions (F). Given these, we set
\begin{equation}
M_{S,2k+1}(\zetav):= (2\pi i)^k\frac{F_{2k+1}^{\text{Lechner}}(\zetav)}{T_{S,2k+1}(\zetav)}.
\end{equation}
Since the poles of the numerator and of the denominator cancel, and since they are both $S$-symmetric, these $M_{S,2k+1}$ will have the properties \eqref{MScondition1}--\eqref{MScondition3}. Of course, this does not solve completely the construction problem, but it shows that our approach is consistent.

However the main challenge in constructing interacting operators for general $S$ is in verifying the various bounds conditions.

The bounds \ref{it:fboundsimag} should be easy to verify with similar methods as in Sec.~\ref{sec:oddexample}, since they essentially depend on the growth of $\tilde g$ except for slower growing terms.

More difficult is verifying the bounds \ref{it:fboundsreal}, i.e. the question whether
\begin{equation}
\gnorm{F_{2k+1}(\thetav, \etav +i\piv)}{m\times n} < \infty,
\end{equation}
and even more whether the summability conditions of Prop.~\ref{proposition:summable2} are fulfilled. To that end, the estimates of Sec.~\ref{sec:operatorbounddom} need to be generalized and improved. In particular, these are more difficult to show than the hypothesis of Prop.~\ref{proposition:summable1} which we used for the case $S=-1$; the extra condition \eqref{eq:weyldomaincond} needs to be taken into account. In other words, one needs to track the dependence of the constant $c'$ on $k$ in Prop.~\ref{pro:kprototype} well enough to prove summability, which requires a big improvement.

%% file: conclusions.tex
\chapter{Conclusions and outlook}\label{sec:conclusion}

We have established existence and uniqueness of the series expansion \eqref{eq:expansionz} for any quadratic form $A$ in two-dimensional factorizing scattering models. We have given an explicit expression for the expansion coefficients $\cme{m,n}{A}$ in terms of matrix elements of $A$, and analysed their properties (independent of locality) with respect to spacetime symmetry transformation of $A$; of particular interest are the spacetime reflections, which also play an important role in the study of local observables in bounded regions, see Sec~\ref{sec:fshifted}. We discussed how to generalize the expression \eqref{eq:coefffree} for the expansion coefficients in terms of a string of nested commutators, valid in the free field theory, to the factorizing scattering models described by \cite{GrosseLechner:2007}, by defining a ``deformed commutator'' with the notion of warped convolution \cite{BuchholzSummersLechner:2011}.

We investigated the necessary and sufficient conditions on the coefficients $\cme{m,n}{A}$ that make a quadratic form $A$ of a certain ``regularity class'' \emph{$\omega$-local} in a bounded spacetime region (see definitions in Sections~\ref{sec:qforms} and \ref{sec:weaklocality}). These are in particular analyticity properties of the coefficients $\cme{m,n}{A}$, and bounds for their analytic continuation.

Extra conditions on the summability of certain $\omega$-norms of $\cme{m,n}{A}$ (see Sec.~\ref{sec:zgen} for definitions) will imply the extension of the quadratic form to a closed, possibly unbounded, operator.

Further, we showed that a family of functions $F_k$ which satisfies the conditions Def.~\ref{def:conditionF} for the characterization of the $\omega$-local quadratic forms, and the condition of Prop.~\ref{proposition:summable1} for the extension of the quadratic form to a closed operator, inserted in \eqref{ArakiF}, yields an operator \emph{affiliated} with the local algebra of bounded operators, see Prop.~\ref{proposition:locality}.

Finally, we used these conditions to construct concrete examples of local observables in the case $S=-1$ in Chapter~\ref{sec:localexamples}.

This construction applies to two dimensional scattering models with particle spectrum described by one kind of particle, scalar, massive and without charge. However, one can generalize it to models with a richer particle spectrum, see for example \cite{LechnerSchuetzenhofer:2012}. This would give a more formal complication to the general setting, for example the scattering function would be a matrix-valued function, rather than scalar-valued; but the expansion \eqref{eq:expansionz} and the characterization of locality of $A$ in terms of properties of the expansion coefficients would remain essentially the same.

We have shown that the expansion \eqref{eq:expansionz} is related to the deformation methods applied in quantum field theory, and in particular to the notion of warped convolution, see for example \cite{GrosseLechner:2007,Lechner:2011}, \cite{BuchholzSummers:2008,BuchholzSummersLechner:2011}. These methods can be applied to any theories with arbitrary spacetime dimensions with the purpose of constructing interacting models of Buchholz-Summers type \cite{BuchholzSummers:2008,BuchholzSummersLechner:2011}. This would suggest the possibility to generalize the expansion \eqref{eq:expansionz} and our analysis to arbitrary spacetime dimensions using techniques of Appendix~\ref{sec:warped}. On a formal level, we would get an expansion in a basis which depends on a deformation parameter $Q$. However, in higher dimensions there is a much larger choice for $Q$ with each $Q$ corresponding to a wedge-region in Minkowski space. One could think to follow the same characterization programme as for $1+1$ dimensional models in theories with higher spacetime dimensions. In higher dimensions double cones are intersection of more than two, in fact infinitely many, wedges, and each of these wedge localizations would give an analyticity condition on the coefficients $\cme{m,n}{A}$ with details to be determined. However, in line with the expectations of the authors
\cite{BuchholzSummers:2007,BuchholzSummersLechner:2011}, one will possibly find that these conditions on the holomorphic functions are so strong that they are fulfilled only by constant functions, and therefore the set of local observables contains only multiples of the identity operator. This may lead to a \emph{no-go} theorem in the class of models described by \cite{BuchholzSummersLechner:2011}.

The expansion \eqref{eq:expansionz} is not only useful for the characterization of local operators in two dimensional factorizing scattering models, but also to analyse the pointlike field structure of these theories, using techniques as in \cite{Bos:short_distance_structure}. Here one would consider the expansion \eqref{eq:expansionz} and write the coefficient functions $\cme{m,n}{A}$ in a series expansion which is adapted to the short distance limit:
\begin{equation}
\cme{m,n}{A}(\thetav,\etav)=\sum_{k=0}^{\infty}c_{kmn}^{[A]}g_{mnk}(\thetav,\etav).\label{coeffpointlike}
\end{equation}
In the free field case, this is a Taylor expansion in momentum variables and the $g_{mnk}$ are polynomials in momentum components, or, in other words, hyperbolic polynomials in rapidity space.

In the factorizing scattering models, $g_{mnk}$ need to be chosen so that they fulfil conditions similar to Def.~\ref{def:conditionF} for radius $r=0$, details regarding the bounds need to be determined: $3+1$ dimensional free field theory can serve as a guidance, since the $g_{mnk}$ are explicitly known there \cite{Bos:short_distance_structure,BostelmannDAntoniMorsella:2010}. In the case $S=-1$, an example of one basis element can be found from Sec.~\ref{sec:oddexample} by formally setting $\tilde g = 1$:
\begin{equation}
h_{2\ell +1}(\zetav)=\frac{1}{(4 \pi i)^{\ell} \ell !}\sum_{\sigma \in \perms{2\ell +1}} \sign \sigma \prod_{j=1}^{\ell} \tanh \frac{ \zeta_{\sigma(2j-1)} - \zeta_{\sigma(2j)}}{2}, \quad \ell \in \nbb_0,
\end{equation}
and setting
\begin{equation}
g_{mn,1}(\thetav,\etav)=\begin{cases}
  h_{m+n}(\thetav,\etav+i\piv) \quad & \text{if $m+n$ is odd},
  \\ 0 & \text{if $m+n$ is even}.
  \end{cases}
\end{equation}
Inserting \eqref{coeffpointlike} in the expansion \eqref{eq:expansionz}, one finds
\begin{equation}
A=\sum_{m,n,k}\frac{1}{m!n!}c_{kmn}^{[A]}\int d\thetav d\etav\, g_{mnk}(\thetav,\etav){\zd}^{m}(\thetav)z^{n}(\etav).
\end{equation}
After reorganizing the sum and the terms in the expression above, one should arrive at an expansion of the form
\begin{equation}
A= \sum_{\ell,k}{c'}^{[A]}_{k,\ell}\phi_{k,\ell}(0).
\end{equation}
Here $\phi_{k,\ell}$ are pointlike localized objects at the origin as a consequence of the analyticity conditions fulfilled by $g_{mnk}$. Note also that $\phi_{k,\ell}$ are independent of $A$.

In this way, one would be able to determine all the interacting pointlike fields of the theory, and would have shown that every operator $A$ can be expanded in their terms. In this sense, all the local observables would be known up to approximation.

Further, note that we can apply the expansion \eqref{eq:expansionz} only to theories where the scattering function $S$ has no poles on the physical strip; in particular, the results of Lechner \cite{Lechner:2008} are valid only in this situation. However, the operator expansion as such (Thm.~\ref{theorem:arakiexpansion}) does not require an analytic continuation of $S$, and therefore can in principle be extended beyond theories where the scattering function is restricted by this condition. For these theories, on the other hand, the Hilbert space as defined in
Sec.~\ref{sec:hilbertspace} is not suitable to allow local operators. So, the Hilbert space needs to be extended in order to include extra states, so called ``bound states'', and Thm.~\ref{theorem:arakiexpansion} needs to be generalized to this extended Hilbert space.

Certainly, an important task in our programme would be to exhibit a concrete example of local operator for a general scattering function $S$. We have presented in Chapter~\ref{sec:localexamples} an approach for finding functions $F_k$ which might yield an example of this type, without having verified all the conditions (F) established in Def.~\ref{def:conditionF} and Prop.~\ref{proposition:summable1}. These $F_k$ are derived as a natural generalization of the examples for $S=-1$ discussed in Chapter~\ref{sec:localexamples}. A further significant step would be to complete the proof of the conditions (F) and to show closability in the general case.

Finally, we would like to emphasize one important message of this thesis: namely that the examples in Chapter~\ref{sec:localexamples} suggests that the expansion Eq.~\eqref{eq:expansionz} in terms of wedge-local objects is more efficient than in terms of local asymptotic free fields (the usual form factor program); in our examples bounds on the coefficients can be exploited to ensure convergence of the expansion series and to establish ($\omega$-) locality, avoiding uncontrolled infinite sums as in usual FFP.

%% file: warped.tex
\chapter{Warped convolution}\label{sec:warped}

In the free field theory, corresponding to the case $S=1$, and where the Zamolodchikov operators $\zd$ and $z$ are the usual Bose annihilation and creation operators $\ad$ and $a$, we can express explicitly the coefficients $\cme{m,n}{A}$ of the Araki expansion in terms of a string of nested commutators, namely as:
\begin{equation}\label{eq:fmnfree}
 \cme{m,n}{A}(\thetav,\etav) = \bighscalar{\Omega}{[\ldots [a(\theta_1),[ \ldots [a(\theta_m),A],\ad(\eta_n)]\ldots ,\ad(\eta_1)] \Omega  }.
\end{equation}
We can verify this formula by direct computation in the case $A=a^{\dagger m'} a^{n'}(f)$, by using repeatedly the relations of the CCR algebra. Then, we have that this formula holds for all quadratic forms $A$ by expressing the quadratic forms with the Araki expansion and by using linearity.

We can extend this kind of formula \eqref{eq:fmnfree} to other examples. For example, in the case of the Ising model, corresponding to the case $S=-1$, and where the Zamolodchikov creation and annihilation operators fulfil the CAR algebra, we can define a graded commutator $[\cdotarg, \cdotarg]_g$; note that the graded commutator between even operators is equal to the commutator and between odd operators is equal to the anticommutator (where even and odd operators are defined with respect to the adjoint action of $(-1)^N$). Then, using this  graded commutator, we can write the Araki coefficients in analogous way as in \eqref{eq:fmnfree}, using the following formula:
\begin{equation}\label{eq:fmnising}
 \cme{m,n}{A}(\thetav,\etav) = \bighscalar{\Omega}{[\ldots [z(\theta_1),[ \ldots [z(\theta_m),A]_g,\zd(\eta_n)]_g\ldots ,\zd(\eta_1)]_g \Omega  },
\end{equation}
As before, we can prove this formula by computing explicitly the commutators in the case $A=z^{\dagger m'} z^{n'}(f)$, using repeatedly the relations of the CAR algebra. Then, we can extend this formula to all quadratic forms $A$ by using the Araki expansion and by using linearity.

We would like to generalize this kind of expressions for the Araki coefficients to more general models. Here we will try to make this generalization to the family of models obtained by Buchholz, Summers and Lechner in \cite{BuchholzSummers:2008,BuchholzSummersLechner:2011}, using the \emph{warped convolution} construction.

In this construction, one starts from a given quantum field theory and deforms the algebras of observables, and in this way constructs a new theory.
This deformation uses as a deformation parameter a skew symmetric matrix $Q$; this deformation is equivalent to a Rieffel deformation \cite{Rieffel:1993} with respect to the action of the translation group (see \cite[Lemma 2.1(i),Eq.~(2.2)]{BuchholzSummersLechner:2011}); moreover, we can alternatively interpret the deformed theory in terms of a quantum field theory on noncommutative space-time \cite{GrosseLechner:2007}. Buchholz, Summers and Lechner in \cite{BuchholzSummersLechner:2011} wanted to apply this deformation to a general and possibly interacting quantum field theory, in particular in 2+1 and in more space-time dimensions. However, in our case, we start from a 1+1 dimensional free field theory. Then we know from~\cite{GrosseLechner:2007} that the deformed theory that one obtains is equivalent to an integrable model with a certain simple type of scattering function $S$. We will see explicitly below this equivalence.

We want to define for this particular class of models, a ``deformed commutator'' $[\cdotarg, \cdotarg]_Q$, so that we can write an analogue of the formula \eqref{eq:fmnfree} also in the case of this class of models.

First, we introduce some notation and preliminaries. In this section, $\hcal$ and related spaces are associated with the free field $S=1$. We also consider only the case where $\omega=0$ and hence we drop the superscript $\omega$ from all objects that we will consider.

Now, following the conventions in \cite{BuchholzSummersLechner:2011}, we introduce some spaces of ``smooth'' operators and quadratic forms. We start with the operators of $\boundedops$ and we consider $x \mapsto A(x)$ using the adjoint action of translations. Then, we introduce the seminorms $A \mapsto \gnorm{ \partial^\kappa A}{}$, where $\partial^\kappa$ with a multi-index $\kappa$ are partial derivatives with respect to the action of space-time translations. We call $\ccal^\infty$ the subalgebra of $\boundedops$ consisting of ``norm-smooth'' operators, namely operators such that $\gnorm{ \partial^\kappa A}{} < \infty$ for all multi-indices $\kappa$. We equip $\ccal^\infty$ with the usual Fr\'echet topology, namely the topology given by the seminorms $A \mapsto \gnorm{ \partial^\kappa A}{}$.

Correspondingly, we also consider the space $\qf^\infty \subset \qf$ of quadratic forms $A$ that fulfil $\fpnp_k A \fpnp_k \in \ccal^\infty$ for all $k$.

We also consider the subspace $\fcal^\infty$ of $\qf^\infty$ defined as follows: For $A \in \fcal^\infty$ and any $k \in \nbb$, there exists $k' \in \nbb$ so that $\fpnp_{k'} A \fpnp_k=A \fpnp_k$ and  $\fpnp_{k} A \fpnp_{k'}=\fpnp_{k} A$. An equivalent way to formulate this definition is to say: For $A \in \fcal^\infty$ and any $k \in \nbb$, there exists $k' \in \nbb$ so that $A \fpnp_k \in \ccal^\infty$, $A\st \fpnp_k \in \ccal^\infty$, $A \fpnp_k\hcal \subset \fpnp_{k'}\hcal$, and $A\st \fpnp_k\hcal \subset\fpnp_{k'} \hcal$.
Then we say that $\qf^\infty$ is a bimodule over $\fcal^\infty$. We note that $\fcal^\infty \subset \qf^\infty$, $\ccal^\infty\subset\qf^\infty$, but $\ccal^\infty \not \subset \fcal^\infty$.

Moreover, we consider the so called $\fcal^\infty$-valued distributions on $\rbb^m$: They are linear maps $\dcal(\rbb^m) \to \fcal^\infty$, $f \mapsto A(f)$, such that for any $k$, the number $k'$ above can be chosen independent of $f$, and such that the maps $f \mapsto A(f) \fpnp_k$ and $f \mapsto \fpnp_k  A(f)$ are continuous in the Fr\'echet topology. We also have that products of $\fcal^\infty$-valued distributions in independent variables are again $\fcal^\infty$-valued distributions. We will write as usual these distributions in terms of their formal kernels, $A(f) = \int A(\thetav) f(\thetav) d\thetav$.

Considering the (anti)unitary representation $U$ of the proper Poincar\'e group, we want to show that if an operator $A$ is an element of $\ccal^\infty$, $\qf^\infty$, $\fcal^\infty$, and of $\fcal^\infty$-valued distributions, then also the operator transformed by the adjoint action of $U$, $UAU^\ast$, is an element of these spaces, respectively.

For $\ccal^\infty$: The first order derivative of $UAU^\ast$ is given by $[P_\mu,UAU^\ast]=U[P_\mu,A]U^\ast$, where $P_\mu$, $\mu=0,1$, is the momentum operator. Then $\gnorm{[P_\mu,UAU^\ast]}{}=\gnorm{U[P_\mu,A]U^\ast}{}=\gnorm{\partial A}{}<\infty$. Similarly, the second order derivative of $UAU^\ast$ is given by the multi-commutator $U[P_\mu ,[P_\kappa , A]]U^\ast$; calling $B:=[P_\kappa , A]$, we can use the result before and conclude that $\gnorm{\partial^2 (UAU^\ast)}{}<\infty$. The same apply to higher order derivatives of $UAU^\ast$.

For $\qf^\infty$: $U A U^\ast \in \qf^\infty$ if and only if we can show that $Q_k U A U^\ast Q_k \in \ccal^\infty$. But this follows from the fact that $Q_k U A U^\ast Q_k = U Q_k A Q_k U^\ast$ (uses that $U$ commutes with the particle number operator), and the fact that $Q_k A Q_k \in \ccal^\infty$ (since $A \in \qf^\infty$).

For $\fcal^\infty$: Since $A\in \fcal^\infty$, then $AQ_k \in \ccal^\infty$; this implies $UAU^\ast Q_k \in \ccal^\infty$, since $U$ commutes with the particle number operator. Hence, $UAU^\ast \in \fcal^\infty$.

For $\fcal^\infty$-valued distributions: It was proved before that $UA(f)U^\ast \in \fcal^\infty$. It remains to show that the map $f \mapsto UA(f)U^\ast Q_k = UA(f)Q_kU^\ast$ (uses that $U$ commutes with the particle number operator) is continuous. For this, we consider the map $f \mapsto A(f)Q_k \mapsto UA(f)Q_k U^\ast$. Since $A$ is a $\fcal^\infty$-valued distribution, then $f \mapsto A(f)Q_k$ is continuous; call $B:=A(f)Q_k$, we have $\gnorm{\partial^\mu (U B U^\ast)}{}= \gnorm{U(\partial^\mu B)U^\ast}{}=\gnorm{\partial^\mu B}{}$. This implies that the map $B \mapsto UBU^\ast$ is continuous. Therefore, $f \mapsto UA(f)U^\ast Q_k$ is a continuous map. Similarly, we can show that the map $f \mapsto Q_kUA(f)U^\ast$ is also continuous. Therefore, $UAU^\ast$ is a $\fcal^\infty$-valued distribution.

In particular here we are interested in the action of the translations operators $U(x):=U(x,0)$.

We say that an $\fcal^\infty$-valued distribution $A$ is \emph{homogeneous} if there is a smooth function $\varphi_A:\rbb^m \to \rbb^2$ such that
\begin{equation}\label{homogeneous}
\forall x \in \rbb^2: \quad U(x)A(\thetav)U(x)\st=e^{i\varphi_A(\thetav)\cdot x}A(\thetav).
\end{equation}
We call $\varphi_A$ the \emph{momentum transfer} of $A$. If $A(\thetav),B(\etav)$ are both homogeneous, then also $A(\thetav)B(\etav)$ is homogeneous, and has momentum transfer $\varphi_{AB}(\thetav,\etav)=\varphi_A(\thetav)+\varphi_B(\etav)$.
There are some important examples of homogeneous distributions: $\ad(\theta)$, $a(\eta)$, and $a^{\dagger m}a^{n}(\thetav,\etav)$, which have momentum transfer $p(\theta)$, $-p(\eta)$, and $p(\thetav)-p(\etav)$, respectively; other examples are their deformed versions, that we will consider below.

Now we introduce the warped convolution. We denote with $dE(p)$ the (joint) spectral measure of the momentum operator, and we denote with $Q$ a skew symmetric $2\times 2$ matrix. The warped convolution $\tau_Q$ of an operator $A$ is defined by
\begin{equation}
\tau_{Q}(A) := \int U(Qp) A U(Qp)\st\;dE(p)=\int dE(p) U(Qp) A U(Qp)\st.
\end{equation}
Note that we must take this integral with care, since the integrand has constant norm. However, Buchholz, Lechner and Summers managed in \cite{BuchholzSummersLechner:2011} to define it in the case where $A$ are smooth operators and in the sense of an oscillatory integral, and to give a bijective map $\tau_Q:\ccal^\infty \to \ccal^\infty$.

We need to extend this map to our space of quadratic forms, and in order to obtain this we will use the projectors $\fpnp_k$.

Let $A \in \ccal^\infty$, since the $\fpnp_k$ commute with $U(x)$, we can write,
\begin{equation}\label{eq:qkcommute}
   \tau_Q(A \fpnp_k) = \tau_Q(A) \fpnp_k, \quad \tau_Q(\fpnp_k A ) = \fpnp_k \tau_Q(A).
\end{equation}
Using this, we can extend $\tau_Q$ to quadratic forms $A \in \qcal^\infty$: Let $\psi,\chi\in\fpn$ we define,
\begin{equation}\label{eq:qkcompat}
   \hscalar{\psi}{ \tau_Q(A) \chi} := \hscalar{\psi}{ \tau_Q(\fpnp_k A \fpnp_k) \chi}
\end{equation}
 where $k$ is chosen large enough for $\psi,\chi$. Indeed, for $k$ large, the expression on the right hand side becomes independent of $k$: If $\psi, \chi \in Q_m \hcal$ and if $k\geq m$, then
\begin{equation}\label{kindepm}
\hscalar{\psi}{ \tau_Q(\fpnp_k A \fpnp_k) \chi}=\hscalar{\psi}{ \fpnp_m\tau_Q(\fpnp_k A \fpnp_k)\fpnp_m \chi}=\hscalar{\psi}{ \tau_Q(\fpnp_m A \fpnp_m) \chi},
\end{equation}
where in the second equality we applied \eqref{eq:qkcommute} since the operator $\fpnp_k A \fpnp_k$ is bounded; moreover we used that if $k\geq m$, then $\fpnp_m \fpnp_k = \fpnp_m$.

Now we want to show that the relations \eqref{eq:qkcommute}, which hold for operators on $\ccal^\infty$, hold also for all $A \in \qcal^\infty$: For $A\in \qf^\infty$, $k \in \mathbb{N}$, the right hand side of the first relation in \eqref{eq:qkcommute} gives
\begin{equation}
\hscalar{\psi}{ \tau_Q(A) Q_k \chi} = \hscalar{\psi}{ \tau_Q(Q_\ell A Q_\ell) Q_k \chi}=  \hscalar{\psi}{ \tau_Q(Q_\ell A Q_k Q_\ell) \chi}=\hscalar{\psi}{ \tau_Q(A Q_k ) \chi},
\end{equation}
where in the first equality we made use of \eqref{eq:qkcompat} with $\ell$ large; in the second equality we applied \eqref{eq:qkcommute} since the operator $Q_\ell A Q_\ell$ is bounded. In the third equality we used again \eqref{eq:qkcompat} with $\ell$ large. Analogously for the second relation in \eqref{eq:qkcommute}.

Note that $A Q_k \in \qf^\infty$ because $Q_k \in \fcal^\infty$ and $A\in \qf^\infty$. Indeed, in general one has for $A\in \qf^\infty$ and $B\in \fcal^\infty$, that $AB \in \qf^\infty$, $BA \in \qf^\infty$ (i.e. $\qf^\infty$ is a bimodule over $\fcal^\infty$). The proof of this statement works as follows: $AB \in \qf^\infty$ if we can show that $Q_k AB Q_k \in \ccal^\infty$. Since $B \in \fcal^\infty$, then $Q_k AB Q_k = Q_k A Q_{k'}BQ_k$, where $Q_k A Q_{k'}\in \ccal^\infty$ and $BQ_k \in \ccal^\infty$. Now, it was already shown in \cite{BuchholzSummersLechner:2011} that for $C,D \in \ccal^\infty$, then $CD \in \ccal^\infty$. Analogously for the product $BA$.

We present the most important properties of the map $\tau_Q$ in the following proposition, which is mostly due to H.Bostelmann:

\begin{proposition}\label{proposition:tauq}
  For any skew symmetric matrices $Q,Q'$, we have:

\begin{enumerate}
  \renewcommand{\theenumi}{(\roman{enumi})}
  \renewcommand{\labelenumi}{\theenumi}
  \item\label{it:warpedcontinuous}
    $\tau_Q: \ccal^\infty \to \ccal^\infty$ is continuous.
  \item\label{it:warpedcompose}
    $\tau_Q \tau_{Q'} = \tau_{Q+Q'}$, $\tau_0 = \id$, $\tau_Q^{-1} = \tau_{-Q}$.
  \item\label{it:warpedtranslate}
    $\tau_Q(U(x) A) = U(x) \tau_Q(A)$, $\tau_Q(A U(x)) = \tau_Q(A)U(x)$ for any $x \in \rbb^2$ and $A \in \qf^\infty$.
  \item\label{it:warpediso}
    $\tau_Q: \ccal^\infty \to \ccal^\infty$, $\tau_Q:\fcal^\infty\to\fcal^\infty$, $\tau_Q:\qcal^\infty \to \qcal^\infty$ are $\ast$-preserving vector space isomorphisms.
  \item\label{it:warpeddistribution}
    If $A$ is an $\fcal^\infty$-valued distribution, then $\tau_Q(A) : f \mapsto \tau_Q(A(f))$ is an $\fcal^\infty$-valued distribution as well.
    If $A$ is homogeneous, then so is $\tau_Q(A)$, with the same momentum transfer as $A$.
\end{enumerate}
\end{proposition}

\begin{proof}

For part~\ref{it:warpedcontinuous}: this part can be proved similarly to \cite[Prop.~2.7(ii)]{BuchholzSummersLechner:2011}: in their notation, they considered the inclusion $\imath$ of $\ccal^\infty$, equipped with the Fr\'echet topology induced by $\gnorm{\cdotarg}{}$, into itself equipped with the Fr\'echet topology induced by $\gnorm{\cdotarg}{Q}$; it was shown in \cite[Lemma~7.2]{Rieffel:1993} that this map is continuous. In turn, they showed that the map $\pi_Q$ is norm-preserving between $\gnorm{\cdotarg}{Q}$ and $\gnorm{\cdotarg}{}$ and hence it also intertwines the associated Fr\'echet topologies. In our notation, we have that $\tau_Q = \pi_Q \circ \imath$; so from the properties of the maps $\pi_Q$ and $\imath$, it follows that the map $\tau_Q$ is continuous.

For part~\ref{it:warpedcompose}: the relation $\tau_Q \tau_{Q'} = \tau_{Q+Q'}$ was shown for $A \in \ccal^\infty$ in \cite[Prop.~2.11]{BuchholzSummersLechner:2011}. To extend it to $\qcal^\infty$, we need to show for $A \in \qcal^\infty$, that
\begin{equation}
\langle \psi, \tau_Q \tau_{Q'}(A)\chi \rangle= \langle \psi, \tau_{Q + Q'}(A)\chi \rangle.
\end{equation}
The left hand side gives:
\begin{eqnarray}
\langle \psi, \tau_Q(\tau_{Q'}(A))\chi \rangle &=& \langle \psi, \tau_Q(Q_k \tau_{Q'}(A)Q_k)\chi \rangle\nonumber\\
&=&\langle \psi, \tau_Q (\tau_{Q'}( Q_k A Q_k)) \chi \rangle \nonumber\\
&=& \langle \psi, \tau_{Q+Q'}(Q_k A Q_k)\chi \rangle \nonumber\\
&=& \langle \psi, \tau_{Q+Q'}(A) \chi \rangle,
\end{eqnarray}
where in the first equality we made use of \eqref{eq:qkcompat}, where in the second equality we used \eqref{eq:qkcommute}, where in the third equality we applied \cite[Prop.~2.11]{BuchholzSummersLechner:2011} since $Q_k A Q_k \in \ccal^\infty$. In the fourth equality we made use of \eqref{eq:qkcompat} again.

The equality $\tau_0 = \id$ can be proved on $\ccal^\infty$ using \cite[Eq.~(2.4)]{BuchholzSummersLechner:2011} and setting $Q=0$. We can extend it to $\qf^\infty$ by computing:
\begin{equation}
\langle \psi, \tau_0(A)\chi \rangle = \langle \psi, \tau_0(Q_k A Q_k)\chi \rangle =\langle \psi, Q_k A Q_k\chi \rangle =\langle \psi, A \chi \rangle,
\end{equation}
where in the first equality we used \eqref{eq:qkcompat}. In the second equality we used the relation $\tau_0 =\id$ on $\ccal^\infty$, since $Q_k A Q_k \in\ccal^\infty$. In the third equality we used \eqref{eq:qkcompat} again, assuming that $k$ is large.

The relation $\tau_Q^{-1} = \tau_{-Q}$ is a direct consequence of $\tau_Q \tau_{Q'} = \tau_{Q+Q'}$ and $\tau_0 =\id$: From $\tau_Q \tau_{Q'} = \tau_{Q+Q'}$ we have $\tau_Q \tau_{-Q}=\tau_{Q-Q}= \tau_0$; inserting $\tau_0 =\id$, we find $\tau_Q \tau_{-Q}= \id$.

Part~\ref{it:warpedtranslate}: these relations can be obtained for the case $A \in \ccal^{\infty}$ by explicit computation, e.g., from \cite[Eq.~(2.4)]{BuchholzSummersLechner:2011}. The result for $A \in \qf^\infty$ can be obtained as follows:
\begin{eqnarray}
\langle \psi, \tau_Q (U(x)A)\chi \rangle &=& \langle \psi, \tau_Q (Q_k U(x) A Q_k) \chi \rangle \nonumber\\
&=& \langle \psi, \tau_Q (U(x)Q_k A Q_k) \chi \rangle\nonumber\\
&=& \langle \psi, U(x) \tau_Q(Q_k A Q_k)\chi \rangle \nonumber\\
&=& \langle \psi, U(x) \tau_Q(A) \chi \rangle,
\end{eqnarray}
where in the first equality we used \eqref{eq:qkcompat}, where in the second equality we used that $U(x)$ commutes with $Q_k$. In the third equality we used the corresponding result for $\ccal^\infty$ obtained above, since $Q_k A Q_k \in \ccal^\infty$. In the fourth equality we used \eqref{eq:qkcompat} again. The second relation in \ref{it:warpedtranslate} follows analogously.

For \ref{it:warpediso}: a map is a vector space isomorphisms if it is linear and bijective. Moreover, it is $\ast$-preserving if $\tau_Q(A^\ast)=\tau_Q(A)^\ast$.

For the case $\ccal^\infty$: the inclusion $\tau_Q(\ccal^\infty)\subset \ccal^\infty$ was already shown in \ref{it:warpedcontinuous}. The map $\tau_Q$ is linear, bijective and $\ast$-preserving as a consequence of \cite[Prop.~2.7(ii)]{BuchholzSummersLechner:2011}, \cite[Lemma~2.2(ii)]{BuchholzSummersLechner:2011} and \cite[Remark after Def.~2.3]{BuchholzSummersLechner:2011}.

For the case $\qf^\infty$, we first need to show that for $A \in \qf^\infty$, we have $Q_k \tau_Q(A)Q_k \in \ccal^\infty$: If $A\in \qf^\infty$, then $Q_k \tau_Q(A)Q_k = \tau_Q( Q_k A Q_k)$ by \eqref{eq:qkcommute}, and $Q_kAQ_k \in \ccal^\infty$. Hence, by part~\ref{it:warpedcontinuous}, $\tau_Q(Q_k A Q_k)\in \ccal^\infty$. That  implies $Q_k \tau_Q(A)Q_k \in \ccal^\infty$, therefore $\tau_Q(A)\in \qf^\infty$.

As for linearity: In \eqref{eq:qkcompat}, we know from the result above for $\ccal^\infty$ that $\tau_Q( Q_k A Q_k)$ is linear in $A$, since $Q_k A Q_k \in \ccal^\infty$. Therefore $\tau_Q$ is linear on $\qf^\infty$.

The map $\tau_Q$ is bijective on $\qf^\infty$ as a consequence of relation $\tau_{Q}^{-1}=\tau_{-Q}$ in part~\ref{it:warpedcompose}.

Finally, $\tau_Q$ is $\ast$-preserving on $\qf^\infty$ for the following reason: In \eqref{eq:qkcompat}, we know from the result above for $\ccal^\infty$ that $\tau_Q( Q_k A Q_k)$ is $\ast$-preserving in $A$, since $Q_k A Q_k \in \ccal^\infty$. Therefore, $\tau_Q$ is $\ast$-preserving on $\qf^\infty$.

Now, it remains to show that $\tau_Q$ is a linear, bijective, $\ast$-preserving map from $\fcal^\infty$ to $\fcal^\infty$. For the inclusion $\tau_Q(\fcal^\infty)\subset \fcal^\infty$, we compute
\begin{equation}\label{inclfcal}
\tau_Q(A)Q_k = \tau_Q(AQ_k)=\tau_Q(Q_{k'}AQ_k)= Q_{k'}\tau_Q(A)Q_k,
\end{equation}
where in the first equality we used \eqref{eq:qkcommute}, since $A\in \fcal^\infty$ is in particular an element of $\qf^\infty$. In the second equality we used that $A \in \fcal^\infty$, and therefore that by definition $Q_{k'}AQ_k=AQ_k$ for some $k'$. In the third equality we made use of \eqref{eq:qkcommute} again.
Similarly, we show that $\tau_Q(A^\ast)Q_k = Q_{k'}\tau_Q(A^\ast)Q_k$. This implies $\tau_Q(A)\in \fcal^\infty$.

The linearity of $\tau_Q$ on $\fcal^\infty$ is a consequence of the linearity of the map on $\qf^\infty$ shown above.

The map $\tau_Q$ is also bijective by part~\ref{it:warpedcompose}, where we proved the relation $\tau_Q^{-1} = \tau_{-Q}$: we only need to show that $\tau_Q^{-1}:\fcal^\infty \rightarrow \fcal^\infty$; this can be done using $\tau_Q^{-1}=\tau_{-Q}$ and \eqref{inclfcal} with $-Q$ in place of $Q$.

$\tau_Q$ is also $\ast$-preserving on $\fcal^\infty$ because it is a $\ast$-preserving map on $\qf^\infty$, as shown above.

For \ref{it:warpeddistribution}: First of all, $\tau_Q(A(f)) \in \fcal^\infty$ by \ref{it:warpediso}. This implies that $\tau_Q(A(f))Q_k \in \ccal^\infty$. We need to show that $f\mapsto \tau_Q(A(f))Q_k$ is a continuous map in the $\ccal^\infty$-topology: Since $A(f)\in \fcal^\infty$ is in particular an element of $\qf^\infty$, then by an application of \eqref{eq:qkcommute}, we have $\tau_Q(A(f))Q_k = \tau_Q(A(f)Q_k)$. Now we consider the map $f\mapsto A(f)Q_k \mapsto \tau_Q(A(f)Q_k)$. Since $A$ is an $\fcal^\infty$-valued distribution, the map $f\mapsto A(f)Q_k$ is continuous by definition. Also, the map $A(f)Q_k \mapsto \tau_Q(A(f)Q_k)$, from $\ccal^\infty$ to $\ccal^\infty$, is continuous due to \ref{it:warpedcontinuous}. Hence, we have that $f \mapsto \tau_Q(A(f)Q_k)$ is continuous. By \eqref{eq:qkcommute}, this implies that the map $f\mapsto  \tau_Q(A(f)) \fpnp_k$ is continuous as well. With a similar argument we show that the map $f \mapsto \fpnp_k \tau_Q(A(f))$ is also continuous.

Since $A$ is an $\fcal^\infty$-valued distribution, we have $Q_{k'}A(f)Q_k=A(f)Q_k$ and $Q_{k}A(f)Q_{k'}=Q_k A(f)$, where $k'$ depends on $k$ but not on $f$; then we want to show that $Q_{k'}\tau_Q(A(f))Q_k=\tau_Q(A(f))Q_k$, where $k'$ is large enough and where the choice of $k'$ does not depend on $f$; similarly for $Q_k\tau_Q(A(f))$. For this, we compute:
\begin{equation}
Q_{k'}\tau_Q(A(f))Q_k = \tau_Q(Q_{k'}A(f)Q_k)=\tau_Q(A(f)Q_k)=\tau_Q(A(f))Q_k,
\end{equation}
where the first and the last equalities are due to an application of \eqref{eq:qkcommute}. In the second equality we used that for $k'$ large enough we can apply \eqref{eq:qkcompat}; here note also that the choice of $k'$ does not depend on $f$, since $A$ is an $\fcal^\infty$-valued distribution.

Hence, we can conclude that $f \mapsto  \tau_Q(A(f))$ is a well-defined $\fcal^\infty$-valued distribution.

Finally, we prove the homogeneity of $\tau_Q(A)$ as follows:
\begin{equation}
U(x)\tau_Q(A(f))U(x)\st = \tau_Q( U(x)A(f)U(x)\st)=\tau_Q(A(e^{i\varphi_A(\cdot)x}f)),
\end{equation}
where in the first equality we applied \ref{it:warpedtranslate}, and where in the second equality we used the homogeneity of $A$.
\end{proof}

 We note that $\tau_Q$ is a vector space isomorphism, i.e. it is a linear ($\tau_Q(A + B)=\tau_Q(A)+\tau_Q(B)$), bijective and $\ast$-preserving map between vector spaces, but it does not preserve the operator product: $\tau_Q(A \cdot B) \neq  \tau_Q(A)\cdot \tau_Q(B)$. On the contrary, it deforms the operator product in the sense given by Lemma~\ref{lemma:homogdeformprod}.

The following lemma describes explicitly the action of $\tau_Q$ on homogeneous distributions:
\begin{lemma}\label{lemma:homogdeformprod}
If $A,B$ are two homogeneous $\fcal^\infty$-valued distribution on $\rbb^m$, then, in the sense of formal kernels,
\begin{equation}\label{eq:homogdeformprod}
\tau_{Q}(A(\thetav)) \tau_{Q}(B(\thetav)) = e^{i\varphi_A(\thetav)Q\varphi_B(\etav)}\tau_{Q}(A(\thetav)B(\etav)).
\end{equation}
\end{lemma}

\begin{proof}
We start with some remarks about a general homogeneous $\fcal^\infty$-valued distribution with kernel $C(\xiv)$.

First, we will show using the basic properties of the warped convolution that the following equation \cite[Eq.~(2.4)]{BuchholzSummersLechner:2011} holds in the sense of distributions:
 \begin{equation}\label{eq:QCvector}
    \tau_Q(C(\xiv)) \rvector{}{\etav} = e^{i \varphi_C(\xiv) Q p(\etav) } C(\xiv) \rvector{}{\etav}.
 \end{equation}
To prove this equation, we consider test functions $f\in\dcal(\rbb^m)$, $g\in\dcal(\rbb^n)$. We choose $h_1,h_2\in\scal(\rbb^2)$ such that $h_1=1$ on a neighbourhood of 0, $h_2(0)=1$, and such that $\tilde h_2$ has compact support. Moreover, we consider $\rvector{}{g}$ to be smooth with respect to translations. Then, under these assumptions, we can apply \cite[Eq.~(2.4)]{BuchholzSummersLechner:2011} and we find,
\begin{equation}
  \tau_Q(C(f)) \rvector{}{g} = \lim_{\epsilon \to 0} (2\pi)^{-2} \iint dx\,dy\, h_1(\epsilon x) h_2(\epsilon y) e^{-ix \cdot y} U(Qx) C(f) U(Qx)^{-1} U(y) \rvector{}{g}.
\end{equation}
Using homogeneity \eqref{homogeneous}, we have
\begin{equation}\label{eq:crlimit}
  \tau_Q(C(f)) \rvector{}{g} = \lim_{\epsilon \to 0} (2\pi)^{-2} \iint dx\,dy\, h_1(\epsilon x) h_2(\epsilon y) e^{-ix \cdot y} C(e^{i \varphi_C(\cdotarg) Qx}f) \rvector{}{e^{i p(\cdotarg) y}g}.
\end{equation}
We call
\begin{equation}\label{defFepsilon}
\begin{aligned}
 F_\epsilon (\thetav,\etav) &= (2\pi)^{-2} \iint dx\,dy\, h_1(\epsilon x) h_2(\epsilon y) e^{i \varphi_C(\thetav)Qx + i (p(\etav)-x) \cdot y}f(\thetav)g(\etav)
\\
 &= \int dx\,  h_1(\epsilon x) \frac{1}{2\pi\epsilon} \tilde h_2\big(\epsilon^{-1}(p(\etav)-x)\big) \,
   e^{i\varphi_C(\thetav)Qx} f(\thetav)g(\etav).
\end{aligned}
\end{equation}
Note that $F_\epsilon\in\dcal(\rbb^{m+n})$ (this follows by computing explicitly the derivatives of $F_\epsilon$ in \eqref{defFepsilon}, taking into account that $f,g$ are smooth and with compact support and that $\varphi_C$ is also smooth) and that it plays the role of test function for the vector valued distributions $C\rvector{}{\cdotarg}$; hence we can write \eqref{eq:crlimit} as
\begin{equation}
 \tau_Q(C(f)) \rvector{}{g} = \lim_{\epsilon \to 0}C\rvector{}{F_\epsilon}.
\end{equation}
In \eqref{defFepsilon}, note that since $\tilde h_2$ has compact support, it restricts the integral to a compact set; moreover, for sufficiently small $\epsilon$, we can replace $h_1(\epsilon x)$ with $1$. We also note that for $\epsilon \to 0$, $\frac{1}{2\pi\epsilon} \tilde h_2(\epsilon^{-1} \cdotarg)$ is a delta sequence. By all these considerations, we find for $\epsilon \to 0$,
\begin{equation}
 F_\epsilon (\thetav,\etav) \to e^{ip(\thetav) Q p(\etav)} f(\thetav)g(\etav)\quad \text{in } \dcal(\rbb^{m+n}).
\end{equation}
Inserting this into \eqref{eq:crlimit}, we find
\begin{equation}
 \tau_Q(C(f)) \rvector{}{g}=\int d\thetav d \etav\,e^{ip(\thetav) Q p(\etav)} f(\thetav)g(\etav) C(\thetav)\rvector{}{\etav},
\end{equation}
which coincides with equation \eqref{eq:QCvector}.

Second, we notice that the support of the distribution $\hscalar{\lvector{}{\thetav}}{ C(\xiv) \rvector{}{\etav}}$ is on the hypersurface $p(\thetav)-p(\etav) = \varphi_C(\xiv)$. To show this, we compute, for $x \in \rbb^2$,
\begin{multline}\label{eq:trans1}
  \hscalar{\lvector{}{\thetav}}{ C(\xiv) \rvector{}{\etav}}
 = \hscalar{U(x) \lvector{}{\thetav}}{ U(x)C(\xiv) U(x)^\ast\,U(x) \rvector{}{\etav}}\\
 = e^{i(-p(\thetav)+p(\etav)+\varphi_C(\xiv))x}\hscalar{\lvector{}{\thetav}}{ C(\xiv) \rvector{}{\etav}},
\end{multline}
where in the second equality we used the homogeneity \eqref{homogeneous} of $C$ and the covariance properties of $\lvector{}{\cdotarg}, \rvector{}{\cdotarg}$ (namely, $U(x) \rvector{}{\etav}=\exp(ip(\etav)x)\rvector{}{\etav}$, and analogously for $\lvector{}{\thetav}$).
But note that equation \eqref{eq:trans1} can hold for all $x$ only if the support of the distribution is contained in the surface $p(\thetav)-p(\etav) = \varphi_C(\xiv)$.

Now we compute both sides of equation \eqref{eq:homogdeformprod} for the distributions $A$ and $B$, between smooth vectors $\lvector{}{\cdotarg},\rvector{}{\cdotarg}$. We start with the left hand side:
\begin{equation}\label{eq:qab1}
\begin{aligned}
 \hscalar{\lvector{}{\thetav'}}{ \tau_Q(A(\thetav)) \tau_Q(B(\etav)) \rvector{}{\etav'}}
&= e^{i \varphi_A(\thetav) Q p(\thetav')} e^{i \varphi_B(\etav) Q p(\etav')}
 \hscalar{\lvector{}{\thetav'}}{ A(\thetav)B(\etav) \rvector{}{\etav'}}
\\
&= e^{i \varphi_A(\thetav) Q \varphi_B(\etav)} e^{i p(\thetav') Q p(\etav')}
 \hscalar{\lvector{}{\thetav'}}{ A(\thetav)B(\etav) \rvector{}{\etav'}},
\end{aligned}
\end{equation}
where in the first equality we applied \eqref{eq:QCvector} twice, and where in the second equality we used that the support of the distribution is restricted to $p(\thetav')-p(\etav')=\varphi_A(\thetav)+\varphi_B(\etav)$, due to the remark above.

As for the right hand side, we obtain
\begin{equation}\label{eq:qab2}
\begin{aligned}
 \hscalar{\lvector{}{\thetav'}}{ \tau_Q(A(\thetav) B(\etav)) \rvector{}{\etav'}}
&= e^{i (\varphi_A(\thetav) +\varphi_B(\etav)) Q p(\etav')}
 \hscalar{\lvector{}{\thetav'}}{ A(\thetav)B(\etav) \rvector{}{\etav'}}
\\
&= e^{i p(\thetav') Q p(\etav')}
 \hscalar{\lvector{}{\thetav'}}{ A(\thetav)B(\etav) \rvector{}{\etav'}},
\end{aligned}
\end{equation}
where in the first equality we applied \eqref{eq:homogdeformprod} with $C(\thetav,\etav)=A(\thetav)B(\etav)$, and where in the second equality we made use of the restriction $p(\thetav')-p(\etav')=\varphi_A(\thetav)+\varphi_B(\etav)$ on the support of the distribution.

Eqs.~\eqref{eq:qab1} and \eqref{eq:qab2} imply the result \eqref{eq:homogdeformprod}.
\end{proof}

Now we can consider the warped convolution of the Bose creation and annihilation operators of the free field theory and we can identify the resulting deformed theory with an integrable model.
To do that, we set $\zd(\theta)=\tau_Q(\ad(\theta))$, $z(\eta)=\tau_Q(a(\eta))$. By applying Lemma~\ref{lemma:homogdeformprod} twice, we find that these $z,\zd$ fulfil the following relations
\begin{equation}\label{eq:zzrfromq}
\begin{aligned}
\zd(\theta)\zd(\eta) &= e^{2ip(\theta)Qp(\eta)}\zd(\eta)\zd(\theta)\\
z(\theta)z(\eta) &= e^{2ip(\theta)Qp(\eta)}z(\eta)z(\theta)\\
z(\theta)\zd(\eta) &= e^{-2ip(\theta)Qp(\eta)}\zd(\eta)z(\theta) + \delta(\theta-\eta)\cdot 1_{\mathcal{H}}.
\end{aligned}
\end{equation}
We can show that there is only a one-parameter family of $2\times 2$ matrices which are skew symmetric with respect to the scalar product in Minkowski space: A matrix is skew-symmetric with respect to the scalar product in Minkowski space if $\pmb{x} \cdot Q \pmb{y}=-(Q\pmb{x})\cdot \pmb{y}$, with $\pmb{x}=(x_0, x_1)$ and $\pmb{y}=(y_0,y_1)$. Solving this equation for $Q=\left( \begin{array}{cc}
a & b  \\
c & d  \end{array} \right)$, we find $a=0$, $d=0$, $b=c$. Hence, we can write this family of matrices explicitly as:
\begin{equation}
 Q = -\frac{a}{2\mu^2} \begin{pmatrix} 0 & 1 \\ 1 & 0 \end{pmatrix}
\end{equation}
where $a$ is a real dimensionless constant.

Using this explicit form of the matrix $Q$ and for $a \geq 0$, we find that the equations \eqref{eq:zzrfromq} are just the Zamolodchikov relations where the scattering function $S$ is given by
\begin{equation}\label{eq:QSform}
   S(\theta) = e^{i a \sinh \theta}.
\end{equation}
Then we can identify unitarily the Hilbert space $\hcal$ of the free theory with the $S$-symmetric Fock space over $\hcal_1$, introduced in Sec.~\ref{sec:hilbertspace}.

As next step we define the $Q$-commutator as follows.

\begin{definition}
For $A,B \in \ccal^\infty$, the $Q$-commutator is
\begin{equation}\label{Qcom}
[A,B]_{Q}:=AB -\tau_{2Q}\Big(\tau_{-2Q}(B)\tau_{-2Q}(A)\Big).
\end{equation}
We use the same definition if $A,B \in \qcal^\infty$ and at least \emph{one} of them is in $\fcal^\infty$.
\end{definition}

For homogeneous distributions $A(\thetav),B(\etav)$, we have the following explicit expression for the $Q$-commutator:
\begin{equation}\label{eq:Qcommhomog}
[A(\thetav),B(\etav)]_{Q}=A(\thetav)B(\etav)-e^{2i\varphi_{A}(\thetav)Q\varphi_{B}(\etav)}B(\etav)A(\thetav).
\end{equation}
This can be computed from \eqref{Qcom} using Lemma~\ref{lemma:homogdeformprod}:
\begin{eqnarray}
[A(\thetav),B(\etav)]_{Q}&=&A(\thetav)B(\etav)-\tau_{2Q}\Big(\tau_{-2Q}(B(\etav))\tau_{-2Q}(A(\thetav))\Big)\nonumber\\
&=& A(\thetav)B(\etav)-e^{2i\varphi_{A}(\thetav)Q\varphi_{B}(\etav)}\tau_{2Q}\Big(\tau_{-2Q}\Big(B(\etav)A(\thetav)\Big)\Big)\nonumber\\
&=&A(\thetav)B(\etav)-e^{2i\varphi_{A}(\thetav)Q\varphi_{B}(\etav)}B(\etav)A(\thetav).
\end{eqnarray}
In particular, we notice that the expression of the $Q$-commutator is again homogeneous.

We note that the $Q$-commutator fulfils the following ``deformed'' versions of the standard properties of a commutator. We formulate these properties only for homogeneous distributions. Indeed, we can show that similar relations then holds for general elements of $\qf^\infty$ or $\ccal^\infty$, by decomposing them into homogeneous distributions using a spectral decomposition in the sense of Arveson \cite{Arveson:1974} with respect to the action of the translation group, or using the Araki expansion.

\begin{proposition}
For homogeneous distributions with kernels $A(\thetav), B(\etav), C(\xiv)$, the $Q$-commutator satisfies

\begin{enumerate}
\renewcommand{\theenumi}{(\roman{enumi})}
\renewcommand{\labelenumi}{\theenumi}
\item\label{it:qcommanti}
 anticommutativity:
\begin{equation}
[A(\thetav),B(\etav)]_{Q}=-e^{2i\varphi_{A}(\thetav)Q\varphi_{B}(\etav)}[B(\etav),A(\thetav)]_{Q};
\end{equation}

\item \label{it:qcommleibniz}
Leibniz rule:
\begin{equation}\label{leibnizbiscuit}
[A(\thetav),B(\etav)C(\xiv)]_{Q}=[A(\thetav),B(\etav)]_{Q}C(\xiv)+e^{2i\varphi_{A}(\thetav)Q\varphi_{B}(\etav)}B(\etav)[A(\thetav),C(\xiv)]_{Q};
\end{equation}

\item\label{it:qcommjacobi}
 Jacobi identity:
\begin{equation}
e^{-2i\varphi_{A}(\thetav)Q\varphi_{C}(\xiv)}[A(\thetav),[B(\etav),C(\xiv)]_{Q}]_{Q}+ \mathrm{ cyclic\, permutations } =0.\label{Qjacobi}
\end{equation}
\end{enumerate}
\end{proposition}
\begin{proof}
In~\ref{it:qcommanti}, applying Eq.~\eqref{eq:Qcommhomog} the r.h.s.~gives:
\begin{multline}
-e^{2i\varphi_{A}(\thetav)Q\varphi_{B}(\etav)} [B(\etav),A(\thetav)]_{Q}
=-e^{2i\varphi_{A}(\thetav)Q\varphi_{B}(\etav)}\Big(B(\etav)A(\thetav)-e^{2i\varphi_{B}(\etav)Q\varphi_{A}(\thetav)}A(\thetav)B(\etav)\Big)\\
=-e^{2i\varphi_{A}(\thetav)Q\varphi_{B}(\etav)}B(\etav)A(\thetav)+A(\thetav)B(\etav)
=[A(\thetav),B(\etav)]_{Q}.
\end{multline}

For~\ref{it:qcommleibniz}, we apply \eqref{eq:Qcommhomog} with respect to $B(\etav)C(\xiv)=:D(\etav,\xiv)$ and $\varphi_{B}(\etav)+\varphi_{C}(\xiv)=:\varphi_{D}(\etav,\xiv)$. The l.h.s.~gives:
\begin{multline}
[A(\thetav),B(\etav)C(\xiv)]_{Q}
=A(\thetav)D(\etav,\xiv)-e^{2i\varphi_{A}(\thetav)Q\varphi_{D}(\etav,\xiv)}D(\etav,\xiv)A(\thetav)\\
=A(\thetav)B(\etav)C(\xiv)-e^{2i\varphi_{A}(\thetav)Q( \varphi_{B}(\etav)+\varphi_{C}(\xiv))}B(\etav)C(\xiv)A(\thetav).
\label{leftside}
\end{multline}
Applying \eqref{eq:Qcommhomog}, the right hand side of \eqref{leibnizbiscuit} gives:
\begin{equation}
[A(\thetav),B(\etav)]_{Q}C(\xiv)=A(\thetav)B(\etav)C(\xiv)-e^{2i\varphi_{A}(\thetav)Q\varphi_{B}(\etav)}B(\etav)A(\thetav)C(\xiv)
\label{rightside1}
\end{equation}
and
\begin{equation}\label{rightside2}
\begin{aligned}
e^{2i\varphi_{A}(\thetav)Q\varphi_{B}(\etav)} & B(\etav)[A(\thetav),C(\xiv)]_{Q}\\
&=e^{2i\varphi_{A}(\thetav)Q\varphi_{B}(\etav)}B(\etav)\Big( A(\thetav)C(\xiv)-e^{2i\varphi_{A}(\thetav)Q\varphi_{C}(\xiv)}C(\xiv)A(\thetav)\Big)\\
&=e^{2i\varphi_{A}(\thetav)Q\varphi_{B}(\etav)}B(\etav)A(\thetav)C(\xiv)-
e^{2i\varphi_{A}(\thetav)Q(\varphi_{B}(\etav)+\varphi_{C}(\xiv))}B(\etav)C(\xiv)A(\thetav).
\end{aligned}
\end{equation}
Combining \eqref{rightside1} and \eqref{rightside2} we get \eqref{leftside}.

For \ref{it:qcommjacobi}: Applying \eqref{eq:Qcommhomog}, the first summand in \eqref{Qjacobi} gives:
\begin{multline}
e^{-2i\varphi_{A}(\thetav)Q\varphi_{C}(\xiv)}[A(\thetav),[B(\etav),C(\xiv)]_{Q}]_{Q}\\
=e^{-2i\varphi_{A}(\thetav)Q\varphi_{C}(\xiv)}[A(\thetav),B(\etav)C(\xiv)]_{Q}
-e^{-2i\varphi_{A}(\thetav)Q\varphi_{C}(\xiv)+2i\varphi_{B}(\etav)Q\varphi_{C}(\xiv)}[A(\thetav),C(\xiv)B(\etav)]_{Q}.
\end{multline}
Applying again Eq.~\eqref{eq:Qcommhomog} with $B(\etav)C(\xiv)=:D(\etav,\xiv)$ and $\varphi_{B}(\etav)+\varphi_{C}(\xiv)=:\varphi_{D}(\etav,\xiv)$ (analogously, $C(\xiv)B(\etav)=:E(\xiv,\etav)$, $\varphi_{C}(\xiv)+\varphi_{B}(\etav)=:\varphi_{E}(\xiv,\etav)$ ), we find from the formula above:
\begin{equation}\label{jacobiterm}
\begin{aligned}
&e^{-2i\varphi_{A}(\thetav)Q\varphi_{C}(\xiv)} [A(\thetav),[B(\etav),C(\xiv)]_{Q}]_{Q}\\
&\qquad=e^{-2i\varphi_{A}(\thetav)Q\varphi_{C}(\xiv)}A(\thetav)D(\etav,\xiv)
-e^{-2i\varphi_{A}(\thetav)Q\varphi_{C}(\xiv)+2i\varphi_{A}(\thetav)Q\varphi_{D}(\etav,\xiv)}D(\etav,\xiv)A(\thetav) \\
&\qquad\quad -e^{-2i\varphi_{A}(\thetav)Q\varphi_{C}(\xiv)+2i\varphi_{B}(\etav)Q\varphi_{C}(\xiv)}A(\thetav)E(\xiv,\etav)
\\
&\qquad\quad +e^{-2i\varphi_{A}(\thetav)Q\varphi_{C}(\xiv)+2i\varphi_{B}(\etav)Q\varphi_{C}(\xiv)
+2i\varphi_{A}(\thetav)Q\varphi_{E}(\xiv,\etav)}E(\xiv,\etav)A(\thetav)\\
&\qquad=e^{-2i\varphi_{A}(\thetav)Q\varphi_{C}(\xiv)}A(\thetav)B(\etav)C(\xiv)
-e^{2i\varphi_{A}(\thetav)Q\varphi_{B}(\etav)}B(\etav)C(\xiv)A(\thetav)\\
&\qquad\quad -e^{2i(-\varphi_{A}(\thetav)+\varphi_{B}(\etav))Q\varphi_{C}(\xiv)}A(\thetav)C(\xiv)B(\etav)
+e^{2i\varphi_{B}(\etav)Q(\varphi_{C}(\xiv)-\varphi_{A}(\thetav))}C(\xiv)B(\etav)A(\thetav).
\end{aligned}
\end{equation}
Taking the sum of the cyclic permuted terms of \eqref{jacobiterm}, we obtain zero.
\end{proof}

 Using Eq.~\eqref{eq:Qcommhomog}, we can rewrite the Zamolodchikov relations \eqref{eq:zzrfromq} in terms of $Q$-commutators in the following way:
\begin{equation}
 [z^{\dagger}(\theta),z^{\dagger}(\theta')]_{Q}=0,
\quad
  [z(\eta),z(\eta')]_{Q}=0,
\quad
  [z(\eta),z^{\dagger}(\theta)]_{Q}=\delta(\theta-\eta).
\end{equation}
Namely, we find that the Zamolodchikov operators $z,\zd$ satisfy relations of the type of the CCR relations with respect to the $Q$-commutator; note the analogy of this with the graded commutator in the case of the CAR relations.

Moreover, using Eq.~\eqref{eq:Qcommhomog}, we also obtain:
\begin{equation}
[z(\xi),z^{\dagger m}(\thetav)z^{n}(\etav)]_{Q}=z(\xi)z^{\dagger m}(\thetav)z^{n}(\etav)
  -e^{-2ip(\xi)Q(p(\thetav)-p(\etav))}z^{\dagger m}(\thetav)z^{n}(\etav)z(\xi),
\end{equation}
which implies by repeated application of the relations of the Zamolodchikov relations \eqref{eq:zzrfromq},
\begin{multline}
[z(\xi),z^{\dagger m}(\thetav)z^{n}(\etav)]_{Q}=\sum_{j=1}^{m}\prod_{l=1}^{j-1}e^{2ip(\theta_{l})Qp(\xi)}\delta(\theta_{j}-\xi)z^{\dagger}(\theta_{1})\ldots \widehat{z^{\dagger}(\theta_{j})}\ldots z^{\dagger}(\theta_{m})z^{n}(\etav)\\
=m\operatorname{Sym}_{S^{-1},\thetav}\Big( \delta(\xi-\theta_{1})z^{\dagger m-1}(\theta_{2},\ldots,\theta_{m})z^{n}(\etav)\Big).\label{QcommAz}
\end{multline}
Similarly, we have
\begin{multline}
[z^{\dagger m}(\thetav)z^{n}(\etav),z^{\dagger}(\xi)]_{Q}\\
=z^{\dagger m}(\thetav)z^{n}(\etav)z^{\dagger}(\xi)-e^{-2ip(\xi)Q\Big( \sum_{j=1}^{m}p(\theta_{j})-\sum_{l=1}^{n}p(\eta_{l}) \Big)}z^{\dagger}(\xi)z^{\dagger m}(\thetav)z^{n}(\etav),
\end{multline}
which implies by repeated application of the relations of the Zamolodchikov relations \eqref{eq:zzrfromq},
\begin{multline}
[z^{\dagger m}(\thetav)z^{n}(\etav),z^{\dagger}(\xi)]_{Q}=\sum_{j=1}^{n}\prod_{l=j+1}^{n}e^{2ip(\xi)Qp(\eta_{l})}\delta(\eta_{j}-\xi)z^{\dagger m}(\thetav)z(\eta_{1})\ldots \widehat{z(\eta_{l})}\ldots z(\eta_{n})\\
=n \operatorname{Sym}_{S^{-1},\etav}\Big( \delta(\xi-\eta_{n})z^{\dagger m}(\thetav)z^{n-1}(\eta_{1},\ldots ,\eta_{n-1})\Big).\label{QcommAzdag}
\end{multline}
Using this, we can now prove the following form of the Araki coefficients in the case where the scattering function is of the form \eqref{eq:QSform}:

\begin{theorem}
Let $S$ be of the form \eqref{eq:QSform}. The coefficients $\cme{m,n}{A}$, where $A\in \qf^\infty$, can be expressed as
\begin{equation}\label{Snested}
\cme{m,n}{A}(\thetav,\etav)=\bighscalar{ \Omega}{ [ z(\theta_{m})\ldots[z(\theta_{1})\ldots [ A,z^{\dagger}(\eta_{n})]_{Q}\ldots z^{\dagger}(\eta_{1})]_{Q}\ldots]_{Q}\Omega }.
\end{equation}
\end{theorem}
\begin{proof}
First we notice that that if $A$ is in $\qf^\infty$, also its expansion terms $z^{\dagger m} z^{n}(\cme{m,n}{A})$ are elements of $\qf^\infty$. Indeed, using Prop.~\ref{proposition:fmnpoincare}, we find that the derivatives $\partial^\kappa$, with a multi-index $\kappa$, of $z^{\dagger m} z^{n}(\cmelong{m,n}{A})$ fulfil the following equality:
\begin{equation}
   \partial^\kappa z^{\dagger m} z^{n}(\cmelong{m,n}{A}) = z^{\dagger m} z^{n}(\cmelong{m,n}{\partial^\kappa A}).
\end{equation}
Then, by Prop.~\ref{pro:zzdcrossnorm} and \ref{proposition:fmnbound}, we have that the right hand side of the equation above have finite norms when applied to fixed particle number vectors.

Therefore, by Thm.~\ref{theorem:arakiexpansion}, we can express a general $A$ using the Araki expansion; hence,  it suffices to prove the equation \eqref{Snested} in the case where $A=z^{\dagger m'} z^{n'}(f)$, since for more general $A$, due to the Araki expansion, it follows by linearity.
Now we have that for this particular $A$, or more precisely for its kernel $A(\thetav',\etav')=z^{\dagger m'}(\thetav') z^{n'}(\etav')$, the nested $Q$-commutator in \eqref{Snested} gives by repeated application of Eqs.~\eqref{QcommAz} and \eqref{QcommAzdag}:
\begin{multline}\label{applcomm}
 [z(\theta_{m})\ldots [z(\theta_{1})\ldots [z^{\dagger m'}(\thetav')z^{n'}(\etav'),z^{\dagger}(\eta_{n})]_{Q}\ldots z^{\dagger}(\eta_{1})]_{Q}\ldots ]_{Q}=
m!n!\operatorname{Sym}_{S^{-1},\etav'}\operatorname{Sym}_{S^{-1},\thetav'}\\
\Big( \prod_{j=1}^{m}\delta(\theta_j - \theta'_{j})\prod_{k=0}^{n-1}\delta(\eta_{n-k}-\eta'_{n'-k})z^{\dagger m'-m}(\theta'_{m+1},\ldots, \theta'_{m'})z^{n'-n}(\eta'_{1},\ldots,\eta'_{n'-n})\Big)
\end{multline}
if $m'\geq m$, $n'\geq n$, and the right hand side vanishes otherwise. Now if $m'>m$ or $n'>n$, the vacuum expectation value of the right hand side of \eqref{applcomm} vanishes. Therefore, we find:
\begin{equation}
 \begin{aligned}
\langle \Omega, [z(\theta_{m})\ldots [z(\theta_{1})\ldots & [z^{\dagger m'}(\thetav')z^{n'}(\etav'),z^{\dagger}(\eta_{n})]_{Q}\ldots z^{\dagger}(\eta_{1})]_{Q}\ldots ]_{Q}\Omega \rangle\\
&=m!n!\delta_{m,m'}\delta_{n,n'}\operatorname{Sym}_{S^{-1},\etav'}\operatorname{Sym}_{S^{-1},\thetav'}\Big( \delta^m(\thetav-\thetav')\delta^n(\etav-\etav')\Big)\\
&= m!n!\delta_{m,m'}\delta_{n,n'}\operatorname{Sym}_{S,\thetav}\delta^m(\thetav-\thetav')\operatorname{Sym}_{S,\etav}\delta^n(\etav-\etav').
 \end{aligned}
\end{equation}
We have used \eqref{symstheta} here. This matches the left hand side of \eqref{Snested} because of Prop.~\ref{proposition:fmnbasis}.
\end{proof}

Hence, we have shown in the case where the scattering function is of the form \eqref{eq:QSform} that the Araki coefficients can be expressed in terms of a string of nested deformed commutators. Now it would be interesting to generalize similar expressions for the Araki coefficients in the case of general $S$. We know that on a formal level this is possible. Indeed, one could use the more general deformation procedure given in \cite{Lechner:2011} to construct a suitable ``S-commutator''; or another more direct way would be to impose as a definition the relations $[z(\eta),\zd(\theta)]_S = \delta(\theta-\eta)$, etc., between homogeneous distributions, and to
use the Araki decomposition to define the deformed commutator for more general operators; we would then obtain the following formula for ``S-commutator'':
\begin{multline}\label{Scomm}
[A,B]_S = AB - \\
\sum_{\substack{m\geq 0, n\geq 0 \\ m'\geq 0, n'\geq 0}}\int \frac{d^{m+n}\thetav}{m!n!} \frac{d^{m'+n'}\thetav'}{m'!n'!} \prod_{i=1}^{m+n}\prod_{j=1}^{m'+n'}S^{(m,m')}(\theta_i-\theta'_{j})
\cme{m',n'}{B}(\thetav')\cme{m,n}{A}(\thetav){\zd}^{m'}z^{n'}(\thetav'){\zd}^{m}z^{n}(\thetav)
\end{multline}
with
\begin{equation}
S^{(m,m')}(\theta_i-\theta'_{j}):=
\begin{cases} S(\theta_i - \theta'_{j}), \quad &\text{ if }\; i\leq m
\wedge j \leq m', \text{ or } i>m \wedge j> m',\\
S(\theta'_{j} - \theta_i), \quad &\text{ if }\; i\leq m \wedge j > m', \text{ or } i> m \wedge j\leq m.
\end{cases}
\end{equation}
But if this definition might make sense on a formal level, its properties in the point of view of functional analysis (for example, whether the $S$-commutator of two bounded operators would be bounded) are still unclear for the moment.

In particular, notice that the right hand side of \eqref{Scomm} is well-defined only if the sum over $m,n$ or the sum over $m',n'$ is finite, namely in the case where either $A$ or $B$ has finite sum in the Araki expansion.

%% file: boundaryvalues.tex
\chapter[Residues and boundary distributions]{Residues and boundary distributions in several variables} \label{sec:bvlemma}

In the proof of our Theorem in Chapter~\ref{sec:localitythm}, see for example Chapter~\ref{sec:fptof}, we study meromorphic functions in several variables and we use their residues. Here, all the poles of these meromorphic functions are of first order and they sit on hyperplanes, $\zv \cdot \av = c$ with $\av \in\rbb^k$, $c \in \cbb$. Our notation for the residues has the following convention: If $F(\zv) = G(\zv)/(\zv\cdot\av - c)$, where $G$ is analytic in a neighbourhood of the hyperplane, then
\begin{equation}
   \res_{\zv \cdot \av = c} F = G \big\vert_{\zv \cdot \av = c}.
\end{equation}
One has to be careful with this notation, because of the following fact: For $\alpha \in \rbb \backslash\{0\}$, we have
\begin{equation} \label{rescaleonepole}
   \res_{\zv \cdot (\alpha\av) = \alpha c} F = \alpha \res_{\zv \cdot \av = c} F,
\end{equation}
even if $\zv \cdot \av = c$ and $\zv \cdot (\alpha \av) = \alpha c$ describe the same geometric set. We accept this because this notation is simpler, indeed the alternative would be to work with oriented manifolds, and with differential forms rather than functions, and the notation would become a bit more involved.

We consider that the residue of a meromorphic function on $\cbb^k$ is again a meromorphic function on a lower-dimensional complex manifold, and we identify this lower-dimensional complex manifold with $\cbb^{k-1}$.

In this section of the appendix, we study the boundary values of meromorphic functions, which are distributions, and we try to generalize to the case of several variables the following relation which is valid in one variable: If $F$ is a function of one complex variable, analytic in a strip around the real axis except for a possible first-order pole at $z=0$, then we have the following relation between the boundary distributions,
\begin{equation}
   F(x-i0) = F(x+i0) + 2 \pi i \delta(x) \res_{z=0} F.
\end{equation}
We present a multi-dimensional generalization of this formula in the following lemma, which is mostly due to H. Bostelmann:

\begin{lemma}\label{lemma:onepole}
  Let $\ucal \subset \rbb^k$ be a neighbourhood of zero, $\ccal \subset \rbb^k$ an open convex cone, and $\av \in \rbb^k$.
  Let $F$ be meromorphic on $\tube(\ucal)$ and $(\zv \cdot \av)F(\zv)$ analytic on $\tube(\ccal \cap \ucal)$.
  Let $\bv^+,\bv^-,\bv^\bot \in \ccal$ so that $\pm \av \cdot \bv^\pm  > 0$, $\av\cdot\bv^\bot  = 0$. Then it holds that
  \begin{equation}\label{eq:onepole}
     F(\xv + i 0 \bv^-)  = F(\xv + i 0 \bv^+) + 2 \pi i \delta(\xv \cdot \av) \res_{\zv \cdot \av = 0} F(\xv + i 0 \bv^\bot).
  \end{equation}
\end{lemma}

\begin{proof}
Note that $\ucal$ and $\ccal$ are both regions in $i\rbb^k\subset \cbb^k$.

We note that equation \eqref{eq:onepole} is valid without requiring a particular form for the neighbourhood $\ucal$ of zero, since this formula needs to hold in the limit $\epsilon \rightarrow 0$. Since $\ucal$  is a neighbourhood of zero, it contains an open ball around the origin. So, without loss of generality, we can assume that $\ucal$ is a ball around the origin.

After the rescaling of the neighbourhood  $\ucal$, we also need to rescale the vectors $\bv^+,\bv^-,\bv^\bot$ with some real positive factor. We need to choose this factor small, so that they are still contained in the ball. Then we note that the statement of the lemma above does not change, in particular the scalar products of those vectors with $\av$ and equation \eqref{eq:onepole} are still the same: Since \eqref{eq:onepole} holds in the limit $\epsilon \rightarrow 0$, we can rescale $\epsilon$ and absorb the rescaling factor (let's say $\lambda$) in the argument of $F$, $F(\xv+i\epsilon \lambda \bv^-)$, that would appear after the rescaling of $\ucal$.

We can also change (``rotate'') the system of coordinates and rescale the vector $\av$, so that $\av=\ev^{(1)}$. This also implies that we rotate the vectors $\bv^+,\bv^-,\bv^\bot$. However equation \eqref{eq:onepole} transforms covariantly under this change of coordinates; moreover it does not change after the rescaling of $\av$ (let's say by a positive factor $\alpha$): Indeed, as remarked in \eqref{rescaleonepole}, we have that the residue rescales by $\alpha$, but the delta function rescales by the inverse $\alpha^{-1}$.

Note that the equation \eqref{eq:onepole} holds in the sense of distributions. We prove this equation when smeared with test functions $g \in \dcal(\kcal)$, where $\kcal$ is a fixed convex compact set. We define $G(\zv):=(\zv\cdot\av)F(\zv)$; by hypothesis, we have that this function is analytic on $\tube(\ccal \cap \ucal)$, and we have $G(\zv)=\res_{\zv\cdot\av=0}F(\zv)$ if $\zv\cdot\av=0$.

First we prove the equation \eqref{eq:onepole} using the additional hypothesis that $G$ and its gradient, $\nabla G$, can be extended to a continuous function on the closure $\kcal + i (\bar\ccal \cap \bar\ucal)$.

We compute,
\begin{equation}\label{eq:fdiff}
 \int \Big( F(\xv+i0 \bv^-)-F(\xv+i0 \bv^+) \Big) g(\xv)d\xv
 = \lim_{\epsilon\searrow 0}
 \int \Big( \frac{G(\xv+i \epsilon \bv^-)}{x_1 + i \epsilon b_1^-} - \frac{G(\xv + i \epsilon \bv^+)}{x_1 + i \epsilon b_1^+} \Big) g(\xv)d\xv,
\end{equation}
where we used that $G(\zv):=(\zv\cdot\av)F(\zv)$, $\zv= \xv +i\epsilon \bv^\pm$, $\av=\ev^{(1)}$.

Now we use the notation $\xv=(x_1,\hat\xv)$ and the substitution $y = x_1/\epsilon$, and we can rewrite the right hand side of \eqref{eq:fdiff} as
\begin{multline}\label{eq:gdiff}
 \eqref{eq:fdiff}
 = \lim_{\epsilon\searrow 0}
 \int \epsilon dy\,d\hat\xv\,  \Big(\frac{G(\epsilon y +i\epsilon b_1^-,\hat\xv +i\epsilon \hat \bv^-)}{\epsilon y+i\epsilon b_1^-}-\frac{G(\epsilon y +i\epsilon b_1^+,\hat\xv +i\epsilon \hat \bv^+)}{\epsilon y+i\epsilon b_1^+}\Big)
g(\epsilon y,\hat\xv)
 \\
 =\lim_{\epsilon\searrow 0}
 \int dy\,d\hat\xv\,  \Big(\frac{G(\zv^-_{\epsilon})}{y+i b_1^-}-\frac{G(\zv^+_{\epsilon})}{ y+i b_1^+}\Big)
g(\epsilon y,\hat\xv)
 \\
 =\lim_{\epsilon\searrow 0}
 \int dy\,d\hat\xv\,  \Big(\frac{(y+ib_1^+)G(\zv^-_{\epsilon})-(y+ib^-_1)G(\zv^+_{\epsilon})}{(y+i b_1^-)(y+i b_1^+)}\Big)
g(\epsilon y,\hat\xv)
 \\
 =\lim_{\epsilon\searrow 0}
 \int dy\,d\hat\xv\,  \frac{i b_1^+ G(\zv_\epsilon^-) - i b_1^- G(\zv_\epsilon^+)
+ y \big(G(\zv_\epsilon^-) - G(\zv_\epsilon^+) \big)}{(y+ib_1^-)(y+ib_1^+)}
g(\epsilon y,\hat\xv),
\end{multline}
where $\zv_\epsilon^\pm := (\epsilon y + i \epsilon b_1^\pm , \hat\xv + i \epsilon \hat{\bv} {\vphantom{\bv}}^\pm)$.

We consider the numerator in \eqref{eq:gdiff}, and we compute:
\begin{multline}\label{absvalintegrand}
\Big\lvert \frac{i b_1^+ G(\zv_\epsilon^-) - i b_1^- G(\zv_\epsilon^+)
+ y \big(G(\zv_\epsilon^-) - G(\zv_\epsilon^+) \big)}{(y+ib_1^-)(y+ib_1^+)} \Big\rvert \\
\leq \frac{b_1^+ |G(\zv_\epsilon^-)| +  b_1^- |G(\zv_\epsilon^+)|
+ |y| \cdot \lvert \big(G(\zv_\epsilon^-) - G(\zv_\epsilon^+) \big)\rvert}{|y+ib_1^-|\cdot |y+ib_1^+|}.
\end{multline}
Since $g$ has compact support, we have that $|\xv|<c$ with some $c>0$, and therefore $|y|<c/\epsilon$, for $\epsilon\leq 1$. Since $G$ is analytic on $\tube(\ccal \cap \ucal)$ and, by the additional assumption before \eqref{eq:fdiff}, is also continuous on the closed domain $\kcal + i (\bar\ccal \cap\bar\ucal )$, we can apply the mean value theorem, and we can find a point $\zv$ in $\kcal + i (\bar\ccal \cap\bar\ucal ) $, such that
\begin{equation}
|G(\zv_\epsilon^-) - G(\zv_\epsilon^+) | = \epsilon \gnorm{\bv^+-\bv^-}{}\gnorm{\nabla G(\zv)}{},
\end{equation}
where we used that $\zv^-_\epsilon - \zv^+_\epsilon = i\epsilon(\bv^- - \bv^+)$.

By taking the supremum of $\gnorm{\nabla G(\zv)}{}$ over all $\zv$ on the domain $\kcal + i (\bar\ccal \cap\bar\ucal )$, we find
\begin{equation}
   |G(\zv_\epsilon^-) - G(\zv_\epsilon^+) |  \leq \epsilon \gnorm{\bv^+-\bv^-}{} \;\sup \big\{ \gnorm{\nabla G(\zv)}{} : \zv \in \kcal + i (\bar\ccal \cap\bar\ucal ) \big\}.
\end{equation}
Since by the additional assumption, $G, \nabla G$ are continuous functions on the compact domain $\kcal + i (\bar\ccal \cap\bar\ucal )$, then they are bounded, which means that the supremum in the above equation is finite and $|G(\zv_\epsilon^\pm)|$ is bounded as well;

Hence, we obtain a majorant in \eqref{absvalintegrand}, which is also integrable. Then we can apply the dominated convergence theorem, and we can bring the limit $\epsilon\searrow 0$ inside the integral sign, we find
\begin{multline}\label{fdiffcomput}
 \eqref{eq:fdiff}
 =  \int dy d\hat\xv\, \Big(\frac{ G(0,\hat\xv)}{y+ib_1^-}-\frac{ G(0,\hat\xv)}{y+ib_1^+}\Big)g(0,\hat\xv)\\
=\int dy\,  \Big(\frac{1}{y+ib_1^-}-\frac{1}{y+ib_1^+}\Big)\int d\hat\xv\,G(0,\hat\xv) g(0,\hat\xv)
\\
 =\int dy \, \frac{i b_1^+ - i b_1^-}{(y+ib_1^-)(y+ib_1^+)}
\int d\hat\xv\, G(0,\hat\xv) g(0,\hat\xv).
\end{multline}
We can solve the integral in $y$ applying the Jordan's lemma. Using $\pm b_1^\pm > 0$, the pole $y=-ib_1^-$ is in the upper half complex plane of $y$ and the pole $y=-ib_1^+$ is in the lower half complex plane of $y$. Hence, we have
\begin{equation}
\int dy \, \frac{i b_1^+ - i b_1^-}{(y+ib_1^-)(y+ib_1^+)} = 2i \pi \res_{y=-ib_1^-}\Big( \frac{i b_1^+ - i b_1^-}{(y+ib_1^-)(y+ib_1^+)} \Big)=2\pi i.
\end{equation}
Inserting in \eqref{fdiffcomput}, we find
\begin{equation}\label{fdifffinale}
\eqref{eq:fdiff} =\int dy \, \frac{i b_1^+ - i b_1^-}{(y+ib_1^-)(y+ib_1^+)}
\int d\hat\xv\, G(0,\hat\xv) g(0,\hat\xv)\\
=2\pi i\int d\hat\xv\, G(0,\hat\xv) g(0,\hat\xv).
\end{equation}
Therefore,
\begin{multline}
 \int \Big( F(\xv+i0 \bv^-)-F(\xv+i0 \bv^+) \Big) g(\xv)d\xv = 2\pi i\int d\hat\xv\, G(0,\hat\xv) g(0,\hat\xv)\\
=2\pi i\int d\hat\xv \int dx_1\,\delta(x_1) G(\xv) g(\xv),
\end{multline}
which we can rewrite in the sense of distributions as
\begin{equation}
F(\xv+i0 \bv^-)-F(\xv+i0 \bv^+)=2\pi i \delta(x_1)G(\xv)=2\pi i \delta(\xv \cdot \av) G(\xv)=2\pi i \delta(\xv \cdot \av)\res_{\zv\cdot \av=0}F(\xv),
\end{equation}
where in the second equality we used that $\av=\ev^{(1)}$, and in the third equality that $G(\zv)=\res_{\zv\cdot\av=0}F(\zv)$. We note that, due to our continuity assumption and equation $G(\zv)=\res_{\zv\cdot\av=0}F(\zv)$, the residue is a continuous function for $\zv\cdot\av=0$, so we can omit the factor $+ i 0 \bv^\bot$ in the argument of $F$.

This proves equation \eqref{eq:onepole} with our additional continuity hypothesis.

Now, we consider the general case, without the additional continuity hypothesis. Also in this case, we can replace $\ucal$ with a smaller ball $\ucal'$, such that $\bar{\ucal'}\subset \ucal$: Indeed, the assertion of Lemma~\ref{lemma:onepole} does not require a specific neighbourhood of zero, since equation \eqref{eq:onepole} needs to hold in the limit $\epsilon\rightarrow 0$.

Since $\ccal$ is open and $\bv^+,\bv^-,\bv^\bot\in\ccal$, then $\ccal$ must also contain a small neighbourhood of each vector $\bv^+,\bv^-,\bv^\bot$. Therefore we can choose a smaller open convex cone $\ccal'$, such that $\bv^+,\bv^-,\bv^\bot\in\ccal'$, and $\bar\ccal' \subset \ccal$. Note that replacing $\ccal$ with a smaller cone $\ccal'$ does not change the assertion of Lemma~\ref{lemma:onepole}: In particular the scalar products of those vectors with $\av$ and equation \eqref{eq:onepole} read still the same.

By hypothesis we have that $G$ and $\nabla G$ are continuous on $\kcal + i (\ccal \cap \ucal)$; since $\bar{\ccal'}\subset \ccal$ and $\bar{\ucal'}\subset \ucal$, this implies that they are continuous on $\kcal + i (\bar{\ccal'} \cap \bar{\ucal'})$, except possibly at $\im \zv = 0$, because there is no neighbourhood of zero which is contained in $\ccal$. Now, to keep the notation simpler, we omit the prime indices, and we write that the functions $G$ and $\nabla G$ are continuous on $\kcal + i (\bar\ccal \cap \bar\ucal)$ except possibly at $\im \zv = 0$ (this change of notation should not confuse the reader since we can choose the new domains $\ccal'$ and $\ucal'$ without loss of generality).

By hypothesis $F$ is meromorphic in $\tube(\ucal)$, so $G, \nabla G$ are meromorphic in $\tube(\ucal)$ as well; this implies that they are locally given as a quotient of two analytic functions, where the denominator has zeros of finite order; therefore, $G, \nabla G$ possibly diverge at $\im \zv = 0$ like an inverse power of $\im \zv $. That is, we can find $c>0$, $\ell>0$ such that
\begin{equation}\label{eq:gpoly}
   |G(\zv)| + \gnorm{\nabla G(\zv)}{} \leq c \gnorm{\im \zv}{}^{-\ell} \quad \text{for all }
  \zv \in \kcal + i (\bar\ccal \cap \bar\ucal), \; \im \zv \neq 0.
\end{equation}
Now we denote with $\partial_\bot = \bv^\bot\cdot\nabla$ the partial derivative in the direction of $\bv^\bot$, and with $G^{(-m)}$ the $m$th-order antiderivative of $G$ with respect to that direction. We construct this antiderivative by repeated integration along lines connecting different points in the domain of $G$; the convexity of $\ccal$ guarantees that these lines are in the domain of $G$, and so that we can construct such antiderivatives.

Due to \eqref{eq:gpoly}, we know that by integrating repeatedly with sufficiently large $m$, we can remove the divergences of $G,\nabla G$, and obtain that both $G^{(-m)}$ and $\nabla G^{(-m)}$ are continuous on $\kcal + i (\bar\ccal \cap \bar\ucal)$, including the points where $\im \zv = 0$. (see for example~\cite[Thm.~IX.16]{ReedSimon:1975-2} for details on this technique.)

We compute
\begin{equation}
\int F(\xv + i \epsilon \bv^\pm) g(\xv) d\xv = \int \frac{G(\xv + i \epsilon \bv^\pm)}{x_1 + i \epsilon b^\pm_1 }  g (\xv) \, d\xv,
\end{equation}
where we used that $G(\zv):=(\zv\cdot\av)F(\zv)$, $\zv= \xv +i\epsilon \bv^\pm$, $\av=\ev^{(1)}$.

Now we use integration by parts. Consider that, since $\bv^\bot \cdot\av=0$, $x_1$ and $\partial_\bot$ refer to mutually orthogonal directions, and therefore the derivative $\partial_\bot$ does not apply to the denominator $x_1 + i \epsilon b^\pm_1$. Note also that the boundary term is zero because of the support properties of $g$. Hence, we have
\begin{equation}
   \int F(\xv + i \epsilon \bv^\pm) g(\xv) d\xv =
   (-1)^m \int \frac{G^{(-m)}(\xv + i \epsilon \bv^\pm)}{x_1 + i \epsilon b^\pm_1 }  \partial_\bot^m g (\xv) \, d\xv.
\end{equation}
Now we can apply all the previous analysis to $G^{(-m)}$, $ \partial_\bot^m g$ in place of $G,g$, until equation \eqref{fdifffinale}. This gives:
\begin{equation}\label{derivantiderivGg}
 \int \Big( F(\xv+i0 \bv^-)-F(\xv+i0 \bv^+) \Big) g(\xv)d\xv = (-1)^m 2 \pi i \int G^{(-m)}(0,\hat\xv) \partial_\bot^m g(0,\hat\xv)\, d\hat\xv.
\end{equation}
We would like now to get rid of the derivatives and antiderivatives of $G,g$ using integration by parts once more; but we have an obstruction: Since $G$ diverges at $\im \zv = 0$ , the function $G^{(-m)}$ is not $m$-times differentiable at $\im \zv = 0$.

To solve this problem, we note that since $G^{(-m)}$ is continuous on $\kcal + i (\bar\ccal \cap \bar\ucal)$, we can write $G^{(-m)}(0,\hat\xv) = \lim_{\epsilon \to 0} G^{(-m)}(0,\hat\xv + i \epsilon \hat\bv\vphantom{\bv}^\bot)$ (where on the left hand side $G^{(-m)}$ is evaluated on the boundary of the cone, in particular at $\im \zv = 0$, and on the right hand side it is evaluated in the interior of the cone). Inserting this into \eqref{derivantiderivGg}, and bringing the limit $\epsilon \to 0$ outside the integral sign (the dominated convergence theorem enters here), we find
\begin{multline}
 \int \Big( F(\xv+i0 \bv^-)-F(\xv+i0 \bv^+) \Big) g(\xv)d\xv\\
  = (-1)^m 2 \pi i \lim_{\epsilon \to 0}\int G^{(-m)}(0,\hat\xv + i \epsilon \hat\bv\vphantom{\bv}^\bot) \partial_\bot^m g(0,\hat\xv)\, d\hat\xv.
\end{multline}
We notice that $G^{(-m)}(0,\hat\xv + i \epsilon \hat\bv\vphantom{\bv}^\bot)$ is $m$-times differentiable in $(0,\hat\xv)\in \kcal$ (since we added a non-zero imaginary part to the argument of $G^{(-m)}$). Now, we can integrate by parts, and find
\begin{equation}
\int \Big( F(\xv+i0 \bv^-)-F(\xv+i0 \bv^+) \Big) g(\xv)d\xv = 2 \pi i \lim_{\epsilon \to 0}\int G(0,\hat\xv + i \epsilon \hat\bv\vphantom{\bv}^\bot) g(0,\hat\xv)\, d\hat\xv.
\end{equation}
Using that $b_1^\bot=\bv^\bot \cdot\av=0$, we find
\begin{equation}
\int \Big( F(\xv+i0 \bv^-)-F(\xv+i0 \bv^+) \Big) g(\xv)d\xv =2 \pi i \lim_{\epsilon \to 0}\int \delta(x_1) G(\xv + i \epsilon \bv\vphantom{\bv}^\bot) g(\xv)\, d\xv.
\end{equation}
Therefore,
\begin{equation}
\int \Big( F(\xv+i0 \bv^-)-F(\xv+i0 \bv^+) \Big) g(\xv)d\xv =2 \pi i \int\delta(\xv \cdot \av) \res_{\zv \cdot \av=0}F(\xv+i 0 \bv\vphantom{\bv}^\bot )g(\xv)\, d\xv.
\end{equation}
This gives the result \eqref{eq:onepole}.
\end{proof}
Using the lemma above, we try to write a similar formula in the case of a function which has first-order poles at several distinct hyperplanes.

\begin{proposition}\label{proposition:multivarres}
  Let $\ucal \subset \rbb^k$ be a neighbourhood of zero, $\ccal \subset \rbb^k$ an open convex cone, and $\av_1,\ldots,\av_p \in \rbb^k$ pairwise different.
  Let $F$ be meromorphic on $\tube(\ucal)$ and $(\zv \cdot \av_1)\cdots (\zv \cdot \av_p)F(\zv)$ analytic on $\tube(\ccal \cap \ucal)$.
  For any $M\subset \{1,\ldots,p\}$, let $\bv_M\in\ccal$ such that $\av_j\cdot\bv_M=0$ if $j\in M$, $\av_j\cdot\bv_M>0$ if $j \not\in M$.
  Let $\cv \in \ccal$ such that $\av_j\cdot\cv <0$ for all $j$. Then it holds that
\begin{equation}\label{sevdimres}
F(\xv+i0\, \cv)=\sum_{M\subset\left\{ 1,\ldots,p\right\}}(2i\pi)^{|M|}\Big( \prod_{m\in M}\delta(\xv\cdot \av_{m})\Big) \res_{\zv\cdot \av_{m_{1}}=0}\ldots \res_{\zv\cdot \av_{m_{|M|}}=0}F(\xv+i0\,\bv_{M})
\end{equation}
with the notation $M=\{m_1,\ldots,m_{|M|}\}$.
\end{proposition}

\begin{proof}

We prove this proposition by using induction on $p$. For $p=1$, \eqref{sevdimres} reduces to
\begin{equation}
F(\xv+i0\, \cv)= F(\xv+i0\,\bv_{\emptyset})+(2i\pi)\delta(\xv\cdot \av_{1})\res_{\zv\cdot \av_{1}=0}F(\xv+i0\,\bv_{\{1\}}).
\end{equation}
This follows directly from Lemma~\ref{lemma:onepole} with $\bv^+ = \bv_{\emptyset}$, $\bv^-=\cv$, $\bv^\bot=\bv_{\{1\}}$.

Now we assume that \eqref{sevdimres} holds for $p-1$ in place of $p$.

Since $\ccal$ is convex, we have by definition of convex domain that the straight line from $\cv$ to $\bv_\emptyset$ is contained in $\ccal$. Since the functions $\av_j\cdot\zv$ fulfils $\av_j\cdot\cv <0$ on the end point $\cv$ of this line and $\av_j\cdot\bv_\emptyset$ on the other end point $\bv_\emptyset$ of the line, and they are continuous between these two points, then they have at least one zero ($\av_j\cdot\zv=0$) along the line; we choose the $\av_j$ which gives the first zero and we rename it $\av_1$. If we have several zeros which are attained at the same time, then we can perform a small deformation of the path within the open set $\ccal$ so that the zeros are isolated along the path. We have that the function $(\zv\cdot\av_1) F(\zv) $ is analytic within the tube over the cone $\ccal^- := \{ \xv \in \ccal : \xv\cdot\av_j < 0 \text{ for } j \geq 2 \}$. (It is analytic because $(\zv\cdot\av_j)^{-1}$, $j\geq 2$, is analytic on that domain.) We also note that $\ccal^-$ is convex and open since it is the intersection of two convex open sets: $\ccal^- = \ccal \cap \{ \xv \cdot \av_2<0\}\cap \ldots \cap \{ \xv \cdot \av_p <0\}$.

After we have possibly renumbered the vectors $\av_j$ (see above), we can choose $\cv'\in\ccal^-$ such that $\av_1\cdot \cv' >0$, but $\av_j\cdot \cv' <0$ for $j \geq 2$.

We apply Lemma~\ref{lemma:onepole} with $\bv^+ = \cv'$, $\bv^-=\cv$, and we find
\begin{equation}\label{eq:ccp}
 F(\xv + i  0 \cv) = F(x + i 0 \cv') + 2 \pi i \delta(\xv\cdot\av_1) \res_{\zv \cdot \av_1=0} F(x + i 0 \cv'')
\end{equation}
where we choose $\cv''\in\ccal^-$ such that $\av_1\cdot \cv''=0$. We can apply to the term $F(x + i 0 \cv')$ the induction hypothesis (namely equation \eqref{sevdimres} with $p-1$ in place of $p$) with respect to the cone $\ccal^+:=\{ \xv \in \ccal : \xv\cdot\av_1 > 0 \}$, where $(\zv \cdot \av_2)\cdots (\zv \cdot \av_p)F(\zv)$ is analytic. This gives
\begin{equation}\label{eq:fcpinduction}
F(\xv+i0 \cv')=\sum_{M\subset\left\{ 2,\ldots,p\right\}}(2i\pi)^{|M|}\Big( \prod_{m\in M}\delta(\xv\cdot \av_{m})\Big) \res_{\zv\cdot \av_{m_{1}}=0}\ldots \res_{\zv\cdot \av_{m_{|M|}}=0}F(\xv+i0\,\bv_{M}).
\end{equation}
Moreover, we have that the residue of $F$ in \eqref{eq:ccp} is a meromorphic function on the hyperplane $\zv\cdot\av_1=0$, which we can identify with $\cbb^{k-1}$; we also have that the function is analytic when we multiply it with $(\zv\cdot\av_2)\cdots(\zv\cdot\av_p)$. So, we apply the induction hypothesis (namely equation \eqref{sevdimres} with $p-1$ in place of $p$) with respect to the cone $\ccal^0:=\{ \xv \in \ccal : \xv\cdot\av_1 = 0 \}$, and we have
\begin{multline}\label{eq:fresinduction}
\res_{\zv\cdot\av_1=0} F(\xv+i0 \cv'')\\
=\sum_{M\subset\left\{ 2,\ldots,p\right\}}(2i\pi)^{|M|}\Big( \prod_{m\in M}\delta(\xv\cdot \av_{m})\Big) \res_{\zv\cdot \av_{m_{1}}=0}\ldots \res_{\zv\cdot \av_{m_{|M|}}=0}\res_{\zv\cdot \av_1=0}F(\xv+i0\,\bv_{M\cup\{1\}}).
\end{multline}
Inserting \eqref{eq:fcpinduction} and \eqref{eq:fresinduction} into \eqref{eq:ccp}, we find
\begin{multline}
F(\xv+i0\cv)=
\sum_{M\subset \{ 2,\ldots,p\}}(2\pi i)^{|M|}\Big(\prod_{m\in M}\delta(\zv\cdot \av_{m})\Big)\res_{\zv\cdot \av_{m_{1}}=0}\ldots \res_{\zv\cdot \av_{m_{|M|}}=0}F(\xv+i0\,\bv_{M})\\
+2\pi i\delta(\xv\cdot \av_{1})\sum_{M\subset \{ 2,\ldots,p\}}(2\pi i)^{|M|}\Big(\prod_{m\in M}\delta(\xv\cdot \av_{m})\Big) \times \\ \times \res_{\zv\cdot \av_{m_{1}}=0}\ldots \res_{\zv\cdot \av_{m_{|M|}}=0}\res_{\zv\cdot \av_{1}=0}F(\xv+i0\,\bv_{M \cup \{ 1\}}),\label{seclin}
\end{multline}
that we can rewrite as
\begin{multline}
F(\xv+i0\,\cv)=
\sum_{M\subset \{ 2,\ldots,p\}}(2\pi i)^{|M|}\Big(\prod_{m\in M}\delta(\xv\cdot \av_{m})\Big)\res_{\zv\cdot \av_{m_{1}}=0}\ldots \res_{\zv\cdot \av_{m_{|M|}}=0}F(\xv+i0\,\bv_{M})\\
+\sum_{M\subset \{ 2,\ldots,p\}}(2\pi i)^{|M|+1}\Big(\prod_{m\in M\cup \{ 1\}}\delta(\xv\cdot \av_{m})\Big)\res_{\zv\cdot \av_{1}=0}\res_{\zv\cdot \av_{m_{1}}=0}\ldots \res_{\zv\cdot \av_{m_{|M|}}=0}F(\xv+i0\,\bv_{M\cup\{ 1\}}).\label{1MCM}
\end{multline}
Fix a subset $M'$ of the set $\{1,\ldots,p\}$. Either $1$ is in $M'$ or $1$ is not in $M'$. If $1$ is in $M'$, then we consider the second line of \eqref{1MCM}. Set $M':=M\cup \{ 1\}$, and relabel the summation index $M$ in the second line of \eqref{1MCM} as $M\cup \{ 1\}$, then the second line of \eqref{1MCM} reads
\begin{multline}
\text{ Second line of }\eqref{1MCM}\\
=\sum_{\substack{ {M'\subset\left\{ 1,\ldots,p\right\}} \\ {M' \text{ contains } 1}}}(2i\pi)^{|M'|}\Big(\prod_{m\in M'}\delta(\xv\cdot \av_{m})\Big)\res_{\zv\cdot \av_{m_{1}}=0}\ldots \res_{\zv\cdot \av_{m_{|M'|}}=0}F(\xv+i0\,\bv_{M'}).
\end{multline}
If $1$ is not in $M'$, then we consider the first line of \eqref{1MCM}, we have
\begin{equation}
\text{ First line of }\eqref{1MCM}=\sum_{\substack{M'\subset\left\{ 1,\ldots,p\right\}\\ M' \text{ not contain } 1}} \Big(\prod_{m\in M'}\delta(\xv\cdot \av_{m})\Big)\res_{\zv\cdot \av_{m_{1}}=0}\ldots \res_{\zv\cdot \av_{m_{|M|}}=0}F(\xv+i0\,\bv_{M}).
\end{equation}
Summing these two terms, we find
\begin{equation}
F(\xv+i0\,\cv)=\sum_{M''\subset\left\{ 1,\ldots,p\right\}}\Big(\prod_{m\in M''}\delta(\xv\cdot \av_{m})\Big)\res_{\zv\cdot \av_{m_{1}}=0}\ldots \res_{\zv\cdot \av_{m_{|M''|}}=0}F(\xv+i0\,\bv_{M''}),
\end{equation}
which gives the proposed result \eqref{sevdimres}.
\end{proof}

%% file: graphs.tex
\chapter{CR functions on graphs} \label{sec:graphs}

Most of the material in this section, until Lemma~\ref{lemma:maxmodcross} included, is due to H. Bostelmann.

We call a \emph{graph} $\gcal$ in $\rbb^k$, a collection of points in $\rbb^k$, that we call the \emph{nodes}, together with a set of straight lines which connect some of these nodes, that we call the \emph{edges}.

In our case, the nodes will always be points on the lattice $\pi \zbb^k$, and the edges will always be lines between nodes which are next neighbours and parallel to the axis; this means that the edges are parametrized by $\lambdav(s) = \boldsymbol{\nu} + s \ev^{(j)}$, where $\ev^{(j)}$ is a standard basis vector of $\rbb^k$, where $0 < s < \pi$, and where $\boldsymbol{\nu}$ and $\boldsymbol{\nu} + \pi \ev^{(j)}$ are nodes of $\gcal$.
We call the \emph{tube over $\gcal$} the set of all $\zetav = \thetav + i \lambdav$ with $\thetav \in \rbb^k$ and $\lambdav$ on an edge of $\gcal$. We denote this tube by $\tube(\gcal)$.

We call a \emph{CR function $F$ on $\tube(\gcal)$} a smooth function on $\tube(\gcal)$ which is analytic along the edges; namely, if we consider an edge $\lambdav(s)$ parametrized as above, we have that $F$ is analytic in $\zeta_j$ along that edge, and it is smooth in all (real) variables.

Moreover we require that the boundary values of $F$ and also of all its derivatives exist at the nodes, $s \searrow 0$ and $s \nearrow \pi$,  and that where several edges meet in a common node, the different limits of $F$ along the different edges agree. This means that we can speak of \emph{the} boundary value at a node, without that we indicate the direction of the limit. But in the case that several graphs play a role, we will denote by $F(\nuv)\vert_\gcal$ the boundary value at node $\nuv$ which is obtained within $\gcal$.

A \emph{CR distribution $F$ on $\tube(\gcal)$}, correspondingly, is an analytic function along the edges and a $\dcal(\rbb^{k-1})'$ distribution in the remaining real variables.\footnote{%
See~\cite[Ch.~I, Appendix 2, \S{}3]{GelfandShilov:1964vol1} regarding a discussion of distributions which depends analytically on a parameter.} As above, we also require that all boundary values at nodes exist in the sense of distributions, and that they agree where several edges meet in a common node.

Now we will obtain some general properties of CR functions on $\tube(\gcal)$; these properties are mainly extensions of standard results to our framework.

First, we note that we can make CR distributions ``regular'' by convoluting them with test functions. Indeed, let $F$ be a CR distribution on $\gcal$, and let $g = (g_1,\ldots,g_k) \in \dcal(\rbb)^k$, we define
\begin{equation}\label{eq:fgconvolute}
     (F \ast g)(\zetav) := \int F(\zetav - \xiv) g_1(\xi_1) \ldots g_k(\xi_k) \, d^k\xiv.
\end{equation}
We can show that $F \ast g$ is a CR \emph{function} on $\tube(\gcal)$. For this, we need to show that $F \ast g$ is a smooth function in all (real) variables $\re \zetav$, and therefore that the following derivative exists for any $j=1,\ldots,k$:
\begin{equation}
\frac{\partial (F \ast g)(\zetav)}{\partial \re\zeta_j} =\frac{\partial}{\partial \re\zeta_j}\int F(\zetav - \xiv) g_1(\xi_1) \ldots g_k(\xi_k) \, d^k\xiv.
\end{equation}
By renaming of the variables, we can shift $\re\zetav$ into the argument of $g$ ($\zetav= \thetav +i\lambdav$).
\begin{equation}\label{renameFstarg}
\frac{\partial (F \ast g)(\zetav)}{\partial \theta_j} =\frac{\partial}{\partial \theta_j}\int F(\xiv' +i\lambdav) g_1(\theta_1 - \xi'_{1}) \ldots g_k(\theta_k-\xi'_{k}) \, d^k\xiv'.
\end{equation}
Since a distribution is a continuous functional in the test functions $g_j$ in the $\dcal(\rbb)$ topology, we can move the derivative to the test function $g$, and note that the functions $g$ are differentiable.
\begin{equation}
\frac{\partial (F \ast g)(\zetav)}{\partial \theta_j} =-\int F(\xiv' +i\lambdav) g_1(\theta_1 - \xi'_{1}) \ldots \frac{\partial}{\partial \theta_j}g_j(\theta_j-\xi'_{j}) \ldots g_k(\theta_k-\xi'_{k}) \, d^k\xiv'.
\end{equation}
To prove that $F \ast g$ is a CR function on $\tube(\gcal)$, we also need to show that $F \ast g$ is an analytic function on the edges of the graph $\gcal$, but this follows directly from the fact that $F$ is analytic along those edges.

Moreover, we have that CR distributions fulfils a version of the tube theorem (see \cite{Bochner:1938}, \cite{BochnerMartin:1938}), namely they can be extended to the convex hull of the graph. In order to formulate this theorem, we introduce further notation: we denote with $\bar \gcal \subset \rbb^k$ the closure of the edges of $\gcal$, namely it is the edges together with the nodes, as a subset of $\rbb^k$.  We call $\cch \gcal$ the closed convex hull of $\bar\gcal$, we denote with $\ich \gcal$ the interior of the convex hull, and following a notation used in \cite{Kazlow:1979}, we define the \emph{almost convex hull} $\ach \gcal := (\ich\gcal) \cup \bar\gcal$.

\begin{lemma} \label{lem:graphtube}
  Let $\gcal$ be a connected graph and $F$ a CR distribution on $\tube(\gcal)$. Then, $F$ extends to an analytic function on $\ich{\gcal}$ with distributional boundary values on $\ach \gcal$.
\end{lemma}

\begin{proof}
To prove this lemma, we use results that you can find in \cite{Kazlow:1979}. In his framework, $\bar \gcal$ is a connected, locally closed, locally starlike set. We have that the convolution $F \ast g$ is a CR function on $\tube(\gcal)$, for any $g \in \dcal(\rbb)^k$ (we have shown this after Eq.~\eqref{eq:fgconvolute}).

Even more, since $F \ast g$ is an analytic function along the edges on $\tube(\gcal)$, and it is continuous on the nodes, then by Morera's theorem, we have that at nodes where two edges along the same axis meet, $F \ast g$ analytically continues in the respective variable across the node.

Hence, we have that at any node the function $F \ast g$ is a smooth function; but due to the several directions where the edges can meet in a common node, the derivatives possibly exist only as single-sided derivatives. This $F \ast g$ is defined on lines in at most $k$ independent directions, and therefore we have at most $k$ independent first order partial derivatives. This makes possible to explicitly construct a smooth extension of this function on a neighbourhood of the node. A possible way is the following: Suppose that $f_1(x_1)$ and $f_2(x_2)$ are two smooth functions on the two real axis of $\rbb^2$, and $0$ is the node ; assume, without loss of generality, that $f_1(0)=f_2(0)=1$. Then we can construct a smooth extension on the neighbourhood of zero by setting $f(x_1,x_2):= f_1(x_1)\cdot f_2(x_2)$.
Since we can construct a smooth extension of $F \ast g$ on an open set of the graph, containing the node, then $F \ast g$ is a smooth function in the sense of Whitney (this can be seen from the definition of smooth function in the sense of Whitney in the remark after Def.~1.3 in \cite{Kazlow:1979}).

Now since $F \ast g$ is smooth in the sense of Whitney on $\tube(\bar \gcal)$ and it is a CR function on each open line segment in $\bar\gcal$, then by \cite[Def.~2.12]{Kazlow:1979}, $F \ast g$ is a CR' function on $\tube(\bar \gcal)$.

Using that $F \ast g$ is a CR' function on $\tube(\bar \gcal)$ (which is the same as CR function, see \cite[Proposition~2.13]{Kazlow:1979}), then we can apply \cite[Theorem~6.1]{Kazlow:1979}, which shows that there is a bijection between CR functions on $\tube(\ach \bar\gcal)$ and CR functions on $\tube(\bar \gcal)$; this gives an extension $G$ of $F \ast g$ to $\tube(\ach \gcal)$, and again by \cite[Theorem~6.1]{Kazlow:1979}, we have that this extension is analytic in $\tube(\ich\gcal)$.

It remains only to show that we can write $G = F \ast g$, where $F$ is some function analytic in $\tube(\ich\gcal)$. To show this we will follow an argument that can be found in more details and in a similar situation for example in \cite[p.~530]{Eps:edge_of_wedge}, and that here we will sketch only briefly.

We know that $G$ depends on $g$; due to the remark in \cite[Sec.~12]{Kazlow:1979} (see page 170 bottom), we have also that at each fixed $\lambdav\in\ich(\gcal)$, the map $g \mapsto G(i\lambdav)$ is continuous in $g$ in the $\dcal$-topology.

If now we smear $G$ in $\lambdav$ within $\ich(\gcal)$ with another test function, we obtain a distribution in $2k$ variables which, due to the analyticity of $G$, fulfils the Cauchy-Riemann equations in the sense of distributions (cf.~\cite[p.~530]{Eps:edge_of_wedge}).

Now, we can apply \cite[p.72]{Schwartz:1959b}, which shows that a distribution which fulfils the Cauchy-Riemann equations in a weak sense, it also fulfils the Cauchy-Riemann equations in the strong sense, namely it is an analytic function smeared with test functions. This gives the desired function $F$.
\end{proof}

If we assume that $F$ is a CR function on $\tube(\gcal)$, then the maximum modulus principle holds also for CR functions on $\tube(\gcal)$ due to \cite[Sec.~11, Corollary]{Kazlow:1979}:
\begin{equation}\label{eq:maxmodulusedge}
   \sup_{\zetav \in \tube (\ach \gcal) } |F(\zetav)|
 =   \sup_{\zetav \in \tube(\bar \gcal) } |F(\zetav)|.
\end{equation}
Now suppose also that the function $F$ fulfils on each edge bounds of this type:
\begin{equation}\label{eq:plbounds}
   \log |F(\thetav + i \nuv + i \lambda \ev^{(j)}) |  = o(\cosh \theta_j)
\end{equation}
for large $|\theta_j|$, uniformly in $\lambda$, and with $\theta_m$ ($m \neq j$) fixed.

Then we can apply the Phragm\'en-Lindel\"of argument given in \cite[Theorem~3]{HardyRogosinski:1946}, which states that the function attains the maximum modulus on the boundary of the edge, and therefore on one of the nodes which are connected by the edge. Hence, we can rewrite \eqref{eq:maxmodulusedge} as
\begin{equation}\label{eq:maxmodulusnode}
   \sup_{\zetav \in \tube (\ach \gcal) } |F(\zetav)|
 =   \sup_{\zetav \in \tube(\nodes(\gcal)) } |F(\zetav)|.
\end{equation}
Now we want to obtain a similar maximum modulus principle as above for the norm $\gnorm{\cdotarg}{\times}$, which was defined in Eq.~\eqref{eq:fullcrossnorm}. This is done in the following lemma.

\begin{lemma}\label{lemma:maxmodcross}
If $F$ is a CR distribution on $\tube(\gcal)$, then
\begin{equation}\label{eq:maxmodcrossedge}
    \sup_{\lambdav \in \ach \gcal } \gnorm{ F(\cdot + i \lambdav ) }{\times}
 =  \sup_{\lambdav \in \bar\gcal }\gnorm{ F(\cdot + i \lambdav ) }{\times}.
\end{equation}
If further $F \ast g$ fulfils bounds of the form \eqref{eq:plbounds} for each fixed $g \in \dcal(\rbb)^k$, then
\begin{equation}\label{eq:maxmodcrossnode}
    \sup_{\lambdav \in \ach \gcal } \gnorm{ F(\cdot + i \lambdav ) }{\times}
 =  \sup_{\lambdav \in \nodes(\gcal) }\gnorm{ F(\cdot + i \lambdav ) }{\times}.
\end{equation}
\end{lemma}

\begin{proof}
We consider $g=(g_1, \ldots, g_k)$ with $\gnorm{g_j}{2}\leq 1$, we consider $F \ast g(\zetav)$ as defined in \eqref{eq:fgconvolute}. We have shown in the proof of Lemma~\ref{lem:graphtube} that this function is analytic on $\tube(\ich \gcal )$ and it is a CR function on $\tube(\gcal)$. We can also show that
\begin{equation}\label{eq:fgsup}
   \gnorm{F(\cdot + i \lambdav)}{\times} = \sup_g |F \ast g( i \lambdav)|
   = \sup_{g,\thetav} |F \ast g( \thetav + i \lambdav)|.
\end{equation}
The first equality can be proved by using the definitions \eqref{eq:fgconvolute} and \eqref{eq:fullcrossnorm}, and by performing a renaming of the variables ($\thetav :=-\xiv$):
\begin{multline}
\sup_g |F \ast g( i \lambdav)|=\sup_g \Big| \int F(i\lambdav - \xiv) g_1(\xi_1) \ldots g_k(\xi_k) \, d^k\xiv \Big|\\
=\sup_g \Big| \int F(\thetav+i\lambdav) g_1(-\theta_1) \ldots g_k(-\theta_k) \, d^k\thetav \Big|.
\end{multline}
Calling $g'_{j}(\theta):=g_{j}(-\theta)$, we find
\begin{equation}
\sup_g |F \ast g( i \lambdav)|=\sup_{g'} \Big| \int F(\thetav+i\lambdav) {g'}_1(\theta_1) \ldots {g'}_k(\theta_k) \, d^k\thetav \Big|=\gnorm{F(\cdot + i \lambdav)}{\times}.
\end{equation}
Similarly, we can prove the second equality in \eqref{eq:fgsup} using a renaming of variables ($\thetav':=\thetav - \xiv$):
\begin{multline}
\sup_{g,\thetav} |F \ast g( \thetav + i \lambdav)|=\sup_{g,\thetav} \Big| \int F(\thetav +i\lambdav - \xiv) g_1(\xi_1) \ldots g_k(\xi_k) \, d^k\xiv \Big|\\
=\sup_{g,\thetav} \Big| \int F(\thetav' +i\lambdav) g_1(\theta_1-\theta'_{1}) \ldots g_k(\theta_k-\theta'_{k}) \, d^k\thetav' \Big|.
\end{multline}
Calling $g'_{j}(\theta'_{j}):=g_{j}(\theta_j-\theta'_{j})$, we have
\begin{multline}
\sup_{g,\thetav} |F \ast g( \thetav + i \lambdav)|=\sup_{g',\thetav} \Big| \int F(\thetav' +i\lambdav) g'_1(\theta'_{1}) \ldots g_k(\theta'_{k}) \, d^k\thetav' \Big|\\
=\sup_{g'} \Big| \int F(\thetav' +i\lambdav) g'_1(\theta'_{1}) \ldots g_k(\theta'_{k}) \, d^k\thetav' \Big|=\gnorm{F(\cdot + i \lambdav)}{\times},
\end{multline}
where in the first equality we have used that the supremum over $g$ equals the supremum over $g'$, and in the second equality we used that, after taking the supremum over all $g'$, the integral expression does not depend on $\thetav$ any more, so we can drop the supremum over $\thetav$.

Applying the maximum modulus principle \eqref{eq:maxmodulusedge} to $F \ast g$ and using the relation \eqref{eq:fgsup}, we can find \eqref{eq:maxmodcrossedge}: Indeed, we have
\begin{equation}
 \sup_{\zetav \in \tube (\ach \gcal) } |(F\ast g)(\zetav)|
=   \sup_{\zetav \in \tube(\bar \gcal) } |(F\ast g)(\zetav)|,
\end{equation}
that we can rewrite as:
\begin{equation}
 \sup_{\thetav}\sup_{ \lambdav \in \ach \gcal } |(F\ast g)(\zetav)|
=   \sup_{\thetav} \sup_{\lambdav \in \bar \gcal } |(F\ast g)(\zetav)|.
\end{equation}
Taking the supremum over $g$ of both sides of the equation, we find
\begin{equation}
\sup_{g}\, \sup_{\thetav}\sup_{ \lambdav \in \ach \gcal } |(F\ast g)(\zetav)|
=  \sup_{g}\, \sup_{\thetav} \sup_{\lambdav \in \bar \gcal } |(F\ast g)(\zetav)|.
\end{equation}
Using \eqref{eq:fgsup}, we find from the equation above our result \eqref{eq:maxmodcrossedge}.

Moreover, if $F \ast g$ fulfils bounds of the form \eqref{eq:plbounds}, then we can apply \eqref{eq:maxmodulusnode} to $F \ast g$; then using the relation \eqref{eq:fgsup}, we find \eqref{eq:maxmodcrossnode}. The details of this argument are as follows: By \eqref{eq:maxmodulusnode}, we have
\begin{equation}
\sup_{\zetav \in \tube (\ach \gcal) } |(F\ast g)(\zetav)|
 =   \sup_{\zetav \in \tube(\nodes(\gcal)) } |(F\ast g)(\zetav)|,
\end{equation}
that we can rewrite as
\begin{equation}
 \sup_{\thetav}\sup_{ \lambdav \in \ach \gcal } |(F\ast g)(\zetav)|
=   \sup_{\thetav} \sup_{\lambdav \in \nodes(\gcal) } |(F\ast g)(\zetav)|.
\end{equation}
Taking the supremum over $g$ of both sides of the equation, we have
\begin{equation}
\sup_{g} \,\sup_{\thetav}\sup_{ \lambdav \in \ach \gcal } |(F\ast g)(\zetav)|
=  \sup_{g} \,\sup_{\thetav} \sup_{\lambdav \in \nodes(\gcal)} |(F\ast g)(\zetav)|.
\end{equation}
Using \eqref{eq:fgsup}, we find from the equation above our result \eqref{eq:maxmodcrossnode}.
\end{proof}

Now we prove a result regarding pointwise bounds on analytic functions, which are estimated using the supremum of the norm $\gnorm{\cdotarg}{\times}$ of these functions. These bounds do not hold only for CR functions on graphs, but they are valid for any analytic functions defined on tube domains, namely are useful in conjunction with Lemma~\ref{lemma:maxmodcross}. The following proposition can be proved by using the mean value property; one can find a similar application of this technique for the computation of certain uniform bounds for example in~\cite[Prop.~4.4]{Lechner:2008}.

\begin{proposition}\label{proposition:pointwise}
Let $\ical\subset \rbb^k$ be open and $F$ analytic on $\tube(\ical)$. Then, for all $\zetav\in \tube(\ical)$,
\begin{equation}\label{pointwiseboundF}
|F(\zetav)|\leq \frac{(4/\pi)^{k} \; k^{k/4}}{\operatorname{dist}(\im\zetav,\partial \ical)^{k/2}}\sup_{\lambdav \in \ical} \gnorm{F(\cdotarg +i\lambdav)}{\times}.
\end{equation}
\end{proposition}

\begin{figure}
\begin{center}
\input{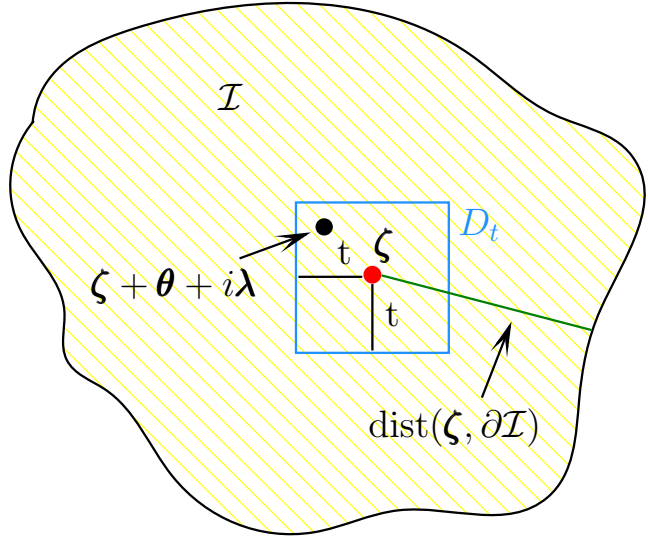}
\caption{The computation of the pointwise bound \eqref{pointwiseboundF}} \label{fig:mean}
\end{center}
\end{figure}

\begin{proof}
We fix $\zetav\in\tube(\ical)$, we consider $D_{t}\subset \mathbb{C}$ to be the disc around the origin with radius $t:=\frac{1}{2}k^{-1/2}\operatorname{dist}(\im\zetav,\partial \mathcal{I})$, cf.~Fig.~\ref{fig:mean}. Then, for this value of the radius, we have that the polydisc $(D_{t}\times\cdots \times D_{t})+\zetav$ is contained in $\tube(\ical)$. We can apply the mean value property for analytic functions and we get
\begin{equation}\label{eq:meanval}
\begin{aligned}
F(\zetav) &=
(\pi t^{2})^{-k}\int_{D_{t}}d\theta_{1}d\lambda_{1}\ldots \int_{D_{t}}d\theta_{k}d\lambda_{k}\;F(\zetav+\thetav+i\lambdav)\\
&=(\pi t^{2})^{-k}\int_{[-t,t]^{\times k}}d^k\lambdav \int_{-(t^2 - \lambda^2_j)^{1/2}}^{(t^2-\lambda_j^2)^{1/2}}d^k\thetav\; F(\zetav + \thetav +i\lambdav)\\
&=(\pi t^2)^{-k}\int_{[-t,t]^{\times k}}d^k\lambdav \int_{-\infty}^{+\infty}d^k\thetav \; F(\zetav +\thetav +i\lambdav)\chi_{\lambdav}(\thetav)
\\
&= (\pi t^{2})^{-k}\int_{[-t,t]^{\times k}}d^{k}\lambdav\; \big( F \ast \chi_{\lambdav} \big) (\zetav+i\lambdav),
\end{aligned}
\end{equation}
where $\chi_{\lambdav}(\thetav) = \prod_{j=1}^k \chi_j(\theta_j)$, and where we denote with $\chi_j$ the characteristic function of the interval $[-(t^2\!-\!\lambda_j^2)^{1/2}, +(t^2\!-\!\lambda_j^2)^{1/2}]$.
Since we have by construction that $\zetav+i\lambdav \in \tube(\ical)$, we can estimate
\begin{multline}
| (F \ast \chi_{\lambdav})(\zetav+i\lambdav) |
=\Big\lvert \int F(\zetav +i\lambdav- \xiv) \chi_{1}(\xi_1) \ldots \chi_k(\xi_k) \, d^k\xiv \Big\rvert
\\
=\Big\lvert \int F(\thetav- \xiv +i\lambdav+i\lambdav'') \chi_{1}(\xi_1) \ldots \chi_k(\xi_k) \, d^k\xiv \Big\rvert
\\
=\Big\lvert \int F(\thetav' +i\lambdav') \chi_{1}(\theta_1-\theta'_{1}) \ldots \chi_k(\theta_k-\theta'_{k}) \, d^k\thetav' \Big\rvert,
\end{multline}
where in the first equality we used \eqref{eq:fgconvolute}, where in the second equality we denoted $\zetav= \thetav +i\lambdav''$ and where in the third equality we called $\lambdav':=\lambdav + \lambdav''$, $\thetav' :=\thetav - \xiv$.

Now we apply the definition of the norm $\gnorm{\cdotarg}{\times}$ given by Eq.~\eqref{eq:fullcrossnorm}, we find
\begin{multline}
\Big\lvert \int F(\thetav' +i\lambdav') \chi_{1}(\theta_1-\theta'_{1}) \ldots \chi_k(\theta_k-\theta'_{k}) \, d^k\thetav' \Big\rvert\\
\leq  \gnorm{F(\cdotarg +i\lambdav')}{\times} \cdot \prod_{j=1}^{k} \gnorm{\chi_j(\theta_j-\cdotarg)}{2}\\
\leq \sup_{\lambdav' \in \ical} \gnorm{F(\cdotarg +i\lambdav')}{\times} \cdot \prod_{j=1}^{k} \gnorm{\chi_j}{2},
\end{multline}
where in the first inequality the estimate holds for some fixed $\lambdav'$ depending on $\lambdav$. (Note that the relation above can be continued to $L^2$ functions $g_j$ by continuity.)

Taking into account $\gnorm{\chi_j}{2} \leq \sqrt{2t}$, we find from \eqref{eq:meanval} and from the equation above,
\begin{equation}
|F(\zetav)| \leq (\pi t^{2})^{-k}  (2t)^{k} (2t)^{k/2}
\sup_{\lambdav' \in \ical} \gnorm{F(\cdotarg +i\lambdav')}{\times}.
\end{equation}
Inserting the definition of $t$ in this equation, we find \eqref{pointwiseboundF}.
\end{proof}

%% file: acknowledgements.tex
\clearpage
\markboth{}{}
\section*{Danksagung}
\addcontentsline{toc}{chapter}{Danksagung}

Ich danke Prof. Rehren f\"ur die Einf\"uhrung in das faszinierende Thema dieser Doktorarbeit und f\"ur seine Hilfe bei ihrer Entstehung. Ich danke H. Bostelmann daf\"ur, dass er das Programm konkret gemacht hat, und f\"ur die gemeinsame Arbeit: Er war f\"ur mich einer der besten Lehrer, den man in diesem Gebiet finden kann.

Ich danke meinen zwei besten Freunden Lishia Teh und Yoh Tanimoto, die f\"ur mich in G\"ottingen wie eine Familie waren. Meiner Familie danke ich f\"ur ihre Unterst\"utzung und im Besonderen meiner Mutter, die mich oft in G\"ottingen besucht hat.

Des Weiteren m\"ochte ich dankend die zahlreichen Diskussionen mit Kollegen hervorheben, von denen ich sehr profitiert habe (Entschuldigung, falls ich hier jemanden vergessen haben sollte.): Yoh Tanimoto, Wojcieh Dybalski, Marcel Bischoff, Ko Sanders, Christian Koehler, Jacques Bros, Gandalf Lechner und Detlev Buchholz. Ich danke Max Gulde, der mir mit Enthusiasmus und Motivation deutsch beigebracht hat. Weiterhin danke ich all denjenigen Freunden und Kollegen, die meinen Aufenthalt in G\"ottingen versch\"onert haben: Antonia Kukhtina, Giovanni Marelli, Christoph Solveen, Lena Wallenhorst, Gennaro Tedesco, Riccardo Catena, Federico Dradi, Navdeep Sidhu und vielen anderen. Ich danke Valter Moretti, der mir vorschlug, nach G\"ottingen zu gehen.

Ich bedanke mich f\"ur die Gastfreundschaft w\"ahrend meiner Besuche an der Universit\"at York und am Erwin Schr\"odinger Institut in Wien, wo diese Arbeit zu Teilen entstanden ist.

Diese Arbeit wurde unterst\"utzt durch die Deutsche Forschungsgemeinschaft (DFG) im Rahmen des Zukunftskonzepts der Georg-August-Universit\"at G\"ottingen , durch das Courant-Forschungs\-zentrum ``Strukturen h\"oherer Ordnung in der Mathematik'', das GRK 1493 ``Mathematische Strukturen in der modernen Quantenphysik'' und das Institut f\"ur Theoretische Physik der Universit\"at G\"ottingen.